\documentclass[10pt]{article}
\usepackage[margin=.8in,left=.8in]{geometry}

\usepackage{amsthm}
\usepackage{amssymb}
\usepackage{amsmath}
\usepackage{mathtools}
\usepackage{adjustbox}

\usepackage[table]{xcolor}

\usepackage{stmaryrd}    % for \sslash
\usepackage{xfrac}
\usepackage{enumitem}\setlist{nosep}

\usepackage[safe]{tipa} % for \esh

\usepackage{hhline}

\usepackage{tensor}

\usepackage{tikz}
\usetikzlibrary{
  backgrounds, 
  decorations.pathmorphing, 
  decorations.markings,
  decorations.text
}
\usetikzlibrary{cd}

\usepackage{mathrsfs} %mathscr
\DeclareMathAlphabet{\mathpzc}{OT1}{pzc}{m}{it} %for math script

\usepackage{hyperref}
\hypersetup{
  colorlinks = true,
  linkcolor= darkblue, citecolor=darkblue
}

\hypersetup{
        colorlinks=true,
        linkcolor=darkblue,
        filecolor=darkblue,      
        urlcolor=black,
        citecolor=darkblue,
        linktoc=page
    }

\usepackage{cleveref}

\definecolor{darkblue}{rgb}{0.05,0.25,0.65}
\definecolor{darkgreen}{RGB}{20,140,10}
\definecolor{lightgray}{rgb}{0.9,0.9,0.9}
\definecolor{darkorange}{RGB}{200,100,5}
\definecolor{darkyellow}{rgb}{.91,.91,0}
\definecolor{orangeii}{RGB}{200,100,5}
\definecolor{lightblue}{RGB}{243, 250, 255}

\newtheorem{theorem}{Theorem}[section]

\newtheorem{lemma}[theorem]{Lemma}

\newtheorem{proposition}[theorem]{Proposition}

\theoremstyle{definition}
\newtheorem{definition}[theorem]{Definition}

\newtheorem{example}[theorem]{Example}

\newtheorem{remark}[theorem]{Remark}

\newcommand{\PoincareForm}{\Omega}

\newcommand{\pr}{\scalebox{.6}{$\prime$}}

\newcommand{\IIA}{\mathrm{II}\mathfrak{A}}

\newcommand{\FDGCA}{\mathbb{R}_{\mathrm{d}}}

\usepackage{accents}
\newlength{\dhatheight}

\let\PLAINthebibliography\thebibliography
\renewcommand\thebibliography[1]{
  \PLAINthebibliography{#1}
  \setlength{\parskip}{0.5pt}
  \setlength{\itemsep}{0.5pt plus .3ex}
}

\newcommand{\proofstep}[1]{\scalebox{.85}{#1}}

\newcommand{\yields}{\Rightarrow}

\newcommand{\DFSpinor}{\phi}
\newcommand{\vectorSpinor}{\psi}
\newcommand{\tensorSpinor}{\psi}

\newcommand{\Brn}{\mathfrak{Brn}}

\newcommand{\ten}{1\!0}

\newcommand{\paramOne}{\delta}
\newcommand{\paramTwo}{\gamma_1}
\newcommand{\paramFive}{\gamma_2}

\newcommand{\paramOnePrime}{\paramOne'}
\newcommand{\paramTwoPrime}{\paramTwo'}

\newcommand{\paramOneDoublePrime}{\paramOne''}
\newcommand{\paramTwoDoublePrime}{\paramTwo''}
\newcommand{\paramFiveDoublePrime}{\paramFive''}

\newcommand{\alphaZero}{\alpha_0}
\newcommand{\alphaOne}{\alpha_1}
\newcommand{\alphaTwo}{\alpha_2}
\newcommand{\alphaThree}{\alpha_3}
\newcommand{\alphaFour}{\alpha_4}

\newcommand{\betaOne}{\beta_1}
\newcommand{\betaTwo}{\beta_2}
\newcommand{\betaThree}{\beta_3}

\newcommand{\betapOne}{\beta'_1}

\newcommand{\ZTwo}{\mathbb{Z}_2}

\newcommand{\defneq}{\equiv}

\newcommand\bos[1]{\mathstrut\mkern2.5mu#1\mkern-14mu\raise1.7ex%
  \hbox{$\scriptstyle\rightsquigarrow$}}

\newcommand\bosonic[1]{\mathstrut\mkern2.5mu#1\mkern-14mu\raise1.7ex%
  \hbox{$\scriptstyle\rightsquigarrow$}}

\newcommand\longbosonic[1]{\mathstrut\mkern2.5mu#1\mkern-20mu\raise1.7ex%
  \hbox{$\scriptstyle\rightsquigarrow$}}

\usepackage[cmtip,all]{xy}
\newcommand{\longsquiggly}{\xymatrix{{}\ar@{~>}[r]&{}}}

\usepackage{bbm} %GG: Needed for macros from SmoothSets article below 

\usepackage[new]{old-arrows}   %For \longhookrightarrow

\newcommand{\FR}{\mathbb{R}}

\newcommand{\dd}{\mathrm{d}}

\newcommand{\Todd}{T_{\!\scalebox{.55}{odd}}}
\newcommand{\bosonicTodd}{{\bosonic T}_{\!\!\scalebox{.55}{odd}}}

\newcommand{\grayunderbrace}[2]{\color{gray}\underbrace{\color{black}#1}_{\color{gray}#2}\color{black}}

\begin{document}

%%%%%%%%%%%%%%%%%%%%%%%%%%%%%%%%%%%%%%%%%%%%%%
%vertical spacing around displayed equations %
\setlength{\abovedisplayskip}{3pt}
\setlength{\belowdisplayskip}{3pt}
\setlength{\abovedisplayshortskip}{-4pt}
\setlength{\belowdisplayshortskip}{3pt}
%%%%%%%%%%%%%%%%%%%%%%%%%%%%%%%%%%%%%%%%%%%%%%

%%%%%%%%%%%%%%%%%%%%%%%%%%%%%%%%%%%%%%%%%%%%%%%%%%%%%
\title{The Hidden M-Group}
%%%%%%%%%%%%%%%%%%%%%%%%%%%%%%%%%%%%%%%%%%%%%%%%%%%%

\author{
  Grigorios Giotopoulos${}^{\ast}$,
  \;\;
  Hisham Sati${}^{\ast \dagger}$,
  \;\;
  Urs Schreiber${}^{\ast}$
}

\maketitle

\begin{abstract}
Following arguments that the (hidden) M-algebra 
%-- both in its basic version and via its ``hidden'' extension -- 
serves as the maximal super-exceptional tangent space for 11D supergravity, we make explicit here its integration to a (super-Lie) {\it group}. This is equipped with a left-invariant extension of the ``decomposed'' M-theory 3-form,
such that it constitutes the Kleinian space on which super-exceptional spacetimes are to be locally modeled as Cartan geometries. 

\smallskip 
As a simple but consequential application, we highlight how to describe lattice subgroups $\mathbb{Z}^{k \leq 528}$ of the hidden M-group that allow to toroidially compactify also the ``hidden'' dimensions of a super-exceptional spacetime, akin to the familiar situation in topological T-duality.

\smallskip 
In order to deal with subtleties in these constructions, we (i) provide a computer-checked re-derivation of the ``decomposed'' M-theory 3-form, and (ii) present a streamlined conception of super-Lie groups, that is both rigorous while still close to physics intuition and practice.

\smallskip 
% Thereby this article is largely a combined review of some modernized super-Lie theory along the example of the hidden M-algebra, with an eye towards laying foundations for super-exceptional geometry. 

Thereby this article highlights modernized super-Lie theory along the example of the hidden M-algebra, with an eye towards laying foundations for super-exceptional geometry. 
Among new observations is the dimensional reduction of the hidden M-algebra to a ``hidden IIA-algebra'' which in a companion article \cite{GSS25-TDuality} we explain as the exceptional extension of the T-duality doubled super-spacetime. 
\end{abstract}

\vspace{.8cm}

\begin{center}
\begin{minipage}{9.5cm}
  \tableofcontents
\end{minipage}
\end{center}

\medskip

\vfill

\hrule
\vspace{5pt}

{
\footnotesize
\noindent
\def\arraystretch{1}
\tabcolsep=0pt
\begin{tabular}{ll}
${}^*$\,
&
Mathematics, Division of Science; and
\\
&
Center for Quantum and Topological Systems,
\\
&
NYUAD Research Institute,
\\
&
New York University Abu Dhabi, UAE.  
\end{tabular}
\hfill
\adjustbox{raise=-15pt}{
\href{https://ncatlab.org/nlab/show/Center+for+Quantum+and+Topological+Systems}{\includegraphics[width=3cm]{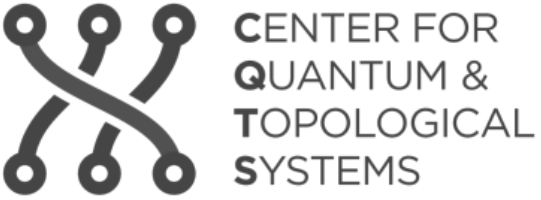}}
}

\vspace{1mm} 
\noindent ${}^\dagger$The Courant Institute for Mathematical Sciences, NYU, NY.

\vspace{.2cm}

\noindent
The authors acknowledge the support by {\it Tamkeen} under the 
{\it NYU Abu Dhabi Research Institute grant} {\tt CG008}.
}

\newpage 

%%%%%%%%%%%%%%%%%%%%%%%%%%%%%%
\section{Introduction}
\label{Introduction}
%%%%%%%%%%%%%%%%%%%%%%%%%%%%%%

\noindent
{\bf The problem of formulating M-theory} (cf. \cite{Duff99-MTheory}) remains open \cite[\S 6]{Duff96}\cite{Duff20}, but considerable attention has been paid -- and convincing progress has been made -- towards its structure visible {\bf locally}, in the (infinitesimal) neighborhood of any spacetime point. This concerns (i) brane-extended super-symmetry (e.g. \cite{Townsend99}) and (ii) exceptional duality symmetry (e.g. \cite{Samtleben23}), which (iii) may be argued \cite[\S 4]{West03}\cite{GSS24-E11} to be neatly unified, locally, via the maximal {\it super-exceptional tangent space} to be identified with the ``M-algebra'' (recalled in a moment) equipped with its ``brane-rotating automorphy'' via the $\mathfrak{sl}_{32}$-quotient of the local Lorentzian form $\mathfrak{k}_{1,10}$ of $\mathfrak{e}_{11}$-symmetry.

\vspace{-.5cm}
$$
  \begin{tikzcd}[
    row sep=10pt, 
    column sep=40pt
  ]
    &
    \overbrace{
      \mathbb{R}^{1,10\,+\,517\,\vert\, \mathbf{32}}
    }^{ 
      \mathfrak{M} 
    }
    \ar[
      dl, 
      ->>,
      shorten=-1.2pt
    ]
    \ar[
      dr, 
      ->>,
      shorten=-1.2pt
    ]
    \ar[
      dd,
      phantom,
      "{
    \mathclap{
      \raisebox{7pt}{
      \scalebox{.7}{
        \color{darkblue}
        \bf
        \def\arraystretch{.9}
        \begin{tabular}{c}
          Basic
          \\
          M-algebra
        \end{tabular}
      }
      }
    }          
      }"{pos=.2}
      ]
    \\[-5pt]
    \mathbb{R}^{1,10\,\vert\,\mathbf{32}}
    \ar[
      dr, 
      ->>,
      shorten=-1.2,
      "{
      \scalebox{.7}{
        \color{darkblue}
        \bf
        \def\arraystretch{.9}
        \begin{tabular}{c}
          Super 
          \\
          tangent space
        \end{tabular}
      }      
      }"{swap,  xshift=-16pt, yshift=7pt}
    ]
    &&
    \mathbb{R}^{1,10\,+\,517}
    \ar[
      dl, 
      ->>,
      shorten=-1.2,
      "{
      \scalebox{.7}{
        \color{darkblue}
        \bf
        \def\arraystretch{.9}
        \begin{tabular}{c}
          Exceptional
          \\
          tangent space
        \end{tabular}
      }      
      }"{
        xshift=19pt,
        yshift=6pt
      }
    ]
    \\[-5pt]
    &
    \underset{
      \mathclap{
        \raisebox{-9pt}{
          \scalebox{.7}{
            \color{darkblue}
            \bf
            \def\arraystretch{.9}
            \begin{tabular}{c}
              Ordinary
              \\
              tangent space
            \end{tabular}
          }
        }
      }
    }{
      \mathbb{R}^{1,10}
    }
  \end{tikzcd}
$$

However, the hallmark of non-perturbative physics in general and hence of M-theory in particular should be visible  {\bf globally} in topologically stabilized field configurations (solitons, skyrmions and anyons, cf. e.g. \cite{Rajaraman84}\cite[\S V.6]{Zee10}), which in discussion of M-theory have received less attention. Global effects arise particularly due to completion of dynamics by {\it flux/charge quantization laws} \cite{Freed02}\cite{MathaiSati}\cite{SS24-Flux} for the (higher) gauge fields: This is classical for the electromagnetic field (where Dirac charge quantization in ordinary cohomology stabilizes Abrikosov vortex solitons in type-II super-conductors, recalled in \cite[\S 2.1]{SS24-Flux}), famous for the RR-fields (where, conjecturally, twisted K-cohomology stabilizes D-branes, see \cite{GS22}, recalled in \cite[\S 4.1]{SS24-Flux}) and as has been hypothesized for the M-theory C-field \cite{DFM07}\cite{tcu}\cite{FSS20-H} (reviewed in \cite[\S 4.2]{SS24-Flux}, where twisted Cohomotopy stabilizes M-branes subject the notorious half-integral shift of the 4-flux and the tadpole cancellation of the 7-flux). 

\medskip

\noindent
{\bf Global topology of super-exceptional geometry.} 
While flux-quantization is ordinarily considered on ordinary supergravity spacetimes (e.g. \cite{LazaroiuShahbazi22}\cite{GSS24-SuGra}) or on brane worldvolume submanifolds (e.g. \cite{FSS21Hopf}\cite{GSS24-FluxOnM5}), the same process should be applied after geometrically manifesting hidden duality-symmetries, but now on the vastly higher dimensional (super-)exceptional geometric enhancement of spacetime, whose (choice of) global topology thereby gains physical significance: Solitonic field configurations that are ordinarily localized in spacetime now also depend on and may be localized along the exceptional geometric spacetime directions!

\smallskip 
This effect
may not before have received due attention in generality (we will discuss it further in \cite{GSS25-AnyonsOnExceptional}), but it is apparent in the more well-studied special case of T-duality, where ``doubled spacetime'' (e.g. \cite{HLZ13}) globally has the structure of a torus-bundle (the ``correspondence space'' in topological T-duality, cf. e.g. \cite[\S 1]{Waldorf24}).

\medskip

\noindent
{\bf Toroidal M-geometry.}
Towards a discussion of such toroidal (and eventually other) global topological structure for super-exceptional geometries, we here consider globalizing the (hidden) M-algebra to a super-Lie group -- the (hidden) {\it M-group}

\vspace{-.4cm}
$$
  \begin{tikzcd}[
    row sep=0pt,
    column sep=40pt
  ]
    \scalebox{.7}{
      \color{gray}
      \bf
      \def\arraystretch{.9}
      \begin{tabular}{c}
        super-
        \\
        Lie groups
      \end{tabular}
    }
    \ar[
      rr,
      phantom,
      shift right=10pt,
      "{
        \scalebox{.7}{
          \color{gray}
          \bf
          \begin{tabular}{c}
            super-Lie
            \\
            differentiation
          \end{tabular}
        }
      }"
    ]
    &&
    \scalebox{.7}{
      \color{gray}
      \bf
      \def\arraystretch{.9}
      \begin{tabular}{c}
        super-
        \\
        Lie algebras
      \end{tabular}
    }
    \\
    \mathrm{sLieGrp}
    \ar[
      rr
    ]
    &&
    \mathrm{sLieAlg}
    \\
    \mathllap{
      \adjustbox{
        scale=.7,
        raise=1pt
      }{
        \color{darkblue}
        \bf 
        Hidden M-group
      }
      \;\;
    }
    \widehat{\mathcal{M}}
    \ar[
      d,
      ->>
    ]
    &&
    \widehat{\mathfrak{M}}
    \mathrlap{
      \;\;
      \adjustbox{
        scale=.7,
        raise=1pt
      }{
        \color{darkblue}
        \bf 
        Hidden M-algebra
      }
    }
    \ar[
      d,
      ->>
    ]
    \\[10pt]
    \mathllap{
      \adjustbox{
        scale=.7,
        raise=1pt
      }{
        \color{darkblue}
        \bf 
        Basic M-group
      }
      \;\;
    }
    \mathcal{M}
    &&
    \mathfrak{M}
    \mathrlap{
      \;\;
      \adjustbox{
        scale=.7,
        raise=1pt
      }{
        \color{darkblue}
        \bf 
        Basic M-algebra
      }
    }
  \end{tikzcd}
$$
such that there are lattice subgroups,  quotienting by which yields toroidally compactified super-exceptional geometries:
\begin{equation}
  \label{ToroidalMGeometryInIntro}
  \begin{tikzcd}[row sep=10pt, column sep=large]
    \mathbb{Z}^{k}
    \ar[
      r,
      hook
    ]
    \ar[
      d,
      equals
    ]
    &
    \widehat{\mathcal{M}}
    \ar[
      d,
      ->>
    ]
    \ar[
      r,
      ->>
    ]
    &
    \widehat{\mathcal{M}}/\mathbb{Z}^k
    \mathrlap{
      \adjustbox{
        scale=.7
      }{
        \color{darkblue}
        \bf
        \def\arraystretch{.9}
        \def\tabcolsep{0pt}
        \begin{tabular}{c}
          Toroidally compactified
          \\
          super-exceptional 
          \\
          spacetime
        \end{tabular}
      }
    }
    \ar[
      d,
      ->>,
      shorten <=-8pt
    ]
    \\
    \mathbb{Z}^{k}
    \ar[
      r,
      hook
    ]
    &
    \mathcal{M}
    \ar[
      r,
      ->>
    ]
    &
    \mathcal{M}/\mathbb{Z}^k
    \mathrlap{
      \hspace{1cm}
      (0 \leq k \leq  528)
    }
  \end{tikzcd}
\end{equation}

\smallskip

It is worthwhile and our aim here to dwell on the details of this construction, because it plays such an interesting role for the physics while being simple (namely: nilpotent, see Rem. \ref{NilpotentSuperLieAlgebras} below) as far as examples of super-Lie groups go, thus potentially enriching both the supergravity literature (which tends to shun super-manifold theory, cf. \cite[\S II.2.4, p. 338]{CDF91}) as well  as the mathematical super-geometry literature (e.g. \cite{Varadarajan04}, which in turn is short of more cutting-edge physics examples).

\medskip

To say this in a little more detail before we get to the full development:

\smallskip

\noindent
{\bf The M-Algebra}
in some generality was named by \cite{Sezgin97}, but in its ``basic'' form it was already highlighted in \cite[(13)]{Townsend95} (further so in \cite[(1)]{Townsend98}\cite{Townsend99}), and in its subtle ``hidden'' extension of the basic form it actually goes way back to \cite{DF82} (later generalized by \cite{BDIPV04}, and reviewed many times, e.g. \cite[\S 5]{AndrianopoliDAuria24}), discovered already with the ambition to identify hidden symmetries of 11D SuGra.

\smallskip

The {\bf basic M-algebra} $\mathfrak{M}$ (as we shall call it here, just for disambiguation from its further extensions) is the maximal extension of the (translational) 11D super-symmetry algebra by central charges identifiable (e.g. \cite{SS17-BPS}) with conserved charges of probe M2- and M5-branes, the non-trivial super-Lie bracket having the emblematic form \eqref{TheFullyExtendedSusyBracket}:
$$
  \begin{array}{c}
    \grayunderbrace{
      [Q_\alpha,\,Q_\beta]
    }{
      \mathclap{
      \scalebox{.7}{
        \def\arraystretch{.9}
        \begin{tabular}{c}
          super-bracket of
          \\
          super-charges
        \end{tabular}
      }
      }
    }
    \;=\;
    -
    \,
    2 \,\Gamma^a_{\alpha \beta}\,
    \grayunderbrace{
      P_a
    }{
      \mathclap{
      \scalebox{.7}{
        \def\arraystretch{.9}
        \begin{tabular}{c}
          space-time
          \\
          momenta
        \end{tabular}
      }
      }
    }    
    \,
    \;+\;
    2 \,\Gamma^{a_1 a_2}_{\alpha \beta}\, 
    \grayunderbrace{
      Z_{a_1 a_2}
    }{
      \mathclap{
      \scalebox{.7}{
        \def\arraystretch{.9}
        \begin{tabular}{c}
          M2-brane
          \\
          charges
        \end{tabular}
      }
      }    
    }
    \,
    \,-\,
    2 \,\Gamma^{a_1 \cdots a_5}_{\alpha \beta}
    \grayunderbrace{
     Z_{a_1 \cdots a_5}
    }{
      \mathclap{
      \scalebox{.7}{
        \def\arraystretch{.9}
        \begin{tabular}{c}
          M5-brane
          \\
          charges
        \end{tabular}
      }
      }        
    }
    \,.
  \end{array}
$$
It is exceedingly useful to re-express this (and all other finite-dimensional super-Lie structure) equivalently in terms of the linear-dual free graded super-algebra
(the CE-algebra, recalled in \S\ref{SuperLieAlgebras}) on which the above super-Lie bracket is incarnated as the differential given on generators by (see \S\ref{SpinorsIn11d} for our Clifford algebra conventions):
$$
  \mathrm{d}
  \, 
  e^a
  \;=\;
  +
  \big(\hspace{1pt}
    \overline{\psi}
    \,\Gamma^a\, 
    \psi
  \big)
  \,,\;\;\;\;\;\;\;\;\;
  \mathrm{d}
  \, 
  e_{a_1 a_2}
  \;=\;
  -
  \big(\hspace{1pt}
    \overline{\psi}
    \,\Gamma_{a_1 a_2}\, 
    \psi
  \big)
  \,,\;\;\;\;\;\;\;\;\;
  \mathrm{d}
  \, 
  e_{a_1 \cdots a_5}
  \;=\;
  +
  \big(\hspace{1pt}
    \overline{\psi}
    \,\Gamma_{a_1 \cdots a_5}\, 
    \psi
  \big)
  \,.
$$

Here we are to think of the ordinary translational super-symmetry algebra as being super-Minkowski spacetime $\mathbb{R}^{1,10\,\vert\, \mathbf{32}}$ equipped (just) with its infinitesimal super-translational structure, and so the basic M-algebra may be thought of as an extended super-spacetime with no less than $11 + \binom{11}{2} + \binom{11}{5} \,=\, 11 \,+\, 517 \,=\, 528$ bosonic dimensions
$$
  \begin{tikzcd}[
    row sep=10pt
  ]
    \mathllap{
      \widehat{\mathfrak{M}}
      \,=\;
    }
    \mathbb{R}^{
      1,10 + 517 
        \,\vert\,
      \mathbf{32}
      \oplus 
      \mathcolor{purple}{
      \mathbf{32}
      }
    }
    \ar[
      d,
      ->>,
      shorten <= -3pt
    ]
    \ar[
      dd,
      ->>,
      rounded corners,
      to path={ 
            ([xshift=+00pt]\tikztostart.east)
        --  ([xshift=+7pt]\tikztostart.east)
        --  node[xshift=5pt]{
           \scalebox{.7}{
             $\widehat{\phi}$
           } 
        }
            ([xshift=+23.5pt]\tikztotarget.east)
        --  ([xshift=+00pt]\tikztotarget.east)
      }
    ]
    \\
    \mathllap{
      \mathfrak{M}
      \,=\;
    }
    \mathbb{R}^{
      1,10 + \mathcolor{purple}{517} \,\vert\,
      \mathbf{32}
    }
    \ar[
      d,
      ->>,
      shorten <=-3pt,
      "{
        \phi
      }"{pos=.15}
    ]
    \\
    \mathbb{R}^{
      1,10\,\vert\,
      \mathbf{32}
    }
  \end{tikzcd}
$$
We are going to be very explicit (in \S\ref{AsASuperlieGroup}) about what this means for the {\it finite} super-translation {\it group} structure, first in this basic case and then for its further  hidden extension.

\smallskip

The {\bf hidden M-algebra} $\widehat{\mathfrak{M}}$ itself is not hard to describe, either: It is a {\it fermionic} (meaning: odd) extension of the basic M-algebra by one further spinor-valued generator $\DFSpinor$ on which the differential is given by \eqref{ExceptionalCEDifferentialOnEta}\eqref{ParameterizationByS}
$$
  \mathrm{d} \, \DFSpinor
  \;=\;
  2(1 + s)
  \,
  \Gamma_a \psi\; e^a
  \;+\;
  \Gamma^{a_1 a_2} \psi
  \;
  e_{a_1 a_2}
  \;+\;
  2\tfrac{6+s}{6!}
  \,
  \Gamma^{a_1 \cdots a_5} \psi
  \;
  e_{a_1 \cdots a_5}\;,
$$
for any $s \in \mathbb{R} \setminus \{0\}$; 
a standard exercise with Fierz identities checks that this differential really squares to zero (see Prop. \ref{ExistenceOfSuperExceptionalAlgebra}). But it is only through a heavy (and error-prone, Rem. \ref{ComputationInTheLiterature}) computation (see Prop. \ref{TheRestrictedAvatarOfTheBFieldFlux}) that one finds the crucial and maybe surprising property of $\widehat{\mathfrak{M}}$:
There exists a rich super-invariant \eqref{AnsatzForH30} on $\widehat{\mathfrak{M}}$,
$$
  \widehat{P}_3
  \;\;\propto\;\;
  e_{a_1 a_2}
  \, e^{a_1}\, e^{a_2}
  \,+\,
  \mbox{several more terms}
  \,,
$$ 
which is a coboundary
\begin{equation}
  \label{CoboundaryRelationInIntro}
  \mathrm{d}
  \,
  \widehat{P}_3
  \,=\,
  \widehat{\phi}^\ast
  G_4
  \,,
  \;\;\;\;
  \mbox{where}
  \;\;\;
  G_4
  \;:=\;
  \tfrac{1}{2}
  \big(\hspace{1pt}
    \overline{\psi}
    \,\Gamma_{a_1 a_2}\,
    \psi
  \big)
  e^{a_1}\, e^{a_2}
  \,,
\end{equation}

\vspace{1mm} 
\noindent for the super-avatar $G_4$
of the super 4-flux density in 11D SuGra (e.g. \cite[(3.15d)]{DF82}\cite[(2.2.7)]{DNP86}\cite[(8)]{GSS24-SuGra}).

\smallskip 
Since this coboundary relation looks like the relation satisfied {\it locally}, namely on any super-chart $U \xhookrightarrow{\phi^U} X^{1,10\,\vert\,\mathbf{32}}$, by a C-field gauge potential $C^U_3$ (the original ``3-index photon'') for vanishing bosonic 4-flux,
original authors \cite{DF82} (cf. also \cite[p. 9]{Sezgin97}\cite[\S 6.5-6]{Varela06}) thought of $\widehat{P}_3$ as a ``decomposition'' of the gauge potential 3-forms $C_3^U$ into a wedge product of the 1-form generators of (the CE-algebra of) the hidden M-algebra. 
However, we caution that there are alternative interpretations which are not unrelated but different: 
Namely, given a $\sfrac{1}{2}$BPS super-embedding 
$$
\Sigma^{1,5\,\vert\, 2\cdot \mathbf{8}} \xhookrightarrow{\quad \phi^{\mathrm{M5}}\quad } X^{1,10\,\vert\, \mathbf{32}}
$$ 
of an M5-brane super-worldvolume \cite{HoweSezgin97b}\cite{Sorokin00} \cite{GSS24-FluxOnM5} into a fluxless background, then the {\it flux density} $H_3$ of the (non-linearly) self-dual tensor field on $\Sigma$ also is a coboundary for $(\phi^{\mathrm{M5}})^\ast G_4$. Under this interpretation, as the $H_3$-factor in the M5-brane action functional, the ``decomposed'' 3-form $\widehat{P}_3$ has been discussed in \cite[p. 10]{Sezgin97}\cite{FSS20Exc}\cite{FSS21Exc}.

\smallskip 
Yet an alternative interpretation of $\widehat{P}_3$ is suggested by \cite{GSS25-TDuality}, where $\widehat{P}_3$ is shown to be an M-theoretic lift of the ``Poincar{\'e} 2-form'' $P_2$ that controls T-duality on doubled 10D super-spacetime via the coboundary relation $\mathrm{d}\,P_2\,=\, H_3^A \,-\, H_3^{\widetilde A}$.

We will further discuss these interpretations of $\widehat{P}_3$ in \cite{GSS25-AnyonsOnExceptional}; here our focus is on laying some super-geometric groundwork, namely to give a careful treatment of the global extension of $\widehat{P}_3$ to a left-invariant super 3-form on the hidden M-group $\widehat{\mathcal{M}}$
\vspace{-2mm} 
$$
  \begin{tikzcd}[
    row sep=-15pt,
    column sep=100pt
  ]
    &[-4pt]
    \underset{
      \mathclap{
        \raisebox{-10pt}{ 
          \scalebox{.7}{
            \color{darkblue}
            \bf
            \def\arraystretch{.85}
            \begin{tabular}{c}
              Hidden
              \\
              M-group
            \end{tabular}
          }
        }
      }
    }{
      \widehat{\mathcal{M}}
    }
    \ar[
      dl,
      |->,
      shorten=7pt,
      "{
        \scalebox{.7}{
          \color{darkgreen}
          \bf
          \def\arraystretch{.9}
          \begin{tabular}{c}
            underlying
            \\
            super-manifold
          \end{tabular}
        }
      }"{sloped}
    ]
    \ar[
      dr,
      |->,
      shorten=7pt,
      "{
        \scalebox{.7}{
          \color{darkgreen}
          \bf
          \def\arraystretch{.9}
          \begin{tabular}{c}
            underlying
            \\
            super-Lie algebra
          \end{tabular}
        }
      }"{sloped}
    ]
    \\
    \underset{
      \mathclap{
        \raisebox{-3pt}{
          \scalebox{.7}{
            \color{darkblue}
            \bf
            \def\arraystretch{.9}
            \begin{tabular}{c}
              Super-exceptional
              \\
              Minkowski spacetime
            \end{tabular}
          }
        }
      }
    }{\mathbb{R}^{1,10\,+\,517\,\vert\,\mathbf{32}\oplus\mathbf{32}}
    }
    &&
    \quad 
    \underset{
      \mathclap{
        \raisebox{-7pt}{
          \scalebox{.7}{
            \color{darkblue}
            \bf
            \def\arraystretch{.9}
            \begin{tabular}{c}
              Hidden
              \\
              M-algebra
            \end{tabular}
          }
        }
      }
    }  
    {
      \widehat{\mathfrak{M}}\, .
    }
  \end{tikzcd}
$$

\noindent
{\bf Super-geometry.}
Historically, the proper mathematical formulation of global differential super-geometry \cite{Berezin87} (super-manifolds, super-Lie groups and super-differential forms on them, cf. \cite{DeligneMorgan99}, hence what should ultimately be the very foundation for formulations of supergravity on ``super-space'' \cite{WessZumino77}\cite{Howe82}\cite{CDF91})  has had a bit of a rough start, with competing definitions arguably tending to look a little clunky (such as in working over infinite-dimensional Grassmann algebras \cite{DeWitt92} even for finite-dimensional manifolds, or resorting to the notion of locally ringed spaces, \cite{BerezinLeites75}, cf. \cite{Rogers07}), which may have discouraged (cf. \cite[\S II.2.4, p. 338]{CDF91}) its wide adoption in the supergravity literature, of all places.

\smallskip 
As a consequence, authors in supergravity theory tend to either work with super-matrix groups only, or else write symbolic exponentiations of super-Lie algebra elements. While this is useful as far as it goes, imagine the analogous hypothetical situation where all that general relativists would know about manifolds were that, locally, they may be parameterized by symbolic exponentials of vector fields.

\smallskip 
Luckily, there is a rigorous, powerful, and slick
\footnote{
Namely,
{\bf (i)} To avoid the notion of locally ringed spaces one may observe that smooth super-manifolds $X$ are faithfully characterized already by their super-algebras $C^\infty(X)$ of {\it global} super-functions
(which in the language of algebraic geometry means that smooth super-manifolds are in fact all {\it affine} -- 
  the analogous statement for ordinary smooth manifolds, ``Milnor's exercise'',  is classical but also remains under-appreciated)
and by Batchelor's theorem these are always the Grassmann algebras of smooth sections of a smooth vector bundle over an ordinary manifold. 
{\bf (ii)} To avoid infinite-dimensional Grassmann algebras one may observe that what is really needed at any given time are finitely many but arbitrary Grassmann variables such that all constructions are covariant under their choice. This is clearly not unlike the situation with choosing ordinary coordinates, and indeed the most general smooth super-space may hence be characterized by the covariant system of generalized super-coodinate charts that it admits. 
}
modern formulation of super-geometry which is {\it secretly} the most abstract-general (\cite[\S 3.1.3]{SS20-Orb}) but which neatly blends into the actual physics practice \cite{GSS24-SuperGeometry}.
By way of developing the example of the hidden M-group in \S\ref{AsASuperlieGroup}, we give a lightweight explanation also of this underlying super-geometry.

\medskip
\medskip

\noindent
{\bf Acknowledgements.} We thank Zoran {\v S}koda for pointing out the historical references for the coordinate expressions of Maurer-Cartan forms mentioned in the proof of Lem. \ref{MCFormsInCoordinates}.

%%%%%%%%%%%%%%%%%%%%%%%%%%%%%%%%%%%%%%%%%%
\section{The M-algebra}
\label{SuperExceptionalAsLieAlgebra}
%%%%%%%%%%%%%%%%%%%%%%%%%%%%%%%%%%%%%%%%%%

Here we recall the basic M-algebra (\S \ref{BasicMAlgebra}), re-derive 
the ``decomposed'' 3-form
on its ``hidden'' extension (\S\ref{MinimalFermionicExtension})
and discuss various related issues, such as phenomena at special values of the parameter that the hidden M-algebra depends on.

%%%%%%%%%%%%%%%%%%%%%%%%%%%%%%%%%%%%%%%%%%%%%
\subsection{The base case}
\label{BasicMAlgebra}
%%%%%%%%%%%%%%%%%%%%%%%%%%%%%%%%%%%%%%%%%%%%%

\noindent
{\bf The super-Minkowski algebra.}
By the  ($D=11$, $\mathcal{N}=1$) {\it super-Minkowski Lie algebra} we mean the super-translational super-Lie sub-algebra of the super-Poincar{\'e} algebra 
\footnote{
The full super-Poincar{\'e} super Lie algebra (aka: ``supersymmetry algebra'') is the semi-direct product $\mathbb{R}^{1,10} \rtimes \mathfrak{so}(1,10)$ of the super-Minkowski algebra \eqref{SuperMinkowskiLinearBasis} with the Lorentz Lie algebra $\mathfrak{so}(1,10)$ acting on $\mathbb{R}\big\langle (P_a)_{a = 0}^{10}\big\rangle$ as its defining/vector representation and on $\mathbb{R}\big\langle (Q_\alpha)_{\alpha=1}^{32} \big\rangle \,\simeq\, \mathbf{32}$ as its irreducible Majorana spin representation \eqref{The11dMajoranaRepresentation}. Similarly, there is the semidirect product with $\mathfrak{so}(1,10)$ of the basic M-algebra \eqref{TheBasicMAlgebra} and the hidden M-algebra \eqref{GeneratorsOfTheSuperExceptionalLieAlgebra}, which may be regarded as the full M-symmetry algebra, see \hyperlink{TableKleinian}{Table 1}. But since no further subtleties are involved in forming these semidirect products with the Lorentz algebra, we do not further dwell on them here.} 
(commonly known as the {\it supersymmetry algebra}) whose underlying super-vector space is (cf. our super-algebra conventions in \S\ref{SuperAlgebraConventions})
\begin{equation}
  \label{SuperMinkowskiLinearBasis}
  \mathbb{R}^{1,10\,\vert\,\mathbf{32}}
  \;\simeq\;
  \mathbb{R}\Big\langle \!
     \grayunderbrace{
       (Q_\alpha)_{\alpha=1}^{32}
     }{ 
       \mathrm{deg}
       \,=\, 
       (0,\mathrm{odd}) 
    }
    \,,
    \grayunderbrace{
      (P_a)_{a = 0}^{10}
    }{ 
      \mathrm{deg} \,=\, (0,\mathrm{evn}) 
    }
  \!\!\Big\rangle
\end{equation}
with the only non-trivial super-Lie brackets on basis elements being \footnote{
  \label{PrefactorConvention}
  Our prefactor convention in \eqref{TheBifermionicSuperBracket} -- ultimately enforced via the translation 
  \eqref{RelationBetweenStructureConstants}
  by our convention for the super-torsion tensor in \cite{GSS24-SuGra} and \cite{GSS24-SuGra} --
  coincides with that in \cite[(1.16)]{DeligneFreed99}\cite[p. 52]{Freed99}.
}
\begin{equation}
  \label{TheBifermionicSuperBracket}
  \big[
    Q_\alpha
    ,\,
    Q_\beta
  \big]
  \;=\;
  -
  2\,
  \Gamma^a_{\alpha \beta}
  P_a
  \,.
\end{equation}

Its Chevalley-Eilenberg algebra \eqref{RelationBetweenStructureConstants}
therefore has the underlying graded super-algebra
\vspace{1mm}
\begin{equation}
  \label{CEOfSuperMinkowksi}
  \mathrm{CE}\big(
    \mathbb{R}^{1,10\,\vert\,\mathbf{32}}
  \big)
  \;\simeq\;
  \mathbb{R}\big[\!
    \grayunderbrace{
     (\psi^\alpha)_{\alpha = 0}^{32}
    }{
      \mathrm{deg}\,=\, (1,\mathrm{odd})
    }
    ,\,
    \grayunderbrace{
    (e^a)_{a=0}^{10}
    }{ 
      \mathrm{deg} = (1,\mathrm{evn})
    }
  \!\big]
\end{equation}
with the differential given on generators by
\begin{equation}
  \label{CEDifferentialForSuperMinkowski}
  \begin{array}{cll}
    \mathrm{d}\, \psi
    &=&
    0
    \\
    \mathrm{d}\, 
    e^a
    &=&
    \big(\hspace{1pt}
      \overline{\psi}
      \,\Gamma^a\,
      \psi
    \big)
    \,.
  \end{array}
\end{equation}

\vspace{1mm} 
For the following, it is instructive to note that the 2-forms $\big(\hspace{1pt}\overline{\psi}\,\Gamma^a\,\psi\big) \in \mathrm{CE}(\mathbb{R}^{0\,\vert\,\mathbf{32}})$ are non-trivial 2-cocycles on the purely fermionic abelian subalgebra $\mathbb{R}^{0\vert 32}$ --- the {\it super-point} -- whence \eqref{CEDifferentialForSuperMinkowski} exhibits the super-Minkowski algebra as a central extension of the superpoint (cf. \cite[\S 2.1]{CdAIPB00}\cite{HS18}):
\begin{equation}
  \label{SuperMinkowskiAlgebrasAsCentralExtension}
  \begin{tikzcd}[
    sep=10pt
  ]
    0
    \ar[r]
    &
    \mathbb{R}^{1,10}
    \ar[rr, hook]
    &&
    \mathbb{R}^{1,10\,\vert\,\mathbf{32}}
    \ar[rr, ->>]
    &&
    \mathbb{R}^{0\,\vert\,\mathbf{32}}
    \ar[r]
    &
    0
    \,.
  \end{tikzcd}
\end{equation}

\smallskip

\noindent
{\bf The basic M-algebra.}
Concerning $\big(\hspace{1pt}\overline{\psi}\,\Gamma^a\,\psi\big)$ in \eqref{CEDifferentialForSuperMinkowski} being a 2-cocycle, it is obvious that it is closed  and not exact --- since $\psi$ is closed and not exact \eqref{CEDifferentialForSuperMinkowski} --- but what is mildly non-trivial is that it exists as a non-vanishing $\mathrm{Spin}(1,10)$-invariant 2-form in the first place: The only further expressions for which this is the case are 
\begin{equation}
  \label{Further2CocyclesOnSuperPoint}
  \big(\hspace{1pt}
    \overline{\psi}
    \,\Gamma^{a_1 a_2}\,
    \psi
  \big)
  ,\; \;
  \big(\hspace{1pt}
    \overline{\psi}
    \,\Gamma^{a_1 \cdots a_5}\,
    \psi
  \big)
  \;\in\;
  \mathrm{CE}\big(
    \mathbb{R}^{0 \vert \mathbf{32}}
  \big)
  \,,
  \hspace{1cm}
  a_i \in \{0, \cdots, 10\}
  \,,
\end{equation}
since the spinor-valued 1-forms $\psi^\alpha$ are of bi-degree $(1,\mathrm{odd})$, hence mutually commuting \eqref{TheSignRule}, and since \eqref{Further2CocyclesOnSuperPoint} are the only {\it symmetric} $\mathrm{Spin}(1,10)$-invariant pairings \eqref{SymmetricSpinorPairings}.

\smallskip

Therefore, the {\it maximal} $\mathrm{Spin}(1,10)$-invariant central extension of the super-point $\mathbb{R}^{0\,\vert\,\mathbf{32}}$ has further central generators $Z^{a_1 a_2}$, $Z^{a_1 \cdots a_5}$ (skew-symmetric in their indices), corresponding to \eqref{Further2CocyclesOnSuperPoint},
\begin{equation}
  \label{TheBasicMAlgebra}
  \mathfrak{M}
  \;\;\simeq\;\;
  \mathbb{R}
  \Big\langle
    \grayunderbrace{
      (Q_\alpha)_{\alpha=1}^{32}
    }{
      \mathrm{deg}\,=\, (0,\mathrm{odd})}
    \,,
    \grayunderbrace{
      (P_a)_{a=0}^{10}
    }{
      \mathrm{deg}\,=\,(0,\mathrm{evn})
    }
    \,,
    \grayunderbrace{
      (Z_{a_1 a_2} = Z_{[a_1 a_2]})_{a=0}^{10}
    }{
      \mathrm{deg}\,=\,(0,\mathrm{evn})
    }
    \,,
    \grayunderbrace{
      (Z_{a_1 \cdots a_5} = Z_{[a_1 \cdots a_5]})_{a=0}^{10}
    }{
      \mathrm{deg}\,=\,(0,\mathrm{evn})
    }
  \Big\rangle
\end{equation}

\vspace{1mm} 
\noindent with non-vanishing super-Lie bracket on generators now given by \footnote{
  The signs in  \eqref{TheFullyExtendedSusyBracket} are conventional; we use a different sign for the second summand in order to, further below, match conventions used in the literature, see footnote \ref{PrefactorConvention} below.
}
\begin{equation}
  \label{TheFullyExtendedSusyBracket}
  \begin{array}{c}
    [Q_\alpha,\,Q_\beta]
    \;=\;
    -
    \,
    2 \,\Gamma^a_{\alpha \beta}\, P_a
    \,
    \,+\,
    2 \,\Gamma^{a_1 a_2}_{\alpha \beta}\, Z_{a_1 a_2}
    \,
    \,-\,
    2 \,\Gamma^{a_1 \cdots a_5}_{\alpha \beta}\, Z_{a_1 \cdots a_5}
    \,.
  \end{array}
\end{equation}

\smallskip 
This fully extended version of (the translational part of) the $D=11$, $\mathcal{N}=1$ supersymmetry algebra may be understood (\cite[(13)]{Townsend95}\cite[(1)]{Townsend98}, cf. also  \cite{SS17-BPS}) as incorporating charges $Z^{a_1 a_2}$ of M2-branes and $Z^{a_1 \cdots a_5}$ of M5-branes, whence we shall call this the {\it basic M-algebra}, following \cite{Sezgin97}\cite{BDPV05}\cite[(3.1)]{Bandos17}. 
\footnote{
  \cite{Sezgin97} uses the term ``M-algebra'' for a large further extension of \eqref{TheFullyExtendedSusyBracket} which includes the ``hidden M-algebra'' that we are concerned with here; whereas other authors like \cite{BDPV05} say ``M-algebra'' for just \eqref{TheFullyExtendedSusyBracket}. Here we disambiguate this situation by speaking of the ``basic'' M-algebra and its ``hidden'' extension, respectively, the latter term following the terminology introduced much earlier by \cite{DF82} (which, we suggest, nicely matches the terminology of ``hidden symmetries'' in generalized-geometric formulation of supergravity).
} 

Its CE-algebra is 
\begin{equation}
  \label{CEOfBasicMAlgebra}
  \mathrm{CE}\big(
    \mathfrak{M}
  \big)
  \;\;
  \;\simeq\;
  \mathbb{R}\Big[
    \grayunderbrace{
    (\psi^\alpha)_{\alpha=1}^{32}
    }{ 
      \mathrm{deg} = (1,\mathrm{odd})
    }
    ,\,
    \grayunderbrace{
    (e^a)_{a=0}^{10}
    }{
      \mathrm{deg} = (1,\mathrm{evn})
    }
    ,\,
    \grayunderbrace{
    (e_{a_1 a_2} = e_{[a_1 a_2]})_{a_i=0}^{10}
    }{
      \mathrm{deg} = (1,\mathrm{evn})
    }    
    \,,\,
    \grayunderbrace{
    (e_{a_1 \cdots a_5} 
      = 
    e_{[a_1 \cdots a_5]})_{a_i=0}^{10}
    }{
      \mathrm{deg} = (1,\mathrm{evn})
    }
  \Big]
  \,,
\end{equation}
with differential on generators given by
\footnote{
\label{TheSignInde}
We have a minus sign in the equation for $\mathrm{d}\, e_{a_1 a_2}$ in \eqref{ExceptionalCEDifferentialOnEta} to match the sign convention in \cite[(6.2)]{DF82}\cite[(17)]{BDIPV04}, which is natural in view of \eqref{TheBispinorCEElement} below, and hence ultimately due to the relative sign in the formula \eqref{FierzDecomposition} for Fierz expansion. 

Alternatively one could choose any other non-vanishing prefactor. In fact, \cite[(6.2)]{DF82} choose in addition a global factor of $1/2$, while \cite[(15)-(17)]{BDIPV04} choose in addition a global factor of $-1$, compared to our convention in \eqref{ExceptionalCEDifferentialOnEta}. But the relative prefactors agree throughout.}
\begin{equation}
  \label{DifferentialOnCEOfBasicMAlgebra}
  \def\arraystretch{1.3}
  \begin{array}{lcl}
  \mathrm{d}\, 
  \psi 
  &=&
  0
  \\
  \mathrm{d}\,
  e^a 
  &=&
  +
  \big(\, 
    \overline{\psi}
    \,\Gamma^a\,
    \psi
  \big)
  \\
  \mathrm{d}
  \,
  e_{a_1 a_2}
  &=&
  -
  \big(\,
    \overline{\psi}
    \,\Gamma_{a_1 a_2}\,
    \psi
  \big)
  \\
  \mathrm{d}
  \,
  e_{a_1 \cdots a_5}
  &=&
  +
  \big(\,
    \overline{\psi}
    \,\Gamma_{a_1 \cdots a_5}
    \psi
  \big)
  \,.
  \end{array}
\end{equation}

\medskip

\noindent
{\bf Automorphy of the basic M-algebra.}
Essentially the following Prop. \ref{ManifestGL32} has been highlighted in \cite[\S 4]{West03}, following \cite[\S 5]{BaerwaldWest00}, our proof follows \cite[(26)]{BDIPV04}:
\begin{proposition}[\bf Manifestly $\mathrm{GL}(32)$-equivariant incarnation of basic M-algebra]
  \label{ManifestGL32}
  Unifying all the bosonic generators of \eqref{CEOfSuperMinkowksi}
  into a bispinorial form like this
  \begin{equation}
    \label{TheBispinorCEElement}
    e^{\alpha \beta}
    \;\;
    :=
    \;\;
    \tfrac{1}{32}
    \big(
      e^a
      \,
      \Gamma
        _a 
        ^{\alpha \beta}
      +
      \tfrac{1}{2}
      e^{a_1 a_2}
      \,
      \Gamma
        _{a_1 a_2}
        ^{\alpha \beta}
      +
      \tfrac{1}{5!}
      e^{a_1 \cdots a_5}
      \,
      \Gamma
        _{a_1 \cdots a_5}
        ^{\alpha \beta}
    \big)
  \end{equation}
  which is symmetric by \eqref{SymmetricSpinorPairings},
  \begin{equation}
    \label{eAlphaBetaIsSymmetric}
    e^{\alpha \beta}
    \;=\;
    e^{\beta \alpha}
    \,,
  \end{equation}
  the differential \eqref{DifferentialOnCEOfBasicMAlgebra}
  acquires equivalently the compact form
  \begin{equation}
    \label{ManifestEquivariantDifferentialOnBasicMAlgebra}
    \def\arraystretch{1.1}
    \begin{array}{lcl}
      \mathrm{d}
      \,
      \psi^\alpha
      &=&
      0
      \\
      \mathrm{d}
      \,
      e
        ^{\alpha \beta}
      &=&
      \psi^\alpha
      \,
      \psi^\beta
    \end{array}
  \end{equation}
  which makes manifest that $g\in \mathrm{GL}(32)$ acts via super-Lie algebra automorphisms of the M-algebra 
  \vspace{-1mm} 
  \begin{equation}
    \label{ManifestGL32Equivariance}
    \begin{tikzcd}[row sep=-4pt,
      column sep=0pt
    ]
      \mathllap{
        g
        \;:\;
      }
      \mathrm{CE}\big(
        \mathfrak{M}
      \big)
      \ar[rr]
      &&
      \mathrm{CE}\big(
        \mathfrak{M}
      \big)
      \\
      \psi^\alpha
      &\longmapsto&
      g^{\alpha}_{\alpha'}
      \,
      \psi^{\alpha'}
      \\
      e^{\alpha \beta}
      &\longmapsto&
      g^{\alpha}_{\alpha'}
      \,
      g^{\beta}_{\beta'}
      \,
      e^{\alpha'\beta'}
    \end{tikzcd}
  \end{equation}
\end{proposition}
\begin{proof}
 First, to see that the transformation \eqref{TheBispinorCEElement} is invertible,
  the trace-property
  \eqref{VanishingTraceOfCliffordElements}
  allows to recover:
  \begin{equation}
    \label{OriginalBosonicGeneratorsFromManifestlyGL32EquivariantBasis}
    \def\arraystretch{1.2}
    \begin{array}{lcl}
      e^a
      &=&
      \phantom{+}
      \Gamma
        ^a
        _{\alpha \beta}
      \,
      e^{\alpha \beta}
      \\
      e^{a_1 a_2}
      &=&
      -
      \Gamma
        ^{a_1 a_2}
        _{\alpha \beta}
      \,
      e^{\alpha \beta}
      \\
      e^{a_1 \cdots a_5}
      &=&
      \phantom{+}
      \Gamma
        ^{a_1 \cdots a_5}
        _{\alpha \beta}
      \,
      e^{\alpha \beta}.
    \end{array}
  \end{equation}
  Finally, the differential is as claimed due to the Fierz expansion formula \eqref{FierzDecomposition}:
  $$
    \def\arraystretch{1.2}
    \begin{array}{ccll}
      \mathrm{d}
      \;
      e^{\alpha \beta}
     &=&
    \tfrac{1}{32}
    \Big(
      \Gamma
        _a 
        ^{\alpha\beta}
      \,
      \big(\hspace{1pt}
        \overline{\psi}
        \,\Gamma^a\,
        \psi
      \big)
      -
      \tfrac{1}{2}
      \Gamma
        _{a_1 a_2}
        ^{\alpha\beta}
      \big(\hspace{1pt}
        \overline{\psi}
        \,\Gamma^{a_1 a_2}
        \,
        \psi
      \big)
      +
      \tfrac{1}{5!}
      \Gamma
        _{a_1 \cdots a_5}
        ^{\alpha\beta}
      \big(\hspace{1pt}
        \overline{\psi}
        \,\Gamma^{a_1 \cdots a_5}\,
        \psi
      \big)
    \Big)
    &
    \proofstep{
     by
     \eqref{TheBispinorCEElement}
     \&
     \eqref{DifferentialOnCEOfBasicMAlgebra}
    }
    \\
    &=&
    \psi^\alpha
    \,
    \psi^{\beta}
    &
    \proofstep{
      by 
      \eqref{FierzDecomposition}.
    }
    \end{array}
  $$

  \vspace{-4mm} 
\end{proof}

\begin{example}[\bf Exponentiated Clifford elements as brane-rotating symmetries]
  \label{ExponentiatedCliffordElementsAsBraneRotation}
  Since the $\Gamma_{a_1 \cdots a_p} \,\in\, \mathrm{End}_{\mathbb{R}}(\mathbf{32})$ for $1 \leq p \leq 10$ are trace-less \eqref{VanishingTraceOfCliffordElements}, their exponentiations constitute special linear group elements
  $$
    g
    \;:=\;
    \exp\big(
      \textstyle{\sum_{p=1}^{5}}
      \tfrac{1}{p!}
      A_{a_1 \cdots a_p}
      \Gamma^{a_1 \cdots a_p}
    \big)
    \;\in\;
    \mathrm{SL}(32)
    \,\subset\,
    \mathrm{GL}(32)
    \,\subset\,
    \mathrm{End}_{\mathbb{R}}(\mathbf{32})
  $$
  for all coefficients $A_{a_1 \cdots a_p} \,\in\, \mathbb{R}$.

  Observe then that as such, their ``brane-rotating'' action \eqref{ManifestGL32Equivariance}
  on the adapted generators 
  \eqref{TheBispinorCEElement}
  of the M-algebra translates to an action by ``Dirac conjugation'' \eqref{AdjointnessOfCliffordBasisElements} $\overline{(\mbox{-})}$ on the Clifford algebra coefficients of the original defining generators \eqref{DifferentialOnCEOfBasicMAlgebra}, in that for any $\psi$, $\phi$ we have

  $$
    \def\arraystretch{1.6}
    \begin{array}{ccll}
    \psi^\alpha
    \big(
    \Gamma
      ^{a_1 \cdots a_p}
      _{\alpha' \beta'}
    g^{\alpha'}_{\alpha} 
    g^{\beta'}_{\beta}    
    \big)
    \phi^\beta
    &=&
    \big(
      g^{\alpha'}_{\alpha} 
      \psi^\alpha
    \big)
    \Gamma
      ^{a_1 \cdots a_p}
      _{\alpha' \beta'}
    \big(
      g^{\beta'}_{\beta}    
      \phi^\beta
    \big)
    \\
    &=&
    -
    \big(
      (\,\overline{g \cdot \psi}\,)
      \,\Gamma^{a_1 \cdots a_p}\,
      (g \cdot \phi)
    \big)
    &
    \proofstep{
      by
      \eqref{LoweringOfSpinorIndices}
    }
    \\
    &=&
    -
    \big(\,
      \overline{\psi}
      \,
      (\,\overline{g}\cdot\Gamma^{a_1 \cdots a_p}\cdot g)
      \,
      \phi
    \big)
    \\
    &=&
    \psi^\alpha
    \,
    \big(\,
      \overline{g}\cdot
      \Gamma^{a_1 \cdots a_p}
      \cdot
      g
    \big)_{\alpha \beta}
    \,
    \phi^\beta
    &
    \proofstep{
      by 
      \eqref{LoweringOfSpinorIndices},
    }
    \end{array}
  $$
  where, just for emphasis, ``$\cdot$'' denotes matrix multiplication, hence composition in $\mathrm{End}_{\mathbb{R}}(\mathbf{32})$.
\end{example}
\begin{example}[\bf Spinorial Lorentz-symmetry among brane-rotating symmetry]
  Restricting Ex. \ref{ExponentiatedCliffordElementsAsBraneRotation} to $p = 2$ makes manifest a canonical inclusion
  $$
    \begin{tikzcd}
      \mathrm{Spin}(1,10)
      \ar[r, hook]
      &
      \mathrm{SL}(32)
    \end{tikzcd}
  $$
  of the ordinary local spacetime symmetry into the generalized/exceptional brane-rotating symmetry.
\end{example}
\begin{example}[\bf Mixing of T-dual coordinates among brane-rotating symmetry]
  Consider the special case of Ex. \ref{ExponentiatedCliffordElementsAsBraneRotation} for 
  $$
    \def\arraystretch{1.6}
    \begin{array}{ccl}
      g &=& \exp(r \Gamma_{\!\ten})
      \\
      &=&
      \cosh(r)
      \,
      \mathrm{id}
      \,+\,
      \mathrm{sinh}(r)
      \,
      \Gamma_{\!\ten}
      \mathrlap{\,.}
    \end{array}
    \hspace{1.4cm}
    \mbox{for $r \,\in\, \mathbb{R}$}
    \,
  $$ 
  Using \eqref{SkewSelfAdjointnessOfCliffordGenerators}, the resulting brane-rotating symmetry acts by (where all $a_i, b_i < \ten$):
  $$
  \small 
    \def\arraystretch{2}
    \def\arraycolsep{2pt}
    \begin{array}{ccccccccc}
      e^{\ten}
      &=&
      \Gamma^{\ten}_{\alpha \beta}
      \,
      e^{\alpha\beta}
      &\mapsto&
      \Big(
        \exp(-r\Gamma_{\!\ten})
        \!\cdot\!
        \Gamma^{\ten} 
        \!\cdot\!
        \exp(r\Gamma_{\!\ten})
      \Big)_{\alpha \beta}
      \,
      e^{\alpha\beta}
      &=&
      \Gamma^{\ten}_{\alpha \beta}
      \, 
      e^{\alpha \beta}
      &=&
      e^{\ten}
      \\
      e^{a}
      &=&
      \Gamma^a_{\alpha \beta}
      \, 
      e^{\alpha \beta}
      &\mapsto&
      \Big(
        \exp(-r\Gamma_{\!\ten})
        \!\cdot\!
        \Gamma^{a} 
        \!\cdot\!
        \exp(r\Gamma_{\!\ten})
      \Big)_{\alpha \beta}      
      \,
      e^{\alpha \beta}
      &=&
      \Big(
        \Gamma^{a} 
        \cdot
        \exp(2r \Gamma_{\!\ten})
      \Big)_{\alpha \beta}      
      \, 
      e^{\alpha \beta}
      &=&
      \mathrm{cosh}(2r)
      \,
      e^a
      -
      \mathrm{sinh}(2r)
      \,
      e^{a\ten}
      \\
      e^{a \ten}
      &=&
      -
      \Gamma^{a \ten}_{\alpha \beta}
      \,
      e^{\alpha \beta}
      &\mapsto&
      -
      \Big(
        \exp(-r\Gamma_{\!\ten})
        \cdot
        \Gamma^{a \ten}_{\alpha \beta}
        \cdot
        \exp(r\Gamma_{\!\ten})
      \Big)_{\alpha \beta}
      \,
      e^{\alpha \beta}
      &=&
      -
      \Big(
        \Gamma^{a\ten}_{\alpha \beta}
        \cdot 
        \exp(2r \Gamma_{\!\ten})
      \Big)_{\alpha \beta}
      \,
      e^{\alpha \beta}
      &=&
      \cosh(2r)\,
      e^{a\ten}
      -
      \sinh(2r)\,
      e^a
      \\
      e^{a b}
      &=&
      -
      \Gamma^{a b}_{\alpha \beta}
      \,
      e^{\alpha \beta}
      &\mapsto&
      -
      \Big(
        \exp(-r\Gamma_{\!\ten})
        \cdot
        \Gamma^{a b} 
        \cdot
        \exp(r\Gamma_{\!\ten})
      \Big)_{\alpha \beta}      
      \, 
      e^{\alpha \beta}
      &=&
      -
      \Gamma^{ab}_{\alpha \beta}
      \,
      e^{\alpha \beta}
      &=&
      e^{a b}
    \end{array}
  $$
  and similarly one finds
  \begin{equation}
    \label{TDualicM5BraneRotation}
    \def\arraystretch{1.4}
    \begin{array}{ccl}
      e^{a_1 \cdots a_5}
      &\mapsto&
      \cosh(2r)
      \,
      e^{a_1 \cdots a_5}
      \,+\,
      \sinh(2r)
      \tfrac{1}{5!}
      \epsilon^{
        a_1 \cdots a_5
        \,\ten\,
        b_1 \cdots b_5
      }
      e_{b_1 \cdots b_5}
      \\
      e^{a_1 \cdots a_4\, \ten}
      &\mapsto&
      e^{a_1 \cdots a_4\, \ten}
      \,.
    \end{array}
  \end{equation}
  To interpret this, note that (this is discussed in detail by\cite{GSS25-TDuality}), the generators 
  \begin{equation}
    \label{StringChargeGenerators}
    \tilde e_a
    \;:=\;
    e_{a \ten}
    \,,
  \end{equation}
  appear as the ``M2-brane charges wrapping the M-theory circle'',
  and as such are to be understood as the type IIA string-charges associated with ``doubled'' coordinates for T-duality in type IIA theory along all 10 spacetime dimensions -- cf. also \eqref{FiberIntegrationOfBasicP3} below --, and in view of \eqref{TDualicM5BraneRotation} note that NS5-branes, and hence their charges, are supposed to transform among each other under T-duality.

\smallskip 
  Therefore the above transformation may be seen to ``admix'' T-dual doubled coordinates. Beware that this is not quite a T-duality transformation as such, which instead {\it swaps} $e^a \,\leftrightarrow\, \tilde e_a$. We discuss in \cite{GSS25-TDuality} how T-duality proper is enacted on the M-algebra.
\end{example}

\vspace{-1mm} 
%%%%%%%%%%%%%%%%%%%%%%%%%%%%%%%%%%%%%%%%%%%%%%%%
\subsection{The hidden extension}
\label{MinimalFermionicExtension}
%%%%%%%%%%%%%%%%%%%%%%%%%%%%%%%%%%%%%%%%%%%%%%%%

We turn to the further extension of the basic M-algebra \eqref{TheBasicMAlgebra} by odd generators $Z_\alpha$ spanning another copy of the $\mathrm{Spin}(1,10)$-representation $\mathbf{32}$.
The idea and the following Propositions \ref{ExistenceOfSuperExceptionalAlgebra} and \ref{TheRestrictedAvatarOfTheBFieldFlux} are due to \cite[(6.4)]{DF82}\cite[(20)]{BDIPV04} (see also \cite{BDPV05}\cite[\S 5]{Azcarraga05}\cite{FIdO15} \cite{ADR16}\cite{ADR17}\cite{Ravera21}\cite{AndrianopoliDAuria24}), but here we spell out the computations in order to secure crucial prefactors (cf. Rem. \ref{ComputationInTheLiterature}) below.

\smallskip

\begin{proposition}[\bf CE-Algebra of the hidden M-algebra]
\label{ExistenceOfSuperExceptionalAlgebra}
The free graded commutative algebra
\vspace{1mm} 
\begin{equation}
  \label{UnderlyingAlgebraOfSuperExceptionalLieAlgebra}
  \mathrm{CE}\big(\,
    \widehat{\mathfrak{M}}
 \, \big)
  \;\;
  \defneq
  \;\;
  \mathbb{R}\Big[
    \underbrace{
    \big(
      e^a
    \big)_{a=0}^{\ten}
    }_{\color{gray}
      \mathrm{deg} = (1,0)
    }
    \,,\,
    \underbrace{
    \big(
      e_{a_1 a_2}
      =
      e_{[a_1 a_2]}
    \big)_{a_i=0}^{\ten}
    }_{\color{gray}
      \mathrm{deg} = (1,0)
    }
    \, ,\,
    \underbrace{
    \big(
      e_{a_1 \cdots a_5}
      =
      e_{[a_1 \cdots a_5]}
    \big)_{a_i=0}^{\ten}
    }_{\color{gray}
      \mathrm{deg} = (1,0)
    }
    \, ,\,
    \underbrace{
    \big(
      \psi^\alpha
    \big)_{\alpha=1}^{32}
    }_{\color{gray}
      \mathrm{deg} = (1,1)
    }
   \, ,\,
    \underbrace{
    \big(
      \DFSpinor
        ^\alpha
    \big)_{\alpha=1}^{32}
    }_{\color{gray}
      \mathrm{deg} = (1,1)
    }
  \, \Big]
\end{equation}

\vspace{1mm} 
\noindent carries a differential $\mathrm{d}$ making it a super-DGC algebra, defined by
\footnote{
  On the sign in the second line, see again footnote \ref{TheSignInde}.
}
\begin{equation}
  \label{ExceptionalCEDifferentialOnEta}
  \def\arraystretch{1.3}
  \begin{array}{lcl}
  \mathrm{d}
  \,
  \psi
  &=&
  0
  \\
  \mathrm{d}\,
  e^a 
  &=&
  +
  \big(\, 
    \overline{\psi}
    \,\Gamma^a\,
    \psi
  \big)
  \\
  \mathrm{d}
  \,
  e_{a_1 a_2}
  &=&
  -
  \big(\,
    \overline{\psi}
    \,\Gamma_{a_1 a_2}\,
    \psi
  \big)
  \\
  \mathrm{d}
  \,
  e_{a_1 \cdots a_5}
  &=&
  +
  \big(\,
    \overline{\psi}
    \,\Gamma_{a_1 \cdots a_5}
    \psi
  \big)
  \\[+4pt]
  \mathrm{d} \, \DFSpinor
  &=&
  \paramOne
  \,
  \Gamma_a \psi\, e^a
  +
  \paramTwo
  \,
  \Gamma^{a_1 a_2} \psi
  \,
  e_{a_1 a_2}
  +
  \paramFive
  \,
  \Gamma^{a_1 \cdots a_5} \psi
  \,
  e_{a_1 \cdots a_5}\;,
  \end{array}
\end{equation}
for any triple of parameters 
$
 \paramOne
 \,,
 \paramTwo
 \,,
 \paramFive
 \, \in\, \mathbb{R}$ satisfying 
\begin{equation}
\label{RelationBetweenFactorsOfSummandsInDifferentialOfExceptionalFermion}
  \paramOne
  + 10 \cdot 
  \paramTwo
  - 
  6! \cdot 
  \paramFive 
  = 0
  \,.
\end{equation}
\end{proposition}
\begin{proof}
Direct inspection shows that the only non-trivial condition to check is $\mathrm{d}^2 \, \DFSpinor = 0$. For that we get with \eqref{ExceptionalCEDifferentialOnEta}:
\begin{equation}
  \label{ConditionForVanishingOfDSquareOnEta}
  -
  \,
  \mathrm{d}^2
  \DFSpinor
  \;\;
  =
  \;\;
  \paramOne
  \,
  \Gamma_a \psi
  \,
  \big(\,
    \overline{\psi}
    \,\Gamma^a\,
    \psi
  \big)
  \;-\;
  \paramTwo
  \,
  \Gamma_{a_1 a_2} \psi
  \,
  \big(\,
    \overline{\psi}
    \,\Gamma^{a_1 a_2}\,
    \psi
  \big)
  \;+\;
  \paramFive
  \,
  \Gamma_{a_1 \cdots a_5} \psi
  \big(\,
    \overline{\psi}
    \,\Gamma^{a_1 \cdots a_5}\,
    \psi
  \big)
  \,.
\end{equation}
By the general cubic Fierz identities \eqref{GeneralCubicFierzIdentities}, this expression vanishes if and only if the following system of equations holds:

\vspace{-6mm} 
$$
  \def\arraystretch{1.9}
  \def\arraycolsep{2pt}
  \begin{array}{crcrcrl}
    \paramOne
    &
    \tfrac{1}{11}
    \,
    \Gamma^a 
    \Gamma_a
    \Xi^{(32)}
    &
    -
    \paramTwo
    &
    \tfrac{1}{11}
    \,
    \Gamma^{a_1 a_2}
    \Gamma_{a_1 a_2}
    \Xi^{(32)}
    &
    +
    \paramFive
    &
    \tfrac{-1}{77}
    \,
    \Gamma^{a_1 \cdots a_5}
    \Gamma_{a_1 \cdots a_5}
    \Xi^{(32)}
    &
    \;\;
    =
    \;\;
    0
    \\
    \paramOne
    &
    \Gamma^a 
    \,
    \Xi_a^{(320)}
    &
    -
    \paramTwo
    &
    \tfrac{-2}{9}
    \,
    \Gamma^{a_1 a_2}
    \Gamma_{[a_1}
    \Xi^{(320)}_{a_2]}
    &
    +
    \paramFive
    &
    \tfrac{5}{9}
    \,
    \Gamma^{a_1 \cdots a_5}
    \Gamma_{[a_1 \cdots a_4}
    \Xi^{(320)}_{a_5]}
    &
    \;\;=\;\;
    0
    \\
    &
    &
    -
    \paramTwo
    &
    \Gamma^{a_1 a_2}
    \Xi^{(1408)}_{a_1 a_2}
    &
    +
    \paramFive
    &
    2
    \,
    \Gamma^{a_1 \cdots a_5}
    \Gamma_{[a_1 a_2 a_3}
    \Xi^{(1408)}_{a_4 a_5]}
    &
    \;\;=\;\;
    0
    \\
    &&&&
    \phantom{+}
    \paramFive
    &
    \Gamma^{a_1 \cdots a_5}
    \Xi^{(4224)}_{a_1 \cdots a_5}
    &
    \;\;=\;\;
    0
    \,.
  \end{array}
$$
Here the last three equations turn out to hold identically (checked in \cite{AncillaryFiles}) for all values of $\paramOne,\, \paramTwo,\,\paramFive$, by the irreducibility of the representations $\Xi$ \eqref{TheHigherTensorSpinors}. On the other hand, the first line is equivalently the claimed condition \eqref{RelationBetweenFactorsOfSummandsInDifferentialOfExceptionalFermion}.
\end{proof}

We consider the following parametrization of those solutions of \eqref{RelationBetweenFactorsOfSummandsInDifferentialOfExceptionalFermion}
for which $\gamma_1 \neq 0$ (in which case we absorb $\gamma_1$ into a rescaling of $\DFSpinor$ and hence assume without essential restriction that $\gamma_1 = 1$, all following \cite[(21)]{BDIPV04}):
\begin{equation}
  \label{ParameterizationByS}
  \hspace{.5cm}
  \left(\!
  \def\arraystretch{1}
  \def\arraycolsep{3pt}
  \begin{array}{rcl}
    \delta(s)
    &=&
    2
    \;
    (1 + s)
    \\
    \gamma_1(s)
    &=&
    1
    \\
    \gamma_2(s)
    &=&
    2
    \,
    \big(
      \tfrac{1}{5!}
      +
      \tfrac{s}{6!}
    \big)
  \end{array}
 \! \right),
  \hspace{1cm}
  s
  \in 
  \mathbb{R}\,.
\end{equation}
The single remaining solution (up to rescaling of $\DFSpinor$) with $\gamma_1 = 0$ may be understood as the re-scaled limit $s \to \infty$ of this parameterization.

\begin{definition}[\bf Hidden M-Algebra]
  \label{SuperExceptionalExtensionPfSuperMinkowski}
  We write $\widehat{\mathfrak{M}}$
  for the super Lie algebra whose CE-algebra obtained in Prop. \ref{ExistenceOfSuperExceptionalAlgebra} parametrized as in \eqref{ParameterizationByS}, and consider it fibered over super-Minkowski spacetime via
  \begin{equation}
    \label{SuperExceptionalProjection}
    \begin{tikzcd}[row sep=-2pt, column sep=small]
      \widehat{\mathfrak{M}}
      \ar[
        rr,
        ->>,
        "{
          \phi_{\mathrm{ex}}
        }"
      ]
      &&
      \mathbb{R}^{1,10\vert \mathbf{32}}
      \\
      \mathrm{CE}\big(
        \mathbb{R}^{1,10\vert \mathbf{32}}
      \big)
      \ar[
        from=rr,
        hook',
        "{
          \phi^\ast_{\mathrm{ex}}
        }"{swap}
      ]
      &&
      \mathrm{CE}\big(\,
        \widehat{\mathfrak{M}}
      \,\big)
      \\
      e^a &\longmapsfrom& e^a
      \\
      \psi^\alpha &\longmapsfrom& \psi^\alpha
      \mathrlap{\,.}
    \end{tikzcd}
  \end{equation}

Concretely, 
$\widehat{\mathfrak{M}}$ has underlying vector space spanned by
\vspace{1mm}  
\begin{equation}
  \label{GeneratorsOfTheSuperExceptionalLieAlgebra}
  \widehat{\mathfrak{M}}
  \;\;\simeq\;\; \mathbb{\FR}
  \Big\langle 
    \grayunderbrace{
    \big(
      P_a
    \big)_{a = 0}^{10}
    \,,\,
    \big(
      Z_{a_1 a_2}
      =
      Z_{[a_1 a_2]}
    \big)_{a_i = 0}^{10}
   \, ,\,
    \big(
      Z_{a_1 \cdots a_5}
      =
      Z_{[a_1 \cdots a_5]}
    \big)_{a_i = 0}^{10}
    }{ 
      \mathrm{deg} \,=\, (0,\mathrm{evn}) }
   \, ,\,
    \grayunderbrace{
    \big(
      Q_\alpha
    \big)_{\alpha=1}^{32}
    ,\,
    \big(
      O_\alpha
    \big)_{\alpha=1}^{32}
    }{ 
      \mathrm{deg} \,=\, (0,\mathrm{odd}) }
  \Big\rangle
\end{equation}
and the non-trivial Lie brackets between these basis elements are found -- by translating \eqref{ExceptionalCEDifferentialOnEta} via \eqref{RelationBetweenStructureConstants} -- to be:
\begin{equation}
  \label{BracketsInTheSuperExceptionalLieAlgebra}
  \def\arraystretch{1.6}
  \begin{array}{ccl}
    \big[
      Q_\alpha
      ,\,
      Q_\beta
    \big]
    &=&
    -\;
    2 \, 
    \Gamma^a_{\alpha \beta} 
    \, 
    P_a
    \;+\;
    2 \, 
    \Gamma^{a_1 a_2}_{\alpha \beta} 
    \, 
    Z_{a_1 a_2}
    \;-\;
    2 \, 
    \Gamma
      ^{a_1 \cdots a_5}
      _{\alpha \beta} 
    \, 
    Z_{a_1 \cdots a_5}
    \\
    \big[
      P_a
      ,\,
      Q_\alpha
    \big]
    &=&
    \paramOne\,
    \,
    \Gamma_a{}^\beta{}_\alpha
    O_\beta
    \\
    \big[
      Z_{a_1 a_2}
      ,\,
      Q_\alpha
    \big]
    &=&
    \paramTwo
    \,
    \Gamma
      _{a_1 a_2}
      {}^{\beta}{}_\alpha
    O_\beta
    \\
    \big[
      Z_{ a_1 \cdots a_5 }
      ,\,
      Q_\alpha
    \big]
    &=&
    \paramFive
    \,
    \Gamma
      _{ a_1 \cdots a_5 }
      {}^{\beta}{}_\alpha
    O_\beta\;.
  \end{array}
\end{equation}
\end{definition}

\begin{remark}[\bf History and literature]
$\,$

\noindent {\bf (i)} 
For a couple of special parameter values $s$ \eqref{ParameterizationByS} this is the ``hidden'' super-Lie algebra of \cite[Table 4]{DF82}; the general form appears in \cite[(1.2-4)]{BDPV05} following \cite{BDIPV04}, while the first line by itself -- disregarding the extra fermionic generators $O_\beta$ -- was independently considered in \cite[(13)]{Townsend95}\cite[(1)]{Townsend98} by analogy with other centrally-extended supersymmetry algebras. 

\vspace{0.5mm} 
\noindent {\bf (ii)}  The term ``M-algebra'' was coined by \cite{Sezgin97} for another extension of the first line in \eqref{BracketsInTheSuperExceptionalLieAlgebra} but has since come to be used (e.g. in \cite{BDPV05}) to refer to the first line itself (within the super-Poincar{\'e} algebra).

\vspace{0.5mm} 
\noindent {\bf (iii)} 
Note that these and authors following them (\cite{BDPV05}\cite[\S 5]{Azcarraga05}\cite{FIdO15}\cite{ADR16}\cite{ADR17}\cite{Ravera21}\cite{AndrianopoliDAuria24}) tend to speak of a ``super-group'' instead of just a super-Lie algebra, without however stating the super-Lie group structure. We construct this in Ex. \ref{IntegratingSuperExceptionalMinkowsliLieAlgebra} below.
\end{remark}

Of course, upon setting to zero the generators $Z_{a_1 a_2}$, $Z_{a_1 \cdots a_5}$ and $O_\alpha$, \eqref{GeneratorsOfTheSuperExceptionalLieAlgebra} reduces to the ordinary super-Minkowski Lie algebra, see Ex. \ref{SuperMinkowskiGroupProbedBySuperPoints} below where we warm up with revisiting the Lie integration of this familiar case.

\newpage 
\begin{remark}[\bf Trinary bracket in super-exceptional Lie algebra]
A key difference between the super-exceptional Lie algebra \eqref{BracketsInTheSuperExceptionalLieAlgebra} and the ordinary super-Minkowski Lie algebra \eqref{TheBifermionicSuperBracket}, for the purpose of their Lie integration (\S\ref{AsASuperlieGroup}),
is that the former has a non-vanishing trilinear super-bracket:
\vspace{-1mm} 
\begin{equation}
  \label{BracketOfThreeQs}
  \def\arraystretch{1.9}
  \begin{array}{rcl}
  \big[
    Q_\gamma
    ,\,
    [
      Q_\alpha
      ,\,
      Q_\beta
    ]
  \big]
  &
  =
  &
  \Big[
    Q_\gamma
    \;,\;
    -2
    \,
    \Gamma^a_{\alpha\beta}
    \, 
    P_a
    \;+2\,
    \,
    \Gamma^{a_1 a_2}_{\alpha\beta}
    \, 
    Z_{a_1 a_2}
    \;-2\,
    \,
    \Gamma^{a_1 \cdots a_5}_{\alpha\beta}
    \, 
    Z_{a_1 \cdots a_5}
  \Big]
  \\
  &=&
  \underbrace{
  2
  \Big(
  \paramOne
  \, 
  \Gamma^a_{\alpha \beta}
  \Gamma
    _a
    {}^\delta{}_\gamma
  \;-\,
  \paramTwo
  \, 
  \Gamma
    ^{a_1 a_2}
    _{\alpha \beta}
  \Gamma
    _{a_1 a_2}
    {}^\delta{}_\gamma
  \;+\,
  \paramFive
  \, 
  \Gamma
    ^{a_1 \cdots a_5}
    _{\alpha \beta}
  \Gamma
    _{a_1 \cdots a_5}
    {}^\delta{}_\gamma
  \Big)
  }_{\color{gray}
    =:
    \,
    [QQQ]
      ^\delta
      _{\gamma\alpha\beta}
  }
  O_\delta
  \,.
  \end{array}
\end{equation}
When the parameter \eqref{ParameterizationByS}
takes the special value $s = 0$ (cf. \S\ref{CaseS=0}), then equation \eqref{BracketOfThreeQs} simplifies to
\begin{equation}
  \label{BracketOfThreeAtSEqualsZero}
  \mathllap{
  s = 0
  \hspace{1.2cm}
  \Rightarrow
  \hspace{1.2cm}
  }
  \big[
    Q^\gamma
    ,\,
    [
      Q_\alpha
      ,\,
      Q_\beta
    ]
  \big]
  \;\;
  =
  \;\;
  64
  \,
  \big(
  \delta^\gamma_\beta 
  \, O_{\alpha}
  +
  \delta^\gamma_\alpha 
  \, O_{\beta}
  \big)
  \,,
\end{equation}
because (by a standard argument, e.g. \cite[(3.65)]{FreedmanVOnProeyen12})
\begin{equation}
  \label{CompletenessRelation}
  \def\arraystretch{2}
  \begin{array}{rcll}
  \delta_\alpha^\delta
  \delta_\gamma^\beta
  &=&
    \tfrac{1}{32}
    \sum_{p = 0}^5
    \;
    \tfrac{
      (-1)^{p(p-1)/2}
    }{ p! }
    \,
    \mathrm{Tr}
    \Big(
      \delta_\alpha^\bullet
      \delta^\beta_\bullet
      \cdot
      \Gamma_{a_1 \cdots a_p}
    \Big)
    \,
    (
      \Gamma^{a_1 \cdots a_p}
    )^\delta{}_\gamma
  &
  \proofstep{
    by \eqref{CliffordExpansionOfEndomorphismOf32}
  }
  \\
  &=&
    \tfrac{1}{32}
    \sum_{p = 0}^5
    \;
    \tfrac{
      (-1)^{p(p-1)/2}
    }{ p! }
    \,
    \Big(
      \delta_\alpha^{\delta'}
      \delta_{\gamma'}^\beta
      (
        \Gamma_{a_1 \cdots a_p}
      )^{\gamma'}{}_{\delta'}
    \Big)
    (
      \Gamma^{a_1 \cdots a_p}
    )^\delta{}_\gamma
  \\
  &=&
    \tfrac{1}{32}
    \sum_{p = 0}^5
    \;
    \tfrac{
      (-1)^{p(p-1)/2}
    }{ p! }
    \,
      (
        \Gamma_{a_1 \cdots a_p}
      )^{\beta}{}_{\alpha}
    \,
    (
      \Gamma^{a_1 \cdots a_p}
    )^\delta{}_\gamma
    \mathrlap{\,,}
  \end{array}
\end{equation}
which upon lowering spinor-indices with the spinor-metric $\eta_{\alpha\beta}$ \eqref{TheSpinorMetric} and symmetrizing the indices gives:

\vspace{-.55cm}
\begin{equation}
  \label{MetricFormOfFierzExpansion}
  \hspace{-4mm} 
  \def\arraystretch{1.7}
  \begin{array}{lll}
  \eta_{\delta(\alpha}
  \eta_{\beta)\gamma}
  &=\;\;
  \tfrac{1}{2}
  \big(
  \eta_{\delta\alpha}
  \eta_{\beta\gamma}
  +
  \eta_{\delta\beta}
  \eta_{\alpha\gamma}
  \big)
  \\
  &=\;\;
  \tfrac{1}{32}
  \Big(
    (\Gamma_a)_{\alpha\beta}
    (\Gamma^a)_{\gamma\delta}
    -
    \tfrac{1}{2}
    (\Gamma_{a_1 a_2})_{\alpha\beta}
    (\Gamma^{a_1 a_2})_{\gamma\delta}
    +
    \tfrac{1}{5!}
    (\Gamma_{a_1 \cdots a_5})_{\alpha\beta}
    (\Gamma^{a_1 \cdots a_5})_{\gamma\delta}
  \Big)
  &
  \proofstep{
    by
    \eqref{CompletenessRelation}
    \& 
    \eqref{SymmetricSpinorPairings}
    \eqref{SkewSpinorPairings}
  }
  \\
  & \underset{
    \mathclap{s=0}
  }{=} \;\;
  \tfrac{1}{64}
  \Big(
    \paramOne
    (\Gamma_a)_{\alpha\beta}
    (\Gamma^a)_{\gamma\delta}
    -
    \paramTwo
    (\Gamma_{a_1 a_2})_{\alpha\beta}
    (\Gamma^{a_1 a_2})_{\gamma\delta}
    +
    \paramFive
    (\Gamma_{a_1 \cdots a_5})_{\alpha\beta}
    (\Gamma^{a_1 \cdots a_5})_{\gamma\delta}
  \Big)
  &
  \proofstep{
    by
    \eqref{ParameterizationByS}
    .
  }
  \end{array}
\end{equation}
Therefore, for general $s \in \mathbb{R}$ the expression \eqref{BracketOfThreeQs} may 
equivalently be re-rewritten as
\begin{equation}
  \label{BracketOfThreeQsRewritten}
  \big[
     Q^\gamma
     ,\,
     [
       Q_\alpha
       ,\,
       Q_\beta
     ]
  \big]
  \;\;
  =
  \;\;
  65
  \big(
    \delta^\gamma_\alpha
    O_\beta
    +
    \delta^\gamma_\beta
    O_\alpha
  \big)
  \,+\,
  \big(
  4s
  \,
  \Gamma^a_{\alpha \beta}
  \Gamma_{a}
    ^{\gamma\delta}
  +
  \tfrac{4s}{6!}
  \Gamma^{a_1 \cdots a_5}_{\alpha \beta}
  \Gamma
    _{a_1 \cdots a_5}
    ^{\gamma\delta}
  \big) O_\delta
  \,.
\end{equation}
\end{remark}

\medskip

\noindent
{\bf Automorphy of the hidden M-algebra.} While the hidden extension breaks the $\mathrm{GL}(32)$-automorphy of the basic M-algebra 
(Prop. \ref{ManifestGL32})
to $\mathrm{Pin^+}(1,10)$, that action is still interesting, as it captures the ``parity'' symmetry of the C-field in 11D SuGra under spatial reflection (following \cite[Prop. 4.26]{FSS20-HigherT}):

\begin{example}[\bf Parity symmetry/MO9-orientifolding]
  \label{ParitySymmetry}
  Consider the Clifford generator $\Gamma^{10}$ as an element 
  $$
    g 
      \,\defneq\, 
    \Gamma 
      \,\defneq\, 
      \big(
        \Gamma_{\ten}
          {}^\alpha
          {}_\beta
        \big)
     \;\in\;
     \mathrm{GL}(32)
     \,.
  $$ 
  Note that as such it is in fact a conformal symplectic transformation, $\Gamma \in \mathrm{CSp}(32)$ \eqref{ConformalSymplecticGroup} for the spinor metric, since it preserves the spinor pairing up to a sign:
  \begin{equation}
    \label{SpinorMetricPreservedUpToSignByReflection}
    \big(\hspace{1pt}
      \overline{
        \Gamma^{10}
        \psi
      }
      \,
      \Gamma^{10}\phi
    \big)
    \underset{
      \scalebox{.7}{
        \eqref{SkewSelfAdjointnessOfCliffordGenerators}
      }
    }{
      \;=\;
    }
    \big(\hspace{1pt}
      \overline{
        \psi
      }
      (-\Gamma^{10})
      \,
      \Gamma^{10}\phi
    \big)
    \underset{
      \scalebox{.7}{
        \eqref{CliffordDefiningRelation}
      }
    }{
    \;=\;
    }
    -
    \big(\,
      \overline{\psi}
      \,
      \phi
    \big)
    \,,
  \end{equation}
  as equivalently seen in components:
  $$
    \def\arraystretch{.9}
    \begin{array}{l}
      \Gamma^\alpha{}_{\alpha'}
      \,
      \eta_{\alpha \beta}
      \,
      \Gamma^{\beta}{}_{\beta'}
      \underset{
       \scalebox{.7}{
         \eqref{LoweringOfSpinorIndices}
       }
      }{
        \;=\;
      }
      \Gamma_{\beta \alpha'}
      \,
      \Gamma^{\beta}{}_{\beta'}
      \underset{
        \scalebox{.7}{
          \eqref{SymmetryOfCliffordBasisElements}
        }
      }{
        \;=\;
      }
      \Gamma_{\alpha'\beta}
      \,
      \Gamma^{\beta}{}_{\beta'}
      \underset{
       \scalebox{.7}{
         \eqref{LoweringOfSpinorIndices}
       }
      }{
        \;=\;
      }
      \eta_{\alpha \alpha'}
      \Gamma^{\alpha}{}_{\beta}
      \,
      \Gamma^{\beta}{}_{\beta'}
      \underset{
       \scalebox{.7}{
         \eqref{CliffordDefiningRelation}
       }
      }{
        \;=\;
      }
      \eta_{\alpha \alpha'}
      \delta^\alpha_{\beta'}
      \;=\;
      \eta_{\beta' \alpha'}
      \underset{
        \scalebox{.7}{
          \eqref{SkewSymmetryOfSpinorMetric}
        }
      }{
        \;=\;
      }
      -
      \eta_{\alpha' \beta'}
      \,.
    \end{array}
  $$
  To see its action on the 
  bosonic generators note that
  $$
    \Gamma^\alpha{}_{\alpha'}
    \,e^{\alpha' \beta'}\,
    \Gamma^{\beta}{}_{\beta'}
    \underset{
      \scalebox{.7}{
       \eqref{LoweringOfSpinorIndices}
      }
    }{
      \;=\;
    }
    -
    \Gamma^\alpha{}_{\alpha'}
    \,e^{\alpha'}{}_{\beta'}\,
    \Gamma^{\beta \beta'}
    \underset{
      \scalebox{.7}{
        \eqref{SymmetryOfCliffordBasisElements}
      }
    }{
      \;=\;
    }
    -
    \Gamma^\alpha{}_{\alpha'}
    \,e^{\alpha'}{}_{\beta'}\,
    \Gamma^{\beta'\beta}
    \,,
  $$
  whence the vector components of the bispinorial $e^{\alpha\beta}$ are mapped as:
  $$
   \def\arraystretch{1.5}
   \begin{array}{ccccl}
      \Gamma_a \, e^a 
      &\longmapsto&
      -
      \big(
        \Gamma_{\ten}
        \cdot
        \Gamma_a
        \cdot
        \Gamma_{\ten}
      \big)
      e^a
      &=&
      +
      \sum_{a \neq 10}
      \,
      \Gamma_a \, e^a 
      -
      \Gamma_{\ten} \, e^{10}
      \\
      \Gamma_{a_1 a_2}
      \,
      e^{a_1 a_2}
      &\longmapsto&
      -
      \big(
        \Gamma_{\ten}
        \cdot
        \Gamma_{a_1 a_2}
        \cdot
        \Gamma_{\ten}
      \big)
      e^{a_1 a_2}      
      &=&
      -
      \sum_{a_i \neq 10}
      \Gamma_{a_1 a_2}
      \, e^{a_1 a_2}
      \,+\,
      2
      \sum_{a \neq 10}
      \Gamma_{a\, 10}
      e^{a\, 10}
      \\
      \Gamma_{a_1 \cdots a_5}
      \,
      e^{a_1 \cdots a_5}
      &\longmapsto&
      -
      \big(
        \Gamma_{\ten}
        \cdot
        \Gamma_{a_1 \cdots a_5}
        \cdot
        \Gamma_{\ten}
      \big)
      e^{a_1 \cdots a_5}      
      &=&
      +
      \sum_{a_i \neq 10}
      \Gamma_{a_1 \cdots a_5}
      \, e^{a_1 \cdots a_5}
      \,+\,
      5
      \sum_{a \neq 10}
      \Gamma_{a_1 \cdots a_4 \, 10}
      e^{a_1 \cdots a_4 \, 10}
      .
    \end{array}
  $$
  In fact, this is an automorphism of the hidden M-algebra for all values of the parameter $s$, as verified by the following computations:
  $$
    \begin{tikzcd}[row sep=small, column sep=large]
      \psi
      \ar[
        rr,
        |->,
        "{ \Gamma_{\ten} }"
      ]
      \ar[
        dd,
        |->,
        "{ \mathrm{d} }"
      ]
      &&
      \Gamma_{\ten}\psi
      \ar[
        dd,
        |->,
        "{ \mathrm{d} }"
      ]
      \\
      \\
      0 
      \ar[
        rr,
        |->,
        "{ \Gamma_{\ten} }"
      ]
        && 
      0
    \end{tikzcd}
  $$
  $$
    \hspace{-1cm}
    \begin{tikzcd}[
      ampersand replacement=\&,
      sep=9pt
    ]
      e^a
      \ar[
        rrr,
        |->,
        "{ \Gamma_{\ten} }"
      ]
      \ar[
        ddd,
        |->,
        "{ \mathrm{d} }"
      ]
      \&\&\&[-50pt]
      \left\{
      \hspace{-3pt}
      \def\arraycolsep{3pt}
      \def\arraystretch{1}
      \begin{array}{ll}
        + e^a & 
        \mathrlap{\mbox{for $a \neq 10$}}
        \\
        - e^a & 
        \mathrlap{\mbox{otherwise}}
      \end{array}
      \right.
      \ar[
        dd,
        |->,
        "{ \mathrm{d} }"
      ]
      \\
      \\
      \&\&\&
      \left\{
      \hspace{-3pt}
      \def\arraycolsep{3pt}
      \def\arraystretch{1.3}
      \begin{array}{ll}
        +
        \big(\hspace{1pt}
          \overline{\psi}
          \,\Gamma^a\,
          \psi
        \big)
        & 
        \mathrlap{\mbox{for $a \neq 10$}}
        \\
        -
        \big(\hspace{1pt}
          \overline{\psi}
          \,\Gamma^a\,
          \psi
        \big)
        & 
        \mathrlap{\mbox{otherwise}}
      \end{array}
      \right.
      \\
      \big(\hspace{1pt}
        \overline{\psi}
        \,\Gamma^a\,
        \psi
      \big)
      \ar[
        rr,
        |->,
        "{ \Gamma_{\ten} }"
      ]
      \&\&
      \big(\hspace{1pt}
        \overline{\Gamma_{\ten}\psi}
        \,\Gamma^a\,
        \Gamma_{\ten}\psi
      \big)      
      \ar[
        ur,
        equals
      ]
    \end{tikzcd}
    \hspace{2cm}
    \begin{tikzcd}[
      ampersand replacement=\&,
      sep=9pt
    ]
      e^{a_1 a_2}
      \ar[
        rrr,
        |->,
        "{ \Gamma_{\ten} }"
      ]
      \ar[
        ddd,
        |->,
        "{ \mathrm{d} }"
      ]
      \&\&\&[-55pt]
      \left\{
      \hspace{-3pt}
      \def\arraycolsep{3pt}
      \def\arraystretch{1}
      \begin{array}{ll}
        - e^{a_1 a_2} & 
        \mathrlap{\mbox{for $a_i \neq 10$}}
        \\
        + e^{a_1 a_2} & \mathrlap{\mbox{otherwise}} 
      \end{array}
      \right.
      \ar[
        dd,
        |->,
        "{ \mathrm{d} }"
      ]
      \\
      \\
      \&\&\&
      \left\{
      \hspace{-3pt}
      \def\arraycolsep{3pt}
      \def\arraystretch{1.3}
      \begin{array}{ll}
        +
        \big(\hspace{1pt}
          \overline{\psi}
          \,\Gamma^{a_1 a_2}\,
          \psi
        \big)
        & 
        \mathrlap{\mbox{for $a_i \neq 10$}}
        \\
        - 
        \big(\hspace{1pt}
          \overline{\psi}
          \,\Gamma^{a_1 a_2}\,
          \psi
        \big)
        & 
        \mathrlap{\mbox{otherwise}}
      \end{array}
      \right.
      \\
      \mathllap{-}
      \big(\hspace{1pt}
        \overline{\psi}
        \,\Gamma^{a_1 a_2}\,
        \psi
      \big)
      \ar[
        rr,
        |->,
        "{ \Gamma_{\ten} }"
      ]
      \&\&
      -
      \big(\hspace{1pt}
        \overline{\Gamma_{\ten}\psi}
        \,\Gamma^{a_1 a_2}\,
        \Gamma_{\ten}\psi
      \big)      
      \ar[
        ur,
        equals
      ]
    \end{tikzcd}
  $$
  \vspace{.2cm}
  $$
    \hspace{-1.5cm}
    \begin{tikzcd}[
      ampersand replacement=\&,
      sep=9pt
    ]
      e^{a_1 \cdots a_5}
      \ar[
        rrr,
        |->,
        "{ \Gamma_{\ten} }"
      ]
      \ar[
        ddd,
        |->,
        "{ \mathrm{d} }"
      ]
      \&\&\&[-60pt]
      \left\{
      \hspace{-3pt}
      \def\arraycolsep{3pt}
      \def\arraystretch{1}
      \begin{array}{ll}
        + e^{a_1 \cdots a_5} & 
        \mathrlap{\mbox{for $a_i \neq 10$}}
        \\
        - e^{a_1 \cdots a_5} & \mathrlap{\mbox{otherwise}} 
      \end{array}
      \right.
      \ar[
        dd,
        |->,
        "{ \mathrm{d} }"
      ]
      \\
      \\
      \&\&\&
      \left\{
      \hspace{-3pt}
      \def\arraycolsep{3pt}
      \def\arraystretch{1.3}
      \begin{array}{ll}
        +
        \big(\hspace{1pt}
          \overline{\psi}
          \,\Gamma^{a_1 \cdots a_5}\,
          \psi
        \big)
        & 
        \mathrlap{\mbox{for $a_i \neq 10$}}
        \\
        - 
        \big(\hspace{1pt}
          \overline{\psi}
          \,\Gamma^{a_1 \cdots a_5}\,
          \psi
        \big)
        & 
        \mathrlap{\mbox{otherwise}}
      \end{array}
      \right.
      \\
      \big(\hspace{1pt}
        \overline{\psi}
        \,\Gamma^{a_1 \cdots a_5}\,
        \psi
      \big)
      \ar[
        rr,
        |->,
        "{ \Gamma_{\ten} }"
      ]
      \&\&
      \big(\hspace{1pt}
        \overline{\Gamma_{\ten}\psi}
        \,\Gamma^{a_1 \cdots a_5}\,
        \Gamma_{\ten}\psi
      \big)      
      \ar[
        ur,
        equals
      ]
    \end{tikzcd}
  $$
  \vspace{.2cm}
  $$
    \begin{tikzcd}[row sep=small, column sep=large,
      ampersand replacement=\&
    ]
      \DFSpinor
      \ar[
        rrr,
        |->,
        "{ \Gamma_{\ten} }"
      ]
      \ar[
        rrr,
        |->,
        "{ \Gamma_{\ten} }"
      ]
      \ar[
        ddd,
        |->,
        "{ \mathrm{d} }"
      ]
      \&[-10pt]
      \&[-10pt]
      \&[-150pt]
      - \Gamma_{\ten}
      \DFSpinor
      \ar[
        dd,
        |->,
        "{ \mathrm{d} }"
      ]
      \\
      \\
      \&\&\&     
      \def\arraystretch{1.3}
      \def\arraycolsep{0pt}
      \begin{array}{ccl}
      &-&
      \paramOne
      \;\,
      \Gamma_{\ten}
      \Gamma_a\psi\,
      e^a
      \\
      &-&
      \paramTwo
      \,
      \Gamma_{\ten}
      \Gamma_{a_1 a_2}\psi\,
      e^{a_1 a_2}
      \\
      &-&
      \paramFive
      \,
      \Gamma_{\ten}
      \Gamma_{a_1 \cdots a_5}\psi\,
      e^{a_1 \cdots a_5}
      \end{array}
      \\
      \def\arraystretch{1.3}
      \def\arraycolsep{0pt}
      \begin{array}{ccl}
      &&
      \paramOne\;
      \,
      \Gamma_a\psi\,
      e^a
      \\
      &+&
      \paramTwo
      \,
      \Gamma_{a_1 a_2}\psi\,
      e^{a_1 a_2}
      \\
      &+&
      \paramFive
      \,
      \Gamma_{a_1 \cdots a_5}\psi\,
      e^{a_1 \cdots a_5}
      \end{array}
      \ar[
        rr,
        |->,
        "{ \Gamma_{\ten} }"
      ]
      \&\&
      \def\arraystretch{1.3}
      \def\arraycolsep{0pt}
      \begin{array}{ccl} 
        &&
        \paramOne\;
        \big(
        +
        \sum_{a \neq 0}
        \Gamma_a \Gamma_{\ten}\psi
        \, e^a
        -
        \Gamma_{\ten}\Gamma_{\ten}\psi
        \, e^{10}
        )
        \\
        &+&
        \paramTwo
        \big(
        -
        \sum_{a_i \neq 0}
        \Gamma_{a_1 a_2} \Gamma_{\ten}\psi
        \, e^{a_1 a_2}
        +
        2
        \Gamma_{a\, 10}
        \Gamma_{\ten}\psi
        \, e^{a\, 10}
        )
        \\
        &+&
        \paramFive
        \big(
        +
        \sum_{a_i \neq 0}
        \Gamma_{a_1 \cdots a_5} \Gamma_{\ten}\psi
        \, e^{a_1 \cdots a_5}
        -
        5
        \Gamma_{a_1 \cdots a_4\, 10}
        \Gamma_{\ten}\psi
        \, e^{a_1\cdots a_4 \, 10}
        )
        \ar[
          ur,
          equals
        ]
      \end{array}
    \end{tikzcd}
  $$
   This reflection automorphism acts by sign inversion on $G_4$:
   $$
     \def\arraystretch{1.4}
     \begin{array}{ccl}
       G_4   
       &\defneq&
       \tfrac{1}{2}
       \big(\hspace{1pt}
         \overline{\psi}
         \,\Gamma_{a_1 a_2}\,
         \psi
       \big)
       e^{a_1}\, e^{a_2}
       \\
       &
       \overset{
         \Gamma_{\ten}
       }{\mapsto}
       &
       \tfrac{1}{2}
       \sum_{a_i \neq 0}
       \big(\hspace{1pt}
         \overline{\Gamma_{\ten}\psi}
         \,\Gamma_{a_1 a_2}\,
         \Gamma_{\ten}\psi
       \big)
       e^{a_1}\, e^{a_2}
       -
       \big(\hspace{1pt}
         \overline{\Gamma_{\ten}\psi}
         \,\Gamma_{a_1 \, 10}\,
         \Gamma_{\ten}\psi
       \big)
       e^{a_1}\, e^{10}
       \\
       &=&
       -
       \tfrac{1}{2}
       \sum_{a_i \neq 0}
       \big(\hspace{1pt}
         \overline{\psi}
         \,\Gamma_{a_1 a_2}\,
         \psi
       \big)
       e^{a_1}\, e^{a_2}
       -
       \big(\hspace{1pt}
         \overline{\psi}
         \,\Gamma_{a_1 \, 10}\,
         \psi
       \big)
       e^{a_1}\, e^{10}
       \\
       &=&
       -G_4
     \end{array}
   $$
   as well as on $\widehat{P}_3$ (by similar inspection)
   $$
     \widehat{P}_3
     \;\;\overset{\Gamma_{\ten}}{\longmapsto}\;\;
     -
     \widehat{P}_3
     \,,
   $$
   and hence must be understood \cite[\S 4.8]{FSS20-HigherT}
   as the ``parity symmetry'' of 11D SuGra (e.g. \cite[(2.2.29)]{DNP86}) or equivalently as the
   Ho{\v r}ava-Witten orientifolding (e.g. \cite[(3.1)]{Falkowski99}\cite[p 1-2]{Ovrut04}\cite[p 94]{Carlevaro06}), lifted from Minkowski spacetime $\mathbb{R}^{1,10}$ to $\widehat{\mathfrak{M}}$.
\end{example}

\noindent

%%%%%%%%%%%%%%%%%%%%%%%%%%%%%%%%%%%%%%%%
\subsubsection{The 3-form}
\label{The3Form}
%%%%%%%%%%%%%%%%%%%%%%%%%%%%%%%%%%%%%%%%

Next, we discuss the construction of the coboundary $\widehat{P}_3$ 
\eqref{CoboundaryRelationInIntro}
for the avatar super-flux density $G_4$ pulled back to $\widehat{\mathfrak{M}}$. The idea is that, by the nature of \eqref{ExceptionalCEDifferentialOnEta}, there are two evident elements in $\mathrm{CE}\big(\widehat{\mathfrak{M}}\big)$ whose differential contains $\phi_{\mathrm{ex}}^\ast G_4$ as a summand, namely $\tfrac{1}{2} e_{a_1 a_2} e^{a_1} e^{a_2}$ and $\tfrac{1}{2}\big(\, \overline{\psi} \,\Gamma_a\, \DFSpinor\big)e^a$:
$$
  \def\arraystretch{1.6}
  \def\tabcolsep{2pt}
  \begin{array}{rcl}
    \mathrm{d}
    \big(-
      \tfrac{1}{2}
      e_{a_1 a_2}
      e^{a_1} e^{a_2}
    \big)
    &=&
    \phi_{\mathrm{ex}}^\ast
    G_4
    +
    \cdots
    \\
    \mathrm{d}
    \big(
      \tfrac{1}{2}
      (\,
        \overline{\psi}
        \,
        \Gamma_a
        \,
        \DFSpinor
      )
      e^a
    \big)
    &=&
    \phi_{\mathrm{ex}}^\ast
    G_4
    +
    \cdots
    \,.
  \end{array}
$$
However, both of these expressions contain different further summands ``$\cdots$'', and a fairly rich correction term needs to be found to cancel these off against each other. The remarkable result of the following Prop. \ref{TheRestrictedAvatarOfTheBFieldFlux} is that such a correction term exists at all (this is originally due to \cite[(6.6)]{DF82} and in more generality due to \cite[(30)]{BDIPV04}; we aim to show the full computation in transparent form, as much as possible).

\begin{proposition}[\bf The hidden 3-form]
  \label{TheRestrictedAvatarOfTheBFieldFlux}
  For $s \in \mathbb{R} \setminus \{0\}$ {\rm  \eqref{ParameterizationByS}}, the left-invariant form 
  $
    \widehat{P}_3
    \,\in\,
    \mathrm{CE}\big(
      \widehat{\mathfrak{M}}
    \big)
  $
  on the hidden M-algebra {\rm (Def. \ref{SuperExceptionalExtensionPfSuperMinkowski})}
 given by
\begin{equation}
  \label{AnsatzForH30}
  \def\arraystretch{1.5}
  \begin{array}{lll}
    \mathllap{
      \widehat{P}_3
      \;:=\;
    }
    \phantom{+}\;\;
    \alphaZero
    \,
    e_{
      {\color{darkblue} a_1 } 
      {\color{darkorange} a_2}
    }
    \, e^{
      \color{darkblue} a_1
    } 
    \, 
    e^{
     \color{darkorange} a_2
    }
    \\
    \;+\;
    \alphaOne
    \, e^{\color{darkgreen}a_1}{}_{\color{darkblue}a_2}
    \, e^{\color{darkblue}a_2}{}_{\color{darkorange}a_3}
    \, e^{\color{darkorange}a_3}{}_{\color{darkgreen}a_1}
        & 
  \;+\;
    \betaOne
    \,
    \big(\,
      \overline{\psi}
      \,\Gamma
        _{\color{darkblue} a}
      \,
      \DFSpinor
    \big)
    e^{\color{darkblue}a}
    \\
    \;+\;
    \alphaTwo
    \, 
    e^{
      {
        \color{darkblue} a_1 \cdots a_4} 
        \color{darkorange} b_1
      }
    \, e_{
        \color{darkorange}b_1
      }{}^{
        \color{darkgreen}
        b_2
      }
    \, e_{
      { \color{darkgreen} b_2 }
      { \color{darkblue} a_1 \cdots a_4}
    }
&
 \;+\;
    \betaTwo 
    \,
    \big(\,
      \overline{\psi}
      \,\Gamma
        _{\color{darkblue}a_1 a_2}\,
      \DFSpinor
    \big)
    \,
    e^{\color{darkblue} a_1 a_2}
    \\
    \;+\;
    \alphaThree
    \,
    \epsilon_{
      { \color{darkblue} a_1 \cdots a_5 }
      { \color{darkorange} b_1 \cdots b_5 }
      { \color{darkgreen} c}
    }
    \,
    e^{
     \color{darkblue}
     a_1 \cdots a_5
    }
    e^{
     \color{darkorange}
     b_1 \cdots b_5
    }
    e^{
      \color{darkgreen}
      c
    }
    &
    \;+\;
    \betaThree 
    \,
    \big(\,
      \overline{\psi}
      \,\Gamma
        _{\color{darkblue}a_1 \cdots a_5}\,
      \DFSpinor
    \big)
    \,
    e^{
      \color{darkblue}
      a_1 \cdots a_5
    }
    \\
    \;+\;
    \alphaFour
    \,
    \epsilon_{
      { \color{darkblue} a_1 a_2 a_3 }
      { \color{darkorange} b_1 b_2 b_3 }
      { \color{darkgreen} c_1 \cdots c_5 }
    }
    \,
    e^{
      {\color{darkblue} a_1 a_2 a_3}
      {\color{purple} d_1 d_2}
    }
    \,
    e_{
      {\color{purple} d_1 d_2}      
     }
     {}
     ^{
      {\color{darkorange} b_1 b_2 b_3}
    }
    \,
    e^{
      \color{darkgreen}
      c_1 \cdots c_5\
    }
    % \\
    % \;+\;
    % \betaOne
    % \,
    % \big(\,
    %   \overline{\psi}
    %   \,\Gamma
    %     _{\color{darkblue} a}
    %   \,
    %   \DFSpinor
    % \big)
    % e^{\color{darkblue}a}
    % \\
    % \;+\;
    % \betaTwo 
    % \,
    % \big(\,
    %   \overline{\psi}
    %   \,\Gamma
    %     _{\color{darkblue}a_1 a_2}\,
    %   \DFSpinor
    % \big)
    % \,
    % e^{\color{darkblue} a_1 a_2}
    % \\
    % \;+\;
    % \betaThree 
    % \,
    % \big(\,
    %   \overline{\psi}
    %   \,\Gamma
    %     _{\color{darkblue}a_1 \cdots a_5}\,
    %   \DFSpinor
    % \big)
    % \,
    % e^{
    %   \color{darkblue}
    %   a_1 \cdots a_5
    % }
    % \mathrlap{\,,}
    % \\
    % \\
  \end{array}
\end{equation}
  satisfies the Bianchi-equation
  \begin{equation}
    \label{H30BianchiIdentity}
    \mathrm{d}
    \,
    \widehat{P}_3
    \;=\;
    \phi^\ast_{\mathrm{ex}}
    \,
    G_4
    \mathrlap{
    \;\;\;\;\;\;\;
    \in
    \;\;
    \mathrm{CE}\big(
      \widehat{\mathfrak{M}}
    \big)
    }
  \end{equation}
  if and only if its coefficients take the following values:
  \vspace{-1mm} 
\begin{equation}
  \label{RestrictedParametersOfH3FluxAvatar}
  \begin{tikzcd}[
    column sep=-8pt,
    row sep=-4pt
  ]
    \alphaZero
    &\;=\;&
    \phantom{+}\tfrac{1}{2}
    &
    \frac
      { -1 }
      {  5 }
    &[-7pt]
    \tfrac
      {
        6 + 2 s + s^2
      }
      {
        s^2
      }
    \\
    \alphaOne
    &\;=\;&
    \phantom{+}\tfrac{1}{2}
    &
    \tfrac{1}{15}
    &
    \frac
      {
        6 + 2 s
      }
      {
        s^2
      }
&[30pt]
    \betaOne
    &\;=\;&
    -1
    &
    \tfrac
      { 1 }
      { 10 \, \paramTwo }
    &
    \tfrac
      { 3 - 2s }
      { s^2 }
        \\
    \alphaTwo
    &\;=\;&
    \phantom{+}\tfrac{1}{2}
    &
    \tfrac{1}{6!}
    &
    \frac
      {
        (6 + s)^2
      }
      {
        s^2
      }
      & 
    \betaTwo
    &\;=\;&
    -1
    &
    \tfrac
      { 1 }
      { 20 \, \paramTwo }
    &
    \frac
      { 3 + s }
      { s^2 }    
    \\
    \alphaThree
    &\;=\;&
    \phantom{+}\tfrac{1}{2}
    &
    \tfrac
      {1}
      {5 \cdot 5! \cdot 6!}
    &
    \frac
      {(6 + s)^2}
      {s^2}
&
\betaThree
    &\;=\;&
    -1
    &
    \tfrac
      { 3 }
      { 10 \cdot 6! \cdot \paramTwo }
    &
    \frac
      { 6 + s }
      { s^2 }
    \\
    \alphaFour
    &\;=\;&
    \phantom{+}\tfrac{1}{2}
    &
    \tfrac{-1}{9 \cdot 5! \cdot 6!}
    &
    \frac
      { (6 + s)^2 }
      { s^2 }
      \,.
  \end{tikzcd}
\end{equation}
\end{proposition}

\vspace{-5mm} 
\begin{proof}
It is essentially straightforward to work out the differential of $\widehat{P}_3$ via \eqref{ExceptionalCEDifferentialOnEta} (cf. \cite[p. 134]{DF82}):
\label{DifferentialOfH30}

\vspace{-5mm}
$$
  \def\arraystretch{1.8}
  \begin{array}{l}
%  \mathllap{
  \mathrm{d}
  \,\widehat{P}_3
  \;=\;
 % }
  \\[-4pt]
  \alphaZero
  \Big(
  -
  \big(
    \overline{\psi}
    \,\Gamma_{a_1 a_2}\,
    \psi
  \big)
  \,
  e^{a_1}\, e^{a_2}
  \,-\,
  {\color{darkblue}
  2
  \,
  e_{a_1 a_2}
  \big(\,
    \overline{\psi}
    \,\Gamma^{a_1}\,
    \psi
  \big)
  e^{a_2}
  }
  \Big)
  \\
    \;\;\;+\;
  \alphaOne
  \Big(
    {\color{darkorange}
    -3
    \,
    \big(\,
      \overline{\psi}
      \,\Gamma^{a_1}{}_{a_2}\,
      \psi
    \big)
    e^{a_2}{}_{a_3}
    \,
    e^{a_3}{}_{a_1}
    }
  \Big)
  \\
  \;\;\;+\;
  \alphaTwo
  \Big(
  {
  \color{olive}
  2
  \big(\,
    \overline{\psi}
    \Gamma^{a_1 \cdots a_4 b_1}
    \psi
  \big)
  e_{b_1}{}^{b_2}
  e_{b_2 a_1 \cdots a_4}
  }
  {
  \color{orange}
  \,+\,
  \big(\,
    \overline{\psi}
    \Gamma_{b_1}{}^{\!b_2}
    \psi
  \big)
  e^{a_1 \cdots a_4 b_1}
  e_{b_2 a_1 \cdots a_4}
  }
  \Big)
  \\
  \;\;\;+\;
  \alphaThree
  \Big(
  {
  \color{blue}
  2
  \,
  \epsilon_{a_1 \cdots a_5 b_1 \cdots b_5 c}
  \big(\,
    \overline{\psi}
    \Gamma^{a_1 \cdots a_5}
    \psi
  \big)
  e^{b_1 \cdots b_5}
  e^c
  }
  {
  \color{brown}
  \,+\,
  \epsilon_{
    a_1 \cdots a_5
    b_1 \cdots b_5
    c
    }
  \big(\,
    \overline{\psi}
    \Gamma^c
    \psi
  \big)
  e^{a_1 \cdots a_5}
  e^{b_1 \cdots b_5}
  }
  \Big)
  \\
  \;\;\;+\;
  \alphaFour
  \Big(
  \underbrace{
  {
    \color{gray}
    2
    \,
    \epsilon_{
      a_1 a_2 a_3
      b_1 b_2 b_3
      c_1 \cdots c_5
      }
      \big(\,
        \overline{\psi}
        \,
        \Gamma^{a_1 a_2 a_3 d_1 d_2}
        \,
        \psi
      \big)
      e_{d_1 d_2}{}^{b_1 b_2 b_3}
      \,
      e^{c_1 \cdots c_5}
    }
  \,
  {
  \color{gray}
  +\,
  \epsilon_{
    a_1 a_2 a_3
    b_1 b_2 b_3
    c_1 \cdots c_5
  }
  \big(\,
    \overline{\psi}
    \Gamma^{c_1 \cdots c_5}
    \psi
  \big)
  e^{a_1 a_2 a_3 d_1 d_2}
  e_{d_1 d_2}{}^{b_1 b_2 b_3}
  }
  }_{
    \overset{
      \scalebox{.6}{
        \eqref{Mixed5IndexContractions}
      }
    }{
      =
    }
    \,
    \color{magenta}
    3
    \,
    \epsilon_{
      a_1 a_2 a_3
      b_1 b_2 b_3
      c_1 \cdots c_5
    }
    \big(\,
      \overline{\psi}
      \,
      \Gamma^{c_1 \cdots c_5}
      \,
      \psi
    \big)
    e^{a_1 a_2 a_3 d_1 d_1}
    \,
    e_{d_1 d_1}{}^{b_1 b_2 b_3}
  }
  \Big)
  \\
  \;\;\;+\;
  \betaOne
  \Big(
  \underbrace{
  \color{gray}
  \paramOne
  \,
  \big(\,
    \overline{\psi}
    \,\Gamma_{a}
    \Gamma_b
    \,\psi
  \big)
  \,
  e^a
  \, 
  e^b
  }_{
    \paramOne
    \big(\,
      \overline{\psi}
      \,
      \Gamma_{a b}
      \,
      \psi
    \big)
    e^a 
    e^b
  }
  \, 
  \;+\;
  \underbrace
  {\color{gray}  
  \paramTwo
  \big(\,
    \overline{\psi}
    \,\Gamma_a
    \,\Gamma_{b_1 b_2}\,
    \psi
  \big)
  \,
  e^a
  \, 
  e^{b_1 b_2}
  }_{
    \color{darkblue}
    2\paramTwo
    \big(\,
      \overline{\psi}
      \,
      \Gamma^b
      \,
      \psi
    \big)
    e_a \, 
    e^{a b}
  }
  \;+\;
  \underbrace
  {
  \color{gray}
  \paramFive
  \big(\,
    \overline{\psi}
    \,
    \Gamma_a 
    \Gamma_{b_1 \cdots b_5}
    \,
    \psi
  \big)
  e^a
  \,
  e^{b_1 \cdots b_5}
  }_{
  \mathclap{
  \color{blue}
  \tfrac
    {\paramFive}
    {5!}
  \,
  \epsilon_{
    a
    \,
    b_1 \cdots b_5
    \,
    c_1 \cdots c_5
  }
  \big(\,
    \overline{\psi}
    \,
    \Gamma_{
      c_1 \cdots b_5
    }
    \,
    \psi
  \big)
  e^a
  \,
  e^{b_1 \cdots b_5}  
  }
  }
  \\
  \hspace{.7cm}
  \;+\;
  {\color{darkgreen}
  \big(\,
    \overline{\DFSpinor}
    \,\Gamma_a\,
    \psi
  \big)
  \big(\,
    \overline{\psi}
    \,\Gamma^a\,
    \psi
  \big)
  }
  \Big)
  \\
    \;\;\;+\;
    \betaTwo
  \Big(
    \underbrace
    {\color{gray}
    \paramOne
    \,
    \big(\,
      \overline{\psi}
      \,\Gamma_{a_1 a_2}\,
      \Gamma_b\,
      \psi
    \big)
    \, e^{a_1 a_2}
    \,
    e^{b}
    }_{
    \color{darkblue}
    2
    \, 
    \paramOne
    \,
    \big(\,
      \overline{\psi}
      \,\Gamma_{a}\,
      \psi
    \big)
    \, e^{a b}
    \,
    e_{b}    
    }
    \;+\;
    \underbrace
    {\color{gray}
    \paramTwo
    \,
    \big(\,
      \overline{\psi}
      \,\Gamma_{a_1 a_2}\,
      \Gamma_{b_1 b_2}\,
      \psi
    \big)
    \,
    \, e^{a_1 a_2}
    \, 
    e^{b_1 b_2}
    }_{
      \color{darkorange}
    4\paramTwo
    \,
    \big(\,
      \overline{\psi}
      \,
      \Gamma
        _{a}
        {}
        ^{b}
      \,
      \psi
    \big)
    \,
    \, 
    e
      ^{a}
      {}
      _{c}
    \, 
    e^c{}_{b}
    }
    \;+\;
    \underbrace
    {
    \color{gray}
    \paramFive
    \big(\,
      \overline{\psi}
      \Gamma_{a_1 a_2}
      \Gamma_{b_1 \cdots b_5}
      \psi
    \big)
    e^{a_1 a_2}
    \,
    e^{b_1 \cdots b_5}
    }_{
      \color{olive}
    10
    \,
    \paramFive
    \big(\,
      \overline{\psi}
      \Gamma
        _{
          a 
          \,
          b_1 \cdots b_4
        }
      \psi
    \big)
    e^{a c}
    \,
    e{}_c{}^{b_1 \cdots b_4}
    }
    \\
    \hspace{.7cm}
    \;-\;
    {\color{darkgreen}
    \,
    \big(\,
      \overline{\DFSpinor}
      \,\Gamma_{a_1 a_2}\,
      \psi
    \big)
    \big(\,
      \overline{\psi}
      \,\Gamma^{a_1 a_2}\,
      \psi
    \big)
    }
    \Big)
    \\
    \;\;\;+\;
    \betaThree
    \Big(
    \underbrace{
    \color{gray}
    \paramOne
    \big(\,
      \overline{\psi}
      \Gamma_{a_1 \cdots a_5}
      \Gamma_b
      \psi
    \big)
    e^{a_1 \cdots a_5}
    \,
    e^b
    }_{
     \mathclap{
      \hspace{-15pt}
      \color{blue}
    \tfrac
      {\paramOne}
      {5!}
    \epsilon_{
      a_1 \cdots a_5
      \,
      b
      \,
      c_1 \cdots c_5
    }
    \big(\,
      \overline{\psi}
      \Gamma_{
        c_1 \cdots c_5
      }
      \psi
    \big)
    e^{a_1 \cdots a_5}
    \,
    e^b
    }
    }
    \,+\,
    \underbrace{
    \color{gray}
    \paramTwo
    \big(\,
      \overline{\psi}
      \Gamma_{a_1 \cdots a_5}
      \Gamma_{b_1 b_2}
      \psi
    \big)
    e^{a_1 \cdots a_5}
    \,
    e^{b_1 b_2}
    }_{
      \color{olive}
    \hspace{-13pt}
    10
    \,
    \paramTwo
    \big(\,
      \overline{\psi}
      \,
      \Gamma_{
        a_1 \cdots a_4
        \,
        b
      }
      \,
      \psi
    \big)
    e^{a_1 \cdots a_4 c}
    \,
    e
      _{c}
      {}
      ^{b}
    }
    \,+\,
    \underbrace{
    \color{gray}
    \paramFive
    \big(\,
      \overline{\psi}
      \Gamma_{a_1 \cdots a_5}
      \Gamma_{b_1 \cdots b_5}
      \psi
    \big)
    e^{a_1 \cdots a_5}
    \,
    e^{b_1 \cdots b_5}
    }_{
      \mathclap{
        \def\arraystretch{1.7}
        \begin{array}{l}
          \color{brown}
          \phantom{+}
          \paramFive
          \,
          \epsilon_{
            a_1 \cdots a_5
            b_1 \cdots b_5
            c
          }
          \big(\,
            \overline{\psi}
            \,\Gamma_c\,
            \psi
          \big)
          e^{a_1 \cdots a_5}
          \,
          e^{b_1 \cdots b_5}
          \\
          \hspace{-2.2cm}
          {
          \color{magenta}
          -
          \,
          \tfrac{
            200
          }{5!}
          \paramFive
          \,
          \epsilon_{
            a_1 a_2 a_3
            \,
            b_1 b_2 b_3
            \,
            c_1 \cdots c_5
          }
          \big(\,
          \overline{\psi}
          \,
          \Gamma
            ^{
              c_1 \cdots c_5
            }
          \,
          \psi
          \big)
          e^{
            a_1 a_2 a_3
            \,
            d_1 d_2
          }
          \,
          e _{d_1 d_2}
            {}
            ^{b_1 b_2 b_3}
          }
          \\
          \color{orange}
          +
          \,
          600
          \,
          \paramFive
          \big(\,
            \overline{\psi}
            \,
            \Gamma
              _{a}
              {}
              ^{b}
            \,
            \psi
          \big)
          e^{a c_1 \cdots c_4}
          \, 
          e_{c_1 \cdots c_4 b}
        \end{array}
      }
    }
    \\[-1.5cm]
    \hspace{.7cm}
    \;+\;
    {
    \color{darkgreen}
    \big(
      \overline{\DFSpinor}
      \,
      \Gamma_{a_1 \cdots a_5}
      \psi
    \big)
    \big(
      \overline{\psi}
      \Gamma^{a_1 \cdots a_5}
      \psi
    \big)
    }
    \Big)
    \,,
    \\[+1.2cm]
  \end{array}
$$

\vspace{-4mm} 
\noindent (where the equalities under the braces use,
unless otherwise indicated,
Clifford-Hodge duality
\eqref{HodgeDualityOnCliffordAlgebra}, 
Clifford expansion
\eqref{GeneralCliffordProduct}
and the symmetry properties \eqref{SymmetricSpinorPairings} \eqref{SkewSpinorPairings}
of the spinor pairings).

Therefore the Bianchi identity 
\eqref{H30BianchiIdentity}
holds if and only if the constants in \eqref{AnsatzForH30} satisfy the following system of linear equations:
\begin{equation}
  \label{EquationsSpecifyingRestrictedH30Coefficients}
  \mathrm{d}\, \widehat{P}_3
  \;=\;
  \tfrac{1}{2}
  \big(\,
    \overline{\psi}
    \,\Gamma_{a_1 a_2}\,
    \psi
  \big)
  \,
  e^{a_1}\, e^{a_2}
  \;\;\;\;
  \Leftrightarrow
  \;\;\;\;
  \left\{\!\!
  \def\arraystretch{1.1}
  \begin{array}{r}
    -\alphaZero 
    + 
    \paramOne\, \betaOne 
    \;=\;
    \tfrac{1}{2}
    \\
    \color{darkblue}
    -2
    \,
    \alphaZero
    \;+\;
    2
    \, 
    \paramTwo 
    \betaOne
    \;+\;
    2
    \,
    \paramOne 
    \betaTwo
    \;=\;
    0
    \\
    \color{darkorange}
    -3
    \,
    \alphaOne
    \;-\;
    4
    \paramTwo
    \betaTwo
    \;=\;
    0
    \\
    \color{olive}
    2 
    \, 
    \alphaTwo
    +
    10
    \,
    \paramFive
    \,
    \betaTwo
    +
    10
    \,
    \paramTwo
    \betaThree
    \;=\;
    0
    \\
    \color{orange}
    \alphaTwo
    \,+\,
    600
    \,
    \paramFive
    \,
    \betaThree
    \;=\;0
    \\
    {
    \color{blue}
    2 
    \, 
    \alphaThree
    +    
    \tfrac
      { \paramFive }
      { 5! }
    \,
    \betaOne
    +
    \tfrac
      {\paramOne}
      {5!}
    \,
    \betaThree
    \;=\;
    0
    }
    \\
    \color{brown}
    \alphaThree 
    \,+\,
    \paramFive
    \,
    \betaThree
    \;=\;
    0
    \\
    \color{magenta}
    3
    \,
    \alphaFour
    \,-\,
    \tfrac
      {200}
      {5!}
    \paramFive
    \,
    \betaThree
    \;=\;
    0
    \\
    \color{darkgreen}
    \betaOne
    +
    10
      \cdot
    \betaTwo
    -
    6! 
      \cdot 
    \betaThree
    \;=\;
    0
    \mathrlap{\,,}
  \end{array}
  \right.
\end{equation}
where the {\color{darkgreen}last line} follows as \eqref{RelationBetweenFactorsOfSummandsInDifferentialOfExceptionalFermion} from  \eqref{ConditionForVanishingOfDSquareOnEta}.

\noindent Using mechanical algebra, one checks \cite{AncillaryFiles} that these equations have the unique solution \eqref{RestrictedParametersOfH3FluxAvatar}, as claimed.
\end{proof}

\begin{remark}[\bf Comparison to the literature]
\label{ComputationInTheLiterature}
$\,$

\noindent {\bf (i)} Essentially, the system of equations \eqref{EquationsSpecifyingRestrictedH30Coefficients} was reported in \cite[(6.6)]{DF82} and in generality in \cite[footnote 7]{BDIPV04} --- except for our factor {\color{magenta}200/5!}, which there instead (after normalizing conventions) is a $5$. (Incidentally, our factor of {\color{blue}1/5!} is not shown in \cite[footnote 7]{BDIPV04} either, but does appear in \cite[(6.6iv)]{DF82} and later again in \cite[footnote 11]{BDPV05}.)

\vspace{.5mm}  
\noindent {\bf (ii)} Accordingly, the general solution \eqref{RestrictedParametersOfH3FluxAvatar} is essentially that reported in \cite[(30)]{BDIPV04}: The global prefactors of $\tfrac{1}{2}$ and $-1$ that we show in \eqref{RestrictedParametersOfH3FluxAvatar} are due to different normalization of $\mathrm{d} \DFSpinor$ \eqref{ExceptionalCEDifferentialOnEta} and are thus not substantial; similarly notice from \cite[(28)]{BDIPV04} that $\lambda$ in \cite[(30)]{BDIPV04} is our $-\alpha_0$, up to a global sign.

\vspace{.5mm}  
\noindent {\bf (iii)} But this leaves one small actual difference, namely in the sign of $\alpha_1$ in \cite[(30)]{BDIPV04} compared to our \eqref{RestrictedParametersOfH3FluxAvatar}.
Our sign comes out as shown because the second dark-orange term on p. \pageref{DifferentialOfH30} has an intrinsic sign difference to the first term, since 
$$
  {
  \color{darkorange}
  \big(\,
    \overline{\psi}
    \,
    \Gamma_a{}^b
    \,
    \psi
  \big)
  e^a{}_c \, e^c{}_b
  }
  \;=\;
  -
  {
  \color{darkorange}
  \big(\,
    \overline{\psi}
    \,
    \Gamma^{a_1}{}_{a_2}
    \,
    \psi
  \big)
  e^{a_2}{}_{a_3} 
  \, 
  e^{a_3}{}_{a_1}
  }
  \,.
$$
\end{remark}

\smallskip

Above, we used the following identity:

\begin{lemma}[\bf Mixed 5-index contractions]
In $\mathrm{CE}\big(\, \widehat{\mathfrak{M}}\, \big)$, we have the following relation:
\begin{equation}
  \label{Mixed5IndexContractions}
  \def\arraystretch{1.6}
  \begin{array}{cl}
  &
  \epsilon_{
    {\color{darkblue}
    a_1 a_2 a_3
    }
    {\color{darkgreen}
    b_1 b_2 b_3
    }
    {\color{darkorange}
    c_1 \cdots c_5
    }
  }
  \big(\,
    \overline{\psi}
    \,
    \Gamma^{
      {\color{darkblue}
      a_1 a_2 a_3 
      }
      {\color{purple}
      d_1 d_2
      }
    }
    \,
    \psi
  \big)
  \,
  e_{
    \color{purple}
    d_1 d_2 
  }{}^{
    \color{darkgreen}
    b_1 b_2 b_3
  }
  \,
  e^{
    \color{darkorange}
    c_1 \cdots c_5
  }
  \\
  =
  &
  \epsilon_{
    {\color{darkblue}
    a_1 a_2 a_3
    }
    {\color{darkgreen}
    b_1 b_2 b_3
    }
    {\color{darkorange}
    c_1 \cdots c_5
    }
  }
  \big(\,
    \overline{\psi}
    \,
    \Gamma^{
      \color{darkorange}
      c_1 \cdots c_5
    }
    \,
    \psi
  \big)
  e^{
    {
    \color{darkblue}
    a_1 a_2 a_3
    }
    {
    \color{purple}
    d_1 d_2
    }
  }
  \;
  e_{
    \color{purple}
    d_1 d_2
  }{}^{
    \color{darkgreen}
      b_1 b_2 b_3
    }.
  \end{array}
\end{equation}
\end{lemma}
\begin{proof}
$\,$

\vspace{-5mm} 
$$
  \def\arraystretch{1.8}
  \begin{array}{ll}
  \epsilon_{
    a_1 a_2 a_3
    b_1 b_2 b_3
    c_1 \cdots c_5
  }
  \big(\,
    \overline{\psi}
    \,
    \Gamma^{a_1 a_2 a_3 d_1 d_2}
    \,
    \psi
  \big)
  e
    _{d_1 d_2}
    {}
    ^{b_1 b_2 b_3}
  \,
  e^{c_1 \cdots c_5}
  \\
  \;=\;
  -\tfrac{1}{6!}
  \epsilon_{
    a_1 a_2 a_3
    b_1 b_2 b_3
    c_1 \cdots c_5
  }
  \epsilon^{
    a_1 a_2 a_3 
    d_1 d_2
    f_1 \cdots f_6
  }
  \big(\,
    \overline{\psi}
    \,
    \Gamma_{f_1 \cdots f_6}
    \,
    \psi
  \big)
  e^{d_1 d_2 b_1 b_2 b_3}
  \,
  e^{c_1 \cdots c_5}
  &
  \proofstep{
    by \eqref{ExamplesOfHodgeDualCliffordElements}
  }
  \\
  \;=\;
  \tfrac{
    3! \cdot 8!
  }{6!}
  \delta^{ 
    d_1 d_2 
    f_1 \cdots f_6
  }_{
    b_1 b_2 b_3
    c_1 \cdots c_5
  }
  \big(\,
    \overline{\psi}
    \,
    \Gamma_{f_1 \cdots f_6}
    \,
    \psi
  \big)
  e
    _{d_1 d_2}
    {}^{b_1 b_2 b_3}
  \,
  e^{c_1 \cdots c_5}
  &
  \proofstep{
    by \eqref{ContractingKroneckerWithSkewSymmetricTensor}
  }
  \\
  \;=\;
  \tfrac{
    3! \cdot 8!
  }{6!}
  \big(
  {
    5 \atop 2
  }
  \big)
  \tfrac{
    2! \cdot 6!
  }{8!}
  \delta^{
    d_1 d_2
  }_{
    c_4 c_5
  }
  \,
  \delta^{
    f_1 \cdots f_6
  }_{
    b_1 b_2 b_3
    c_1 c_2 c_3
  }
  \big(\,
    \overline{\psi}
    \,
    \Gamma_{f_1 \cdots f_6}
    \,
    \psi
  \big)
  e
    _{d_1 d_2}
    {}^{b_1 b_2 b_3}
  \,
  e^{c_1 \cdots c_5}
  &
  \proofstep{
    \def\tabcolsep{-10pt}
    \def\arraystretch{.9}
    \begin{tabular}{c}
      combinatorics using that
      \\
      $d/b$-contraction vanishes
    \end{tabular}
  }
  \\
  \;=\;
  120
  \cdot
  \big(\,
    \overline{\psi}
    \,
    \Gamma_{
      b_1 b_2 b_3
      c_1 c_2 c_3
    }
    \,
    \psi
  \big)
  e
    _{d_1 d_2}
    {}^{b_1 b_2 b_3}
  \,
  e^{d_1 d_2 c_3 \cdots c_5}
  &
  \proofstep{
    by \eqref{ContractingKroneckerWithSkewSymmetricTensor}
  }
  \\
  \;=\;
  \tfrac
    { 120 } 
    { 5! }
  \,
  \epsilon_{
    b_1 b_2 b_3
    c_1 c_2 c_3
    a_1 \cdots a_5
  }
  \big(\,
    \overline{\psi}
    \,
    \Gamma^{
      a_1 \cdots a_5
    }
    \,
    \psi
  \big)
  e
    ^{b_1 b_2 b_3}
    {}
    _{d_1 d_2}
  \,
  e^{d_1 d_2 c_3 \cdots c_5}
  &
  \proofstep{
    by 
    \eqref{ExamplesOfHodgeDualCliffordElements}.
  }
  \end{array}
$$

\vspace{-5mm}
\end{proof}

\medskip

\begin{remark}[\bf Induced 7-cocycle]
  Given $\widehat{P}_3$ on $\widehat{\mathfrak{M}}$
  satisfying \eqref{H30BianchiIdentity}, there exists a 7-form $\widetilde G_7 \,\in\, \mathrm{CE}\big(\widehat{\mathfrak{M}}\big)$ of the famous form
  \begin{equation}
    \label{Induced7Cocycle}
    \widetilde
    G_7
    \;:=\;
    \big(
    \phi_{\mathrm{ex}}^\ast
    G_7
    \big)
    -
    \tfrac{1}{2}
    \widehat{P}_3 \,
    \big(
    \phi_{\mathrm{ex}}^\ast
    G_4
    \big)
    \,,
    \;\;\;\;\;\;\;
    \mbox{where}
    \;\;
    G_7
    \;:=\;
    \tfrac{1}{5!}
    \big(\hspace{1pt}
      \overline{\psi}
      \,\Gamma_{a_1 \cdots a_5}\,
      \psi
    \big)
    e^{a_1}\cdots e^{a_5}
    \,\in\,
    \mathrm{CE}\big(
      \mathbb{R}^{
        1,10\,\vert\,\mathbf{32}
      }
    \big)
    \,,
  \end{equation}
  which is closed
  $$
    \mathrm{d}\,
    \widetilde G_7
    \;=\;
    0
  $$
  due to the fundamental quartic Fierz identity that governs 11D supergravity (recalled e.g. in \cite{GSS24-SuGra})
  $$
    \mathrm{d}\, G_7
    \;=\;
    \tfrac{1}{2}
    \,
    G_4 \, G_4
    \,.
  $$
  
\vspace{-2mm} 
\noindent  A natural question then is whether with $G_4$ also $\widetilde G_7$ admits a coboundary on $\widehat{\mathfrak{M}}$. At least for the special parameter value $s = -1$ we answer this to the negative, below in \S\ref{CaseS=-1}.
\end{remark}

\noindent
{\bf Special values of the parameter.}
Some values of the parameter $s \in \mathbb{R}$ in \eqref{ParameterizationByS} are noteworthy for special properties enjoyed by the corresponding hidden M-algebra \eqref{ExceptionalCEDifferentialOnEta}
and/or its super-invariant 3-form \eqref{AnsatzForH30}.

\begin{itemize}[
  leftmargin=1.5cm,
  topsep=2pt,
  itemsep=2pt
]

\item [\colorbox{lightgray}{$s = 0$}]: At exactly this parameter value a super-invariant $\widehat{P}_3$ satisfying the basic Bianchi identity \eqref{H30BianchiIdentity} does {\it not} exist. On the other hand, at $s = 0$ the hidden M-algebra
\begin{itemize}[
  leftmargin=.5cm,
  topsep=2pt,
  itemsep=2pt
]
\item
carries a {\it closed} super-invariant 3-form $\Omega_3$ (Rem. \ref{TheClosed3Form}),
\item
has automorphism symmetry enhanced from $\mathfrak{so}_{1,10}$ to the conformal symplectic algebra $\mathfrak{csp}_{32}$ (Prop. \ref{EnhancedCSpSymmetry}).
\end{itemize}

\item [\colorbox{lightgray}{$s = -6$}]:  At exactly this parameter value, the differential of $\DFSpinor$ is independent of the M5-brane charges (the 5-index generators $e^{a_1 \cdots a_5}$), as is the 3-form $\widehat{P}_3$, so that these may entirely be discarded from the discussion.

\item [\colorbox{lightgray}{$s = -1$}]:  At exactly this value, the differential $\DFSpinor$ is independent of the spacetime coframe $e^a$, so that the hidden M-algebra in this case is the fiber product of 11D super-spacetime with an extended ``pure brane charge''-algebra.

\end{itemize}

We now discuss further aspects of these special cases.

%%%%%%%%%%%%%%%%%%%%%%%%%%%%%
\subsubsection{The case $s=0$: $\mathrm{CSp}$-symmetry}
\label{CaseS=0}
%%%%%%%%%%%%%%%%%%%%%%%%%%%%%%

\noindent
{\bf Enhanced symmetry.} 
At generic parameter value $s$ \eqref{ParameterizationByS} the hidden extension $\widehat{\mathfrak{M}}$ breaks the $\mathrm{GL}(32)$-equivariance of the basic M-algebra (Prop. \ref{ManifestGL32}) down to the spinorial Lorentz subgroup $\mathrm{Pin}^+(1,10) \subset \mathrm{GL}(32)$. However, at the special parameter value $s = 0$ a much larger symmetry remains intact:

\medskip

First, the following was noted in \cite[(26-7)]{BDIPV04}:
\begin{proposition}[\bf The hidden M-algebra at $s = 0$]
  \label{TheHiddenAlgebraAtHighSymmetry}
  At parameter value $s = 0$ \eqref{ParameterizationByS} and in terms of the unified bosonic generators $e^{\alpha \beta}$ \eqref{TheBispinorCEElement},
  the differential \eqref{ExceptionalCEDifferentialOnEta} may equivalently be re-written as
  \begin{equation}
    \label{ManifestlySp32EquivariantDifferential}
    \def\arraystretch{1.1}
    \begin{array}{lcl}
      \mathrm{d}
      \,
      \psi^\alpha
      &=&
      \phantom{
        -64\,       
      }
      0
      \\
      \mathrm{d}
      \,
      e
        ^{\alpha}{}_{\beta}
      &=&
      \phantom{
        -64\,       
      }
      \psi^\alpha
      \,
      \psi_\beta
      \\
      \mathrm{d}
      \,
      \DFSpinor
        ^\alpha
      &=&
      64
      \, 
      e^{\alpha}{}_{\beta}
      \,
      \psi^\beta
      \,,
    \end{array}
  \end{equation}
  which makes manifest that the hidden extension inherits from the
  $\mathrm{GL}(32)$-equivariance 
  \eqref{ManifestGL32Equivariance}
  of the basic M-algebra at least the  symplectic subgroup $\mathrm{Sp}(32,\mathbb{R}) \subset \mathrm{GL}(32)$ extended to act on the new spinor $\DFSpinor$ in the same way as on the original spinor $\psi$:
  \vspace{-1mm} 
  \begin{equation}
    \label{ManifestSp32Equivariance}
    \begin{tikzcd}[row sep=-5pt,
      column sep=0pt
    ]
      \mathrm{Sp}(32,\mathbb{R})
      \times
      \mathrm{CE}\big(
      \mathbb{R}
        ^{1,10\,\vert\,\mathbf{32}}
      \big)
      \ar[rr]
      &&
      \mathrm{CE}\big(
      \mathbb{R}
        ^{1,10\,\vert\,\mathbf{32}}
      \big)
      \\
      \big(
        g
        ,\,
        \psi^\alpha
      \big)
      &\longmapsto&
      g^{\alpha}_{\alpha'}
      \,
      \psi^{\alpha'}
      \\
      \big(
        g,\,
        e^{\alpha \beta}
      \big)
      &\longmapsto&
      g^{\alpha}_{\alpha'}
      \,
      g^{\beta}_{\beta'}
      \,
      e^{\alpha'\beta'}
      \\
      \big(
        g
        ,\,
        \DFSpinor^\alpha
      \big)
      &\longmapsto&
      g^\alpha_{\alpha'}
      \,
      \DFSpinor^{\alpha'}
      \mathrlap{\,.}
    \end{tikzcd}
  \end{equation}
\end{proposition}
\begin{proof}
The first two lines in \eqref{ManifestlySp32EquivariantDifferential} are as in Prop. \ref{ManifestGL32}. From this the third line follows by
$$
  \def\arraystretch{1.3}
  \begin{array}{ccll}
  \big(
  \mathrm{d}
  \,
  \phi
  \big)_\gamma
  &=&
  \paramOne
  \,
  (\Gamma_a\psi)_\gamma
  \, e^a
  \,+\,
  \paramTwo
  \,
  (\Gamma_{a_1 a_2}\psi)_\gamma
  \, e^{a_1 a_2}
  \,+\,
  \paramFive
  \,
  (\Gamma_{a_1 \cdots a_5}\psi)_\gamma
  \, e^{a_1 \cdots a_5}
  &
  \proofstep{
    by
    \eqref{ExceptionalCEDifferentialOnEta}
  }
  \\
  &=&
  \Big(
  \paramOne
  \,
  (\Gamma_a)_{\gamma \delta} 
  \,
  \Gamma^a_{\alpha \beta}
  \,-\,
  \paramTwo
  \,
  (\Gamma_{a_1 a_2})_{\gamma \delta} 
  \,
  \Gamma^{a_1 a_2}_{\alpha \beta} 
  \,+\,
  \paramFive
  \,
  (\Gamma_{a_1 \cdots a_5})_{\gamma\delta}
  \,
  \Gamma^{a_1 \cdots a_5}_{\alpha \beta}
  \Big)
  \psi^\delta
  \,
  e^{\alpha \beta}
  &
  \proofstep{
    by
    \eqref{OriginalBosonicGeneratorsFromManifestlyGL32EquivariantBasis}
  }
  \\
  &=&
  64
  \,
  \eta_{\delta(\alpha}
  \eta_{\beta)\gamma}
  \,
  \psi^\delta
  \,
  e^{\alpha \beta}
  &
  \proofstep{
    by
    \eqref{MetricFormOfFierzExpansion}
  }
  \\
  &=&
  +64
  \,
  \psi_\alpha\, e^{\alpha}{}_\gamma
  &
  \proofstep{
    by 
    \eqref{eAlphaBetaIsSymmetric}
  }
  \\
  &=&
  -64
  \,
  \psi^\alpha\, e_{\alpha\gamma}
  &
  \proofstep{
    by \eqref{LoweringOfSpinorIndices}
  }
  \\
  &=&
  +64
  \,
  e_{\gamma \alpha}
  \,
  \psi^\alpha
  &
  \proofstep{
    by
    \eqref{eAlphaBetaIsSymmetric}.
  }
  \end{array}
$$
This makes the $\mathrm{Sp}(32)$-action fairly evident, but just to make it also explicit:
We extend a transformation $g \,\in\, \mathrm{GL}(32)$ as in \eqref{ManifestGL32Equivariance} from the basic M-algebra to the hidden extension by letting it act in the obvious way also on the new spinor $\DFSpinor$ (more generally there is also a less obvious way, to which we come below in Prop. \ref{EnhancedCSpSymmetry}):
  \begin{equation}
    \label{32Equivariance}
    \begin{tikzcd}[row sep=-4pt,
      column sep=0pt
    ]
      \mathllap{
        g
        \;:\;
      }
      \mathrm{CE}\big(\,
        \widehat{\mathfrak{M}}
    \,  \big)
      \ar[rr]
      &&
      \mathrm{CE}\big(\,
        \widehat{\mathfrak{M}}
     \, \big)
      \\
      \psi^\alpha
      &\longmapsto&
      g^{\alpha}_{\alpha'}
      \,
      \psi^{\alpha'}
      \\
      e^{\alpha \beta}
      &\longmapsto&
      g^{\alpha}_{\alpha'}
      \,
      g^{\beta}_{\beta'}
      \,
      e^{\alpha'\beta'}
      \\
      \mathcolor{purple}
      {\DFSpinor^\alpha}
      &\mathcolor{purple}{\longmapsto}&
      \mathcolor{purple}{g^{\alpha}_{\alpha'}}
      \,
      \DFSpinor^{\alpha'}
      \mathrlap{\,.}
    \end{tikzcd}
  \end{equation}
This preserves also the third line in \eqref{ManifestlySp32EquivariantDifferential} iff
$$
  g 
  \;:\;
  e^\alpha{}_\beta
  \;\longmapsto\;
  g^\alpha_{\alpha'}
  \,
  e^{\alpha'}{}_{\beta'}
  \,
  \bar g^{\beta'}_{\beta}
  \,,
$$
where $\bar g$ denotes the inverse matrix. Now since $e^{\alpha}{}_\beta \,=\, e^{\alpha\gamma}\eta_{\gamma\beta}$ this means equivalently that
$$
  (
  g^\alpha_{\alpha'}
  \,
  e^{\alpha' \gamma'}
  )
  \,
  g^\gamma_{\gamma'}
  \,
  \eta_{\gamma\beta}
  \;=\;  
  (
  g^\alpha_{\alpha'}
  e^{\alpha' \gamma'}
  )
  \eta_{\gamma' \beta' }
  \bar g^{\beta'}_{\beta}
  \,,
$$
and hence equivalently that $g$ preserves the spinor metric $\eta_{\alpha\beta}$ in that
$
  g^{\gamma}_{\gamma'}
  \,
  \eta_{\gamma\beta}
  \,
  g^\beta_{\beta'}
  \,=\,
  \eta_{\gamma' \beta'}
$.
But since the spinor metric is skew-symmetric \eqref{SkewSymmetryOfSpinorMetric}, this means by definition that $g$ must be an element of the subgroup $\mathrm{Sp}(32) \,\subset\,\mathrm{GL}(32)$.
\end{proof}

However, we highlight that the automorphism group of $\widehat{\mathfrak{M}}$ is larger than the $\mathrm{Sp}(32)$ of Prop. \ref{TheHiddenAlgebraAtHighSymmetry}, due to the fact there is extra freedom in transforming the new spinorial generator:
\begin{definition}[{\bf Conformal symplectic group} (e.g. {\cite[p. 7]{MallerTesterman12}})]
  For $n \in \mathbb{N}$ the {\it conformal symplectic group}
  \vspace{1mm} 
  \begin{equation}
    \label{ConformalSymplecticGroup}
    \mathrm{CSp}(2n)
    \;:=\;
    \big\{
      g \,\in\,
      \mathrm{GL}(n)
      \;\big\vert\;
      \eta\big(
        g(-)
        ,\,
        g(-)
      \big)
      \,=\,
      \lambda(g)
      \cdot
      \eta\big(
        -
        ,\,
        -
      \big)
      \,\; , \, \;
      \lambda(g)
      \,\in\,
      \mathbb{R}^\times
    \big\}
  \end{equation}

  \vspace{1mm} 
\noindent  is the group of linear automorphisms of $\mathbb{R}^{2n}$ which preserve the canonical (or any fixed) symplectic form up to rescaling by a non-vanishing real number.

  Extracting the rescaling multiplier $\lambda$ is evidently a group homomorphsism onto the multiplicative group $\mathbb{R}^\times$, whose kernel is the ordinary symplectic group:
  \begin{equation}
    \label{SESOFConformalSymplecticGroup}
    \begin{tikzcd}
      0
      \ar[r]
      &
      \mathrm{Sp}(2n)
      \ar[
        r,
        hook
      ]
      &
      \mathrm{CSp}(2n)
      \ar[
        r,
        ->>,
        "{ \lambda }"
      ]
      &
      \mathbb{R}^\times
      \ar[r]
      &
      0
      \,.
    \end{tikzcd}
  \end{equation}
\end{definition}
\begin{proposition}[\bf Enhanced $\mathrm{CSp}(32,\mathbb{R})$-symmetry of the hidden M-algebra]
  \label{EnhancedCSpSymmetry}
  At $s = 0$ the automorphism group of the hidden M-algebra contains the conformal symplectic group \eqref{ConformalSymplecticGroup}, acting on generators as
  $$
    \begin{tikzcd}[row sep=-4pt,
      column sep=0pt
    ]
      \mathrm{CSp}(32)
      \times
      \mathrm{CE}\big(\,
        \widehat{\mathfrak{M}}
      \,\big)
      \ar[
        rr
      ]
      &&
      \mathrm{CE}\big(\,
        \widehat{\mathfrak{M}}
      \,\big)
      \\
      \big(
        g
        ,\,
        \psi^\alpha
      \big)
      &\longmapsto&
      g^\alpha_{\alpha'}
      \,
      \psi^\alpha
      \\
      \big(
        g
        ,\,
        e^{\alpha\beta}
      \big)
      &\longmapsto&
      g^\alpha_{\alpha'}
      \,
      g^\beta_{\beta'}
      \,
      e^{\alpha' \beta'}
      \\
      \big(
        g
        ,\,
        \DFSpinor^\alpha
      \big)
      &\longmapsto&
      \lambda(g)
      \cdot
      g^\alpha_{\alpha'}
      \,
      \DFSpinor^{\alpha'}
      \mathrlap{\,.}
    \end{tikzcd}
  $$
\end{proposition}
\begin{proof}
  The $\mathrm{CSp}$-property of $g$ says in components that
  $$
    g^\alpha_{\alpha'}
    \,
    \eta_{\alpha \beta}
    \,
    g^\beta_{\beta'}
    \;=\;
    \lambda(g)
    \cdot
    \eta_{\alpha' \beta'}
    \,.
  $$
  With this, we find the respect of $g$ for the differential of $\DFSpinor$ as:
  $$
    \begin{tikzcd}[row sep=small, 
      column sep=20pt
    ]
      \DFSpinor^\alpha
      \ar[
        rrr,
        |->,
        "{ g }"
      ]
      \ar[
        ddd,
        |->,
        "{ \mathrm{d} }"
      ]
      &&
      &[-85pt]
      \lambda(g)
      g^\alpha_{\alpha'}
      \,
      \DFSpinor^{\alpha'}
      \ar[
        dd,
        |->,
        "{ \mathrm{d} }"
      ]
      \\
      \\
      &&&
      -2
      \lambda(g)
      \,
      g^\alpha_{\alpha'}
      \,
      e^{\alpha' \beta'}
       \eta_{\beta' \gamma'}
      \psi^{\gamma'}
      \\
      -2
      \,
      e^{\alpha \beta}
      \eta_{\beta \gamma}
      \psi^\gamma
      \ar[
        rr,
        "{ g }"
      ]
      &&
      -2
      \,
      g^{\alpha}_{\alpha'}
      e^{\alpha'\beta'}
      \!
      g^{\beta}_{\beta'}
      \,
      \eta_{\beta \gamma}
      \,
      g^{\gamma}_{\gamma'}
      \psi^{\gamma'}
      \mathrlap{\,.}
      \ar[
        ur,
        equals
      ]
    \end{tikzcd}
  $$
  \vspace{-4mm} 

\end{proof}

\begin{example}[\bf Pin-action among automorphisms of hidden M-algebra]  
  It is only the $\mathrm{CSp}(32)$ action from Prop. \ref{EnhancedCSpSymmetry} --
  but not the $\mathrm{Sp}(32)$-action from \eqref{ManifestGL32Equivariance} -- which contains the reflection/parity automorphisms from Ex. \ref{ParitySymmetry} (due to \eqref{SpinorMetricPreservedUpToSignByReflection} there):
  $$
    \begin{tikzcd}[row sep=10pt]
      \mathrm{Pin}^+(1,10)
      \ar[r, hook]
      &
      \mathrm{CSp}(32)
      \ar[r, hook]
      &
      \mathrm{Aut}\big(
        \widehat{\mathfrak{M}}
      \big)
      \ar[
        d,
        equals
      ]
      \\
      \mathrm{Spin}(1,10)
      \ar[
        r,
        hook
      ]
      \ar[
        u,
        hook
      ]
      &
      \mathrm{Sp}(32)
      \ar[
        r,
        hook
      ]
      \ar[
        u,
        hook
      ]
      &
      \mathrm{Aut}\big(
        \widehat{\mathfrak{M}}
      \big)      
      \mathrlap{\,.}
    \end{tikzcd}
  $$
\end{example}

\newpage
\noindent
{\bf The 3-Form at $s = 0$.}
The following point was amplified in \cite[(3.13)]{ADR17}:
\begin{remark}[\bf The closed 3-form]
  \label{TheClosed3Form}
  While at $s=0$ the super-invariant 3-form $ \widehat{P}_3$ according to \eqref{RestrictedParametersOfH3FluxAvatar} is not defined, its rescaled limit is well-defined, as follows:
  $$
    \hspace{3cm}
  \def\arraystretch{1.5}
  \begin{array}{l}
    \label{ThePoincareThreeForm}
    \mathllap{
    \PoincareForm_3
    \;
    :=
    \;
    \underset{
      s \to 0
    }{\lim}    
    \;\;
    s^2 \cdot 
    \widehat{P}_3
    \;=\;\;\;\;
    }
    -\,
    \tfrac{3}{5}
    \,
    e_{
      {\color{darkblue} a_1 } 
      {\color{darkorange} a_2}
    }
    \, e^{
      \color{darkblue} a_1
    } 
    \, 
    e^{
     \color{darkorange} a_2
    }
    \\
    \;+\;
    \tfrac{1}{5}
    \, e^{\color{darkgreen}a_1}{}_{\color{darkblue}a_2}
    \, e^{\color{darkblue}a_2}{}_{\color{darkorange}a_3}
    \, e^{\color{darkorange}a_3}{}_{\color{darkgreen}a_1}
    \\
    \;+\;
    \tfrac{18}{6!}
    \, 
    e^{
      {
        \color{darkblue} a_1 \cdots a_4} 
        \color{darkorange} b_1
      }
    \, e_{
    \color{darkorange}b_1
      }{}^{
        \color{darkgreen}
        b_2
      }
    \, e_{
      { \color{darkgreen} b_2 }
      { \color{darkblue} a_1 \cdots a_4}
    }
    \\
    \;+\;
    \tfrac{18}{
      5 \cdot 5! \cdot 6!
    }
    \,
    \epsilon_{
      { \color{darkblue} a_1 \cdots a_5 }
      { \color{darkorange} b_1 \cdots b_5 }
      { \color{darkgreen} c}
    }
    \,
    e^{
     \color{darkblue}
     a_1 \cdots a_5
    }
    e^{
     \color{darkorange}
     b_1 \cdots b_5
    }
    e^{
      \color{darkgreen}
      c
    }
    \\
    \;-\;
    \tfrac{2}{5! \cdot 6!}
    \,
    \epsilon_{
      { \color{darkblue} a_1 a_2 a_3 }
      { \color{darkorange} b_1 b_2 b_3 }
      { \color{darkgreen} c_1 \cdots c_5 }
    }
    \,
    e^{
      {\color{darkblue} a_1 a_2 a_3}
      {\color{purple} d_1 d_2}
    }
    \,
    e_{
      {\color{purple} d_1 d_2}      
     }
     {}
     ^{
      {\color{darkorange} b_1 b_2 b_3}
    }
    \,
    e^{
      \color{darkgreen}
      c_1 \cdots c_5\
    }
    \\
    \;-\;
    \tfrac{3}{10}
    \,
    \big(\,
      \overline{\psi}
      \,\Gamma
        _{\color{darkblue} a}
      \,
      \DFSpinor
    \big)
    e^{\color{darkblue}a}
    \\
    \;-\;
    \tfrac{3}{20}
    \,
    \big(\,
      \overline{\psi}
      \,\Gamma
        _{\color{darkblue}a_1 a_2}\,
      \DFSpinor
    \big)
    \,
    e^{\color{darkblue} a_1 a_2}
    \\
    \;-\;
    \tfrac{3}{10 \cdot 5!}
    \,
    \big(\,
      \overline{\psi}
      \,\Gamma
        _{\color{darkblue}a_1 \cdots a_5}\,
      \DFSpinor
    \big)
    \,
    e^{
      \color{darkblue}
      a_1 \cdots a_5
    }
    \mathrlap{\,,}
    \end{array}
  $$
 and by Prop. \ref{TheRestrictedAvatarOfTheBFieldFlux} has differential equal to $\underset{
      s \to 0
    }{\lim}    
    \;
    s^2 \cdot G_4
    \;=\;
    0
    ,
$
hence is closed:
$$
  \mathrm{d}
  \;
  \Omega_3
  \;=\;
  0
  \,. 
$$
\end{remark}
\begin{proposition}[\bf Space of super-Poincar\'{e} 3-forms]
  The space of solutions of the equation $\mathrm{d}\, \PoincareForm_3 \,=\, 0$ 
  for $\PoincareForm_3$ parameterized as in \eqref{AnsatzForH30} is (0-dimensional for parameter $s \neq 0$ and) 1-dimensional for $s = 0$, spanned by  \eqref{ThePoincareThreeForm}.
\end{proposition}

%%%%%%%%%%%%%%%%%%%%%%%%%%%%%
\subsubsection{The case $s=-1$: IIA-Algebra}
\label{CaseS=-1}
%%%%%%%%%%%%%%%%%%%%%%%%%%%%%

This special case has not received further attention before, we further put it into perspective in \cite{GSS25-TDuality}.

\smallskip

\medskip

\noindent
{\bf The hidden IIA-algebra.}
For $s = -1$ \eqref{ParameterizationByS} 
the differential \eqref{ExceptionalCEDifferentialOnEta}
of $\DFSpinor$ is independent of the space-time generators $(e^a)_{a=0}^{10}$. This means that here the hidden extension exists already on sub-algebras of the M-algebra where some or all of the space-time generators are discarded.
Of particular interest are the cases of 
\begin{itemize}[
  leftmargin=.5cm,
  topsep=1pt,
  itemsep=2pt
]
\item[--] discarding just $e^{\ten}$ from the M-algebra because the result may be understood as the (translational) extended IIA super-algebra,

\item[--] discarding all $e^a$, because the result may be understood as the pure brane charge algebra
\end{itemize}

\begin{definition}[\bf The fully brane-extended type IIA algebra]
\label{FullyExtebdedIIASpacetime}
The translational type IIA fully extended supersymmetry algebra 
$\IIA$
is (e.g. \cite[(2.16)]{CdAIPB00} \footnote{
  In \cite[(2.16)]{CdAIPB00} also the D0-brane charge with differential $(\overline{\psi} \,\Gamma_{\!\ten}\, \psi)$ --  is included in the extended IIA-algebra \eqref{CEFullyExtendedIIA}.
  But condensing D0-brane charge of course means opening up the 11th dimension, and hence
  here we regard this term instead as providing the further extension to the M-algebra, see Ex. \ref{BasicMAlGebraAsExtensionOfFullyExtendedIIASpacetime}.
}) given by \footnote{
  The signs in \eqref{CEFullyExtendedIIA} are a convention that is natural in view of the further extension by the M-algebra \eqref{CEOfBasicMAlgebra}, 
  where these signs 
  align with the Fierz identity \eqref{FierzDecomposition}, and makes the exceptional brane rotating symmetry in Prop. \ref{ManifestGL32} come out naturally.
}
\vspace{1mm} 
\begin{equation}
  \label{CEFullyExtendedIIA}
  \mathrm{CE}\big(
    \IIA
  \big)
  \;\;
  \simeq
  \;\;
  \FDGCA
  \left[\!
  \def\arraystretch{1.5}
  \begin{array}{c}
    (\psi^\alpha)_{\alpha=1}^{32}
    \\
    (e^a)_{a = 1}^9
    \\
    (\tilde e_a)_{a = 1}^9
    \\
    (e_{a_1 a_2} = e_{[a_1 a_2]})
      _{ a_i = 0 }^9
    \\
    (e_{a_1 \cdots a_4} = e_{[a_1 \cdots a_4]})
      _{ a_i = 0 }^9
    \\
    (e_{a_1 \cdots a_5} = e_{[a_1 \cdots a_5]})
      _{ a_i = 0 }^9
  \end{array}
  \!\right]
  \Big/
  \left(
  \def\arraystretch{1.5}
  \def\arraycolsep{2pt}
  \begin{array}{ccl}
    \mathrm{d}\, \psi 
      &=& 
    0
    \\
    \mathrm{d}\, e^a
    &=&
    +
    \big(\hspace{1pt}
      \overline{\psi}
      \,\Gamma^a\,
      \psi
    \big)
    \\
    \mathrm{d}\, \tilde e_a
    &=&
    -
    \big(\hspace{1pt}
      \overline{\psi}
      \,\Gamma_a\Gamma_{\!\ten}\,
      \psi
    \big)
    \\
    \mathrm{d}\, e_{a_1 a_2}
    &=&
    -
    \big(\hspace{1pt}
      \overline{\psi}
      \,\Gamma_{a_1 a_2}\,
      \psi
    \big)
    \\
    \mathrm{d}\, e_{a_1 \cdots a_4}
    &=&
    +
    \big(\hspace{1pt}
      \overline{\psi}
      \,\Gamma_{a_1 \cdots a_4}
      \Gamma_{\!\ten}\,
      \psi
    \big)
    \\
    \mathrm{d}\, e_{a_1 \cdots a_5}
    &=&
    +
    \big(\hspace{1pt}
      \overline{\psi}
      \,\Gamma_{a_1 \cdots a_5}\,
      \psi
    \big)
  \end{array}
  \!\!\right)
  \,.
\end{equation}
\end{definition}
 
\begin{remark}[\bf Extended IIA-algebra and brane charges]
\label{ExtendedIIAAlgebraAndBraneCharges} The bosonic body of the fully extended type IIA algebra \eqref{CEFullyExtendedIIA} may suggestively be re-arranged as
  \begin{equation}
    \label{BraneChargesInFullyExtededIIAAlgebra}
    \def\arraystretch{1.6}
    \def\arraycolsep{1pt}
    \begin{array}{ccccccccccc}
    \big(
      \IIA
    \big)_{\mathrm{bos}}
    &\;\simeq_{{}_{\mathbb{R}}}\;&
    \mathbb{R}^{1,9}
    &\oplus&
    (\mathbb{R}^{1,9})^\ast
    &\oplus&
    \wedge^2 (\mathbb{R}^{1,9})^\ast
    &\oplus&
    \wedge^4 (\mathbb{R}^{1,9})^\ast
    &\oplus&
    \wedge^5 (\mathbb{R}^{1,9})^\ast
    \\
    &\simeq_{{}_{\mathbb{R}}}&
    \underset{
      \mathclap{
        \hspace{4pt}
        \adjustbox{
          rotate=-25
        }{
          \rlap{
          \hspace{-14pt}
          \scalebox{.7}{
            \color{gray}
            space-time 
          }
          }
        }
      }
    }{
      \mathbb{R}^{1,9}
    }
    &\oplus&
    \underset{
      \mathclap{
        \hspace{4pt}
        \adjustbox{
          rotate=-25
        }{
          \rlap{
          \hspace{-24pt}
          \scalebox{.7}{
            \color{gray}
            \def\arraystretch{.8}
            \def\tabcolsep{0pt}
            \begin{tabular}{c}
              string
              \\
              charges
            \end{tabular}
          }
          }
        }
      }    
    }{
      (\mathbb{R}^{1,9})^\ast
    }
    &\oplus&
    \underset{
      \mathclap{
        \hspace{4pt}
        \adjustbox{
          rotate=-25
        }{
          \rlap{
          \hspace{-18pt}
          \scalebox{.7}{
            \color{gray}
            \def\arraystretch{.8}
            \def\tabcolsep{0pt}
            \begin{tabular}{c}
              D2-brane
              \\
              charges
            \end{tabular}
          }
          }
        }
      }        
    }{
      \wedge^2(\mathbb{R}^{9})^\ast
    }
    \oplus
    \underset{
      \mathclap{
        \hspace{4pt}
        \adjustbox{
          rotate=-25
        }{
          \rlap{
          \hspace{-18pt}
          \scalebox{.7}{
            \color{gray}
            \def\arraystretch{.8}
            \def\tabcolsep{0pt}
            \begin{tabular}{c}
              D8-brane
              \\
              charges
            \end{tabular}
          }
          }
        }
      }            
    }{
      \wedge^8(\mathbb{R}^{9})
    }
    &\oplus&
    \underset{
      \mathclap{
        \hspace{4pt}
        \adjustbox{
          rotate=-25
        }{
          \rlap{
          \hspace{-18pt}
          \scalebox{.7}{
            \color{gray}
            \def\arraystretch{.8}
            \def\tabcolsep{0pt}
            \begin{tabular}{c}
              D4-brane
              \\
              charges
            \end{tabular}
          }
          }
        }
      }
    }{
      \wedge^4(\mathbb{R}^9)^\ast
    }
    \oplus
    \underset{
      \mathclap{
        \hspace{4pt}
        \adjustbox{
          rotate=-25
        }{
          \rlap{
          \hspace{-18pt}
          \scalebox{.7}{
            \color{gray}
            \def\arraystretch{.8}
            \def\tabcolsep{0pt}
            \begin{tabular}{c}
              D6-brane
              \\
              charges
            \end{tabular}
          }
          }
        }
      }    
    }{
      \wedge^6(\mathbb{R}^9)
    }
    &\oplus&
    \underset{
      \mathclap{
        \hspace{4pt}
        \adjustbox{
          rotate=-20
        }{
          \rlap{
          \hspace{-18pt}
          \scalebox{.7}{
            \color{gray}
            \def\arraystretch{.8}
            \def\tabcolsep{0pt}
            \begin{tabular}{c}
              NS5-brane
              \\
              charges
            \end{tabular}
          }
          }
        }
      }    
    }{
      \wedge^5(\mathbb{R}^{1,9})^\ast
      \mathrlap{\,,}
    }    
    \end{array}
  \end{equation}
  \vspace{.2cm}

  \noindent
  where in the second line we Hodge-dualized all temporal components (following \cite[(2.12)]{Hull98}) by the rule
  $$
    \wedge^p(\mathbb{R}^{1,d})^\ast
    \;\simeq_{\mathbb{R}}\;
    \grayunderbrace{
    \wedge^p(\mathbb{R}^d)^\ast
    }{
      \mathclap{
        \scalebox{.7}{
          spatial
        }
      }
    }
    \,\oplus\,
    \grayunderbrace{
    \wedge^{1+d-p}( \mathbb{R}^d )
    }{
      \mathclap{
        \scalebox{.7}{
          \def\arraystretch{.9}
          \begin{tabular}{c}
            dualized
            \\
            temporal
          \end{tabular}
        }
      }
    }.
  $$
\end{remark}

\vspace{-2mm} 
At $s = -1$, 
this construction lifts to the hidden M-algebra by discarding its $e^{\ten}$-generator: 
\begin{proposition}[\bf The hidden IIA-algebra]
There exists a fermionic super-Lie algebra extension $\widehat{\IIA}$ of the IIA-algebra \eqref{CEFullyExtendedIIA} given by
\begin{equation}
  \label{HiddenIIA}
  \mathrm{CE}\big(
    \widehat{\IIA}
  \big)
  \;\;
  \simeq
  \;\;
  \FDGCA
  \left[\!
  \def\arraystretch{1.5}
  \begin{array}{c}
    (\psi^\alpha)_{\alpha=1}^{32}
    \\
    (e^a)_{a = 1}^9
    \\
    (\tilde e_a)_{a = 1}^9
    \\
    (e_{a_1 a_2} = e_{[a_1 a_2]})
      _{ a_i = 0 }^9
    \\
    (e_{a_1 \cdots a_4} = e_{[a_1 \cdots a_4]})
      _{ a_i = 0 }^9
    \\
    (e_{a_1 \cdots a_5} = e_{[a_1 \cdots a_5]})
      _{ a_i = 0 }^9
    \\
    (\DFSpinor^\alpha)_{\alpha = 1}^{32}
  \end{array}
  \!\right]
  \Big/
  \left(
  \def\arraystretch{1.4}
  \def\arraycolsep{2pt}
  \begin{array}{ccl}
    \mathrm{d}\, \psi 
      &=& 
    0
    \\
    \mathrm{d}\, e^a
    &=&
    +
    \big(\hspace{1pt}
      \overline{\psi}
      \,\Gamma^a\,
      \psi
    \big)
    \\
    \mathrm{d}\, \tilde e_a
    &=&
    -
    \big(\hspace{1pt}
      \overline{\psi}
      \,\Gamma_a\Gamma_{\!\ten}\,
      \psi
    \big)
    \\
    \mathrm{d}\, e_{a_1 a_2}
    &=&
    -
    \big(\hspace{1pt}
      \overline{\psi}
      \,\Gamma_{a_1 a_2}\,
      \psi
    \big)
    \\
    \mathrm{d}\, e_{a_1 \cdots a_4}
    &=&
    +
    \big(\hspace{1pt}
      \overline{\psi}
      \,\Gamma_{a_1 \cdots a_4}
      \Gamma_{\!\ten}\,
      \psi
    \big)
    \\
    \mathrm{d}\, e_{a_1 \cdots a_5}
    &=&
    +
    \big(\hspace{1pt}
      \overline{\psi}
      \,\Gamma_{a_1 \cdots a_5}\,
      \psi
    \big)
    \\
    \mathrm{d}\,
    \DFSpinor
    &=&
    \Gamma_{a_1 a_2}\psi
    \, e^{a_1 a_2}
    \,+\,
    2
    \,
    \Gamma_{a \ten}\psi
    \, \tilde e^a
    \\
    &&
    \,+\,
    \tfrac{10}{6!}
    \,
    \Gamma_{a_1 \cdots a_5}\psi
    \,
    e^{a_1 \cdots a_5}
    \\
    &&
    \,+\,
    \tfrac{50}{6!}
    \,
    \Gamma_{a_1 \cdots a_4 \ten}\psi
    \,
    e^{a_1 \cdots a_4}
  \end{array}
  \!\!\right)
  \,.
\end{equation}
\end{proposition}
\begin{proof}
  This is just the hidden M-algebra
  \eqref{ExistenceOfSuperExceptionalAlgebra}
  at $s= -1$ \eqref{ParameterizationByS}
  with the generator $e^{\ten}$ discarded and the remaining generators decomposed into those that do or do not carry a $\ten$-index, according to the isomorphism
$$
  \def\arraystretch{1.5}
  \def\arraycolsep{1pt}
  \begin{array}{ccccccccc}
    \mathfrak{M}_{\mathrm{bos}}
    &\;\simeq_{{}_{\mathbb{R}}}\;&
    \mathbb{R}^{1,10}
    &\oplus&
    \wedge^2(\mathbb{R}^{1,10})^\ast
    &\oplus&
    \wedge^5(\mathbb{R}^{1,10})^\ast
    \\
    &\simeq_{{}_{\mathbb{R}}}&
    \mathbb{R}
    \oplus
    \mathbb{R}^{1,9}
    &\oplus&
    (\mathbb{R}^{1,9})^\ast
    \oplus
    \wedge^2(\mathbb{R}^{1,9})^\ast
    &\oplus&
    \wedge^4(\mathbb{R}^{1,9})^\ast
    \oplus
    \wedge^5(\mathbb{R}^{1,9})^\ast
    \\[0pt]
    &\;\simeq_{{}_{\mathbb{R}}}\;&
    \hspace{-34pt}
    \mathrlap{
    \mathbb{R}
    \oplus
    \big(
      \IIA
    \big)_{\mathrm{bos}}
    \mathrlap{\,,}
    }
    
  \end{array}
$$
where in the second line we have decomposed into components that are parallel resp. orthogonal to the $\ten$-coordinate axis, by the rule
$$
  \wedge^p\big(
    \mathbb{R}^{1,d}
  \big)^\ast
  \;
  \simeq_{{}_{\mathbb{R}}}
  \;
  \wedge^{p-1}\big(
    \mathbb{R}^{1,d-1}
  \big)^\ast
  \oplus
  \wedge^{p}\big(
    \mathbb{R}^{1,d-1}
  \big)^\ast
  \mathrlap{\,.}
$$

\vspace{-5mm} 
\end{proof}

Another way to say this:
\begin{remark}[\bf M-Algebra as extension of IIA algebra]
\label{BasicMAlGebraAsExtensionOfFullyExtendedIIASpacetime}
The basic M-algebra 
\eqref{CEOfBasicMAlgebra}
is a central extension
of the fully extended type IIA algebra \eqref{CEFullyExtendedIIA} by (the pullback of) the same 2-cocycle 
that classifies the M/IIA extension:
\vspace{-2mm} 
\begin{equation}
  \label{MAlgebraExtensionOfFullyExtendedIIA}
  \adjustbox{
    raise=-2.5cm
  }{
  \begin{tikzpicture}
  \node at (0,0) {
  \begin{tikzcd}[
    row sep=-1pt,
   column sep=5pt
  ]
    \widehat{\mathfrak{M}}
    \ar[rr, ->>]
    \ar[
      d, ->>
    ]
    &&
    \widehat{\IIA}
    \ar[d, ->>]
    \ar[
      rr,
      "{
        (\,\overline{\psi}
        \Gamma^{\ten} \psi)
      }"
    ]
    &\phantom{-----}&
    b\mathbb{R}
    \ar[d, equals]
    \\[15pt]
    \mathfrak{M}
    \ar[
      rr,
      ->>
    ]
    &&
    \IIA
    \ar[
      rr,
      "{
        (\,\overline{\psi}
        \Gamma^{\ten} \psi)
      }"
    ]
    &\phantom{-----}&
    b\mathbb{R}
    \\
    \psi
    &\longmapsfrom&
    \psi
    \\
    e^a
    &\longmapsfrom&
    e^a
    \\
    e_{a \ten}
    &\longmapsfrom&
    \tilde e_a
    \\
    e_{a_1 a_2}
    &\longmapsfrom&
    e_{a_1 a_2}
    \\
    e_{a_1 \cdots a_4 \ten}
    &\longmapsfrom&
    e_{a_1 \cdots a_4}
    \\
    e_{a_1 \cdots a_5}
    &\longmapsfrom&
    e_{a_1 \cdots a_5}
    \mathrlap{\,.}
  \end{tikzcd}
  };
  \begin{scope}[
    shift={(0,-15pt)}
  ]
  \draw[
    draw=olive,
    fill=olive,
    draw opacity=.3,
    fill opacity=.3,
  ]
    (-3.4,-.4) rectangle
    (.38,.13);
  \node at
     (1.45,-.16)  {
       \scalebox{.7}{
         \color{gray}
         \def\arraystretch{.9}
         \begin{tabular}{c}
           string charges /
           \\
           doubled spacetime 
         \end{tabular}
       }
     };
    \node at
     (-4.35,-.16)  {
       \scalebox{.7}{
         \color{gray}
         \def\arraystretch{.9}
         \begin{tabular}{c}
           wrapped M2- 
           \\
           brane charges
         \end{tabular}
       }
     };
     \end{scope}
  \end{tikzpicture}
  }
\end{equation}
\end{remark}

Alternatively, we may discard {\it all} the spacetime generators $e^a$ from the M-algebra, retaining only the brane charges (equivalently the M-brane charges or IIA-brane charges, according to the above isomorphisms):
\begin{definition}[\bf Pure brane charge algebra]
Write $\mathfrak{Brn}$ for the super-Lie algebra given by
\begin{equation}
  \label{PureBraneAlgebra}
  \mathrm{CE}\big(
    \Brn
  \big)
  \;\simeq\;
  \FDGCA
  \left[
  \def\arraystretch{1.4}
  \begin{array}{c}
    (\psi^\alpha)_{\alpha=1}^{32}
    \\
    \big(
    e_{a_1 a_2}
    \,=\,
    e_{[a_1 a_2]}
    \big){}_{a_i = 0}^{\ten}
    \\
    \big(
    e_{a_1 \cdots a_5}
    \,=\,
    e_{[a_1 \cdots a_5]}
    \big){}_{a_i = 0}^{\ten}
  \end{array}
  \right]
  \Big/
  \left(
  \def\arraystretch{1.3}
  \def\arraycolsep{2pt}
  \begin{array}{ccl}
    \mathrm{d}\, 
    \psi &=& 0
    \\
    \mathrm{d}\, 
    e_{a_1 a_2}
    &=&
    -
    \big(\hspace{1pt}
      \overline{\psi}
      \,\Gamma_{a_1 a_2}\,
      \psi
    \big)
    \\
    \mathrm{d}\, 
    e_{a_1 \cdots a_5}
    &=&
    +
    \big(\hspace{1pt}
      \overline{\psi}
      \,\Gamma_{a_1 \cdots a_5}\,
      \psi
    \big)
  \end{array}
  \!\!\right)
  \,,
\end{equation}
and $\widehat{\mathfrak{Brn}}$ for its hidden extension given by
\begin{equation}
  \label{ExtendedPureBraneAlgebra}
  \mathrm{CE}\big(
    \widehat{\Brn}
  \big)
  \;\simeq\;
  \FDGCA
  \left[
  \def\arraystretch{1.5}
  \begin{array}{c}
    (\psi^\alpha)_{\alpha=1}^{32}
    \\
    \big(
    e_{a_1 a_2}
    \,=\,
    e_{[a_1 a_2]}
    \big){}_{a_i = 0}^{\ten}
    \\
    \big(
    e_{a_1 \cdots a_5}
    \,=\,
    e_{[a_1 \cdots a_5]}
    \big){}_{a_i = 0}^{\ten}
    \\
    \mathcolor{purple}{
      (\DFSpinor^\alpha)_{\alpha=1}^{32}
    }
  \end{array}
  \right]
  \Big/
  \left(
  \def\arraystretch{1.4}
  \def\arraycolsep{2pt}
  \begin{array}{ccl}
    \mathrm{d}\, 
    \psi &=& 0
    \\
    \mathrm{d}\, 
    e_{a_1 a_2}
    &=&
    -
    \big(\hspace{1pt}
      \overline{\psi}
      \,\Gamma_{a_1 a_2}\,
      \psi
    \big)
    \\
    \mathrm{d}\, 
    e_{a_1 \cdots a_5}
    &=&
    +
    \big(\hspace{1pt}
      \overline{\psi}
      \,\Gamma_{a_1 \cdots a_5}\,
      \psi
    \big)
    \\
    \mathcolor{purple}{
    \mathrm{d}\,
    \DFSpinor
    }
    &\mathcolor{purple}{=}&
    \mathcolor{purple}{
    \Gamma^{a_1 a_2}\psi
    \,
    e_{a_1 a_2}
    \,+\,
    \tfrac{10}{6!}
    \,
    \Gamma^{a_1 \cdots a_5}\psi
    \,
    e_{a_1 \cdots a_5}
    }
  \end{array}
  \!\right)
  \,.
\end{equation}
\end{definition}

\medskip

\noindent
{\bf Dimensional reduction from the hidden M-algebra to the hidden IIA-algebra.}
We are going to consider graded derivations on the underlying graded algebra of $\mathrm{CE}\big(\,\widehat{\mathfrak{M}}\,\big)$. Since this algebra is freely generated, by their graded Leibniz rule these derivations are fixed by their value on generators, and hence the canonical linear basis of all graded derivations as a module over $\mathrm{CE}\big(\,\widehat{\mathfrak{M}}\,\big)$ may be written as
$$
  \mathrm{Der}\big(
    \mathrm{CE}(\,\widehat{\mathfrak{M}}\,)
  \big)
  \;\simeq\;
  \mathrm{CE}(\,\widehat{\mathfrak{M}}\,)
  \Big\langle \!
    \grayunderbrace{
      \partial_{\psi}
    }{
      (-1,\mathrm{odd})
    }
    ,\,
    \grayunderbrace{
      \partial_{e^a}
    }{
      (-1,\mathrm{evn})
    }
    ,\,
    \grayunderbrace{
    \partial_{e_{a_1 a_2}}
    }{  (-1,\mathrm{evn})  }
    ,\,
    \grayunderbrace{
      \partial_{e_{a_1 \cdots a_5}}
    }{ (-1,\mathrm{evn}) }
    ,\,
    \grayunderbrace{
      \partial_{\DFSpinor}
    }{ (-1,\mathrm{odd}) }
  \!\Big\rangle
  \,.
$$
For example, the CE-differential
\eqref{ExceptionalCEDifferentialOnEta}
itself appears in this notation as
\begin{equation}
    \label{CEDifferentialExpandedInBasicDerivations}
    \def\arraystretch{1.6}
    \begin{array}{ccl}
    \mathrm{d}
    &=&
    \big(\hspace{1pt}
      \overline{\psi}
      \,\Gamma^a\,
      \psi
    \big)
    \partial_{e^a}
    \,-\,
    \big(\hspace{1pt}
      \overline{\psi}
      \,\Gamma_{a_1 a_2}\,
      \psi
    \big)
    \partial_{e_{a_1 a_2}}
    \,+\,
    \big(\hspace{1pt}
      \overline{\psi}
      \,\Gamma_{a_1 \cdots a_5}\,
      \psi
    \big)
    \partial_{e_{a_1 \cdots a_5}}
    \\
    && + 
    \big(
      \paramOne
      \,
      \Gamma_a\psi \, e^a
      \,+\,
      \paramTwo
      \,
      \Gamma_{a_1 a_2}\psi \, e^{a_1 a_2}
      \,+\,
      \paramFive
      \,
      \Gamma_{a_1 \cdots a_5}\psi \, e^{a_1 \cdots a_5}
    \big)\partial_{\DFSpinor}
    \,.
    \end{array}
  \end{equation}

\begin{definition}[\bf Dimensional reduction derivation]
We write
\begin{equation}
  \label{DimensionalReductionDerivation}
  p^M_\ast
  \;:\;
  \begin{tikzcd}
    \mathrm{CE}\big(\,
      \widehat{\mathfrak{M}}
    \,\big)
    \ar[r]
    &
    \mathrm{CE}\big(\,
      \widehat{\IIA}
    \,\big)
  \end{tikzcd}
\end{equation}
for the derivation
\begin{equation}
  \label{DimReductionDerivation}
  p^M_\ast
  \;=\;
  \partial_{e^{\ten}}
\end{equation}
but regarded as taking values in the hidden IIA-algebra \eqref{HiddenIIA}.
We may think of this as the operation of ``fiber integration over the M-theory circle'' (cf. \cite{GSS25-TDuality}).
\end{definition}
\begin{example}[\bf Some fiber integrations]
The fiber integration 
\begin{itemize}[
  leftmargin=.8cm,
  topsep=1pt,
  itemsep=2pt
]
\item[\bf (i)] of the avatar super 4-flux density \eqref{CoboundaryRelationInIntro} is:
\begin{equation}
  \label{FiberIntegrationOf4Flux}
  p^M_\ast
  \phi^\ast_{\mathrm{ex}}
  G_4
  \;\defneq\;
  p^M_\ast
  \Big(
    \tfrac{1}{2}
    \big(\hspace{1pt}
      \overline{\psi}
      \,\Gamma_{a_1 a_2}\,
      \psi
    \big)
    e^{a_1}e^{a_2}
  \Big)
  \;=\;
  -
  \grayunderbrace{
  \textstyle{\underset{a < \ten}{\sum}}
  \big(\hspace{1pt}
    \overline{\psi}
    \,\Gamma_{a \ten}\,
    \psi
  \big)
  e^a}{
    H_3^A
  }
\end{equation}
\vspace{-.4cm}

\item[\bf (ii)] of the hidden 3-form \eqref{AnsatzForH30} 
is
\begin{equation}
  \label{FiberIntegrationOf3Form}
  \grayunderbrace{
  p^M_\ast
  \widehat{P}_3
  }{
    \widehat{P_2}
  }
  \;=\;
  -
  2\alphaZero
  \,
  \grayunderbrace{
  \textstyle{
    \underset
      {a < \ten}
      {\sum}
  }
  e_{a \ten} \, e^a
  }{
    P_2
  }
  \;+\;
  \alphaThree\,
  \epsilon_{
    a_1 \cdots a_5
    \,
    b_1 \cdots b_5
    \,
    \ten
  }
  \,
  e^{
    a_1 \cdots a_5
  }
  e^{
    b_1 \cdots b_5
  }
  \;+\;
  \betaOne
  \big(\hspace{1pt}
    \overline{\psi}
    \,\Gamma_{\ten}\,
    \DFSpinor
  \big),
\end{equation}

\vspace{-2mm} 
\item[\bf (iii)] and that of its first summand alone gives, at $s = -1$:
\begin{equation}
  \label{FiberIntegrationOfBasicP3}
  p^M_\ast
  \big(
    -
    \tfrac{1}{2}
    e_{a_1 a_2}
    \, e^{a_1} e^{a_2}
  \big)
  \;=\;
  \grayunderbrace{
  \textstyle{
    \underset{a < \ten}{\sum}
  }
  e_{\ten a} 
  \,
  e^a
  }{
    P_2
  }
  \;\;
  \underset{
    \mathclap{
      \scalebox{.7}{
        \eqref{StringChargeGenerators}
      }
    }
  }{=}
  \;\;
    e^a \, \widetilde e_a
  \,.
\end{equation}
\end{itemize}

\vspace{-1mm} 
\noindent (The symbols under the braces are explained and discussed in \cite{GSS25-TDuality}, here the reader may take them just as shorthands.)
\end{example}
From \eqref{CEDifferentialExpandedInBasicDerivations} and \eqref{DimReductionDerivation}, we have:
\begin{lemma}[\bf Hidden Lie derivative along M-theory circle]
\label{HiddenLieDerivativeAlongMTHeoryCircle}
The graded commutator of the derivation \eqref{DimReductionDerivation} with the CE-differential
\begin{equation}
  \label{DeclaringGradedCommutator}
  [\mathrm{d}, p^M_{\ast}]
  \;\defneq\;
  \mathrm{d}
    \circ 
  p^M_\ast
  \,+\,
  p^M_\ast
    \circ 
  \mathrm{d}
\end{equation}
equals 
\begin{equation}
  \label{HiddenLieDerivativeAlongMCircleFormula}
  [\mathrm{d}, p^M_{\ast}]
  \;=\;
  -
  \paramOne 
  \big(
  \,
  \Gamma_{\ten}\psi
  \big)
  \partial_{\DFSpinor}
  \,.
\end{equation}
\end{lemma}
It is then interesting to work out the fiber integration of the 3-form $\widehat{P}_3$ \eqref{The3Form} on the hidden M-algebra. For completeness we first state this for general $s$, though only for $s = -1$ may the result be understood as being in $\mathrm{CE}(\, \widehat{\IIA}\,)$.
\begin{example}[\bf Hidden Lie derivative of the 3-form]
The hidden Lie derivative \eqref{HiddenLieDerivativeAlongMCircleFormula} of $\widehat{P}_3$ \eqref{RestrictedParametersOfH3FluxAvatar} is 
\begin{equation}
  \label{HiddenLieDerivativeOf3Form}
  \def\arraystretch{1.6}
  \begin{array}{ccll}
    [\mathrm{d}, p^M_{\ast}]
    \widehat{P}_3
    &=&
    \betaOne\paramOne
    \,
    \big(\hspace{1pt}
      \overline{\psi}
      \,\Gamma_a\Gamma_{\ten}\,
      \psi
    \big)
    e^a
    \;+\;
    \betaTwo\paramOne
    \,
    \big(\hspace{1pt}
      \overline{\psi}
      \,\Gamma_{a_1 a_2}\Gamma_{\ten}\,
      \psi
    \big)
    e^{a_1 a_2}
    \;+\;
    \betaThree\paramOne
    \,
    \big(\hspace{1pt}
      \overline{\psi}
      \,\Gamma_{a_1 \cdots a_5}\Gamma_{\ten}\,
      \psi
    \big)
    e^{a_1 \cdots a_5}
    \\
    &=&
    \betaOne\paramOne
    \grayunderbrace{
    \textstyle{
      \underset
        {a < \ten}
        {\sum}
    }
    \big(\hspace{1pt}
      \overline{\psi}
      \,\Gamma_{a\ten}\,
      \psi
    \big)
    e^a
    }{
      H_3^A
    }
    \;-\;
    2
    \betaTwo\paramOne
    \grayunderbrace{
    (-1)
    \textstyle{
    \underset
      {a < \ten}
      {\sum}
    }
    \big(\hspace{1pt}
      \overline{\psi}
      \,\Gamma^{a}\,
      \psi
    \big)
    e_{a \ten}
    }{
      H_3^{\widetilde A}
    }
    \;+\;
    \betaThree\paramOne
    \grayunderbrace{
    \textstyle{
    \underset
      {a_i < \ten}
      {\sum}
    }
    \big(\hspace{1pt}
      \overline{\psi}
      \,\Gamma_{a_1 \cdots a_5 \ten}\,
      \psi
    \big)
    e^{a_1 \cdots a_5}\
    \mathrlap{\,,}
    }{
      =:\, H_3^C
    }
  \end{array}
\end{equation}
where the second step follows by Clifford expansion \eqref{GeneralCliffordProduct} and the vanishing of resulting skew terms \eqref{SkewSpinorPairings}, and where under the braces we recognized the avatar super-flux densities of the NS B-field of type IIA and type IIB, pulled back to the M-algebra (this is explained in \cite{GSS25-TDuality}, but for the present purpose the reader may take these symbols to be defined thereby).
\end{example}

This then leads to the following: 

\begin{example}[\bf Differential of fiber integration of the 3-form]
For $s \neq 0$, the differential of the fiber integration $\widehat{P}_2$ \eqref{FiberIntegrationOf3Form}
of the 3-form \eqref{AnsatzForH30} is
$$
  \def\arraystretch{1.6}
  \begin{array}{ccll}
   \mathrm{d}\, \widehat{P}_2
   &\defneq&
   \mathrm{d}
   (
     p^M_\ast \widehat{P}_3
   )
   &
   \proofstep{
     by
     \eqref{FiberIntegrationOf3Form}
   }
   \\
   &=&
   -
   p^M_\ast \mathrm{d}\widehat{P}_3
   +
   [\mathrm{d},p^M_\ast]\widehat{P}_3
   &
   \proofstep{
     by \eqref{DeclaringGradedCommutator}
   }
   \\
   &=&
   -
   p^M_\ast 
   \phi^\ast_{\mathrm{ex}}
   G_4
   +
   [\mathrm{d},p^M_\ast]\widehat{P}_3
   &
   \proofstep{
     by \eqref{H30BianchiIdentity}
   }
   \\
   &=&
   (1 + \betaOne\paramOne)
   \,
   H_3^A
   \;-\;
   2\betaTwo\paramOne
   \,
   H_3^{\widetilde A}
   \;+\;
   2\betaThree\paramOne
   \, 
   H_3^C
   &
   \proofstep{
     by
     \eqref{FiberIntegrationOf4Flux}
     \&
     \eqref{HiddenLieDerivativeOf3Form}
   }
   \\
   &=&
   \left\{
   \def\arraystretch{1.6}
   \def\arraycolsep{2pt}
   \begin{array}{rcrcrcl}
     H^A_3 &&&& & 
     &\mbox{for $s = -1$}
     \\[-3pt]
     \tfrac{17}{12} H_3^A
     &+&
     \tfrac{1}{12}
     H_3^{\widetilde A}
     && 
     &&\mbox{for $s = -6$}
     \\
     \tfrac{2 s^2}{5}
     H_3^A
     &+&
     \tfrac{3 s^2}{5}
     H_3^{\widetilde A}
     &-&
     \tfrac{6 s^2}{5 \cdot 5!}
     H_3^C
     & 
     &\mbox{for $s \to 0$}\,
   \end{array}
   \right.
  \end{array}
$$
So in particular, at the parameter value $s = -1$ of interest, where the dimensional reduction of the hidden 3-form exists on the hidden IIA-algebra \eqref{DimensionalReductionDerivation}, it satisfies the direct IIA-analog of the Bianchi identity of the 3-form in M-theory:
\begin{equation}
  \begin{tikzcd}[
    row sep=12pt,
    column sep=3pt
  ]
    \mathrm{d} \, \widehat{P}_3
    \ar[
      d,
      shift right=20pt,
      bend right=35,
      shorten >=-8pt,
      |->,
      "{
        \scalebox{.7}{
          \color{darkgreen}
          \bf
          \def\arraystretch{.9}
          \begin{tabular}{c}
            dimensional
            \\
            reduction
          \end{tabular}
        }
      }"{swap, pos=.6}
    ]
    &=&
    \phi^\ast_{\mathrm{ex}}
    G_4
    &
    \in\;
    \mathrm{CE}\big(\,
      \widehat{\mathfrak{M}}
  \,  \big)
    \\
    \mathrm{d} \, \widehat{P}_2
    &=&
    \phi^\ast_{\mathrm{ex}}
    H_3^A
    &
    \in\;
    \mathrm{CE}\big(\,
      \widehat{\IIA}
    \,\big)
    \,.
  \end{tikzcd}
\end{equation}
\end{example}

\noindent
{\bf The 7-Form on the hidden M-algebra.} At $s=-1$ we may also say more about the avatar 7-flux: 
\begin{lemma}[\bf Induced 7-cocycle is non-trivial]
  At $s = -1$, at least, there does -not- exists a $\mathrm{Spin}(1,10)$-invariant coboundary for the induced 7-cocycle $\widetilde G_7$ \eqref{Induced7Cocycle}.
\end{lemma}
\begin{proof}
We are looking for 
$$
  P_6
  \,\in\,
  \mathrm{CE}\big(
    \,
    \widehat{\mathfrak{M}}
    \,
  \big)^{\mathrm{Spin}(1,10)}
$$
such that

\begin{equation}
  \label{BianchiIdentityForP6}
  \mathrm{d}
  \,
  P_6
  \;=\;
  \grayunderbrace{
  \tfrac{1}{5!}
  \big(\hspace{1pt}
    \overline{\psi}
    \,
    \Gamma_{a_1 \cdots a_5}
    \,
    \psi
  \big)
  e^{a_1}
  \cdots 
  e^{a_5}
  }{
    G_7
  }
  \,-\,
  \tfrac{1}{2}
  \grayunderbrace{
  \tfrac{1}{2}
  \big(\hspace{1pt}
    \overline{\psi}
    \,\Gamma_{a_1 a_2}\,
    \psi
  \big)
  e^{a_1} e^{a_2}
  }{
    G_4
  }
  \,
  \grayunderbrace{
  \big(
    \tfrac{1}{2}
    e^{a_1}
    \,
    e_{a_1 a_2}
    \,
    e^{a_2}
    +
    \cdots
  \big)
  }{
    \widehat{P}_3
  }.
\end{equation}
In priving the statement, we will now use repeatedly that at $s = -1$ the differential \eqref{ExceptionalCEDifferentialOnEta} does not increase the number of $e^a$-s in monomials, since $\paramOne = 0$. Therefore the only term which can give the first summand in \eqref{BianchiIdentityForP6},
under the differential, 
is $\tfrac{1}{5!}e_{a_1 \cdots a_5} e^{a_1} \cdots e^{a_5}$. The other summand that this term gives under the differential, shown in dark blue below, does not appear in \eqref{BianchiIdentityForP6} and hence must be cancelled by a suitable counter-term.
But again since the differential does not increase the order of $e^a$-s, the only possible counter-term is of the form $\big(\hspace{1pt}\overline{\psi} \,\Gamma_{a_1\cdots a_4}\,\DFSpinor\big)e^{a_1}\cdots e^{a_4}$.
Therefore, any candidate $P_6$ must start out with monomials of this form
$$
  \def\arraystretch{1.2}
  \begin{array}{ccl}
    P_6
    &:=&
    \tfrac{1}{5!}
    e_{a_1 \cdots a_5}
    e^{a_1} \cdots e^{a_5}
    \\
    &+&
    r
    \;
    \big(\hspace{1pt}
      \overline{\psi}
      \,
      \Gamma_{a_1 \cdots a_4}
      \,
      \DFSpinor
    \big)
    e^{a_1} \cdots e^{a_4}
    \\
    &+& \cdots
    \,,
  \end{array}
$$
for some $r \in \mathbb{R}$.

\noindent Its differential thus is:
$$
  \def\arraystretch{1.4}
  \begin{array}{l}
    \mathrm{d}\,P_6
    \;=\;
    \Big(
    \tfrac{1}{5!}
    \big(\hspace{1pt}
      \overline{\psi}
      \,\Gamma_{a_1 \cdots a_5}\,
      \psi
    \big)
    e^{a_1} \cdots e^{a_5}
    \mathcolor{darkblue}{
    \;-\;
    \tfrac{1}{4!}
    \,
    e_{b\, a_1 \cdots a_4}
    \big(\overline{\psi}
      \,\Gamma^b\,
      \psi
    \big)
    e^{a_1} \cdots e^{a_4}
    }
    \Big)
    \\[5pt]
    \hspace{1cm} +\;
    r
    \Big(
    -
    \grayunderbrace{
    \big(\hspace{1pt}
      \overline{\psi}
      \,
      \Gamma_{a_1 \cdots a_4}
      \Gamma^{b_1 b_2}
      \,
      \psi
    \big)
    e_{b_1 b_2}
    e^{a_1} \cdots e^{a_4}
    }{
      \scalebox{.7}{$
      \def\arraystretch{1.5}
      \begin{array}{r}
        \big(\hspace{1pt}
          \overline{\psi}
          \,
          \Gamma_{
            a_1 \cdots a_4
            \,
            b_1 b_2
          }
          \,
          \psi
        \big)
        e^{b_1 b_2}
        \,
        e^{a_1} \cdots e^{a_4}
        \\
        -
        \mathcolor{purple}{
        12
        \,
        \Big(
        \big(\hspace{1pt}
          \overline{\psi}
          \,
          \Gamma_{a_1 a_2}
          \,
          \psi
        \big)
        e^{a_1}e^{a_2}
        \Big)
        e_{b_1 b_2}
        e^{b_1} e^{b_2}
        }
      \end{array}
      $}
    }
    -
    \tfrac{10}{6!}
    \grayunderbrace{
    \big(\hspace{1pt}
      \overline{\psi}
      \,
      \Gamma_{a_1 \cdots a_4}
      \Gamma^{b_1 \cdots b_5}
      \,
      \psi
    \big)
    e_{b_1 \cdots b_5}
    e^{a_1} \cdots e^{a_4}
    }{
    \mathclap{
      \scalebox{.7}{$
      \def\arraystretch{1.4}
      \begin{array}{r}
      \big(\hspace{1pt}
        \overline{\psi}
        \,
        \Gamma_{
          a_1 \cdots a_4
          \,
          b_1 \cdots b_5
        }
        \,
        \psi
      \big)
      e^{b_1 \cdots b_5}
      e^{a_1} \cdots e^{a_4}
      \\
      -\,
      120
      \,
      \big(\hspace{1pt}
        \overline{\psi}
        \,
        \Gamma_{a_3 a_4 b_3 b_4 b_5}
        \,
        \psi
      \big)
      e^{c_1 c_2 b_3 b_4 b_5}
      e_{c_1} e_{c_2}
      e^{a_3} e^{a_4}
      \\
      +\,
      \mathcolor{darkblue}{
      120
      \,
      \big(\hspace{1pt}
        \overline{\psi}
        \,
        \Gamma_{b_5}
        \,
        \psi
      \big)
      e^{c_1 \cdots c_4 b_5}
      e_{c_1} \cdots e_{c_4}
      }
      \end{array}
      $}
      }
    }
\\
    \hspace{1.85cm} +\,
    4
    \big(\hspace{1pt}
      \overline{\psi}
      \,\Gamma_{b\, a_1 a_2 a_3 }\,
      \phi
    \big)
    \big(\hspace{1pt}
      \overline{\psi}
      \,\Gamma^b\,
      \psi
    \big)
    e^{a_1}
    e^{a_2}
    e^{a_3}
    \Big)
    \;\;\;+\;
    \cdots
    \mathrlap{\,,}
  \end{array}
$$
where under the braces we used Clifford expansion \eqref{GeneralCliffordProduct} and the fact 
\eqref{SymmetricSpinorPairings}
that $\big(\hspace{1pt}\overline{\psi}\,\Gamma_{a_1 \cdots a_p}\,\psi\big)\,=\, 0$ if $p \,\in\,\{0, 3,4,7,8, 11\}$.

Now again since the differential does not increase the order of the $e^a$-s, it follows that the omitted summands do not contain monomials of either the darkblue or the purple kind. But since the monomials of the darkblue form clearly do not appear in the induced 7-cocycle on the right of \eqref{BianchiIdentityForP6}, the darkblue summands above must cancel among each other, which is equivalent to
$$
  -
  r
  \,
  \tfrac{1200}{6!}
  \,-\,
  \tfrac{1}{4!}
  \;=\;
  0
  \;\;\;\;\;\;\;
  \Leftrightarrow
  \;\;\;\;\;\;\;
  r
  \;=\;
  -
  \tfrac{
    6!
  }{
    1200 \cdot 4!
  }
  \;=\;
  -
  \tfrac{1}{40}
  \,.
$$
With this, the contribution of the purple monomial is fixed as
$$
  \mathrm{d}\, P_6
  \;=\;
  \tfrac{1}{5!}
  \big(\hspace{1pt}
    \overline{\psi}
    \,\Gamma_{a_1 \cdots a_5}\,
    \psi
  \big)
  e^{a_1} \cdots e^{a_5}
  \;-\;
  \grayunderbrace{
    \tfrac{
      12 \cdot 8
    }{
      40
    }
  }{\color{red}
    \neq\, 1
  }
  \tfrac{1}{2}
  \Big(
  \grayunderbrace{
  \mathcolor{purple}{
  \tfrac{1}{2}
  \big(\hspace{1pt}
    \overline{\psi}
    \,
    \Gamma_{a_1 a_2}
    \,
    \psi
  \big)
  e^{a_1} e^{a_2}
  }
  }{
    G_4
  }
  \;
  \grayunderbrace{
  \mathcolor{purple}{
  \tfrac{1}{2}
  e^{b_1}
  e_{b_1 b_2}
  e^{b_2}
  }}{
    \widehat{P}_3
    \,-\,
    \cdots
  }
  \Big)
  \;+\;
  \cdots
$$
But this has the wrong coefficient with respect to \eqref{BianchiIdentityForP6}. 
Since, again, there is no other way to get this monomial under the differential, it follows that $P_6$ as in \eqref{BianchiIdentityForP6} does not exist.
\end{proof}

\begin{remark}[\bf The hidden M-algebra as a correspondence space for twisted non-abelian cocycles]The non-existence of a cobounding $P_6$ on $\widehat{\mathfrak{M}}$ reinforces the intepretation of the hidden M-algebra advocated in \cite{GSS25-TDuality}, namely as the correspondence space of an M-theoretic lift of T-duality, on which \textit{only} the twisting cocycle $G_4$ underlying the full $\mathfrak{l}S^4$-cocycle $(G_4,G_7)$ is trivialized, with the latter viewed as a twisted non-abelian cocycle as in  \cite[Ex. 2.19]{GSS25-TDuality}.  
\end{remark}

%%%%%%%%%%%%%%%%%%%%%%%%%%%%%%%%%%%%%%%%%%%%%%%%%
\subsection{Further extensions}
\label{FurtherFermionicExtension}
%%%%%%%%%%%%%%%%%%%%%%%%%%%%%%%%%%%%%%%%%%%%%%%%%

For completeness, we give a streamlined account of the further fermionic extensions of the hidden M-algebra, making transparent the available choices.

\smallskip 
To this end, note that what  \eqref{RelationBetweenFactorsOfSummandsInDifferentialOfExceptionalFermion} really says is that the right hand side of the last line of \eqref{ExceptionalCEDifferentialOnEta} varies in a {\it 2-dimensional space of 2-cocycles} on the basic M-algebra.
Hence instead of just extending by one of them, we may extend by two of them at once, such as the ones for $s = 0$ and for $s = -6$:
\begin{equation}
  \label{ExtensionByTwoCocycles}
  \def\arraystretch{1.6}
  \begin{array}{l}
    \mathrm{d}
    \;
    \DFSpinor_{{}_{(0)}}
    \;=\;
    2\big(
      \Gamma_a
      \psi
      \,e^a
      +
      \tfrac{1}{2}
      \Gamma_{a_1 a_2}
      \psi
      \,e^{a_1 a_2}
      +
      \tfrac{1}{5!}
      \Gamma_{a_1 \cdots a_5}
      \psi
      \,e^{a_1 \cdots a_5}
    \big)
    \\
    \mathrm{d}
    \;
    \DFSpinor_{{}_{(-6)}}
    \;=\;
    -10
    \,
    \Gamma_a \psi\, e^a
    +
      \Gamma_{a_1 a_2}
      \psi
      \,e^{a_1 a_2}.
  \end{array}
\end{equation}
While explicitly considered in this form in \cite[(3.6-7)]{ADR17}, we find below in Ex. \ref{GammaTraceOfVectorSpinorIsPreviousSpinor} 
that this further generator is essentially implicit already in \cite[p. 5]{Sezgin97}\cite[(3.19)]{Castellani11}.

\medskip

\noindent
{\bf Further tensor-spinor generator.}

\begin{lemma}[\bf Cubic Fierz relations]
  \label{CubicFierzRelations}
  In $\mathrm{CE}\big(\mathbb{R}^{1,10\,\vert\,\mathbf{32}}\big)$ 
 from \eqref{CEOfSuperMinkowksi},
  the following identities hold
  \begin{equation}
    \label{FurtherFierzRelations}
    \def\arraystretch{1.6}
    \begin{array}{lcccccc}
      0
      &=&
      \Gamma_{a \, b}\psi
      \, 
      \big(\hspace{1pt}
        \overline{\psi}
        \,\Gamma^b\,
        \psi
      \big)
      &+&
      \Gamma^{b}\psi
      \big(\hspace{1pt}
        \overline{\psi}
        \,\Gamma_{a\, b}\,
        \psi
      \big),
      \\
      0
      &=&
      \Gamma_{a_1 \cdots a_4\, b}\psi
      \, 
      \big(\hspace{1pt}
        \overline{\psi}
        \,\Gamma^b\,
        \psi
      \big)
      &-&
      \Gamma_{[a_1 a_2}\psi
      \, 
      \big(\hspace{1pt}
        \overline{\psi}
        \,\Gamma_{a_3 a_4]}
        \,
        \psi
      \big)
      &+&
      6
      \,
      \Gamma^{b}\psi
      \, 
      \big(\hspace{1pt}
        \overline{\psi}
        \,\Gamma_{a_1 \cdots a_4\, b}\,
        \psi
      \big)
      \,.
    \end{array}
  \end{equation}
\end{lemma}
\begin{proof}
We are looking for coefficients solving the following equations:
  \begin{equation}
    \label{EquivalentJacobiIdentityForTensorSpinors}
    \def\arraystretch{1.6}
    \begin{array}{lcccccc}
      0
      &=&
      \paramOnePrime
      \,
      \Gamma_{a \, b}\psi
      \, 
      \big(\hspace{1pt}
        \overline{\psi}
        \,\Gamma^b\,
        \psi
      \big)
      &-&
      \paramTwoPrime
      \,
      \Gamma^{b}\psi
      \big(\hspace{1pt}
        \overline{\psi}
        \,\Gamma_{a\, b}\,
        \psi
      \big)
      \\
      0
      &=&
      \paramOneDoublePrime
      \,
      \Gamma_{a_1 \cdots a_4\, b}\psi
      \, 
      \big(\hspace{1pt}
        \overline{\psi}
        \,\Gamma^b\,
        \psi
      \big)
      &-&
      \paramTwoDoublePrime
      \,
      \Gamma_{[a_1 a_2}\psi
      \, 
      \big(\hspace{1pt}
        \overline{\psi}
        \,\Gamma_{a_3 a_4]}
        \,
        \psi
      \big)
      &+&
      \paramFiveDoublePrime
      \,
      \Gamma^{b}\psi
      \, 
      \big(\hspace{1pt}
        \overline{\psi}
        \,\Gamma_{a_1 \cdots a_4\, b}\,
        \psi
      \big)
      \,.
    \end{array}
  \end{equation}  
  Substituting the cubic Fierz identities \eqref{GeneralCubicFierzIdentities} for the $(\psi^3)$ terms and using the $\Gamma$-tracelessness \eqref{TheHigherTensorSpinors} of the resulting representations $\Xi$, one finds that the summands appearing above evaluate as follows (cf. \cite[\S A]{Vaula07}, and mechanical checks in \cite{AncillaryFiles}).

  For the first equation, we have

  $$
    \def\arraystretch{1.5}
    \begin{array}{ll}
      \Gamma_{a\, b}
      \psi
      \,
      \big(\hspace{1pt}
        \overline{\psi}
        \,\Gamma^b\,
        \psi
      \big)
      &
      \;=\;
      \quad 
      \underbrace{
        \Gamma_{a\, b}
      }_{\color{gray} 
        \mathclap{
          \Gamma_{a}\Gamma_b
          -
          \eta_{a b}        
        }
      }      
      \Big(
      \tfrac{1}{11}
      \Gamma^b
      \Xi^{(32)}
      +
      \Xi^{(320) b}
      \Big)
      \\
      &\;=\;
      \tfrac{10}{11}
      \Gamma_a
      \Xi^{(32)}
      \,-\,
      \Xi^{(320)}_a,
   %  \end{array}
   % $$ 
   %  \hspace{.7cm}
   %  $$
   %  \def\arraystretch{1.5}
   %  \begin{array}{ll}
   \\[+5pt]
      \Gamma^b
      \psi
      \big(\hspace{1pt}
        \overline{\psi}
        \,\Gamma_{a\, b}\,
        \psi
      \big)
      &
      \;=\;
      \Gamma^b\big(
        \tfrac{1}{11}
        \Gamma_{a\, b}
        \Xi^{(32)}
        \,-\,
        \tfrac{2}{9}
        \Gamma_{[a}
        \Xi^{(320)}_{b]}
      \big)
      \\
      &\;=\;
      -\tfrac{10}{11}
      \Gamma_a
      \Xi^{(32)}
      -
      \tfrac{1}{9}
      (\Gamma^b\Gamma_a+\Gamma_a
        \grayunderbrace{
        \Gamma^b) 
        \Xi_b^{(320)}
        }{0}
      +
        \tfrac{1}{9}
        \Gamma^b\Gamma_b 
        \Xi_a^{(320)}
      \\[-9pt]
      &\;=\;
      -\tfrac{10}{11}
      \Gamma_a \Xi^{(32)}
      +
      \Xi_a^{(320)}
      \,,
    \end{array}
  $$
  whence the first condition is equivalently the system
  $$
    \def\arraystretch{1.5}
    \def\arraycolsep{2pt}
    \begin{array}{crcrrcc}
      \big(
      \tfrac{10}{11}
      &
      \paramOnePrime
      &
      +
      \tfrac{10}{11}
      &
      \paramTwoPrime
      \big)
      &
      \Gamma_a \Xi^{(32)}
      &=&
      0
      \\
      \big(
      -
      &
      \paramOnePrime
      &
      -
      &
      \paramTwoPrime
      \big)
      &
      \Xi^{(320)}_a
       &=&
    0
    \mathrlap{\,,}
    \end{array}
  $$
  which is clearly solved as claimed.

For the second equation, we have
$$
    \def\arraystretch{1.6}
    \begin{array}{lll}
      \Gamma_{a_1 \cdots a_4\, b}
      \psi
      \big(\hspace{1pt}
        \overline{\psi}
        \,\Gamma^b\,
        \psi
      \big)
      &
      \;=\;
      \tfrac{1}{11}
      \Gamma_{a_1 \cdots a_4\,b}
      \Gamma^b
      \Xi^{(32)}
      \,+\,
      \Gamma_{a_1 \cdots a_4\, b}
      \Xi^{(320)\, b}
      \\
      &\;=\;
      \tfrac{7}{11}
      \,
      \Gamma_{a_1 \cdots a_4}
      \Xi^{(32)}
      \,-\,
      4
      \,
      \Gamma_{[a_1 a_2}
      \Xi^{(320)}_{a_3 a_4]}\;,
  %   \end{array}
  %   \hspace{.7cm}
  % \def\arraystretch{1.5}
  % \begin{array}{ll}
  \\[5pt]
    \Gamma_{[a_1 a_2}
    \psi
    \big(\hspace{1pt}
      \overline{\psi}
      \,\Gamma_{a_3 a_4]}\,
      \psi
    \big)
    &
    \;=\;
    \Gamma_{[a_1 a_2}
    \big(
      \tfrac{1}{11}
      \Gamma_{a_3 a_4]}
      \Xi^{(32)}
      -
      \tfrac{2}{9}
      \Gamma_{a_3}\Xi^{(320)}_{a_4]}
      +
      \Xi^{(1408)}_{a_3 a_4]}
    \big)
    \\
    &\;=\;
    \tfrac{1}{11}
    \Gamma_{a_1 \cdots a_4}
    \Xi^{(32)}
    \,-\,
    \tfrac{2}{9}
    \Gamma_{[a_1 a_2 a_3}
    \Xi^{(320)}_{a_4]}
    +
    \Gamma_{[a_1 a_2}
    \Xi^{(1408)}_{a_3 a_4]},
%   \end{array}
% $$
% $$
%   \def\arraystretch{1.5}
%   \begin{array}{ll}
   \\[5pt]
    \Gamma^b
    \psi
    \big(\hspace{1pt}
      \overline{\psi}
      \,\Gamma_{a_1 \cdots a_4\, b}\,
      \psi
    \big)
    &
    \;=\;
    -
    \tfrac{1}{77}
    \Gamma^b\Gamma_{a_1 \cdots a_4\, b}
    \Xi^{(32)}
    \,+\,
    \tfrac{5}{9}
    \Gamma^b
    \Gamma_{[a_1 \cdots a_4}
    \Xi^{(320)}_{b]}
    \,+\,
    2
    \,
    \Gamma^b
    \Gamma_{[a_1 a_2 a_3}
    \Xi^{(1408)}_{a_4 \, b]}
    \\
& \quad \;\;   -
    \tfrac{1}{11}
    \Gamma_{a_1 \cdots a_4}
    \Xi^{(32)}
    \,+\,
    \tfrac{24}{9}
    \,
    \Gamma_{[a_1 a_2 a_3}
    \Xi^{(320)}_{a_4]}
    \,+\,
    6
    \,
    \Gamma_{[a_1 a_2}
    \Xi^{(1408)}_{a_3 a_4]}
    \,,
  \end{array}
$$
whence the second condition is equivalent to the following system of linear equations:
$$
  \def\arraystretch{1.6}
  \def\arraycolsep{2pt}
  \begin{array}{crcrcrccc}
    \big(
    \tfrac{7}{11}
    &
    \paramOneDoublePrime
    &
    -
    \tfrac{1}{11}
    &
    \paramTwoDoublePrime
    &
    -
    \tfrac{1}{11}
    &
    \paramFiveDoublePrime
    \big)
    &
    \Gamma_{a_1 \cdots a_4}
    \Xi^{(32)}
    &=&
    0
    \\
    \big(
    -4
    &
    \paramOneDoublePrime
    &
    +
    \tfrac{2}{9}
    &
    \paramTwoDoublePrime
    &
    +
    \tfrac{24}{9}
    &
    \paramFiveDoublePrime
    \big)
    &
    \Gamma_{[a_1 a_2 a_3}
    \Xi^{(320)}_{a_4]}
    &=&
    0
    \\
    \big(
    &
    &
    -
    &
    \paramTwoDoublePrime
    &
    +
    6
    &
    \paramFiveDoublePrime
    \big)
    &
    \Gamma_{[a_1 a_2}
    \Xi^{(1408)}_{a_3 a_4]}
    &=&
    0
  \end{array}
$$
whose solution space is readily seen to be as claimed.
\end{proof}

We now observe that given Fierz relations as in Lem. \ref{CubicFierzRelations}, one immediately obtains cocycles on the basic M-algebra by replacing pairs $\psi \overline{\psi} = (\psi^\alpha \psi_\beta)$ with the bispinorial generator $e = (e^{\alpha}{}_\beta)$\eqref{TheBispinorCEElement}; it follows immediately from \eqref{FurtherFierzRelations} that:
\begin{proposition}[\bf The vector-spinor valued form generator]
\label{TheVectorSpinorGenerator}
  In $\mathrm{CE}(\mathfrak{M})$
  we have
  $$
  \def\arraystretch{1.6}
  \begin{array}{c}
    \mathrm{d}\big(
      \Gamma_{a b}
      \,e\, 
      \Gamma^b
      \psi
      \,+\,
      \Gamma^b
      \,e\,
      \Gamma_{a b}
      \psi
    \big)
    \;=\;
    0\;.
  \end{array}
  $$
  Hence there exists an extension of $\mathrm{CE}(\mathfrak{M})$ by generators 
  $\big(\vectorSpinor_a^\alpha\big)_{
    { \alpha \in \{1, \cdots, 32\} }
    \atop
    { a \in \{0,1, \cdots, 10\} }
  }$ 
  in $\mathrm{deg} \,=\, (1,\mathrm{odd})$ with differential
  \begin{equation}
    \label{TheVectorSpinorExtension}
    \begin{array}{l}
      \mathrm{d}
      \;
      \vectorSpinor_a
      \;\;
      =
      \;\;
      \tfrac{1}{16}
      \big(
      \Gamma_{a b}
      \,e\, 
      \Gamma^b
      \psi
      \,+\,
      \Gamma^b
      \,e\,
      \Gamma_{a b}
      \psi
      \big)
      \,.
    \end{array}
  \end{equation}
\end{proposition}

\begin{example}[\bf Recovering the traditional differential of the vector-spinor valued generator]
\label{RecoveringTheTraditionalDifferentialOfTheVectorSpinorValuedGenerator}
Inserting into \eqref{TheVectorSpinorExtension} the defining expression 
\eqref{TheBispinorCEElement} of the generators $e^{\alpha \beta}$ in terms of the generators $e^a$, $e^{a_1 a_2}$ and $e^{a_1 \cdots a_5}$, and then just performing the resulting Clifford contractions, we get
$$
\def\arraystretch{1.5}
  \begin{array}{l}
    \mathrm{d}
    \,
    \vectorSpinor_a
    \\
    \;=
      \Gamma_{a b}
      \,e\, 
      \Gamma^b
      \psi
      \,+\,
      \Gamma^b
      \,e\,
      \Gamma_{a b}
      \psi
    \\
    \;=
    \tfrac{1}{16}
    \,
      \Gamma_{a b}
      \big(
        \Gamma_c \psi
        \,
        e^c
        +
        \tfrac{1}{2}
        \Gamma_{c_1 c_2}
        \psi
        \,
        e^{c_1 c_2}
        +
        \tfrac{1}{5!}
        \Gamma_{c_1 \cdots c_5}
        \psi
        \,
        e^{c_1 \cdots c_5}
      \big) 
      \Gamma^b
      \psi
      +
      \tfrac{1}{16}
      \,
      \Gamma^b
      \big(
        \Gamma_c \psi
        \,
        e^c
        +
        \tfrac{1}{2}
        \Gamma_{c_1 c_2}
        \psi
        \,
        e^{c_1 c_2}
        +
        \tfrac{1}{5!}
        \Gamma_{c_1 \cdots c_5}
        \psi
        \,
        e^{c_1 \cdots c_5}
      \big) 
      \Gamma_{a b}
      \psi
    \\
    \;=
    \Gamma_{a c}\psi \, e^c
    -
    \Gamma_{c}
    \psi\, 
    e^{ac}   
    +
    0
    \,.
  \end{array}
$$
This recovers the equations given in
\cite[p. 5]{Sezgin97}\cite[(3.19)]{Castellani11}\cite[(2.36)]{Vaula07} (up to normalization conventions).
\end{example}

\medskip

\noindent
{\bf Reducibility of the extra generators.}
The vector-spinor valued generator from Ex. \ref{TheVectorSpinorGenerator} is actually reducible (which seems not to have been remarked before).
Generally, given a tensor spinor 
$\tensorSpinor_a$, we may split it into:
\begin{itemize}[
  topsep=1pt,
  itemsep=2pt
]
\item 
its {\it $\Gamma$-trace} $\Gamma^a \tensorSpinor_a$ (a plain spinor), and 
\item 
its {\it $\Gamma$-trace free part} 
$
  \big(
    \tensorSpinor_a
    \,-\,
    \tfrac{1}{11}
    \Gamma_a
    \Gamma^b\tensorSpinor_b
  \big)
$ (a vector-spinor with vanishing $\Gamma$-trace).
\end{itemize}

\begin{example}
  \label{GammaTraceOfVectorSpinorIsPreviousSpinor}
  The $\Gamma$-trace of the vector-spinor $\vectorSpinor_a$ \eqref{TheVectorSpinorExtension} behaves just as the spinor $\DFSpinor_{{}_{(-6)}}$ \eqref{ExtensionByTwoCocycles}:
  $$
    \def\arraystretch{1.3}
    \begin{array}{lll}
      \mathrm{d}
      (\Gamma^a\vectorSpinor_a)
      &
      \;=\;
      16 \, \Gamma^a
      \big(
        \Gamma_{ac}
        \psi
        \, e^c
        \,-\,
        \Gamma_c
        \psi
        \,
        e^{a c}
      \big)
      &
      \proofstep{
        Ex. \ref{RecoveringTheTraditionalDifferentialOfTheVectorSpinorValuedGenerator}
      }
      \\
      &\;=\;
      16
      \big(
        10 
        \,
        \Gamma_c \psi\, e^c
        -
        \Gamma_{ac}
        \psi
        \,
        e^{ac}
      \big)
      \,.
    \end{array}
  $$
\end{example}

\smallskip

\noindent
{\bf Further terms in the super-invariant 3-form.}
With the further vector-spinor valued generator \eqref{TheVectorSpinorExtension} included, there is a further term that
may be added to the ansatz \eqref{AnsatzForH30} for $\widehat{P}_3$, namely proportional to
\vspace{2mm}
\begin{equation}
  \label{VectorSpinorContributionToH3}
    \big(\,
      \overline{\psi}
      \,\Gamma^{a b}\,
      \vectorSpinor_a
    \big)
    e_b
    \;-\;
    \big(\,
      \overline{\psi}
      \,\Gamma_{b}\,
      \vectorSpinor_a
    \big)
    e^{a b}
  \;\;
  \in
  \;\;
  \mathrm{CE}\big(\,
\widehat{\vphantom{\rule{1pt}{5.5pt}}\smash{\widehat{\mathfrak{M}}}}
  \,\big)
  \,.
\end{equation}
Here, the relative factor between these two summands is already fixed by the requirement that in the differential of this term the summands proportional to $\tensorSpinor_a$ cancel out among each other, analogous to the dark-green terms proportional to $\DFSpinor$ in \eqref{EquationsSpecifyingRestrictedH30Coefficients}. Namely by \eqref{FurtherFierzRelations} the following term over the brace vanishes:
\begin{equation}
  \label{CheckingTheFormOfTherVectorSpinorContributionToH3}
  \def\arraystretch{1.6}
  \begin{array}{l}
  \mathrm{d}
  \Big(
    \big(\,
      \overline{\vectorSpinor}_a
      \,\Gamma^{a b}\,
      \psi
    \big)
    e_b
    \;-\;
    \big(\,
      \overline{\vectorSpinor}_a
      \,\Gamma_{b}\,
      \psi
    \big)
    e^{a b}
  \Big)
  \\
  \;=\;
  \underbrace{
  \Big(
    \big(\,
      \overline{\vectorSpinor}_a
      \,\Gamma^{a b}\,
      \psi
    \big)
    \big(\hspace{1pt}
      \overline{\psi}
      \,\Gamma_b\,
      \psi
    \big)
    \;+\;
    \big(\,
      \overline{\vectorSpinor}_a
      \,\Gamma_{b}\,
      \psi
    \big)
    \big(\hspace{1pt}
      \overline{\psi}
      \,\Gamma^{ab}\,
      \psi
    \big)
    \Big)
  }_{ \color{gray} = \, 0 }
  \;-\;
  \Big(
    \big(\,
      \overline{\vectorSpinor}
      \,\Gamma^{a b}\,
      \mathrm{d}\psi_a
    \big)
    e_b
    \;-\;
    \big(\,
      \overline{\vectorSpinor}
      \,\Gamma_{b}\,
      \mathrm{d}\psi_a
    \big)
    e^{a b}
    \Big)
    \,.
  \end{array}
\end{equation}

\smallskip 
\begin{proposition}[\bf Three-form with vector-spinor]
\label{ThreeFormWithVectorSpinor}
With the vector-spinor contribution \eqref{VectorSpinorContributionToH3} adjoined to the ansatz \eqref{AnsatzForH30} parameterized by $\betapOne \in \mathbb{R}$, 
\begin{equation}
  \label{GeneralizedAnsatzForH30}
  \def\arraystretch{1.5}
  \begin{array}{l}
    \mathllap{
      \widehat{P}_3
      \;:=\;
    }
    \phantom{+}\;\;
    \alphaZero
    \,
    e_{
      {\color{darkblue} a_1 } 
      {\color{darkorange} a_2}
    }
    \, e^{
      \color{darkblue} a_1
    } 
    \, 
    e^{
     \color{darkorange} a_2
    }
    \\
    \;+\;
    \alphaOne
    \, e^{\color{darkgreen}a_1}{}_{\color{darkblue}a_2}
    \, e^{\color{darkblue}a_2}{}_{\color{darkorange}a_3}
    \, e^{\color{darkorange}a_3}{}_{\color{darkgreen}a_1}
    \\
    \;+\;
    \alphaTwo
    \, 
    e^{
      {
        \color{darkblue} a_1 \cdots a_4} 
        \color{darkorange} b_1
      }
    \, e_{
        \color{darkorange}b_1
      }{}^{
        \color{darkgreen}
        b_2
      }
    \, e_{
      { \color{darkgreen} b_2 }
      { \color{darkblue} a_1 \cdots a_4}
    }
    \\
    \;+\;
    \alphaThree
    \,
    \epsilon_{
      { \color{darkblue} a_1 \cdots a_5 }
      { \color{darkorange} b_1 \cdots b_5 }
      { \color{darkgreen} c}
    }
    \,
    e^{
     \color{darkblue}
     a_1 \cdots a_5
    }
    e^{
     \color{darkorange}
     b_1 \cdots b_5
    }
    e^{
      \color{darkgreen}
      c
    }
    \\
    \;+\;
    \alphaFour
    \,
    \epsilon_{
      { \color{darkblue} a_1 a_2 a_3 }
      { \color{darkorange} b_1 b_2 b_3 }
      { \color{darkgreen} c_1 \cdots c_5 }
    }
    \,
    e^{
      {\color{darkblue} a_1 a_2 a_3}
      {\color{purple} d_1 d_2}
    }
    \,
    e_{
      {\color{purple} d_1 d_2}      
     }
     {}
     ^{
      {\color{darkorange} b_1 b_2 b_3}
    }
    \,
    e^{
      \color{darkgreen}
      c_1 \cdots c_5\
    }
    \\
    \;+\;
    \betaOne
    \,
    \big(\,
      \overline{\psi}
      \,\Gamma
        _{\color{darkblue} a}
      \,
      \DFSpinor
    \big)
    e^{\color{darkblue}a}
    \\
    \;+\;
    \betaTwo 
    \,
    \big(\,
      \overline{\psi}
      \,\Gamma
        _{\color{darkblue}a_1 a_2}\,
      \DFSpinor
    \big)
    \,
    e^{\color{darkblue} a_1 a_2}
    \\
    \;+\;
    \betaThree 
    \,
    \big(\,
      \overline{\psi}
      \,\Gamma
        _{\color{darkblue}a_1 \cdots a_5}\,
      \DFSpinor
    \big)
    \,
    e^{
      \color{darkblue}
      a_1 \cdots a_5
    }
    \\[+5pt]
    \;+\;
    \mathcolor{purple}{\betapOne}
    \Big(
    \big(\,
      \overline{\psi}
      \,\Gamma
        ^{
        {\color{darkblue}a} 
        {\color{darkorange}b}
        }\,
      \vectorSpinor
        _{\color{darkblue}a}
    \big)
    e_{\color{darkorange}b}
    \;-\;
    \big(\,
      \overline{\psi}
      \,\Gamma_{\color{darkorange}b}\,
      \vectorSpinor_{\color{darkblue}}
    \big)
    e^{
      {\color{darkblue}a} 
      {\color{darkorange}b}
    }
    \Big)
    \mathrlap{\,,}
  \end{array}
\end{equation}
the Bianchi identity \eqref{H30BianchiIdentity}
is solved, in addition to the previous solution \eqref{RestrictedParametersOfH3FluxAvatar} with $\betapOne = 0$, by
\begin{equation}
  \label{FurtherSolutionForH3}
  \def\arraystretch{1.1}
  \begin{array}{clcl}
  \alphaZero
  &=&
  -1/20
  \\
  \alphaOne 
  &=&
  -1/60
  \\
  \alphaTwo  
  &=&
  0
  \\
  \alphaThree  
  &=&
  0
  \\
  \alphaFour  
  &=&
  0
  \\
  \betaOne
  &=&
  0
  \\
  \betaTwo
  &=&
  0
  \\
  \betaThree
  &=&
  0
  \\
  \betapOne
  &=&
  -1/20
  \,,
  \end{array}
\end{equation}
and the convex combinations of these two solutions, \eqref{RestrictedParametersOfH3FluxAvatar} and \eqref{FurtherSolutionForH3},
exhaust the space of all solutions.
\end{proposition}
\begin{proof}
The differential of the last summand in \eqref{GeneralizedAnsatzForH30} is (showing the computation in small steps in order to secure the signs):
$$
  \def\arraystretch{1.7}
  \begin{array}{ll}
    \mathrm{d}
    \Big(
    \big(\,
      \overline{\psi}
      \,\Gamma^{a b}\,
      \vectorSpinor_a
    \big)
    e_b
    \;-\;
    \big(\,
      \overline{\psi}
      \,\Gamma_{b}\,
      \vectorSpinor_a
    \big)
    e^{a b}
    \Big)
    \\
    \;=\;
    -
    \big(\,
      \overline{\psi}
      \,\Gamma^{a b}\,
      \mathrm{d}
      \vectorSpinor_a
    \big)
    e_b
    \;+\;
    \big(\,
      \overline{\psi}
      \,\Gamma_{b}\,
      \mathrm{d}
      \vectorSpinor_a
    \big)
    e^{a b}
    \Big)
    &
    \proofstep{
      by \eqref{CheckingTheFormOfTherVectorSpinorContributionToH3}
    }
    \\
    \;=\;
    -
    \big(\,
      \overline{\psi}
      \,\Gamma^{a b}\,
      (
        \Gamma_{a c}
        \psi
        e^c
        -
        \Gamma^c 
        \psi
        e_{ac}
      )
    \big)
    e_b
    \;+\;
    \big(\,
      \overline{\psi}
      \,\Gamma_{b}\,
      (
        \Gamma_{a c}
        \psi
        e^c
        -
        \Gamma^c 
        \psi
        e_{ac}
      )
    \big)
    e^{a b}
    \Big)
    &
    \proofstep{
      by
      Ex. \ref{RecoveringTheTraditionalDifferentialOfTheVectorSpinorValuedGenerator}
    }
    \\
    \;=\;
    -
    \big(\,
      \overline{\psi}
      \,\Gamma^{a b}
        \Gamma_{a c}
        \psi
      \big)
      e^c
      e_b
      +
    \big(\,
      \overline{\psi}
      \,\Gamma^{a b}
        \Gamma^c 
        \psi
    \big)
    e_{ac}
    e_b
    \;+\;
    \big(\,
      \overline{\psi}
      \,\Gamma_{b}
        \Gamma_{a c}
        \psi
    \big)
     e^c
     e^{a b}        
        -
    \big(\,
      \overline{\psi}
      \,\Gamma_{b}
        \Gamma^c 
        \psi
     \big)
      e_{ac}
      e^{a b}
    \\
    \;=\;
    -(-9)
    \big(\,
      \overline{\psi}
        \,\Gamma_{b c}\,
        \psi
      \big)
      e^c
      e^b
      +
    \big(\,
      \overline{\psi}
      \,\Gamma^a\,
        \psi
    \big)
    e_{ac}
    e^c
    \;-\;
    \big(\,
      \overline{\psi}
      \,\Gamma_a\,
        \psi
    \big)
     e_b
     e^{a b}        
        -
    \big(\,
      \overline{\psi}
      \,\Gamma^{bc}\,
        \psi
     \big)
      e_{a c}
      e^{a}{}_{b}    
    \\
    \;=\;
    -9
    \big(\,
      \overline{\psi}
        \,\Gamma_{b c}\,
        \psi
      \big)
      e^b
      e^c
      +
    \big(\,
      \overline{\psi}
      \,\Gamma^a\,
        \psi
    \big)
    e_{a b}
    e^b
    \;+\;
    \big(\,
      \overline{\psi}
      \,\Gamma_a\,
        \psi
    \big)
     e_{a b}        
     e^b
    +
    \big(\,
      \overline{\psi}
      \,\Gamma^{bc}\,
      \psi
     \big)
      e_{c a}
      e^{a}{}_{b}    
      \,,
  \end{array}
$$
where we used manipulations such as
\begin{equation}
  \def\arraystretch{1.6}
  \begin{array}{lll}
  \big(\,
    \overline{\psi}
    \,
    \Gamma^{ab}
    \Gamma^c
    \,
    \psi
  \big)
  e_{ac} e_b
  &
  \;=\;
  \big(\,
    \overline{\psi}
    \,
    (\eta^{bc}\Gamma^a
    - \eta^{ac}\Gamma^b
    +
    \Gamma^{abc})
    \,
    \psi
  \big)
  e_{ac} e_b
  &
  \proofstep{
    by \eqref{GeneralCliffordProduct}
  }
  \\
  &\;=\;
  \big(\,
    \overline{\psi}
    \,
    \eta^{bc}\Gamma^a    
    \,
    \psi
  \big)
  e_{ac} e_b
  &
  \proofstep{
    by 
    \eqref{SkewSpinorPairings}
    .
  }
  \end{array}
\end{equation}
Therefore the system of linear equations \eqref{EquationsSpecifyingRestrictedH30Coefficients} to be solved
generalizes to picking up the following boxed terms
\begin{equation}
  \label{}
  \mathrm{d}\, \widehat{P}_3
  \;=\;
  \tfrac{1}{2}
  \big(\,
    \overline{\psi}
    \,\Gamma_{a_1 a_2}\,
    \psi
  \big)
  \,
  e^{a_1}\, e^{a_2}
  \;\;\;\;
  \Leftrightarrow
  \;\;\;\;
  \left\{\!\!
  \def\arraystretch{1.1}
  \begin{array}{r}
    -\alphaZero 
    + 
    \paramOne\, \betaOne 
    \,
    \fbox{$-\, 9\, \betapOne$}
    \;=\;
    \tfrac{1}{2}
    \\[+5pt]
    \color{darkblue}
    -2
    \,
    \alphaZero
    \;+\;
    2
    \, 
    \paramTwo 
    \betaOne
    \;+\;
    2
    \,
    \paramOne 
    \betaTwo
    \,
    \fbox{$+\,2\betapOne$}
    \;=\;
    0
    \\[+4pt]
    \color{darkorange}
    -3
    \,
    \alphaOne
    \;-\;
    4
    \paramTwo
    \betaTwo
    \,
    \fbox{$+\,\betapOne$}
    \;=\;
    0
    \\
    \color{olive}
    2 
    \, 
    \alphaTwo
    +
    10
    \,
    \paramFive
    \,
    \betaTwo
    +
    10
    \,
    \paramTwo
    \betaThree
    \;=\;
    0
    \\
    \color{orange}
    \alphaTwo
    \,+\,
    600
    \,
    \paramFive
    \,
    \betaThree
    \;=\;0
    \\
    {
    \color{blue}
    2 
    \, 
    \alphaThree
    +    
    \tfrac
      { \paramFive }
      { 5! }
    \,
    \betaOne
    +
    \tfrac
      {\paramOne}
      {5!}
    \,
    \betaThree
    \;=\;
    0
    }
    \\
    \color{brown}
    \alphaThree 
    \,+\,
    \paramFive
    \,
    \betaThree
    \;=\;
    0
    \\
    \color{magenta}
    3
    \,
    \alphaFour
    \,-\,
    \tfrac
      {200}
      {5!}
    \paramFive
    \,
    \betaThree
    \;=\;
    0
    \\
    \color{darkgreen}
    \betaOne
    +
    10
      \cdot
    \betaTwo
    -
    6! 
      \cdot 
    \betaThree
    \;=\;
    0
    \mathrlap{\,,}
  \end{array}
  \right.
\end{equation}
By mechanical computation \cite{AncillaryFiles} this system is solved as claimed in \eqref{FurtherSolutionForH3}. 
\end{proof}

%%%%%%%%%%%%%%%%%%%%%%%%%%%%%%%%%%%%%%%%%
\section{The M-group}
\label{AsASuperlieGroup}
%%%%%%%%%%%%%%%%%%%%%%%%%%%%%%%%%%%%%%%%%

We now turn to promoting the hidden M-algebra (\S\ref{SuperExceptionalAsLieAlgebra}) --- which is ``just'' a super-Lie algebra --- to an actual group, hence to a super-Lie group (Def. \ref{SuperLieGroup}), to be called the hidden M-group (Ex. \ref{IntegratingSuperExceptionalMinkowsliLieAlgebra} below).
The main effect here is that (in contrast to the case of the basic M-algebra) the ``hidden'' fermionic extension makes, via the Dynkin formula (the Hausdorff series), a trilinear fermionic term appear, first in the group product operation \eqref{TheMSuperGroup}
and thereby in the Maurer-Cartan form \eqref{MCFormOnHiddenMGroup}
and thereby finally in the coordinate expression for the super-invariant 3-form.

\medskip
To make this important point rigorous, we develop, along the way, the relevant notions of super-Lie group theory in a streamlined form that should be satisfactory both for physicists and mathematicians.

\smallskip

%%%%%%%%%%%%%%%%%%%%%%%%%%%%%%%%%%%%%%%
\subsection{Super-Lie groups}
%%%%%%%%%%%%%%%%%%%%%%%%%%%%%%%%%%%%%%%%

Our notation for super-geometry follows \cite[\S 2.1]{GSS24-SuGra}, to which we refer for background and references.

\medskip

\noindent
{\bf Super-Manifolds.}
In view of {\it Batchelor's theorem} \cite{Batchelor79}\cite[\S 1.1.3]{Batchelor84} and {\it Milnor's exercise} \cite[\S 35.8-10]{KMS93}, 
we may considerably shortcut the definition of super-manifolds to the following:

\begin{definition}[\bf Category of supermanifolds]
  \label{CategoryOfSupermanifolds}
  The category of (smooth, real) super-manifolds is the full subcategory of the opposite of super-commutative $\mathbb{R}$-algebras on those objects which are $C^\infty(B)$-Grassmann algebras of smooth sections $\Gamma_B$ of a smooth vector bundle $V$ over a smooth manifold $B$ (the {\it bosonic body} of the supermanifold):
  \begin{equation}
    \label{FullInclusionOfSupermanifoldsIntoDualsOfAlgebras}
    \begin{tikzcd}[
      sep=0pt
    ]
      \mathrm{sSmthMfd}
      \ar[
        rr,
        hook,
        "{
          C^\infty(-)
        }"
      ]
      &&
      \mathrm{sCAlg}
        ^{\mathrm{op}}
        _{\mathbb{R}}
      \\[-7pt]
      X \,\defneq\,
      B\vert V_{\mathrm{odd}}
      &\qquad \longmapsto&
      \qquad 
      \def\arraystretch{1.4}
      \begin{array}{r}
      \wedge
        ^\bullet
        _{ C^\infty(B) }
      \Gamma_B(V^\ast)
      %\\
      \;=\;
      \Gamma_B\big(
        \wedge^\bullet_B
        V^\ast
      \big)
      \,.
      \end{array}
    \end{tikzcd}
  \end{equation}
This means that for a pair of supermanifolds $X^{(1)}, X^{(2)}$, the maps (morphisms) between them are in bijection to reverse super-algebra homomorphisms between their algebras of smooth functions (cf. \cite[Prop. 2.2]{HKST11}) according to \eqref{FullInclusionOfSupermanifoldsIntoDualsOfAlgebras}:
\begin{equation}
  \label{MapsOfSupermanifoldsAsAlgebraHomomorphisms}
  \big\{
    f \,:\,
    X^{(1)}
      \xrightarrow{\;}
    X^{(2)}
  \big\}
  \;\;
  \simeq
  \;\;
  \big\{
    C^\infty\big(X^{(1)}\big)
    \xleftarrow{\;}
    C^\infty\big(X^{(2)}\big)
    \,:\,
    f^\ast
  \big\}
  \,.
\end{equation}
\end{definition}

\smallskip

The archetypical examples of super-manifolds:
\begin{example}[\bf Ordinary smooth manifolds among super-manifolds]
\label{OrdinaryAmongSuperManifolds}
An ordinary smooth manifold $X \in \mathrm{SmthMfd}$ is a super-manifold via its ordinary algebra of smooth functions, $C^\infty(X)$, regarded as a super-commutative algebra without odd elements. 
This identification constitutes a full subcategory inclusion of ordinary into super-manifolds:
$$
  \begin{tikzcd}[row sep=25pt, column sep=large]
    \mathrm{SmthMfd}
    \ar[r, hook]
    \ar[
      d, 
      hook,
      "{ C^\infty(-) }"
    ]
    &
    \mathrm{sSmthMfd}
    \ar[
      d, 
      hook, 
      "{ C^\infty(-) }"
    ]
    \\[-5pt]
    \mathrm{CAlg}
      ^{\mathrm{op}}
      _{\mathbb{R}}
    \ar[r, hook]
    &
    \mathrm{sCAlg}
      ^{\mathrm{op}}
      _{\mathbb{R}}
  \end{tikzcd}
$$
\end{example}
\begin{example}[\bf Super-points]
\label{SuperPoints}
For $q \in \mathbb{R}$, the {\it super-point} $\mathbb{R}^{0\vert q}$ is the supermanifold (Def. \ref{CategoryOfSupermanifolds}) whose bosonic body is the point, $\bosonic{\mathbb{R}}{}^{0\vert q} = \ast$, equipped with the $q$-dimensional fermionic fiber space, so that its algebra of smooth functions is the ordinary Grassmann algebra on $q$ generators:
$$
  C^\infty\big(
    \mathbb{R}^{0\vert q}
  \big)
  \;:=\;
  \wedge^\bullet_{\mathbb{R}}
  (\mathbb{R}^q)^\ast
  \;\simeq\;
  \mathbb{R}\big[
    \vartheta^1, \cdots,  
    \vartheta^q
  \big]
  \,,
  \;\;\;\;\;
  \underset{i}{\forall}
  \;\;\;
  \mathrm{deg}(\vartheta^i)
  \,=\,
  \mathrm{odd}
  \,.
$$
For $n \in \mathbb{N}$ we will abbreviate
\begin{equation}
  \label{ProductOfOddParameters}
  \vartheta^{i_1 i_2 \cdots i_n}
  \;:=\;
  \vartheta^{i_1}
  \vartheta^{i_2}
  \cdots
  \vartheta^{i_n}
  \;=\;
  \epsilon^{i_1 i_2 \cdots i_n}
  \,
  \vartheta^{1}
  \vartheta^{2}
  \cdots
  \vartheta^{n}
  \;\;
  \;\in\;
  C^\infty\big(
    \mathbb{R}^{0\vert q}
  \big)
  \,.
\end{equation}

We denote the full subcategory of super-points among all supermanifolds by
\begin{equation}
  \label{SubCategoryOfSuperpoints}
  \begin{tikzcd}
    \mathrm{sPnt}
    \ar[
      r,
      hook
    ]
    &
    \mathrm{sMfd}
  \end{tikzcd}
\end{equation}
\end{example}
\begin{example}[\bf Super-Cartesian spaces]
For $p,q \in \mathbb{N}$, the {\it super Cartesian space} $\mathbb{R}^{p\vert q}$ is, as a super-manifold (Def. \ref{CategoryOfSupermanifolds}), 
the Cartesian product of the ordinary manifold $\mathbb{R}^p$ (via Ex. \ref{OrdinaryAmongSuperManifolds}) with the super-point $\mathbb{R}^{0\vert q}$ (Ex. \ref{SuperPoints})
$$
  \mathbb{R}^{p\vert q}
  \;=\;
  \mathbb{R}^p \times 
  \mathbb{R}^{0\vert q}
$$

\vspace{-1mm} 
\noindent hence whose algebra of smooth functions is
$$
  C^\infty\big(
    \mathbb{R}^{p\vert q}
  \big)
  \,=\,
  C^\infty(\mathbb{R}^p)
  \otimes_{{}_{\mathbb{R}}}
  C^\infty(\mathbb{R}^{0\vert q})
  \;\simeq\;
  C^\infty(\mathbb{R}^p)
  \big[
    \vartheta^1
    ,\cdots,
    \vartheta^q
  \big]
  \,.
$$
\end{example}

We will need a generalization of the following example (e.g. \cite[\S 3.1]{KS05}\cite[Ex. 2.1, Prop. 3.1]{HKST11}\cite[Ex. 5.3]{CarchediRoytenberg12}):

\begin{example}[\bf The odd tangent bundle]
  \label{TheOddTangentBundle}
  For $X \in \mathrm{SmthMfd}$, the total space of its {\it odd-tangent bundle} is the super-manifold 
  whose super-algebra of smooth functions is the de Rham algebra of differential forms on $X$ with the even/odd degree forms in even/odd super-degree, respectively:
  $$
    \Todd X
    \;:=\;
    X \vert T X
    \,,
    \hspace{.7cm}
    C^\infty\big(
     \Todd X
    \big)
    \;=\;
    \Omega^\bullet_{\mathrm{dR}}(X)
    \,.
  $$   
  Consider more generally a super-manifold $X$ which, just for simplicity of presentation, we take to be super-Cartesian $X \defneq \mathbb{R}^{d \vert N}$. Then a map of super-manifolds from the point
  \begin{equation}
    \label{MapFromPointToSuperCartesian}
    \begin{tikzcd}[row sep=-2pt, column sep=small]
      \mathbb{R}^{0\vert 0}
      \ar[rr, "{ x_0 }"]
      &&
      X
      \\
      \mathbb{R}
      \ar[
        rr,
        <-,
        "{  (x_0)^\ast  }"
      ]
      &&
      C^\infty(\mathbb{R}^d)
      \mathrlap{
      \otimes 
      \mathbb{R}[\theta^1, \cdots, \theta^N]
      }
      \\
      x^a_0
      &\longmapsfrom&
      x^a
      \\[-2pt]
      0
      &\longmapsfrom&
      \theta^\alpha
    \end{tikzcd}
  \end{equation}
  is equivalently the choice of a point $x_0 \in \bosonic X = \mathbb{R}^d$ in the bosonic body of $X$, hence a $d$-tuple of real numbers, while a map from the first-order super-point
  \vspace{-2mm} 
  \begin{equation}
    \label{MapFromFirstSuperpointToSuperCartesian}
    \begin{tikzcd}[row sep=0pt, column sep=small]
      \mathbb{R}^{0\vert 1}
      \ar[
        rr,
        "{
          (x_0, \theta_1)
        }"
      ]
      &&
      X
      \\
      \mathbb{R}[\vartheta^1]
      \ar[
        rr,
        <-,
        "{
          (x_0, \theta_1)^\ast
        }"
      ]
      &&
      C^\infty(\mathbb{R}^d)
      \mathrlap{
        \otimes
        \mathbb{R}[\theta^1, \cdots, \theta^N]
      }
      \\[-2pt]
      x^a_0 &\longmapsfrom&
      x^a
      \\[-3pt]
      \theta^\alpha_1 \vartheta^1
      &\longmapsfrom&
      \theta^\alpha
    \end{tikzcd}
  \end{equation}
  is specified in addition by an $N$-tuple of real numbers $\big(\theta^\alpha_1 \in \mathbb{R}\big)_{\alpha=1}^N$ to be thought of as defining an ``odd tangent vector'' at $x_0$ in $X$.
  The manifold formed by these super-points in $X$ is the bosonic body of the odd-tangent bundle of $X$:
  $$
    C^\infty\big(
      \longbosonic\Todd\;\,
      \mathbb{R}^{d\vert N}
    \big)
    \;\simeq\;
    C^\infty\big(\mathbb{R}^{d+N}\big)
  $$
  coordinatized by $(x^a_0)_{a=1}^d$ and $\big(\theta^\alpha_1\big)_{\alpha=1}^N$.
Thereby the odd coordinates of the original super-manifold $X$ are detected by ordinary bosonic coordinates on the bosonic body $\longbosonic\Todd \;X$ of its odd tangent bundle. 
However, $\longbosonic \Todd \; X$ sees only the linearization of {\it maps} 
$f : X \xrightarrow{\;} X^{\pr}$ between supermanifolds:
\vspace{-2mm} 
$$
  \begin{tikzcd}[
    row sep=0pt, column sep=large
  ]
    \mathbb{R}^{0\,\vert\,1}
    \ar[
      rr,
      "{ 
        (x_0,\, \theta_1) 
      }"
    ]
    \ar[
      rrrr,
      rounded corners,
      to path={ 
           ([yshift=+00pt]\tikztostart.north)  
        -- ([yshift=+08pt]\tikztostart.north)  
        -- node[yshift=6pt]{
          \scalebox{.7}{$
             (x^{\pr}_0, \theta^{\pr}_1)
             \,\defneq\,
             f_\ast(x_0, \theta_1)
          $}
        }
           ([yshift=+10pt]\tikztotarget.north)  
        -- ([yshift=+00pt]\tikztotarget.north)  
      }
    ]
    &&
    X
    \ar[
      rr,
      "{f}"
    ]
    &&
    X^{\pr}
    \\
    \grayunderbrace{
    f^\alpha_{\beta_1}\hspace{-1pt}(x_0)
    }{
      \mathclap{
      \scalebox{.7}{
        \color{darkblue}
        \bf
        \def\arraystretch{.9}
        \begin{tabular}{c}
          Only linear contribution is
          \\
          seen on this super-point
        \end{tabular}
      }
      }
    }
    \theta^{\beta_1}_1
    \cdot
    \vartheta^1
    &\longmapsfrom&
    \grayunderbrace{
    \textstyle{\sum_k}
      f
        ^\alpha
        _{\beta_1 \cdots \beta_{2k+1}}\hspace{-1pt}(x)
    }{
      \mathclap{
      \scalebox{.7}{
        \bf
        \color{darkblue}
        \def\arraystretch{.9}
        \begin{tabular}{c}
          Full polynomial effect of
          \\
          map on odd coordinates
        \end{tabular}
      }
      }
    }
     \cdot
     \,
      \theta^{\beta_1 \cdots \beta_{2k+1}}
    &\longmapsfrom&
    \theta^{\pr\, \alpha}
  \end{tikzcd}
$$

\end{example}

\vspace{-1mm} 
However,  to detect also the higher polynomial effects of maps, there is the following evident generalization of the odd tangent bundle to higher order in the odd coordinates (cf. also \cite{KS05}):
\begin{example}[\bf Odd higher tangent bundles]
  In generalization of \eqref{MapFromFirstSuperpointToSuperCartesian}, a map from the $q$th super-point to a super-Cartesian manifold $X \defneq \mathbb{R}^{d\vert N}$ 
\vspace{-2mm} 
  $$
    \begin{tikzcd}[
      row sep=-2pt,
      column sep=huge
    ]
      \mathbb{R}^{0\vert q}
      \ar[
        rr,
        "{
          \scalebox{1}{$
          \scalebox{1.3}{$($}
            x_{i_1 \cdots i_{2k}}
            ,\,
            \theta_{i_1 \cdots i_{2k+1}}
          \scalebox{1.3}{$)$}
            _{k \leq q/2}
          $}
        }"
      ]
      &&
      X
      \\
      \mathbb{R}[
        \vartheta^1,
          \cdots, 
        \vartheta^q
      ]
      \ar[
        rr,
        <-
      ]
      &&
      C^\infty(\mathbb{R}^d)
      \mathrlap{
        \otimes
        \mathbb{R}[\theta^1, \cdots, \theta^N]
      }
      \\
      \sum_k
        x^a_{i_1 \cdots i_{2k}}
        \vartheta^{i_1 \cdots i_{2k}}
      &\longmapsfrom&
      x^a
      \\
      \sum_k
        \theta^\alpha
          _{i_1 \cdots i_{2k+1}}
        \vartheta^{i_1 \cdots i_{2k+1}}     
      &\longmapsfrom&
      \theta^\alpha
    \end{tikzcd}
  $$
  is specified by tuples of real numbers $x^a_{i_1 \cdots i_{2k}} = x^a_{[i_1 \cdots i_{2k}]} \in \mathbb{R}$ and $\theta^\alpha_{i_1 \cdots i_{2k+1}} = \theta^\alpha_{[i_1 \cdots i_{2k+1}]} \in \mathbb{R}$ which encode 
  \begin{itemize}[
    leftmargin=1cm,
    itemsep=2pt,
    topsep=1pt
  ]
    \item[\bf (i)] a point $x$ in $X$, 
    \item[\bf (ii)] an odd tangent vector $\theta_1$ at this point, 
    \item[\bf (iii)] a $\big({q \atop 2}\big)$-tuple of actual tangent vectors $x_{i_1 i_2}$ at this point,
    \item[\bf (iv)] a $\big({q \atop 3}\big)$-tuple of {\it odd 2-jets} $\theta_{i_1 i_2 i_3}$ at this point,
    \item[\bf (v)] a $\big({q \atop 4}\big)$-tuple of actual 2-jets $x_{i_1 \cdots i_4}$ at the point, 
    
    etc.
  \end{itemize}
  These are coordinates on the bosonic body of the odd super-geometric version of what in the terminology of \cite[Rem. 1.14]{MoerdijkReyes91}) is a {\it prolongation} or {\it generalized jet bundle} (cf. \cite{KS17}) super-manifold: \footnote{
    The super-algebra $C^\infty\big(\Todd^{{}_{(q)}}X\big) $ is called in \cite{KS05} the algebra of {\it differential worms} on $X$.
  }
  $$
    C^\infty\Big(
      \bosonicTodd^{{}_{(q)}}
      \mathbb{R}^{d\vert N}
    \Big)
    \;\simeq\;
    C^\infty\Big(
      \adjustbox{
        raise=-3pt
      }{$
      \mathbb{R}^{
        \big(
        d \sum_k\binom{q}{2k}
        +
        N \sum_k\binom{q}{2k+1}
        \big)
      }
      $}
    \Big)
    \,.
  $$
  These higher order coordinates serve to detect higher polynomial components of odd coordinates under {\it maps} between supermanifolds.
  For instance, the action $f_\ast$ of a quadratic map $f : X \xrightarrow{\;} X$ on the coordinate functions on $\Todd X$ is 
  \vspace{-3mm} 
  \begin{equation}
    \label{QuadraticMapUnderBosonicTodd}
    \hspace{-5mm} 
    \begin{tikzcd}[
      row sep=0pt, 
      column sep=5, 
      ampersand replacement=\&
    ]
    \mathbb{R}^{0\vert q}
    \ar[
      rr,
      "{
          \scalebox{1}{$
          \scalebox{1.3}{$($}
            x_{i_1 \cdots i_{2k}}
            ,\,
            \theta_{i_1 \cdots i_{2k+1}}
          \scalebox{1.3}{$)$}
            _{k =0}
            ^{\lfloor q/2 \rfloor}
          $}
      }"
    ]
    \ar[
      rrrr,
      rounded corners,
      to path={ 
           ([yshift=+00pt]\tikztostart.north)  
        -- ([yshift=+10pt]\tikztostart.north)
        -- node[yshift=9pt] {
           \scalebox{.7}{$
           f_\ast
          \scalebox{1}{$
          \scalebox{1.3}{$($}
            x_{i_1 \cdots i_{2k}}
            ,\,
            \theta_{i_1 \cdots i_{2k+1}}
          \scalebox{1.3}{$)$}
            _{k =0}
            ^{\lfloor q/2 \rfloor}
          $}             
           $}
        }
           ([yshift=+10pt]\tikztotarget.north)
        -- ([yshift=+00pt]\tikztotarget.north)
      }
    ]
    \&\&
      \mathbb{R}
        ^{d\vert N}
      \ar[
        rr,
        "{ f }"
      ]
      \&\&
      \mathbb{R}
        ^{d\vert N}
      \\
      \def\arraystretch{1.4}
      \def\arraycolsep{0pt}
      \begin{array}{rl}
        f^a_b
        \sum_{k=0}^{\lfloor q/2\rfloor}
        x^b_{i_1 \cdots i_{2k}}
        &
         \vartheta^{i_1 \cdots i_{2k}}
        \\
        + 
        f^a_{b_1 b_2}
        \sum_{k=0}^{\lfloor q/2\rfloor}
        \sum_{k^{\pr}=0}^k
        x^{b_1}
          _{i_1 \cdots i_{2k'}}
        x^{b_2}
          _{i_{2k'+1} \cdots i_{2k}}
        &
         \vartheta^{i_1 \cdots i_{2k}}
        \\
        + f^a_{\beta_1 \beta_2}
        \sum_{k=0}^{\lfloor q/2\rfloor}
        \sum_{k'=0}^{k-1}
        \theta^{\beta_1}
          _{i_1 \cdots i_{2k^{\pr}+1}}
        \theta^{\beta_2}
          _{i_{2k'+2} \cdots i_{2k}}
        &
         \vartheta^{i_1 \cdots i_{2k}}
      \end{array}
     \&\longmapsfrom\&
      f^a_b x^b
      +
      f^{a}_{b_1 b_2}
      x^{b_1} x^{b_2}
      +
      f^a_{\beta_1 \beta_2}
      \theta^{\beta_1}
      \theta^{\beta_2}
      \&\longmapsfrom\&
      x^a
      \\[+2pt]
      \def\arraycolsep{0pt}
      \begin{array}{rl}
        f^a_\beta
        \sum_{k=1}
          ^{\lfloor q/2 \rfloor}
        \theta^\beta
          _{i_1 \cdots i_{2k+1}}
        &
        \vartheta^{i_1 \cdots i_{2k+1}}
        \\
        +
        f^\alpha_{b \beta}
        \sum_{k=1}^{\lfloor 
q/2 \rfloor}
        \sum_{k'=1}^{k}
        x^b
          _{i_1 \cdots i_{2k'}}
        \theta^\beta
          _{i_{2k'+1}\cdots i_{2k+1}}
        &
        \vartheta
          ^{i_1 \cdots i_{2k+1} }
      \end{array}
      \&\longmapsfrom\& 
      f^\alpha_\beta
      \theta^\beta
      +
      f^\alpha
        _{b \beta}
      x^b \theta^\beta
      \&\longmapsfrom\&
      \theta^\alpha
      \mathrlap{\,.}
    \end{tikzcd}
  \end{equation}
This makes the construction of the bosonic body of the odd $q$-tangent bundle a functor from super-manifolds to ordinary smooth manifolds 
\vspace{-3mm}
\begin{equation}
  \label{bosonicToddFunctor}
  \begin{tikzcd}[row sep=-3pt, column sep=small]
    \mathrm{sSmthMfd}
    \ar[
      rr,
      "{
       \bosonicTodd^{{}_{(q)}}
       (-)
      }"
    ]
    &&
    \mathrm{SmthMfd}
    \\
    X
    \ar[
      dd,
      "{
        f
      }"
    ]
    &\longmapsto&
    \bosonicTodd^{{}_{(q)}} X
    \ar[
      dd,
      "{
        \scalebox{1}{$
          \bosonicTodd^{{}_{(q)}}
          f
        $}      
      }"
    ]
    &
    x^{a}_{i_1 \cdots i_{2k}}
    &
    \theta^\alpha_{i_1 \cdots i_{2k+1}}
    \\
    &&&
    \rotatebox[origin=c]{-90}{$\longmapsto$}
    &
    \rotatebox[origin=c]{-90}{$\longmapsto$}
    \\
    Y
    &\longmapsto&
    \bosonicTodd^{{}_{(q)}} Y
    &
    f_\ast x^{a}_{i_1 \cdots i_{2k}}
    &
    f_\ast \theta^\alpha_{i_1 \cdots i_{2k+1}}
  \end{tikzcd}
\end{equation}
Moreover, as $q$ ranges, these odd higher tangent bundles naturally pull back along maps between the probing super-points, 
\vspace{-2mm} 
$$
  \begin{tikzcd}[row sep=-2pt, column sep=small]
    &
    \mathrm{sPnt}^{\mathrm{op}}
    \ar[
      rr,
      "{
        \bosonicTodd^{{}_{(-)}}X     
      }"
    ]
    &&
    \mathrm{SmthMfd}
    \\
    \phi^j_i \vartheta^i
    &
    C^\infty\big(
      \mathbb{R}^{0\vert q}
    \big)
    &\longmapsto&
    C^\infty\big(
      \bosonicTodd^{{}_{(q)}}
      X
    \big)
    \ar[
      dd,
      "{
        \phi^\ast
      }"
    ]
    &
    x^a_{i_1 \cdots i_{2k}}
    &[+4pt]
    \theta^\alpha
      _{ i_1 \cdots i_{2k+1} }
    \\
    \rotatebox[origin=c]{+90}{$
      \longmapsto
    $}
    &
    && &
    \rotatebox[origin=c]{-90}{$
      \longmapsto
    $}
    &
    \rotatebox[origin=c]{-90}{$
      \longmapsto
    $}
    \\
    \vartheta^j
    &
    C^\infty\big(
      \mathbb{R}^{0\vert q'}
    \big)
    \ar[
      uu,
      "{ \phi }"
    ]
    &\longmapsto&
    C^\infty\big(
      \bosonicTodd^{{}_{(q')}}
      X
    \big)
    &
    x^a
      _{j_1 \cdots j_{2k}}
    \phi^{j_1}_{i_1}
    \cdots
    \phi^{j_{2k}}_{i_{2k}}
    &
    \theta^\alpha
      _{j_1 \cdots j_{2k}}
    \phi^{j_1}_{i_1}
    \cdots
    \phi^{j_{2k}}_{i_{2k}}
    \mathrlap{\,.}
 \end{tikzcd}
$$
\end{example}
This construction is used below to recognize super-point-wise ordinary Lie groups as being {\it represented} (cf. e.g. \cite[p. 8]{HKST11}) by super-Lie groups, see around \eqref{SuperMinkowskiGroupRepresented} and \eqref{TheMSuperGroup} below.

\medskip

\noindent
{\bf Super-Lie groups.} The notion of super-Lie groups as originating around \cite[Def. 2.1]{Berezin87} is an instance of {\it group objects internal to} an ambient category (\cite[\S 3]{Grothendieck61}, see also \cite[p. 123]{BarrWells85}), here: internal to supermanifolds. 

\begin{definition}[{\bf Super-Lie group} (e.g. {\cite[\S 7.1]{Varadarajan04}})]
  \label{SuperLieGroup}
  A super Lie group is a {\it group object} internal to the category of supermanifolds (Def. \ref{CategoryOfSupermanifolds}), hence a super-manifold $G$ equipped with maps of supermanifolds of the form
  \begin{equation}
    \label{GroupObjectOperations}
    \begin{tikzcd}
      G \times G
      \ar[
        r, 
        "\mathrm{prd}"
      ]
      &
      G
    \end{tikzcd}
    \,,
    \;\;\;\;
    \begin{tikzcd}
      \ast
      \ar[r, "{ \mathrm{e} }"]
      &
      G
    \end{tikzcd}
    \,,
    \;\;\;
    \begin{tikzcd}
      G
      \ar[r, "{\mathrm{inv}}"]
      &
      G
    \end{tikzcd}
  \end{equation}
  making the following diagrams commute
  \begin{equation}
    \label{GroupObjectAxioms}
    \hspace{-8mm} 
    \def\arraystretch{1.5}
    \def\arraycolsep{15pt}
    \begin{array}{ccc}
    \mbox{\bf \color{darkblue} Associativity}
    &
    \mbox{\bf  \color{darkblue} Unitality}
    &
    \mbox{\bf  \color{darkblue} Invertibility}
    \\
    \begin{tikzcd}
      G \times G \times G
      \ar[
        rr,
        "{ 
          \mathrm{prd}
          \,\times\, 
          \mathrm{id}
        }"
      ]
      \ar[
        dd,
        "{
          \mathrm{id}
          \,\times\,
          \mathrm{prd}
        }"
        %{description, pos=.45}
      ]
      &&
      G \times G
      \ar[
        dd,
        "{ \mathrm{prd} }"
        %{description, pos=.45}
      {swap}
      ]
      \\
      \\
      G \times G
      \ar[
        rr,
        "{
          \mathrm{prd}
        }"
      ]
      && 
      G
    \end{tikzcd}
    &
    \begin{tikzcd}[
      column sep=27pt
    ]
          G
          \ar[
            r, 
            shorten=-2pt,
            "{ \sim}"{pos=.4}
          ]
          \ar[
            d, 
            shorten >=-2pt,
            "{ \sim }"{sloped, swap, pos=.4}
          ]
          \ar[
           ddrr,
           equals
          ]
          &[-24pt]
          G \times \ast
          \ar[
            r,
            "{
              \mathrm{id}
              \,\times\, 
              \mathrm{e}        
            }"{pos=.4}
          ]
          &
          G \times G
          \ar[
            dd,
            "{
              \mathrm{prd}
            }"{swap}
          ]
          \\[-10pt]
          \ast \times G
          \ar[
            d,
            "{
              \mathrm{e}
              \,\times\,
              \mathrm{id}
            }"
            %{description, pos=.43}
          ]
          \\
          G \times G
          \ar[
            rr,
            "{ \mathrm{prd} }"
          ]
          &&
          G
        \end{tikzcd}
        &
    \begin{tikzcd}[
      column sep=25pt
    ]
      G 
      \ar[
        rr,
        shift left=3pt,
        "{
          (\mathrm{id},\mathrm{inv})
        }"
      ]
      \ar[
        rr,
        shift right=3pt,
        "{
          (\mathrm{inv}, \mathrm{id})
        }"{swap}
      ]
      \ar[
        dd,
        "{ \exists! }" 
      ]
      &&
      G \times G
      \ar[
        dd,
        "{ 
          \mathrm{prd} 
        }"
        {swap}
      ]
      \\
      \\
      \ast
      \ar[
        rr,
        "{ \mathrm{e} }"
      ]
      &&
      G
    \end{tikzcd}
    \end{array}
  \end{equation}
\end{definition}

\medskip

\noindent
{\bf Examples.}
A first simplistic but important example, showcasing how ordinary Lie groups appear in this dual perspective when regarded as super-Lie groups (with trivial odd components):
\begin{example}[\bf The circle group as a super-Lie group]
\label{CircleGroup}
Consider the short exact sequence of ordinary Lie groups
$$
  \begin{tikzcd}
    \mathbb{Z}
    \ar[r, hook]
    &
    \mathbb{R}
    \ar[r]
    &
    \mathbb{S}^1
  \end{tikzcd}
$$
as seen in the category of super-Lie groups.
First, with respect to the canonical coordinate function $x \,\in\, C^\infty(\mathbb{R})$, the additive group operation on the real line pulls back as
\begin{equation}
  \label{GroupStructureOnRealLine}
  \begin{tikzcd}[row sep=-3pt, column sep=0pt]
    \mathbb{R}
    \times
    \mathbb{R}
    \ar[
      rr,
      "{ + }"
    ]
    &&
    \mathbb{R}
    \\
    \dot x + x
    &\longmapsfrom&
    x
    \,.
  \end{tikzcd}
\end{equation}
Similarly, the algebra of smooth functions on the integers is of course the set of $\mathbb{Z}$-tuples of real numbers
$$
  C^\infty(\mathbb{Z})
  \,\simeq\,
  \mathbb{R}^{\mathbb{Z}}
  \;\simeq\;
  \Big\{
    f\,\defneq\,
    \big(
      f(n)
      \,\in\,
      \mathbb{R}
    \big)_{n \in \mathbb{Z}}
  \Big\}
  \,,
$$
all regarded in even degree, and
equipped with the index-wise addition and multiplication in the real numbers:
$$
  (f \cdot g)(n) 
    \,:=\,
  f(n) \cdot g(n)
  \,,
  \;\;\;\;\;
  (f + g)(n) 
    \,:=\,
  f(n) + g(n)
  \,.
$$
We may say that a ``coordinate function'' on $\mathbb{Z}$ is any injective function $f \,:\, \mathbb{Z} \hookrightarrow \mathbb{Z}$, and that 
the ``canonical coordinate function'' $x \,\in\, C^\infty(\mathbb{Z})$ is the canonical injection
\begin{equation}
  \label{CanonicalCoordinateOnZ}
  x(n)
  \;:=\;
  n
  \,.
\end{equation}
The additive group operation on the integers is uniquely characterized by how it pulls back coordinate functions, and for the canonical coordinate functions it has the simple form
\begin{equation}
  \label{GroupStructureOnIntegers}
  \begin{tikzcd}[row sep=-3pt, column sep=0pt]
    \mathbb{Z}
    \times
    \mathbb{Z}
    \ar[
      rr,
      "{ + }"
    ]
    &&
    \mathbb{Z}
    \\
    \dot x \,+\, x
    &\longmapsfrom&
    x
    \,,
  \end{tikzcd}
\end{equation}
which is of the same form 
\eqref{GroupStructureOnRealLine} as for the real line.
This makes manifest the group homomorphism given by the canonical inclusion of the integers into the real numbers
$$
  \begin{tikzcd}[row sep=-3pt, 
   column sep=0pt
  ]
    \mathbb{Z}
    \ar[
      rr,
      hook
    ]
    &&
    \mathbb{R}
    \\
    x &\longmapsfrom& x
    \,,
  \end{tikzcd}
$$
where the bottom line reflects simply the restriction of the canonical coordinate function $x$ on $\mathbb{R}$ to the integer points.
Forming the quotient of this inclusion of Lie groups, hence the pushout along the map to the trivial group, means dually to consider only those  functions on $\mathbb{R}$ whose restriction to $\mathbb{Z}$ is constant, hence only the 1-periodic functions, hence those on the circle $S^1 \,=\, \mathbb{R}/\mathbb{Z}$:
\vspace{-1mm} 
\begin{equation}
  \label{FunctionsOnCircleAsFiberProduct}
  \begin{tikzcd}[row sep=small, column sep=large]
    \mathbb{Z}
    \ar[r, hook]
    \ar[d, ->>
    ]
    \ar[
      dr,
      phantom,
      "{
        \scalebox{.7}{
          \color{gray}
          (po)
        }
      }"{pos=.8}
    ]
    &
    \mathbb{R}
    \ar[d, ->>]
    \\
    1
    \ar[r, hook]
    &
    S^1
  \end{tikzcd}
  \hspace{1.5cm}
  \begin{tikzcd}[row sep=small, column sep=large]
    C^\infty(\mathbb{Z})
    &
    C^\infty(\mathbb{R})
    \ar[l, ->>]
    \\
    \mathbb{R}
    \ar[
      u, hook,
    ]
    &
    C^\infty(\mathbb{R})_{\mathrm{prdc}}
    \mathrlap{
      \;\;=\;
      C^\infty(S^1)
      \,.
    }
    \ar[l, ->>]
    \ar[u, hook]
    \ar[
      ul,
      phantom,
      "{
        \scalebox{.7}{
          \color{gray}
          (pb)
        }
      }"{pos=.2}
    ]
  \end{tikzcd}
\end{equation}
\end{example}

The following Ex. \ref{GroupStructureOnMinkowskiSuperSpacetime} must be well-known to experts but may not be citable in detail from the literature.\footnote{The base case $\mathbb{R}^{1,0\vert\mathbf{1}}$ of Ex. \ref{GroupStructureOnMinkowskiSuperSpacetime} is described in terms of functorial geometry in \cite[p. 277]{Varadarajan04} and the general product law $\mathrm{prd}$ from \eqref{GroupOperationForSuperMinkowski} appears in \cite[(2.1), (2.6)]{CdAIPB00}.} 
We make it now fully explicit in order to prepare the ground for the construction of its extension by the hidden M-group further below.

\begin{example}[\bf Super-Lie group structure on super-Minkowski spacetime]
\label{GroupStructureOnMinkowskiSuperSpacetime}
Denoting the canonical coordinate functions on the product super-manifold $\mathbb{R}^{1,10\,\vert\,\mathbf{32}} \times \mathbb{R}^{1,10\,\vert\,\mathbf{32}}$ by
$(x^a, \theta^\alpha)$ for the second factor and  $(x^a_{\pr}, \theta^\alpha_{\pr})$ for the first factor (adapted to thinking equivalently in terms of the canonical left-multiplication action of the group on itself), consider the following definition of group operations \eqref{GroupObjectOperations} on the supermanifold $\mathbb{R}^{1,10\,\vert\,\mathbf{32}}$:
\begin{equation}
  \label{GroupOperationForSuperMinkowski}
  \begin{tikzcd}[row sep=-2pt, 
    column sep=0pt
  ]
    \mathbb{R}
      ^{1,10\,\vert\,\mathbf{32}}
    \times
    \mathbb{R}
      ^{1,10\,\vert\,\mathbf{32}}
    \ar[
      rr,
      "{ \mathrm{prd} }"
    ]
    &&
    \mathbb{R}
      ^{1,10\,\vert\,\mathbf{32}}
    \\
    x^{\pr\, a} + x^a
    -
    \big(
      \hspace{1pt}
      \overline{\theta}{}^{\,\pr}
      \,\Gamma^a\,
      \theta
    \big)
    &\overset{\mathrm{prd}^{\mathrlap{\ast}}}{\longmapsfrom}&
    x^a
    \\
    \theta^{\,\pr}
    +
    \theta
    &\overset{\mathrm{prd}^{\mathrlap{\ast}}}{\longmapsfrom}&
    \theta
    \,,
\end{tikzcd}
\hspace{1.2cm}
\begin{tikzcd}[
    row sep=2pt,
    column sep=0pt
 ]
    \ast
    \ar[
      rr,
      "{ \mathrm{e} }"
    ]
    &&
    \mathbb{R}^{1,10\,\vert\,\mathbf{32}}
    \\
    0
    &\overset{\mathrm{e}^\ast}{\longmapsfrom}&
    x^a
    \\
    0 
    &\overset{\mathrm{e}^\ast}{\longmapsfrom}&
    \theta
  \end{tikzcd}
  \hspace{1.2cm}
  \begin{tikzcd}[
    sep=0pt
  ]
    \mathbb{R}^{1,10\,\vert\,\mathbf{32}}
    \ar[
      rr,
      "{ \mathrm{inv} }"
    ]
    &&
    \mathbb{R}^{1,10\,\vert\,\mathbf{32}}
    \\
    -x^a 
    &\overset{\mathrm{inv}^\ast}{\longmapsfrom}&
    x^a
    \\
    -\theta^\alpha
    &\overset{\mathrm{inv}^\ast}{\longmapsfrom}&
    \theta^\alpha
  \end{tikzcd}
\end{equation}
Here the second and third lines specify on coordinate functions the corresponding reverse homomorphisms of super-function algebras via pullback, which uniquely characterize maps of super-manifolds (cf. \cite[Ex. 2.13]{GSS24-SuGra}).

The definition of $\mathrm{e}$ and $\mathrm{inv}$ \eqref{GroupOperationForSuperMinkowski} is obvious, while the extra summand appearing in the definition of $\mathrm{prd}$
is such as to make the co-frame field 
\begin{equation}
\label{CanonicalCoframeOnSuperMinkowski}
\def\arraystretch{1.2}
\begin{array}{ccl}
  e^a 
  &:=& 
  \mathrm{d}x^a + \big(\, \overline{\theta}\,\Gamma^a\mathrm{d}\theta\big)
  \\
  \psi 
  &:=& 
  \mathrm{d}\theta 
\end{array}
\end{equation}
be left-invariant, namely invariant under the operation
\vspace{-1mm} 
$$
  \begin{tikzcd}
    C^\infty
    \big(
      \mathbb{R}^{1,10\,\vert\,\mathbf{32}}
    \big)
    \,\widehat{\otimes}\,
    \Omega^\bullet_{\mathrm{dR}}
    \big(
      \mathbb{R}^{1,10\,\vert\,\mathbf{32}}
    \big)
    \ar[
      r,
      <<-
    ]
    \ar[
      rr,
      <-,
      rounded corners,
      to path={
           ([yshift=+00pt]\tikztostart.north)  
        -- ([yshift=+08pt]\tikztostart.north)  
        -- node[yshift=4pt]
           {
             \scalebox{.7}{$
               \mathrm{act}^\ast
             $}
           }
           ([yshift=+08pt]\tikztotarget.north)  
        -- ([yshift=+00pt]\tikztotarget.north)  
      }
    ]
    &
    \Omega^\bullet_{\mathrm{dR}}
    \big(
      \mathbb{R}^{1,10\,\vert\,\mathbf{32}}
    \big)
    \,\widehat{\otimes}\,
    \Omega^\bullet_{\mathrm{dR}}
      \big(\mathbb{R}^{1,10\,\vert\,\mathbf{32}}
    \big)
    \ar[
      r,
      <-,
      "{ \mathrm{prd}^{\mathrlap{\ast}} }"
    ]
    &
    \Omega^\bullet_{\mathrm{dR}}
    \big(
      \mathbb{R}^{1,10\,\vert\,\mathbf{32}}
    \big)
  \end{tikzcd}
$$
dual to the left action of the supergroup on its odd tangent bundle (Ex. \ref{TheOddTangentBundle}):
\vspace{-1mm} 
$$
  \begin{tikzcd}
    \mathbb{R}^{1,10\,\vert\,\mathbf{32}}
    \times
    T_{\!{}_{\mathrm{odd}}}
    \mathbb{R}^{1,10\,\vert\,\mathbf{32}}
    \ar[
      r,
      hook
    ]
    \ar[
      rr,
      rounded corners,
      to path={
           ([yshift=+00pt]\tikztostart.north)  
        -- ([yshift=+10pt]\tikztostart.north)  
        -- node[yshift=5pt] {
           \scalebox{.7}{$
             \mathrm{act}
           $}
        }
           ([yshift=+10pt]\tikztotarget.north)  
        -- ([yshift=+00pt]\tikztotarget.north)  
      }
    ]
    &
    T_{\!{}_{\mathrm{odd}}}\mathbb{R}^{1,10\,\vert\,\mathbf{32}}
    \times
    T_{\!{}_{\mathrm{odd}}}
    \mathbb{R}^{1,10\,\vert\,\mathbf{32}}
    \ar[
      r,
      "{ \mathrm{prd}_\ast }"
    ]
    &    
    T_{\!{}_{\mathrm{odd}}}
    \mathbb{R}^{1,10\,\vert\,\mathbf{32}}.
  \end{tikzcd}
$$

\vspace{-1mm} 
\noindent This is because
\begin{equation}
  \label{CheckingLeftInvarianceOnSuperMinkowski}
  \hspace{-1cm}
  \def\arraystretch{1.5}
  \begin{array}{ccl}
    \mathrm{act}^\ast e^a
    &=&
    \mathrm{act}^\ast 
    \Big(
      \mathrm{d}x^a
      +
      \big(\,\overline{\theta}
      \,\Gamma^a\,
      \mathrm{d}\theta\big)
    \Big)
    \\
    &=&
    \mathrm{d}\,\mathrm{act}^\ast x^a
    +
    \big(\,\overline{\mathrm{act}^\ast\theta}
    \,\Gamma^a\,
    \mathrm{d}\,\mathrm{act}^\ast\theta\big)
    \\
    &=&
    \mathrm{d}
    \Big(
      x^{\pr \, a} + x^a
      -
      \big(\,
        \overline{\theta}^{\,\pr}
        \,\Gamma^a\,
        \theta
      \big)
    \Big)
      +
      \big(
        (\,\overline{\theta}^{\,\pr} + 
        \overline{\theta})
        \,\Gamma^a\,
        \mathrm{d}(
          \theta^{\,\pr} + \theta
        )
      \big)
    \\
    &=&
    \mathrm{d}x^a 
    -
    \big(\,
      \overline{\theta}^{\, \pr}
      \,\Gamma^a\,
      \mathrm{d}\theta
    \big)
    \,+\,
    \big(\,
      \overline{\theta}^{\,\pr}
      \,\Gamma^a\,
      \mathrm{d}\theta
    \big)
    \,+\,
    \big(\,
      \overline{\theta}
      \,\Gamma^a\,
      \mathrm{d}\theta
    \big)
    \\
    &=&
    \mathrm{d}x^a 
    \,+\,
    \big(\,
      \overline{\theta}
      \,\Gamma^a\,
      \mathrm{d}\theta
    \big)
    \\
    &=&
    e^a
    \mathrlap{\,,}
  \end{array}
  \hspace{1.4cm}
  \def\arraystretch{1.4}
  \begin{array}{ccl}
    \mathrm{act}^\ast
    \psi
    &=&
    \mathrm{act}^\ast
    \mathrm{d}\theta
    \\
    &=&
    \mathrm{d}
    \,
    \mathrm{act}^\ast \theta
    \\
    &=&
    \mathrm{d}\big(
      \theta^{\,\pr}
      +
      \theta
    \big)
    \\
    &=&
    \mathrm{d}\theta
    \\
    &=&
    \psi
    \mathrlap{\,.}
    \\
    {}
  \end{array}
\end{equation}
Here $\mathrm{d}$ denotes the differential on the second factor, hence acting on the un-primed coordinates only (with the primed coordinates instead parametrizing the fixed group ``element'' along which to pull back).

Hence if \eqref{GroupOperationForSuperMinkowski} defines indeed a group structure and $\mathbb{R}^{1,10\,\vert\,\mathbf{32}}$, then it carries a left-invariant coframe field \eqref{CanonicalCoframeOnSuperMinkowski} whose de Rham differential relations (its Maurer-Cartan equations) coincide with those of the CE-algebra of the super-Minkowski super-Lie algebra, thus exhibiting \eqref{CanonicalCoframeOnSuperMinkowski} as the corresponding super-Lie group.

Checking that \eqref{CanonicalCoframeOnSuperMinkowski} indeed does satisfy the group axioms \eqref{GroupObjectAxioms} is straightforward, but it may still be interesting to note how the bifermionic term is involved in making this work:

\noindent
$
  \adjustbox{
    rotate=90,
    margin=0pt,
    fbox,
    raise=-1.2cm
  }{\bf  \color{darkblue} Associativity}
  \hspace{.4cm}
  \begin{tikzcd}[row sep=10pt, 
    column sep=12pt
  ]
    &[-190pt]
    x^{\pr\pr\, a}
    +
    x^{\pr\, a} 
    +
    \big(
      \hspace{1pt}
      \overline{\theta}^{\,\pr\pr}
      \,\Gamma^a\,
      \theta^{\,\pr}
    \big)
    + 
    x^a
    +
    \big(
      \hspace{1pt}
      \overline{(\theta^{\,\pr\pr} + \theta^{\,\pr})}
      \,\Gamma^a\,
      \theta
    \big)    
    \ar[
      rr,
      <-|
    ]
    \ar[dl, equals]
    &[-5pt]&[-5pt]
    x^{\pr\, a} + x^a
    +
    \big(
      \hspace{1pt}
      \overline{\theta}^{\,\pr}
      \,\Gamma^a\,
      \theta
    \big)    
    \ar[
      ddd,
      <-|
    ]
    \\[-5pt]
    x^{\pr\pr\,a} + 
    x^{\pr\,a} + x_a
    +\big(
      \hspace{1pt}
      \overline{\theta}^{\,\pr}
      \,\Gamma^a\,
      \theta
    \big)
    +
    \big(
      \hspace{1pt}
      \overline{\theta}^{\,\pr\pr}
      \,\Gamma^a\,
      (\theta^{\,\pr} + \theta)
    \big)
    \ar[
      dd,
      <-|
    ]
    \\
    \\
    x^{\pr\pr\, a} + x^a
    +
    \big(
      \hspace{1pt}
      \overline{\theta}^{\,\pr\pr}
      \,\Gamma^a\,
      \theta
    \big)
    \ar[rrr, <-|]
    &&& x^a
  \end{tikzcd}
  \hspace{.6cm}
  \begin{tikzcd}[
    row sep=12pt, 
    column sep=8pt
  ]
    \theta^{\,\pr\pr}
    +
    \theta^{\,\pr}
    +
    \theta
    \ar[
      ddd, <-|
    ]
    \ar[rrr, <-|]
    &&&
    \theta^{\,\pr} + \theta
    \ar[ddd, <-|]
    \\
    \\
    \\
    \theta^{\,\pr\pr}
    +
    \theta
    \ar[rrr, <-|]
    &&&
    \theta
  \end{tikzcd}
$

\vspace{.3cm}

\noindent
$
  \adjustbox{
    rotate=90,
    margin=0pt,
    fbox,
    raise=-.8cm
  }{\bf \color{darkblue}  Unitality}
  \hspace{.4cm}
  \begin{tikzcd}[
   row sep=8pt, 
    column sep=15pt
  ]
    x^a
    \ar[
      r, <-|
    ]
    \ar[
      d,
      <-|
    ]
    \ar[
      dddrrr,
      equals
    ]
    &[-25pt]
    x^{\pr\,a} 
    \ar[
      rr, <-|
    ]
    &&
    x^{\pr\,a} + x^a 
    +
    \big(
      \hspace{1pt}
      \overline{\theta}^{\,\pr}
      \,\Gamma^a\,
      \theta
    \big)    
    \ar[
      ddd,
      <-|
    ]
    \\
    x^{\pr\,a}
    \ar[dd, <-|]
    \\
    \\
    x^{\pr \,a} + x^a 
    +
    \big(
      \hspace{1pt}
      \overline{\theta}^{\,\pr}
      \,\Gamma^a\,
      \theta
    \big)
    \ar[
      rrr,
      <-|
    ]
    &&&
    x^a
  \end{tikzcd}
  \hspace{1cm}
  \begin{tikzcd}[row sep=10pt, 
    column sep=15pt
  ]
    \theta
    \ar[
      r, <-|
    ]
    \ar[
      d, <-|
    ]
    \ar[
      dddrrr,
      equals
    ]
    &
    \theta^{\,\pr}
    \ar[
      rr, <-|
    ]
    &&
    \theta^{\,\pr} + \theta
    \ar[
      ddd,
      <-|
    ]
    \\
    \theta
    \ar[
      dd, <-|
    ]
    \\
    \\
    \theta^{\,\pr} + \theta
    \ar[
      rrr,
      <-|
    ]
    &&&
    \theta
  \end{tikzcd}
$

\vspace{.3cm}

\noindent
$
  \adjustbox{
    rotate=90,
    margin=0pt,
    fbox,
    raise=-1.1cm
  }{\bf \color{darkblue} Invertibility}
  \hspace{.4cm}
  \begin{tikzcd}[row sep=12pt, column sep=15pt]
    &[-20pt]
    x^a
    - x^a 
    - 
    \big(
      \hspace{1pt}
      \overline{\theta}
      \,\Gamma^a\,
      \theta
    \big)
    \ar[
      rr,
      <-|
    ]
    &&
    x^{\pr \,a}
    + x^a + 
    \big(
      \hspace{1pt}
      \overline{\theta}
      \,\Gamma^a\,
      \theta
    \big)
    \ar[
      ddd,
      <-|
    ]
    \\[-15pt]
    0
    \ar[
      ur,
      equals
    ]
    \ar[
      dd,
      <-|
    ]
    \\
    \\
    0
    \ar[
      rrr,
      <-|
    ]
    &&&
    x^a
  \end{tikzcd}
  \hspace{1cm}
  \begin{tikzcd}[row sep=8pt, column sep=15pt]
    &
    \theta - \theta
    \ar[
      rr,
      <-|
    ]
    &&
    \theta^{\,\pr}
    +
    \theta
    \ar[
      ddd,
      <-|
    ]
    \\
    0
    \ar[
      dd,
      <-|
    ]
    \ar[
      ur,
      equals
    ]
    \\
    \\
    0
    \ar[
      rrr, <-|
    ]
    &&&
    \theta
  \end{tikzcd}
$

\smallskip 
\noindent
In the last step on the left we used that $\big(\hspace{1pt}\overline{\theta}\,\Gamma^a\,\theta\big) = 0$ because the $\theta^\alpha$ anticommute among each other, while their pairing here is symmetric \eqref{SymmetricSpinorPairings}.
\hfill$\Box$
\end{example}

%%%%%%%%%%%%%%%%%%%%%%%%%%%%%%%%%%%%
\subsection{The Lie integration}
%%%%%%%%%%%%%%%%%%%%%%%%%%%%%%%%%%%%

While the integration in Ex. \ref{GroupStructureOnMinkowskiSuperSpacetime} of the super-Minkowski Lie algebra by ``educated guess followed by checking its consistency'' is efficient in this simple case, more general cases require a more systematic approach:

\medskip

\noindent
{\bf Integrating nilpotent super-Lie algebras.}
We may essentially reduce the question of integration of super-Lie algebras (to super-Lie groups) to the classical theory of integration of ordinary Lie algebras (to ordinary Lie groups) by regarding objects in super-algebra/geometry as systems of ordinary algebraic/geometric objects indexed by super-points whose function algebras provide an arbitrary supply of ``Grassmann variables''.

Here we focus on the nilpotent case (Rem. \ref{NilpotentSuperLieAlgebras}), which covers all super-Minkowski-like examples.

\begin{definition}[{\bf Super-Lie algebras probed by super-points} (cf. {\cite[\S 3]{Sachse08}})]
\label{SuperLieAlgebrasProbedBySuperPoints}
Given a super-Lie algebra $\mathfrak{g} \in \mathrm{sLieAlg}$ and a Grassmann algebra $\wedge^\bullet_{{}_{\mathbb{R}}} (\mathbb{R}^q)^\ast \;\simeq\; C^\infty\big(\mathbb{R}^{0\vert q}\big)$ (Ex. \ref{SuperPoints}),
the even part of the tensor product of the underlying super-vector spaces
\begin{equation}
  \label{SuperPointProbeOfSuperLieAlgebra}
  \mathfrak{g}_{(q)}
  \;\;
  :=
  \;\;
  C^\infty\big(
    \mathbb{R}^{0\vert q}
    ,\,
    \mathfrak{g}
  \big)_{\mathrm{evn}}
  \;\;
  :=
  \;\;
  \big(C^\infty(\mathbb{R}^{0\vert q})
  \otimes_{{}_\mathbb{R}}
  \mathfrak{g}\big)_{\mathrm{evn}}
  \;\;
    \simeq
  \;\;
  \mathbb{R}
  \left\langle
    \vartheta^{i_1 \cdots i_n}
    \!\otimes\!
    T
    \,\bigg\vert\,
    n \in \mathbb{N}
    ,\,
    \def\arraystretch{.9}
    \begin{array}{ll}
      T \in \mathfrak{g}_{\mathrm{evn}}
      & \mbox{for $n$ even}
      \\
      T \in 
      \mathfrak{g}_{\mathrm{odd}}
      & \mbox{for $n$ odd}
    \end{array}\!\!
  \right\rangle
\end{equation}
is an ordinary vector space which carries the structure of an ordinary Lie algebra, with Lie bracket given by
\footnote{
The sign rule of super-algebra demands that \eqref{MakingSuperLieBracketIntoLieBracket} be multiplied by $(-1)$ whenever both $T$ and $\vartheta^{i'_1\cdots i'_{n'}}$ are in odd degree. But this sign rule is readily seen to be equal to changing the formula \eqref{MakingSuperLieBracketIntoLieBracket} by pulling out the Grassmann-elements in reverse order (as in \cite[(25)]{Sachse08}), hence to modify it to
\vspace{-.2cm}
$$
  \big[
    \vartheta^{i_1 \cdots i_n}
    T
    ,\,
    \vartheta^{i'_1 \cdots i'_{n'}}
    T'
  \big]_{\mathrm{sgn}}
  \;:=\;
  \vartheta
    ^{
      i'_1 \cdots i'_{n'}
      \,
      i_1 \cdots i_n
    }
  [T,T']
  \;=\;
  (-1)^{n n'}
  \vartheta
    ^{
      i_1 \cdots i_n
      \,
      i'_1 \cdots i'_{n'}
    }
  [T,T']
  \,,
$$
\vspace{-.4cm}

\noindent
and this is readily seen to be naturally isomorphic to our rule \eqref{MakingSuperLieBracketIntoLieBracket}, by the transformation which reverses the order of Grassmann generators in all products:
\vspace{-.2cm}
$$
  \begin{tikzcd}[
    column sep=40pt,
    row sep=-4pt,
    ampersand replacement=\&
  ]
    \big(
      \mathfrak{g}_{(p)}
      ,\
      [-,-]
    \big)
    \ar[
      rr,
      "{
        \vartheta^{i_1 \cdots i_n}
        T
        \;\,\mapsto\;\,
        \vartheta^{i_n \cdots i_1}
        T
      }",
      "{ \sim }"{swap}
    ]
    \&\&
    \big(
      \mathfrak{g}_{(p)}
      ,\,
      [-,-]_{\mathrm{sgn}}
    \big)
    \\
    \big(
      \vartheta^{i_1 \cdots i_n}
      T
      ,\,
      \vartheta^{i_1 \cdots i'_{n'}}
      T'
    \big)
    \&\longmapsto\&
    \big(
      \vartheta^{i_n \cdots i_1}
      T
      ,\,
      \vartheta^{i'_{n'} \cdots i_1}
      T'
    \big)
    \\
    \rotatebox[origin=c]{-90}{$
      \longmapsto
    $}
    \&\&
    \rotatebox[origin=c]{-90}{$
      \longmapsto
    $}
    \\[-2pt]
    \vartheta^{
      i_1 \cdots i_n
      \,
      i'_1 \cdots i'_{n'}
    }
    [T,T']
    \&\longmapsto\&
    \vartheta^{
      i'_{n'} \cdots i'_{1}
      \,
      i_n \cdots i_1
    }
    [T,T']
    \,.
  \end{tikzcd}
$$

\vspace{-1mm} 
\noindent Therefore we may stick with our rule \eqref{MakingSuperLieBracketIntoLieBracket}, which is convenient because this is the rule actually picked up by functors on $\mathrm{sPnt}$ that are represented by a super-Lie group, see \eqref{SuperMinkowskiProductUnderBosonicTodd} in Ex. \ref{SuperMinkowskiGroupProbedBySuperPoints} below.
}
\begin{equation}
  \label{MakingSuperLieBracketIntoLieBracket}
  \big[
    \vartheta^{i_1 \cdots i_n}
    T
    ,\,
    \vartheta^{i'_1 \cdots i'_{n'}}
    T'
  \big]
  \;\;
    :=
  \;\;
  \vartheta^{i_1 \cdots i_n i'_1 \cdots i'_{n'}}
  \big[
    T
    ,\,
    T'
  \big]
\end{equation}
(with the given super-Lie bracket appearing on the right).

This means that super-algebra homomorphisms $C^\infty(\mathbb{R}^{0\vert q}) \xleftarrow{f^\ast} C^\infty(\mathbb{R}^{0\vert r})$ 
induce Lie algebra homomorphisms
$$
  \begin{tikzcd}[row sep=-2pt, column sep=small]
    \mathfrak{g}_{(q)}
    \ar[
      rr,
      <-,
      "{ f^\ast }"
    ]
    &&
    \mathfrak{g}_{(r)}
    \\
    f^\ast(\vartheta^{i_1 \cdots i_n})
    \otimes 
    T
    &
    \longmapsfrom
    &
    \vartheta^{i_1 \cdots i_n}
    \otimes 
    T
  \end{tikzcd}
$$
thus incarnating the super Lie algebra $\mathfrak{g}$ as a  functor from the opposite of the category of super-points \eqref{SubCategoryOfSuperpoints} to (ordinary) Lie algebras:
\begin{equation}
  \label{SuperLieAlgebrasAsPresheafOnSuperpoints}
  \begin{tikzcd}[column sep=0pt, row sep=-4pt]
    \mathfrak{g}
    \;:\;
    \mathrm{sPnt}^{\mathrm{op}}
    \ar[rr]
    &&
    \mathrm{LieAlg}_{\mathbb{R}}
    \\
    \mathbb{R}^{0\vert q}
    &\longmapsto&
    \mathfrak{g}_{(q)}\;.
  \end{tikzcd}
\end{equation}
\end{definition}

\begin{remark}[\bf Nilpotent super Lie algebras]
\label{NilpotentSuperLieAlgebras}
The super-translation Lie algebras that we are concerned with here are {\it nilpotent}, meaning that their $n$-fold adjoint action vanishes for large enough $n$.
(The definition of nilpotent super Lie algebras, e.g. \cite[\S 26]{FrappatEtAl00}, is just as for ordinary Lie algebras, e.g.  \cite[\S V]{Serre64}). 

Note that if a super Lie algebra $\mathfrak{g} \in \mathrm{sLieAlg}_{\mathbb{R}}$ is nilpotent, then its probes by super-points \eqref{SuperLieAlgebrasAsPresheafOnSuperpoints} evidently take values in ordinary nilpotent Lie algebras:
$$
  \mathfrak{g}
  \in
  \mathrm{sLieAlg}
    ^{\mathrm{nil}}
    _{\mathbb{R}}
  \hspace{.7cm}
  \Rightarrow
  \hspace{.7cm}
  \mathfrak{g}
  \,:\,
  \mathrm{sPnt}^{\mathrm{op}}
  \xrightarrow{\;\;}
  \mathrm{LieAlg}
    ^{\mathrm{nil}}
    _{\mathbb{R}}
  \,.
$$
\end{remark}

Recall now the following classical fact (e.g. from \cite[\S 1.2]{CorwinGreenleaf04}):
\begin{proposition}[\bf Lie theory for nilpotent Lie algebras]
\label{LieTheoryForNilpotentLieAlgebras}
  For {\rm(ordinary)} nilpotent Lie algebras $\mathfrak{g}$, the Dynkin formula {\rm (aka {\it Campbell-Baker-Hausdorff series}, e.g. \cite[\S IV.7]{Serre64}\cite[\S 1.7]{DuistermaatKolk00})}
  \footnote{
    The left hand side of \eqref{DynkinSeries} would more traditionally be written with the exponential map $\mathrm{exp}$ and its local inverse $\mathrm{log}$ as ``$\log\!\scalebox{1.3}{$($}\mathrm{prd}\big(\exp(T_1), \exp(T_2)\big)\scalebox{1.3}{$)$}$'' or ``$\log\!\scalebox{1.1}{$($}\exp(T_1) \ast \exp(T_2)\scalebox{1.1}{$)$}$''. But since the exponential map $\exp$ is globally an isomorphism due to nilpotency, by Prop. \ref{LieTheoryForNilpotentLieAlgebras}, as is hence  its logarithm $\log$, we may as well suppress them notationally. It is with this suppression that the usual expressions in the examples of super-translation groups are obtained.
  }
  \begin{equation}
    \label{DynkinSeries}
    \mathrm{prd}\big(
      T_1
      \,,
      T_2
    \big)
    \;=\;
    T_1 + T_2
    +
    \tfrac{1}{2}
    [T_1, T_2]
    +
    \tfrac{1}{12}
    \Big(
    \big[
      T_1
      ,\,
      [T_1, T_2]
    \big]
    +
    \big[
      T_2
      ,\,
      [T_2, T_1]
    \big]
    \Big)
    +
    \tfrac{1}{2}
    \Big[
     T_2
    \big[
      T_1,
      [T_2,T_1]
    \big]
    \Big]
    +
    \cdots
  \end{equation}
  {\rm (which truncates and hence converges due to nilpotency)} exhibits isomorphy of the exponential map onto the corresponding connected and simply-connected nilpotent Lie group, thereby constituting an equivalence of categories {\rm \cite[Thm. 14.37]{Milne17}}:
  $$
    \textstyle{\int}
    \;:\;
    \begin{tikzcd}
      \mathrm{LieAlg}
        ^{\mathrm{nil}}
        _{\mathbb{R}}
      \ar[
        r,
        "{ \sim }"
      ]
      &
      \mathrm{LieGrp}
        ^{\mathrm{unip}}
      .
    \end{tikzcd}
  $$
\end{proposition}

\begin{example}[\bf Systematic integration of the super-Minkowski Lie algebra]
\label{SuperMinkowskiGroupProbedBySuperPoints}
Probing the super-Minkowski super-Lie algebra
$\mathbb{R}^{1,10\,\vert\,\mathbf{32}}$ \eqref{SuperMinkowskiLinearBasis} 
\eqref{TheBifermionicSuperBracket}
with the super-point $\mathbb{R}^{0\vert 2}$ (via Def. \ref{SuperLieAlgebrasProbedBySuperPoints}), the underlying ordinary vector space \eqref{SuperPointProbeOfSuperLieAlgebra} is
\begin{equation}
  \label{SuperMinkowskiVectorSpaceAtStage2}
  \mathbb{R}^{1,10\vert\mathbf{32}}_{(2)}
  \;\simeq\;
  \mathbb{R}
  \Big\langle
    \big(P_a\big)_{a=0}^{10}
    ,\,
    \big(
      \vartheta^{12}\!P_a
    \big)_{a = 0}^{10}
    ,\,
    \big(
    \vartheta^1 Q_\alpha
    \big)_{\alpha=1}^{32}
    ,\,
    \big(
    \vartheta^2 Q_\alpha
    \big)_{\alpha = 1}^{32}
  \Big\rangle
\end{equation}
(where now the terms in parenthesis are to be regarded as primitive symbols, being the names of linear basis elements, all in degree $(0,\,\mathrm{evn})$),
and the non-vanishing Lie brackets on these basis elements are:
\begin{equation}
  \label{SuperMinkowskiLieBracketAtStage2}
    \big[
      \vartheta^i Q_\alpha
      ,\,
      \vartheta^j Q_\alpha
    \big]
    \;=\;
    -2
    \,
    \Gamma^a_{\alpha\beta}
    \, 
    \vartheta^{ij}P_a
    \,,
\end{equation}
where on the right we are using the notation $\vartheta^{ij} := \vartheta^i \vartheta^j $\eqref{ProductOfOddParameters}.

With $\mathbb{R}^{1,10\vert\mathbf{32}}$ itself, also this ordinary Lie algebra $\mathbb{R}^{1,10\vert\mathbf{32}}_{(2)}$  is clearly nilpotent (cf. Rem. \ref{NilpotentSuperLieAlgebras}) and hence the corresponding 1-connected Lie group has (Prop. \ref{LieTheoryForNilpotentLieAlgebras}) as underlying manifold the vector space \eqref{SuperMinkowskiVectorSpaceAtStage2}, which we think of as parameterized as follows
\begin{equation}
  \label{SuperMinkowskiManifoldAtStage2}
  \mathbb{R}
    ^{1,10\vert\mathbf{32}}
    _{(2)}
  \;\simeq\;
  \left\{
    \def\arraystretch{1.4}
    \def\arraycolsep{1pt}
    \begin{array}{ccr}
      {}
      &
      x^a 
      &
      P_a
      \\
      +&
      x^a_{i_1 i_2} 
      &
      \vartheta^{i_1 i_2}
      P_a
      \\
      +&
      \theta^\alpha_i
      &
      \vartheta^i Q_\alpha
    \end{array}
    \;\middle\vert\;
    \def\arraystretch{1.4}
    \def\arraycolsep{1pt}
    \begin{array}{rr}
      x^a
      \in 
      \mathbb{R}
      \\
      x^a_{i_1 i_2}
      =
      -
      x^a_{i_2 i_1}
      \in 
      \mathbb{R}
      \\
      \theta^\alpha_i
      \in 
      \mathbb{R}
    \end{array}
    \;\;,\;\;
      \begin{array}{l}
        a \in \{0,1, \cdots, 10\}
        \\
        \alpha \in \{1,2, \cdots, 32\}
        \\
        i_1, i_2 \in \{1,2\}
      \end{array}
  \right\}
  \,,
\end{equation}
with group product given by applying the Dynkin formula \eqref{DynkinSeries} to \eqref{SuperMinkowskiLieBracketAtStage2}, as follows
\begin{equation}
  \label{SuperminkowskiGroupStructureAtStage2}
  \begin{tikzcd}[
    column sep=small, row sep=-4pt,
    ampersand replacement=\&
  ]
  \mathbb{R}
    ^{1,10\vert\mathbf{32}}
    _{(2)}
  \times
  \mathbb{R}
    ^{1,10\vert\mathbf{32}}
    _{(2)}
  \ar[
    rr,
    "{
      \mathrm{prd}_{(2)}
    }"
  ]
  \&\&
  \mathbb{R}
    ^{1,10\vert\mathbf{32}}
    _{(2)}
  \\
  \left(
    \def\arraystretch{1.4}
    \def\arraycolsep{1pt}
    \begin{array}{ccr}
      {}
      &
      \dot x^a 
      &
      P_a
      \\
      +&
      \dot x^a_{i_1 i_2} 
      &
      \vartheta^{i_1 i_2}
      P_a
      \\
      +&
      \dot \theta^\alpha_i
      &
      \vartheta^i Q_\alpha
    \end{array}
    \;\;\scalebox{1.4}{,}\;\;
    \def\arraystretch{1.4}
    \def\arraycolsep{1pt}
    \begin{array}{ccr}
      {}
      &
      x^b 
      &
      P_b
      \\
      +&
      x^b_{j_1 j_2} 
      &
      \vartheta^{j_1 j_2}
      P_b
      \\
      +&
      \theta^\beta_j
      &
      \vartheta^j Q_\beta
    \end{array}
  \right)
  \&\longmapsto\&
  \left(
    \def\arraystretch{1.4}
    \def\arraycolsep{1pt}
    \begin{array}{clr}
    &
    \big(
      \dot x^a 
      +
      x^a
    \big)
    &
    P_a
    \\
    +
    &
    \big(
    \dot x_{ij}^a 
    +
    x_{ij}^a
    -
    \dot \theta^\alpha_i
    \theta^\beta_j
    \, 
    \Gamma^a_{\alpha \beta}
    \big)
    &
    \vartheta^{ij}
    P_a
    \\
    +&
    \big(
      \dot\theta_i^\alpha
      +
      \theta_i^\alpha
    \big)
    &
    \vartheta^i Q_\alpha
    \end{array}
  \right)
  ,
  \end{tikzcd}
\end{equation}
where the extra summand in the second line is the one coming from the Dynkin formula \eqref{DynkinSeries}: 
$$
  \mathrm{prd}\Big(
    \dot\theta^\alpha_i
    \;
    \vartheta^i
    Q_\alpha
    \;,\;
    \theta^\beta_j
    \;
    \vartheta^j
    Q_\beta
  \Big)
  \;\;
  =
  \;\;
  \dot\theta^\alpha_i
  \,
  \vartheta^i
  Q_\alpha
  \,+\,
  \theta^\beta_j
  \,
  \vartheta^j
  Q_\beta
  \,+\,
  \dot\theta^\alpha_i
  \theta^\beta_j
  \,
  \tfrac{1}{2}
  \underbrace{
  \big[
    \vartheta^i
    Q_\alpha
    ,\,
    \vartheta^j
    Q_\beta
  \big]
  }_{\color{gray} 
    -2
    \, 
    \Gamma^a_{\alpha \beta}
    \,
    \vartheta^{i j}
    P_a
  }
  \,+\,
  \underbrace{
    \cdots
    \mathclap{\phantom{\vert_{\vert_{\vert}}}}
  }_{\color{gray} 
    0
  }
  \,.
$$

The general case of probes by any super-point $\mathbb{R}^{0\vert q}$, $q \in \mathbb{N}$, is not much different: 
In generalization of \eqref{SuperMinkowskiVectorSpaceAtStage2} we have at any stage $q$ the vector space
\begin{equation}
  \label{SuperMinkowskiVectorSpaceAtAnyStage}
  \mathbb{R}^{1,10\vert\mathbf{32}}_{(q)}
  \;\simeq\;
  \mathbb{R}
  \left\langle
    \big(
      \vartheta
        ^{i_1 \cdots i_{2k}}
      \!
      P_a
      \big)_{
        \scalebox{.6}{$
          \begin{array}{l}
            a \in \{0,\cdots,10\},
            \\
            0 \leq k \leq q/2
            \\
            i_j \in \{1,\cdots, q\}
          \end{array}
        $}
      }
      \;,\;
    \big(
     \vartheta
       ^{i_1 \cdots i_{2k+1}} 
     Q_\alpha
    \big)_{
        \scalebox{.6}{$
          \begin{array}{l}
            a \in \{0,\cdots,10\},
            \\
            o \leq k \leq (q-1)/2
            \\
            i_j \in \{1,\cdots, q\}
          \end{array}
        $}
      }
  \right\rangle
\end{equation}
equipped with the Lie algebra structure whose only non-trivial brackets are, in generalization of \eqref{SuperMinkowskiLieBracketAtStage2},
\begin{equation}
  \label{SuperMinkowskiLieBracketAtAnyStage}
    \big[
      \vartheta
        ^{i_1 \cdots i_{2k'+1}}
      Q_\alpha
      ,\,
      \vartheta
        ^{j_1 \cdots j_{2k+1}}
      Q_\beta
    \big]
    \;=\;
    \Gamma^a_{\alpha\beta}
    \, 
    \vartheta
      ^{
        i_1 \cdots i_{2k'+1}
        j_1 \cdots j_{2k+1}
      }
    \!
    P_a
    \,,
\end{equation}
and in generalization of \eqref{SuperMinkowskiManifoldAtStage2} we may coordinatize this space as
\begin{equation}
  \label{SuperMinkowskiManifoldAtAnyStage}
  \mathbb{R}
    ^{1,10\vert\mathbf{32}}
    _{(q)}
  \;\;
    \simeq
  \;\;
  \left\{
  \def\arraystretch{1.4}
  \def\arraycolsep{1pt}
  \begin{array}{clr}
    &
    \sum_{k}
    x^a_{i_1 \cdots i_{2k}}
    &
    \vartheta^{i_1 \cdots i_{2k}}
    \!
    P_a
    \\
    + 
    &
    \sum_k
    \theta^\alpha_{i_1 \cdots i_{2k+1}}
    &
    \vartheta
      ^{i_1 \cdots i_{2k+1}}
    Q_\alpha
  \end{array}
  \;\middle\vert\;
  \def\arraystretch{1.4}
  \def\arraycolsep{1pt}
  \begin{array}{cr}
    x^a_{i_1 \cdots i_{2k}}
    \,=\,
    x^a_{[i_1 \cdots i_{2k}]}
    &\in\,
    \mathbb{R}
    \\
    \theta^\alpha_{i_1 \cdots i_{2k+1}}
    \,=\,
    \theta^\alpha
      _{[i_1 \cdots i_{2k+1}]}
    &\in\,
    \mathbb{R}
  \end{array}  
    \;\;,\;\;
      \begin{array}{l}
        a \in \{0,1, \cdots, 10\}
        \\
        \alpha \in \{1,2, \cdots, 32\}
        \\
        i_j \in \{1, 2, \cdots, q\}
      \end{array}
  \right\}
  \,,
\end{equation}
which the Dynkin formula \eqref{DynkinSeries}
equips with the following group product, in generalization of \eqref{SuperminkowskiGroupStructureAtStage2}:
\begin{equation}
  \label{SuperminkowskiGroupStructureAtAnyStage}
  \begin{tikzcd}[
    row sep=-7pt,
    ampersand replacement=\&
  ]
    \mathbb{R}
      ^{1,10\vert\mathbf{32}}
      _{(q)}
    \times
    \mathbb{R}
      ^{1,10\vert\mathbf{32}}
      _{(q)}
    \ar[
      dd,
      "{ 
        \mathrm{prd}_{(q)} 
      }"
    ]
    \&
    \left(
  \def\arraystretch{1.4}
  \def\arraycolsep{1pt}
  \begin{array}{clr}
    &
    \sum_{k}
    \dot x^a_{i_1 \cdots i_{2k}}
    &
    \vartheta^{i_1 \cdots i_{2k}}
    \!
    P_a
    \\
    + 
    &
    \sum_k
    \dot \theta^\alpha_{i_1 \cdots i_{2k+1}}
    &
    \vartheta
      ^{i_1 \cdots i_{2k+1}}
    Q_\alpha
  \end{array}
  \;\;
  \scalebox{1.4}{,}
  \;\;
  \def\arraystretch{1.4}
  \def\arraycolsep{1pt}
  \begin{array}{clr}
    &
    \sum_{k}
    x^a_{i_1 \cdots i_{2k}}
    &
    \vartheta^{i_1 \cdots i_{2k}}
    \!
    P_a
    \\
    + 
    &
    \sum_k
    \theta^\alpha_{i_1 \cdots i_{2k+1}}
    &
    \vartheta
      ^{i_1 \cdots i_{2k+1}}
    Q_\alpha
  \end{array}
    \right)
    \\
    \&
    \adjustbox{rotate=-90}{
      $\longmapsto$
    }
    \\
    \mathbb{R}
      ^{1,10\vert\mathbf{32}}
      _{(q)}    
  \&
  \left(
  \def\arraystretch{1.4}
  \def\arraycolsep{1pt}
  \begin{array}{clr}
    &
    \sum_{k}
    \Big(
      \dot x^a
        _{i_1 \cdots i_{2k} }
      +
      x^a_{i_1 \cdots i_{2k}}
      -
      \sum_{\dot k = 0}^{k-1}
      \dot\theta
        ^\alpha
        _{i_1 \cdots i_{2\dot k + 1}}
      \theta
        ^\beta
        _{i_{2\dot k + 2} \cdots i_{2k}}
      \Gamma^a_{\alpha \beta}
    \Big)
    &
    \vartheta^{i_1 \cdots i_{2k}}
    \!
    P_a
    \\
    + 
    &
    \sum_k
    \big(
    \dot \theta
      ^\alpha
      _{i_1 \cdots i_{2k+1}}
    +
    \theta
      ^\alpha
      _{i_1 \cdots i_{2k+1}}
    \big)
    &
    \vartheta
      ^{i_1 \cdots i_{2k+1}}
    Q_\alpha
  \end{array}
  \right)
  \end{tikzcd}
\end{equation}
or expressed dually as:
$$
  \begin{tikzcd}[
    row sep=0pt
  ]
    C^\infty\Big(
    \mathbb{R}
      ^{1,10\vert\mathbf{32}}
      _{(q)}
    \times
    \mathbb{R}
      ^{1,10\vert\mathbf{32}}
      _{(q)}
    \Big)
    \ar[
      dd,
      <-,
      "{
        \mathrm{prd}
          ^\ast
          _{(q)}
      }"
    ]
    &
    \dot x^a_{i_1 \cdots a_{2k}}
    +
    x^a_{i_1 \cdots a_{2k}}
    -
    \sum_{\dot k=0}^{k-1}
    \dot\theta
      ^\alpha
      _{i_1 \cdots i_{2\dot k+1}}
    \theta
      ^\beta
      _{i_{2\dot k + 2} \cdots i_{2k}}
    \Gamma^a_{\alpha \beta}
    &
    \dot \theta
      ^\alpha
      _{i_1 \cdots i_{2k+1}}
    +
    \theta
      ^\alpha
      _{i_1 \cdots i_{2k+1}}
    \\
    &
    \adjustbox{rotate=90}{$
      \longmapsto
    $}
    &
    \adjustbox{rotate=90}{$
      \longmapsto
    $}
    \\
    C^\infty\Big(
    \mathbb{R}
      ^{1,10\vert\mathbf{32}}
      _{(q)}
    \Big)
    &
    x^a_{i_1 \cdots i_{2k}}
    &
    \theta
      ^\alpha
      _{i_1 \cdots i_{2k+1}}
  \end{tikzcd}
$$
These formulas are clearly functorial across stages with respect to maps between the parameterizing super-points:
$$
  \begin{tikzcd}[
    row sep=2pt,
    column sep=90pt
  ]
    \mathrm{sPnt}^{\mathrm{op}}
    \ar[
      dd,
      "{
        \mathbb{R}
          ^{1,10\vert\mathbf{32}}
      }"
    ]
    &[-50pt]
    C^\infty(\mathbb{R}^{0\vert q})
    \ar[
      rr,
      "{
        \vartheta^i
        \;\mapsto\;
        \phi^i_j \vartheta^j
      }"
    ]
    &&
    C^\infty(\mathbb{R}^{0\vert r})
    \\[-5pt]
    &&
    \adjustbox{rotate=-90}{$
      \longmapsto
    $}
    \\
    \mathrm{LieAlg}
      ^{\mathrm{nil}}
      _{\mathbb{R}}
    \ar[
      dd,
      "{
        \int
      }"
    ]
    &
    \mathbb{R}
      ^{1,10\vert\mathbf{32}}
      _{(q)}
    \ar[
      rr,
      "{
        \vartheta^{i_1 \cdots i_{2k}}
        P_a
        \;\mapsto\;
        \big(
        \phi^{i_1}_{j_1}
        \cdots
        \phi^{i_{2k}}_{j_{2k}}
        \big)
        \vartheta
          ^{j_1 \cdots j_{2k}}
        P_a
      }",
      "{
        \vartheta^{i_1 \cdots i_{2k+1}}
        Q_\alpha
        \;\mapsto\;
        \big(
        \phi^{i_1}_{j_1}
        \cdots
        \phi^{i_{2k+1}}_{j_{2k+1}}
        \big)
        \vartheta
          ^{j_1 \cdots j_{2k+1}}
        Q_\alpha
      }"{swap}
    ]
    &&
    \mathbb{R}
      ^{1,10\vert\mathbf{32}}
      _{(r)}
    \\
    &&
    \adjustbox{rotate=-90}{$
      \longmapsto
    $}
    \\
    \mathrm{LieGrp}
      ^{\mathrm{unip}}
    &
    \mathbb{R}
      ^{1,10\vert\mathbf{32}}
      _{(q)}
    \ar[
      rr,
      "{
        x^a_{i_1 \cdots i_{2k}}
        \;\mapsto\;
        \big(
        \phi^{i_1}_{j_1}
        \cdots
        \phi
          ^{i_{2k}}
          _{i_{jk}}
        \big)
        x^a_{j_1 \cdots j_{2k}}
      }"
    ]
    &&
    \mathbb{R}
      ^{1,10\vert\mathbf{32}}
      _{(q)}
  \end{tikzcd}
$$

Thereby, we have lifted the super-Minkowski super-Lie algebra $\mathbb{R}^{1,10\vert\mathbf{32}}$ to a group-valued functor by applying ordinary Lie integration to all its ordinary Lie algebras of probes by super-points
$$
  \begin{tikzcd}[
    row sep=-2pt, column sep=small
  ]
    &&
    \mathrm{LieGrp}
    \ar[
      d,
      ->>
    ]
    \\[15pt]
    \mathrm{sPnt}
      ^{\mathrm{op}}
    \ar[
      rr,
    ]
    \ar[
      urr,
      dashed
    ]
    &&
    \mathrm{LieAlg}_{\mathbb{R}}
    \\
    \mathbb{R}^{0\vert q}
    &\longmapsto&
    \mathbb{R}
      ^{1,10\vert\mathbf{32}}
      _{(q)}
  \end{tikzcd}
$$
This functorial incarnation of super-Lie algebras and their super-Lie groups is an instance of the original definition of {\it internal group objects} due to \cite[p. 270]{Grothendieck60}\cite[\S 3]{Grothendieck61}, for early  discussion along these lines see also \cite{Yagi93}, a brief discussion may also be found in \cite[\S 2.10]{DeligneMorgan99}, more details are in \cite[\S 3]{Sachse08}.

\medskip

But we may observe now that this functor is {\it represented} by the super-Minkowski super-Lie group structure $\big(\mathbb{R}^{1,10\vert\mathbf{32}},\,  \mathrm{prd},\,  \mathrm{e}, \, \mathrm{inv}\big)$ of Ex. \ref{GroupStructureOnMinkowskiSuperSpacetime} in that we have a natural isomorphism as follows, intertwining the (dual) group structures:
\begin{equation}
  \label{SuperMinkowskiGroupRepresented}
  \begin{tikzcd}[row sep=20pt, 
    column sep=40pt,
    ampersand replacement=\&
  ]
    \scalebox{.7}{
      \color{darkblue}
      \bf
      \def\arraystretch{.9}
      \begin{tabular}{c}
        Odd tangents of
        \\
        super-Lie group structure 
      \end{tabular}
    }
    \&
    \clap{
    \scalebox{.7}{
      \color{darkgreen}
      \bf
      naturally isomorphic to
    }
    }
    \&
    \scalebox{.7}{
      \color{darkblue}
      \bf
      \def\arraystretch{.9}
      \begin{tabular}{c}
        integration of system of Lie algebras
        \\
        of probes by any super-point
      \end{tabular}
    }
    \\[-15pt]
    C^\infty
    \Big(
      \bosonic T
      {}^{(q)}_{\!\mathrm{odd}}
      \mathbb{R}
        ^{1,10\vert\mathbf{32}}
    \Big)
    \ar[
      rr,
      "{
      }",
      "{ \sim }"
      %{description}
    ]
    \ar[
      dd,
      "{
        \mbox{$
        \big(
          \bosonic T
          {}^{(q)}_{\!\mathrm{odd}}
          \mathrm{prd}
        \big)^\ast
        $}
      }"
      %{description}
    ]
    \&\&
    C^\infty\big(
      \mathbb{R}
        ^{1,10\vert\mathbf{32}}
        _{(q)}
    \big)
    \ar[
      dd,
      "{ \mathrm{prd}^\ast_{(q)} }"
    ]
    \\
    \\
    C^\infty\Big(
      \bosonic T
      {}^{(q)}_{\!\mathrm{odd}}
      \mathbb{R}
        ^{1,10\vert\mathbf{32}}
      \!\times\!
      \bosonic T
      {}^{(q)}_{\!\mathrm{odd}}
      \mathbb{R}
        ^{1,10\vert\mathbf{32}}
    \Big)_{\mathrm{evn}}
    \ar[
      rr,
      "{ \sim }"
      %{description}
    ]
    \&\&
    C^\infty\big(
      \mathbb{R}
        ^{1,10\vert\mathbf{32}}
        _{(q)}
      \times
      \mathbb{R}
        ^{1,10\vert\mathbf{32}}
        _{(q)}
    \big)
    \,.
  \end{tikzcd}
\end{equation}
For instance, for $q=2$ the operation on the left of \eqref{SuperMinkowskiGroupRepresented}
is given, via \eqref{QuadraticMapUnderBosonicTodd}, by:
\begin{equation}
  \label{SuperMinkowskiProductUnderBosonicTodd}
  \hspace{-4mm} 
  \begin{tikzcd}[row sep=-2pt, column sep=15pt]
    \mathbb{R}^{0\vert2}
    \ar[
      rr,
      "{
        \big( 
          (
          x^{\pr}, 
          x^{\pr}_{12},
          \theta^{\pr}_1,
          \theta^{\pr}_2
          )
          ,\,
          (
            x,
            x_{12},
            \theta_1,
            \theta_2
          )
        \big)
      }"
    ]
    &&
    \mathbb{R}^{1,10\vert\mathbf{32}}
    \times
    \mathbb{R}^{1,10\vert\mathbf{32}}
    \ar[
      rr,
      "{ \mathrm{prd} }"
    ]
    &&
    \mathbb{R}^{1,10\vert\mathbf{32}}
    \\
    (x^{\pr\, a} + x^{\pr\, a}_{ij}\vartheta^{i j})
    +
    (x^{a} + x^{a}_{ij}\vartheta^{i j})
    -
      \Gamma^a_{\alpha\beta}
      (\theta^{\pr\, \alpha}_i
      \vartheta^i)
      (\theta^{\beta}_j
      \vartheta^j)
    &\longmapsfrom& 
    x^{\pr\, a} + x^a
    -
      \Gamma^a_{\alpha\beta}\,
      \theta^{\pr\alpha}
      \theta^\beta
    &\longmapsfrom&
    x^a
    \\
    =
    (
    {\color{purple}
      x^{\pr\, a} + x^a
    }
    )
    +
    \big(
    {\color{purple}
      x^{\pr\, a}_{ij}
      +
      x^a_{i j}
      -
      \theta^{\pr\, \alpha}
      \theta^{\beta}
      \Gamma^a_{\alpha \beta}
    }
    \big)
    \vartheta^{ij}
    \\
    (
    {\color{purple}
    \theta^{\pr\alpha}_i
    +
    \theta^\alpha_i
    }
    )
    \vartheta^i
    &\longmapsfrom&
    \theta^{\pr\alpha} 
    +
    \theta^\alpha
    &\longmapsfrom&
    \theta^\alpha
  \end{tikzcd}
\end{equation}
which manifestly coincides with what we found for the right-hand side in \eqref{SuperminkowskiGroupStructureAtStage2}.

\medskip

In conclusion, we have (re-)obtained the Lie integration of the super-Minkowski Lie algebra to its (1-connected) super-Lie group by applying ordinary Lie integration to the system of ordinary Lie algebras formed by probing the super-Minkowski Lie algebra with super-points.

\smallskip

This integration process may easily appear notationally more cumbersome than the alternative Lie integration via ``educated guess followed by consistency check'' that we showed in Ex. \ref{GroupStructureOnMinkowskiSuperSpacetime}; however: 

\vspace{1mm} 
\begin{itemize}[
  leftmargin=.7cm,
  topsep=1pt,
  itemsep=2pt
]
\item[\bf (i)] the functorial notation here looks heavy only superficially, in effect it just means to tensor everything with any number of auxiliary Grassmann parameters, thereby shifting all expressions into even degree, and to check (a simple observation) that these parameters remain mere ``bystanders'' in all expressions under all operations,

\item[\bf (ii)] the functorial machinery provides a systematic Lie integration of any  (nilpotent) super-Lie algebra, even in cases where an ``educated guess'' does not so easily spring to mind -- as is the case already for the next example.
\end{itemize}
\end{example}

\medskip

Thereby we come to our main example, in variation of Ex. \ref{SuperMinkowskiGroupProbedBySuperPoints}:
\begin{example}[\bf Integrating the hidden M-algebra to the hidden M-group]
\label{IntegratingSuperExceptionalMinkowsliLieAlgebra}

The probes \eqref{SuperPointProbeOfSuperLieAlgebra}
of the hidden M-algebra $\widehat{\mathfrak{M}}$ from \S\ref{MinimalFermionicExtension},  by the super-point $\mathbb{R}^{0\vert q}$ form the following space, in variation of \eqref{SuperMinkowskiManifoldAtAnyStage},
\begin{equation}
  \label{qProbesOfHiddenMAlgebra}
  \widehat{\mathcal{M}}_{(q)}
  \;\;
    \simeq
  \;\;
  \left\{
  \def\arraystretch{1.4}
  \def\arraycolsep{0pt}
  \begin{array}{clrl}
     &
     \sum_k 
     x^a
      _{i_1 \cdots i_{2k}}
     &
     \vartheta
       ^{i_1 \cdots i_{2k}}
     &
      P_a
    \\
    +\;\;
    &
     \sum_k 
     b^{a_1 a_2}
      _{i_1 \cdots i_{2k}}
     &
     \vartheta
       ^{i_1 \cdots i_{2k}}
      &
      Z_{a_1 a_2}
    \\
    +\;\;
    &
     \sum_k 
     b^{a_1 \cdots a_5}
      _{i_1 \cdots i_{2k}}
     &
     \vartheta
       ^{i_1 \cdots i_{2k}}
      &
      Z_{a_1 \cdots a_5}
    \\
    +\;\;
    &
     \sum_k 
     \theta^{\alpha}
      _{i_1 \cdots i_{2k+1}}
     &
     \vartheta
       ^{i_1 \cdots i_{2k+1}}
      &
      Q_\alpha 
    \\
    +\;\;
    &
     \sum_k 
     \xi^{\alpha}
      _{i_1 \cdots i_{2k+1}}
     &
     \vartheta
       ^{i_1 \cdots i_{2k+1}}
      &
      O_\alpha 
  \end{array}
  \;\middle\vert\;
  \def\arraystretch{1.4}
  \def\arraycolsep{2pt}
  \begin{array}{ll}
    x^a
     _{i_1 \cdots i_{2k}}
    &
    \in
    \mathbb{R}
    \\
    b^{a_1 a_2}
     _{i_1 \cdots i_{2k}}
    &
    \in
    \mathbb{R}
    \\
    b^{a_1 \cdots a_5}
     _{i_1 \cdots i_{2k}}
    &
    \in
    \mathbb{R}
    \\
    \theta^{\alpha}
     _{i_1 \cdots i_{2k+1}}
    &
    \in
    \mathbb{R}
    \\
    \xi^{\alpha}
     _{i_1 \cdots i_{2k+1}}
    &
    \in
    \mathbb{R}
  \end{array}
  \;\;
  \scalebox{1.3}{,}
  \;\;
  \def\arraystretch{1.4}
  \begin{array}{l}
    a_j \in \{0,\cdots, 10\}
    \\
    \alpha \in \{0, \cdots, 32\}
    \\
    i_j \in \{1, \cdots, q\}
    \\
    0 \leq k \leq q/2
  \end{array}
  \right\}
  \,,
\end{equation}
which the Dynkin formula \eqref{DynkinSeries} equips, in variation of \eqref{SuperminkowskiGroupStructureAtAnyStage}, with the  group product
$$
  \begin{tikzcd}
  \widehat{\mathcal{M}}_{(q)}
   \times
  \widehat{\mathcal{M}}_{(q)}
    \ar[
      rr,
      "{
        \mathrm{prd}_{(q)}
      }"
    ]
    &&
  \widehat{\mathcal{M}}_{(q)}
  \end{tikzcd}
$$
given by
$$
  \begin{tikzcd}[
    row sep=0pt,
    ampersand replacement=\&
  ]
  \left(
  \def\arraystretch{1.4}
  \def\arraycolsep{0pt}
  \begin{array}{clll}
     &
     \sum_k 
     \dot x^a
      _{i_1 \cdots i_{2k}}
     &
     \vartheta
       ^{i_1 \cdots i_{2k}}
     &
      P_a
    \\
    +\;\;
    &
     \sum_k 
     \dot b^{a_1 a_2}
      _{i_1 \cdots i_{2k}}
     &
     \vartheta
       ^{i_1 \cdots i_{2k}}
      &
      Z_{a_1 a_2}
    \\
    +\;\;
    &
     \sum_k 
     \dot b^{a_1 \cdots a_5}
      _{i_1 \cdots i_{2k}}
     &
     \vartheta
       ^{i_1 \cdots i_{2k}}
      &
      Z_{a_1 \cdots a_5}
    \\
    +\;\;
    &
     \sum_k 
     \dot \theta^{\alpha}
      _{i_1 \cdots i_{2k+1}}
     &
     \vartheta
       ^{i_1 \cdots i_{2k+1}}
      &
      Q_\alpha 
    \\
    +\;\;
    &
     \sum_k 
     \dot\xi^{\alpha}
      _{i_1 \cdots i_{2k+1}}
     &
     \vartheta
       ^{i_1 \cdots i_{2k+1}}
      &
      O_\alpha 
  \end{array}
  \;\;\;
  \scalebox{1.3}{,}
  \;\;\;
  \def\arraystretch{1.4}
  \def\arraycolsep{0pt}
  \begin{array}{clll}
     &
     \sum_k 
     x^a
      _{i_1 \cdots i_{2k}}
     &
     \vartheta
       ^{i_1 \cdots i_{2k}}
     &
      P_a
    \\
    +\;\;
    &
     \sum_k 
     b^{a_1 a_2}
      _{i_1 \cdots i_{2k}}
     &
     \vartheta
       ^{i_1 \cdots i_{2k}}
      &
      Z_{a_1 a_2}
    \\
    +\;\;
    &
     \sum_k 
     b^{a_1 \cdots a_5}
      _{i_1 \cdots i_{2k}}
     &
     \vartheta
       ^{i_1 \cdots i_{2k}}
      &
      Z_{a_1 \cdots a_5}
    \\
    +\;\;
    &
     \sum_k 
     \theta^{\alpha}
      _{i_1 \cdots i_{2k+1}}
     &
     \vartheta
       ^{i_1 \cdots i_{2k+1}}
      &
      Q_\alpha 
    \\
    +\;\;
    &
     \sum_k 
     \xi^{\alpha}
      _{i_1 \cdots i_{2k+1}}
     &
     \vartheta
       ^{i_1 \cdots i_{2k+1}}
      &
      O_\alpha 
  \end{array}
  \right)
  \\
  \adjustbox{rotate=-90}{
    $\longmapsto$
  }
  \\
  \left(
  \def\arraystretch{1.6}
  \def\arraycolsep{0pt}
  \begin{array}{crrl}
     &
     \sum_k 
     \big(
       \dot x^a
        _{i_1 \cdots i_{2k}}
       +
       \dot x^a
        _{i_{i_1} \cdots i_{2k}}
       -
       \sum_{\dot k=0}^{k-1}
       \dot \theta^\alpha
         _{i_1 \cdots i_{2\dot k+1}}
       \theta^\beta
         _{i_{2\dot k + 2} \cdots i_{2k}}
       \Gamma^a_{\alpha \beta}
    \big)
     &
     \vartheta
       ^{i_1 \cdots i_{2k}}
     &
      P_a
    \\
    +\;\;
    &
     \sum_k 
     \big(
       \dot b^{a_1 a_2}
       _{i_1 \cdots i_{2k}}
       +
       b^{a_1 a_2}
       _{i_1 \cdots i_{2k}}
       +
       \sum_{\dot k = 0}^{k-1}
       \dot \theta^\alpha
         _{i_1 \cdots i_{2\dot k + 1}}
       \theta^\beta
         _{i_{2\dot k+2} \cdots i_{2k}}
        \Gamma^{a_1 a_2}
          _{\alpha \beta}
     \big)
     &
     \vartheta
       ^{i_1 \cdots i_{2k}}
      &
      Z_{a_1 a_2}
    \\
    +\;\;
    &
     \sum_k 
     \big(
       \dot b^{a_1 \cdots a_5}
       _{i_1 \cdots i_{2k}}
       +
       b^{a_1 \cdots a_5}
       _{i_1 \cdots i_{2k}}
       -
       \sum_{\dot k = 0}^{k-1}
       \dot \theta^\alpha
         _{i_1 \cdots i_{2\dot k + 1}}
       \theta^\beta
         _{i_{2\dot k+2} \cdots i_{2k}}
        \Gamma^{a_1 \cdots a_5}
          _{\alpha \beta}
     \big)
     &
     \vartheta
       ^{i_1 \cdots i_{2k}}
      &
      Z_{a_1 \cdots a_5}
    \\
    +\;\;
    &
     \sum_k 
     \big(
       \dot\theta^{\alpha}
        _{i_1 \cdots i_{2k+1}}
        +
       \dot\theta^{\alpha}
        _{i_1 \cdots i_{2k+1}}
      \big)
     &
     \vartheta
       ^{i_1 \cdots i_{2k+1}}
      &
      Q_\alpha 
    \\
    +\;\;
    &
     \sum_k 
     \left(
       \dot\xi^{\alpha}
        _{i_1 \cdots i_{2k+1}}
        +
       \xi^{\alpha}
        _{i_1 \cdots i_{2k+1}}
        \begin{array}{l}
        +
        \sum_{\dot k=0}^{k}
        \dot x^a
          _{i_{1} \cdots i_{2\dot k}}
        \theta^\beta
          _{i_{2\dot k+1} \cdots i_{2k+1}}
        \tfrac{\paramOne}{2}
        \Gamma_a{}^\alpha{}_\beta
        \\
        +
        \sum_{\dot k=0}^k
        \dot b^{a_1 a_2}
          _{i_{1} \cdots i_{2\dot k}}
        \theta^\beta
          _{i_{2\dot k+1} \cdots i_{2k+1}}
        \tfrac{\paramTwo}{2}
        \Gamma_{a_1 a_2}{}^\alpha{}_\beta
        \\
        +
        \sum_{\dot k=0}^k
        \dot b^{a_1 \cdots a_5}
          _{i_{1} \cdots i_{2\dot k}}
        \theta^\beta
          _{i_{2\dot k+1} \cdots i_{2k+1}}
        \tfrac{\paramFive}{2}
        \Gamma_{a_1 \cdots a_5}{}^\alpha{}_\beta
        \\
        -
        \sum_{\dot k=0}^k
        x^a
          _{i_{1} \cdots i_{2\dot k}}
        \dot \theta^\beta
          _{i_{2\dot k + 1} \cdots i_{2k+1}}
        \tfrac{\paramOne}{2}
        \Gamma_a{}^\alpha{}_\beta
        \\
        -
        \sum_{\dot k=0}^k
        b^{a_1 a_2}
          _{i_{1} \cdots i_{2\dot k}}
        \dot \theta^\beta
          _{i_{2\dot k + 1} \cdots i_{2k+1}}
        \tfrac{\paramTwo}{2}
        \Gamma_{a_1 a_2}{}^\alpha{}_\beta
        \\
        -
        \sum_{\dot k=0}^{k}
        b^{a_1 \cdots a_5}
          _{i_{1} \cdots i_{2\dot k}}
        \dot \theta^\beta
          _{i_{2\dot k+1} \cdots i_{2k+1}}
        \tfrac{\paramFive}{2}
        \Gamma_{a_1 \cdots a_5}{}^\alpha{}_\beta
        \end{array}
      +
      \cdots
      \right)
     &
     \vartheta
       ^{i_1 \cdots i_{2k+1}}
      &
      O_\alpha 
    \\
    \\[-15pt]
    + & 
  \tfrac{1}{12}
  \displaystyle{\sum_k}
  \displaystyle{\sum_{\ddot k =0}^{k-1}}
  \sum_{\dot k = \ddot k}^{k-1}
  \left(
  \def\arraystretch{1.4}
  \def\arraycolsep{2pt}
  \begin{array}{cl}
   &
  \dot\theta^\alpha
    _{i_1 \cdots i_{2\ddot k+1}}
  \dot\theta^{\alpha'}
    _{i_{2\ddot k+2} \cdots i_{2\dot k+2}}
  \theta^\beta
    _{i_{2\dot k+3} \cdots i_{2 k+1}}
    \\
  + &
  \theta^\alpha
    _{i_1 \cdots i_{2\ddot k+1}}
  \theta^{\alpha'}
    _{i_{2\ddot k+2} \cdots i_{2\dot k+2}}
  \dot\theta^\beta
    _{i_{2\dot k+3} \cdots i_{2 k+1}}
    \end{array}
  \right)
  [QQQ]^\delta_{\alpha\alpha'\beta}
  \;\;
  &
  \vartheta^{i_1 \cdots i_{2k+1}}
  &
  O_\delta
  \end{array}
  \right)
  \end{tikzcd}
$$
Here the last summand in the last row arises via the 4th summand in the Dynkin formula \eqref{DynkinSeries} due to the non-vanishing trilinear bracket \eqref{BracketOfThreeQs}.

This group-valued functor is evidently {\it represented} -- analogous to \eqref{SuperMinkowskiGroupRepresented} --  by the following super-Lie group structure (Def. \ref{SuperLieGroup}).
The underlying super-manifold is
\begin{equation}
  \label{UnderlyingManifoldOfHiddenMGroup}
  \widehat{\mathcal{M}}
  \;:=\;
  \mathbb{R}^{528\,\vert\, 64}
  \;\simeq\;
  \left(\!\!\!
  \def\arraystretch{1.1}
  \begin{array}{rl}
    x^a & P_a
    \\
    + \, b_{a_1 a_2} & Z^{a_1 a_2}
    \\
    + \, b_{a_1 \cdots a_5} & Z^{a_1 \cdots a_5}
    \\
    + \theta^\alpha & Q_\alpha
    \\
    + \xi^\alpha & O_\alpha
  \end{array}
  \;\middle\vert\;
  \def\arraystretch{1.3}
  \begin{array}{r}
    x^a \,\in\, \mathbb{R} 
    \\
    b_{a_1 a_2} \,=\, b_{[a_1 a_2]} \,\in\, \mathbb{R}
    \\
    b_{a_1 \cdots a_5} \,=\, b_{[a_1 \cdots a_5]} \,\in\, \mathbb{R}
    \\
    \theta^\alpha \,\in\, \mathbb{R}
    \\
    \xi^\alpha \,\in\, \mathbb{R}
  \end{array}
  \;,
  \def\arraystretch{.8}
  \begin{array}{l}
\scalebox{0.7}{$a_i \,\in\, \{0,1, \cdots, 10\}$}
    \\
   \scalebox{0.7}{$ \alpha \,\in\, \{1,2, \cdots, 32\} $}
  \end{array}
  \!\!\!\right)
\end{equation}

\vspace{.5mm} 
\noindent on which the group operation is given by
\begin{equation}
  \label{TheMSuperGroup}
  \begin{tikzcd}[
    row sep=-4pt, column sep=large, 
    ampersand replacement=\&
  ]
    \widehat{\mathcal{M}}
      \times
    \widehat{\mathcal{M}}
    \ar[
      rr,
      "{ \mathrm{prd} }"
    ]
    \&\&
    \widehat{\mathcal{M}}
    \\[5pt]
    x^{\pr a}
    \,+\,
    x^a
    \,-\,
    \big(\,
      \overline{\theta}{}^{\pr}
      \,\Gamma^a\,
      \theta
    \big)
    \&\longmapsfrom\&
    x^a
    \\
    b^{\pr}_{a_1 a_2}
    +
    b_{a_1 a_2}
    +
    \big(\,
      \overline{\theta}{}^{\pr}
      \,\Gamma_{a_1 a_2}\,
      \theta
    \big)
    \&\longmapsfrom\&
    b_{a_1 a_2}
    \\
    b^{\pr}_{ a_1 \cdots a_5}
    +
    b_{a_1 \cdots a_5}
    -
    \big(\,
      \overline{\theta}{}^{\pr}
      \,\Gamma_{a_1 \cdots a_5}\,
      \theta
    \big)
    \&\longmapsfrom\&
    b_{a_1 \cdots a_5}
    \\
    \theta^{\pr}
    +
    \theta
    \&\longmapsfrom\&
    \theta
    \\[-1pt]
    \left.
    \def\arraystretch{1.4}
    \begin{array}{c}
    \xi^{\pr}
    +
    \xi
    \\
    +
    \tfrac{\paramOne}{2}
      x^{\pr a}
      \Gamma_a
      \theta
    +
    \tfrac{\paramTwo}{2}
      b^{\pr a_1 a_2}
      \Gamma_{a_1 a_2}
      \theta
    +
    \tfrac{\paramFive}{2}
      b^{\pr a_1 \cdots a_5}
      \Gamma_{a_1 \cdots a_5}
      \theta
     \\
    -
    \tfrac{\paramOne}{2}
      x^{a}
      \Gamma_a
      \theta^{\pr}
    -
    \tfrac{\paramTwo}{2}
      b^{a_1 a_2}
      \Gamma_{a_1 a_2}
      \theta^{\pr}
    -
    \tfrac{\paramFive}{2}
      b^{a_1 \cdots a_5}
      \Gamma_{a_1 \cdots a_5}
      \theta^{\pr}
    \\
    +
    \tfrac{1}{12}
    [QQQ](\theta^{\pr}, \theta^{\pr}, \theta)
    +
    \tfrac{1}{12}
    [QQQ](\theta, \theta, \theta^{\pr})
    \end{array}    
    \!\!\!\right\}
    \&\longmapsfrom\&
    \xi
  \end{tikzcd}
\end{equation}

\vspace{-1mm} 
\noindent We call this super-Lie group the {\it hidden M-group}.
\end{example}

\medskip

\noindent
{\bf The Maurer-Cartan form.}
With the super-Lie group in hand, we may explicitly construct its Maurer-Cartan forms and with that, finally, the left-invariant form of $\widehat{P}_3$.

\begin{lemma}[\bf Maurer-Cartan forms in coordinates]
  \label{MCFormsInCoordinates}
  Consider a Lie algebra $\mathfrak{g} \simeq \big\langle (T_i)_{i =1}^n \big\rangle$ with Lie bracket $[T_i, T_j] = f^k_{i j} T_k$ which is nilpotent to third order, in that $\big[-\left[-[-,-]\right]\big] = 0$, so that the corresponding group product \eqref{DynkinSeries}
  is
\begin{equation}
  \label{GroupProductOnThirdOrderNilpotentGroup}
  \mathrm{prd}\big(
    x^{\pr\, i}\, T_{i}
    \,,\,
    x^{i}\, T_i
  \big)
  \;=\;
  x^{\pr i}
  +
  x^i
  +
  \tfrac{1}{2}
  f^i_{jk}
  x^{\pr j}x^k
  +
  \tfrac{1}{12}
  f^i_{j k}
  f^{k}_{k_1 k_2}
  x^{\pr\, j}
  x^{\pr\, k_1}x^{k_2}
  +
  \tfrac{1}{12}
  f^i_{j k}
  f^{k}_{k_1 k_2}
  x^{j}
  x^{k_1}x^{\pr\, k_2}
  .
\end{equation}
Then the corresponding integrating group's Maurer-Cartan forms may be given in these coordinates by:
\begin{equation}
  \label{MCFormInCoordinates}
  e^i
  \;=\;
  \mathrm{d}x^i
  -
  \tfrac{1}{2}
  f^i_{j k}
  x^j \mathrm{d}x^k
  +
  \tfrac{1}{6}
  f^i_{j k'}
  f^{k'}_{kl}
  x^j x^k \mathrm{d}x^l
\end{equation}
in that
\begin{itemize}[
  topsep=1pt,
  itemsep=2pt
]
  \item[\bf (i)] {\rm (MC equation)}  $\mathrm{d}e^i \,=\, -\tfrac{1}{2}f^i_{j k} e^j e^k \,,$

  \item[\bf (ii)] {\rm (Left-invariance)} $\mathrm{act}^\ast e^i \,=\, e^i \,.$
\end{itemize}
\end{lemma}
\begin{proof}
  Checking this directly is straightforward, if already somewhat tedious. For the terms quadratic in $f$ the check relies heavily on the Jacobi identity.

  Alternatively, the expression follows from the general Hausdorff-like formula of Schur 
  %\cite[(34)]{Schur1891} 
  (see \cite[\S II Thm. 7.4 \& p. 36]{Helgason01}\cite[Thm. C.2 \& p. 99]{Meinrenken13}), according to which $e^i$ is given at any point $X = x^i T_i \in \mathfrak{g}$ as
  $$
    \def\arraystretch{1.6}
    \begin{array}{ccl}
    e^i
    &=&
    \mathrm{d}x^i
    \left(
      \frac
        {1 - \exp(- \mathrm{ad}X)}
        {\mathrm{ad}X}
      (\partial_{x^k})
    \right)
    \mathrm{d}x^k
    \\
    &:=&
    \mathrm{d}x^i
    \left(
      \textstyle{\sum_{n=0}^\infty}
      \tfrac{1}{(n+1)!}
      (-\mathrm{ad} X)^n(\partial_{x^k})
    \right)
    \mathrm{d}x^k
    \,.
    \end{array}
  $$
  Plugging in $\mathrm{ad} X \;=\; \big(x^j f^\bullet_{j \bullet}\big)$ into this formula, its first three summands are as claimed in \eqref{MCFormInCoordinates}.
\end{proof}

\begin{example}[\bf Maurer-Cartan  forms on the hidden M-group]
By plugging in the structure constants \eqref{ExceptionalCEDifferentialOnEta}
\eqref{BracketsInTheSuperExceptionalLieAlgebra}
of the hidden M-algebra at any stage $q$ \eqref{qProbesOfHiddenMAlgebra}
into this formula \eqref{MCFormInCoordinates}, we obtain a coordinate expression for the Maurer-Cartan form on the hidden M-group (cf. also \cite[(6.7.7)]{Varela06}):
\begin{equation}
  \label{MCFormOnHiddenMGroup}
  \def\arraystretch{1.3}
  \def\arraycolsep{1pt}
  \begin{array}{ccll}
    e^a
    &\;=\;&
    \mathrm{d}x^a
    &
    +
    \big(
      \hspace{1pt}
      \overline{\theta}
      \,\Gamma^a\,
      \mathrm{d}\theta
    \big)
    \\
    e_{a_1 a_2}
    &\;=\;&
    \mathrm{d}b_{a_1 a_2}
    &
    -
    \big(
      \hspace{1pt}
      \overline{\theta}
      \,\Gamma_{a_1 a_2}\,
      \mathrm{d}\theta
    \big)
    \\
    e_{a_1 \cdots a_5}
    &\;=\;&
    \mathrm{d}b_{a_1 \cdots a_5}
    &
    +
    \big(
      \hspace{1pt}
      \overline{\theta}
      \,\Gamma_{a_1 \cdots a_5}\,
      \mathrm{d}\theta
    \big)
    \\
    \psi
    &\;=\;&
    \mathrm{d}\theta
    \\
    \DFSpinor
    &\;=\;&
    \mathrm{d}\xi
    &
    -\tfrac{1}{2}
    \paramOne
    \big(
      x^a \Gamma_a \mathrm{d}\theta
      -
      \Gamma_a \theta
      (\mathrm{d} x^a) 
    \big)
    \\
    &&&
    -\tfrac{1}{2}
    \paramTwo
    \big(
      x^{a_1 a_2} \Gamma_{a_1 a_2} \mathrm{d}\theta
      -
      \Gamma_{a_1 a_2} \theta
      (\mathrm{d} x^{a_1 a_2}) 
    \big)
    \\
    &&&
    -\tfrac{1}{2}
    \paramFive
    \big(
      x^{a_1 \cdots a_5} \Gamma_{a_1 \cdots a_5} \mathrm{d}\theta
      -
      \Gamma_{a_1 \cdots a_5} \theta
      (\mathrm{d} x^{a_1 \cdots a_5})
    \big)
    \\
    &&&
    + 
    \tfrac{2}{6}
    \big(
      \paramOne
      \,
      \Gamma^a_{\alpha \beta}
      \Gamma_{a \, \gamma}
      -
      \paramTwo
      \Gamma^{a_1 a_2}_{\alpha \beta}
      \Gamma_{a_1 a_2 \, \gamma}
      +
      \paramFive
      \Gamma^{a_1 \cdots a_5}_{\alpha \beta}
      \Gamma_{a_1 \cdots a_5 \, \gamma}
    \big)
    \theta^\gamma
    \theta^\alpha
    \mathrm{d}\theta^\beta.
  \end{array}
\end{equation}  
In this left-invariant basis of Maurer-Cartan forms, the left-invariant 1-form globalization of $\widehat{P}_3$ takes exactly the same form as in \eqref{AnsatzForH30}, while substituting the right hand side of these equations into \eqref{AnsatzForH30} yields its expansion in the coordinate 1-form basis.
\end{example}

\medskip

\noindent
{\bf Toroidal compactifications.}
The hidden M-group $\widehat{\mathcal{M}}$ in \eqref{IntegratingSuperExceptionalMinkowsliLieAlgebra} is evidently the simply-connected Lie integration of the hidden M-algebra $\widehat{\mathfrak{M}}$. From it we may obtain its non-simply-connected versions by quotienting out lattice subgroups $\mathbb{Z}^k$.  In straightforward variation of Ex. \ref{CircleGroup} this is now immediate, but consequential:

\begin{example}[\bf Fully toroidal version of the hidden M-group]
Just for ease of notation, consider the case of the inclusion of all of $\mathbb{Z}^{528}$ (more generally, just omit any factors of this direct product group in the following).
Denoting the $528$ 
``canonical coordinate functions'' 
\eqref{CanonicalCoordinateOnZ}
on $\mathbb{Z}^{528}$ by the same symbols as the bosonic canonical coordinate functions
on $\widehat{\mathcal{M}}$ \eqref{UnderlyingManifoldOfHiddenMGroup},
it is readily seen 
from \eqref{TheMSuperGroup}
that the following tautological-looking assignment is a super-Lie group homomorphism as anticipated in \eqref{ToroidalMGeometryInIntro}:
$$
  \begin{tikzcd}[row sep=-3pt, column sep=0pt]
    \mathbb{Z}^{528}
    \ar[
      rr,
      hook
    ]
    &&
    \widehat{\mathcal{M}}
    \\
    x^a &\longmapsfrom& x^a
    \\
    b_{a_1 a_2}
    &\longmapsfrom& 
    b_{a_1 a_2}
    \\
    b_{a_1 \cdots a_5}
    &\longmapsfrom&
    b_{a_1 \cdots a_5}
    \\
    0
    &\longmapsfrom&
    \theta
    \\
    0
    &\longmapsfrom&
    \xi 
    \mathrlap{\,.}
  \end{tikzcd}
$$
The same formulas show that this inclusion factors through an inclusion of $\mathbb{R}^{528}$ and descends to an inclusion into the basic M-group:
$$
  \begin{tikzcd}
    \mathbb{Z}^{528}
    \ar[r, hook]
    \ar[d, equals]
    &
    \mathbb{R}^{528}
    \ar[r, hook]
    \ar[d, equals]
    &
    \widehat{\mathcal{M}}
    \ar[d, ->>]
    \\
    \mathbb{Z}^{528}
    \ar[r, hook]
    &
    \mathbb{R}^{528}
    \ar[r, hook]
    &
    {\mathcal{M}}
  \end{tikzcd}
$$
Hence, passing to the quotient of this group inclusion -- the fully toroidal hidden M-group -- means, as in \eqref{FunctionsOnCircleAsFiberProduct}, to restrict the bosonic elements in $C^\infty\big(\widehat{\mathcal{M}}\, \big)$ to those which are suitably periodic. 
This is, of course, just as it should be.
\end{example}

%%%%%%%%%%%%%%%%%%%%%%%%%%%%%%%
\section{Conclusion}
%%%%%%%%%%%%%%%%%%%%%%%%%%%%%%%

Motivated by a recent re-understanding of the relevance --  for potentially formulating M-theory -- of the hidden M-algebra and of the ``decomposed'' M-theory 3-form it carries, we have given first a careful re-derivation and then have discussed in detail its integration/globalization to a super-Lie group, the hidden M-group, carrying a corresponding left-invariant super 3-form. Despite the common abuse of terminology that suggests otherwise, this seems to be the first discussion of this super-Lie group, and therefore, we took the time to review the relevant streamlined theory of super-Lie groups along the way.

\medskip 
In its vanilla form, the hidden M-group is simply-connected. But with this in hand, its toroidially compactified versions are easy to come by, which we discussed as a first simple but important 
example of a topologically non-trivial super-exceptional spacetime.

\medskip 
These are going to be important both in relating super-exceptional formulations of 11D SuGra to topological T-duality, under dimensional reduction, but in particular due to the fact that the global completion of 11D SuGra by a flux quantization law leads to new solitonic states of the C-field not only on topologically non-trivial spacetime domains but now also on their much larger super-exceptional enhancement. Such global effects in exceptional-geometric super-gravity seem not to have received attention before, we will discuss first examples in \cite{GSS25-AnyonsOnExceptional}.

\newpage 

\appendix

%%%%%%%%%%%%%%%%%%%%%%%%%%%%%%%%
\section*{Background}
%%%%%%%%%%%%%%%%%%%%%%%%%%%%%%%

For ease of reference we briefly recall and cite some notation and facts used in the main text.

\vspace{-4mm} 
%%%%%%%%%%%%%%%%%%%%%%%%%%%%%%%%%%
\subsection*{Tensor conventions}
\label{TensorConventionsAnd11dSpinors}
%%%%%%%%%%%%%%%%%%%%%%%%%%%%%%%%%% 

\vspace{-2mm} 
Our tensor conventions are standard, but since the computations below crucially depend on the corresponding prefactors, here to briefly make them explicit:\begin{itemize}[leftmargin=.4cm]
\item
  The Einstein summation convention applies throughout: Given a product of terms indexed by some $i \in I$, with the index of one factor in superscript and the other in subscript, then a sum over $I$ is implied:
  $
    x_i \, y^i
    :=
    \sum_{i \in I} 
    x_i \, y^i
  $.

\item
Our Minkowski metric is ``mostly plus''
\begin{equation}
  \label{MinkowskiMetric}
  \big(\eta_{ab}\big)
    _{a,b = 0}
    ^{ d }
  \;\;
    =
  \;\;
  \big(\eta^{ab}\big)
    _{a,b = 0}
    ^{ d }
  \;\;
    :=
  \;\;
  \big(
    \mathrm{diag}
      (-1, +1, +1, \cdots, +1)
  \big)_{a,b = 0}^{d} \;.
\end{equation}
\item
  Shifting position of frame indices always refers to contraction with the  Minkowski metric \eqref{MinkowskiMetric}:
  $$
    V^a 
      \;:=\;
    V_b \, \eta^{a b}
    \,,
    \;\;\;\;
    V_a \;=\; V^b \eta_{a b}
    \,.
  $$
\item Skew-symmetrization of indices is denoted by square brackets ($(-1)^{\vert\sigma\vert}$ is sign of the permutation $\sigma$):
$$
  V_{[a_1 \cdots a_p]}
  \;:=\;
  \tfrac{1}{p!}
  \sum_{
    \sigma \in \mathrm{Sym}(n)
  }
  (-1)^{\vert \sigma \vert}
  V_{ a_{\sigma(1)} \cdots a_{\sigma(p)} }\,.
$$
\vspace{-1mm} 
\item
We normalize the Levi-Civita symbol to \begin{equation}
  \label{transversalizationOfLeviCivitaSymbol}
  \epsilon_{0 1 2 \cdots} 
    \,:=\, 
  +1
  \;\;\;\;\mbox{hence}\;\;\;\;
  \epsilon^{0 1 2 \cdots} 
    \,:=\, 
  -1
  \,.
\end{equation}
\item
We normalize the Kronecker symbol to
\begin{equation}
  \label{KroneckerSymbol}
  \delta
    ^{a_1 \cdots a_p}
    _{b_1 \cdots b_p}
  \;:=\;
  \delta^{[a_1}_{[b_1}
  \cdots
  \delta^{a_p]}_{b_p]}
  \;=\;
  \delta^{a_1}_{[b_1}
  \cdots
  \delta^{a_p}_{b_p]}
  \;=\;
  \delta^{[a_1}_{b_1}
  \cdots
  \delta^{a_p]}_{b_p}
\end{equation}
so that
\begin{equation}
  \label{ContractingKroneckerWithSkewSymmetricTensor}
  V_{
    \color{darkblue}
    a_1 \cdots a_p
  }
  \tensor*
    {\delta}
    {
    ^{ 
       \color{darkblue}
       a_1 \cdots a_p 
    }
    _{b_1 \cdots b_p}
    }
  \;\;
  =
  \;\;
  V_{[b_1 \cdots b_p]}  
  \;\;\;\;
  \mbox{and}
  \;\;\;\;
  \epsilon^{
    {\color{darkblue}  
      c_1 \cdots c_p
    }
    a_1 \cdots a_q
  }
  \,
  \epsilon_{
    {\color{darkblue}
    c_1 \cdots c_p 
    }
    b_1 \cdots b_q
  }
  \;\;
  =
  \;\;
  -
  \,
  p! \cdot q!
  \,
  \delta
    ^{a_1 \cdots a_q}
    _{b_1 \cdots b_q}
  \,.
\end{equation}
\end{itemize}

\vspace{-2mm} 
%%%%%%%%%%%%%%%%%%%%%%%%%%%%%%%%%%%%%%%%%%%%
\subsection*{Super-algebra}
\label{SuperAlgebraConventions}
%%%%%%%%%%%%%%%%%%%%%%%%%%%%%%%%%%%%%%%%%%%%

\vspace{-2mm} 
\noindent
{\bf Sign rule.}
For homological super-algebra, we consider bigrading in the direct product ring $\mathbb{Z} \times \ZTwo$ ---  where the first factor $\mathbb{Z}$ is the homological degree and the second $\ZTwo \simeq \{\mathrm{evn}, \mathrm{odd}\}$ the super-degree -- with sign rule
\begin{equation}
  \label{Signs}
  \mathrm{deg}_1 = 
  (n_1, \sigma_1),
  \;
  \mathrm{deg}_2
  =
  (n_2, \sigma_2)
    \,\in\,
  \mathbb{Z}\times \ZTwo
  \hspace{1cm}
    \yields
  \hspace{1cm}
  \mathrm{sgn}
  \big(
    \mathrm{deg}_1,
    \,
    \mathrm{deg}_2
  \big)
  \;:=\;
  (-1)^{n_1 \cdot n_2 + \sigma_1 \cdot \sigma_2}
  \,.
\end{equation}

For $(v_i)_{i \in I}$ a set of generators with bi-degrees $(\mathrm{deg}_i)_{i \in I}$ we write:
\begin{itemize}[
  leftmargin=1.2cm,
  topsep=2pt,
  itemsep=4pt
]
\item[\bf (i)]
$
  \mathbb{R}\big\langle
    (v_i)_{i \in I}
  \big\rangle
$
for the graded super-vector space spanned by these elements,
\item[\bf (ii)]
$
  \mathbb{R}\big[
    (v_i)_{i \in I}
  \big]
$
for the graded-commutative polymonial algebra generated by these elements, 

hence the tensor algebra on $\vert I\vert$ generators modulo the relation
\begin{equation}
  \label{TheSignRule}
  v_1 \cdot v_2
  \;=\;
  (-1)^{\mathrm{sgn}(
    \mathrm{deg}_1
    ,
    \mathrm{deg}_2
  )}
  \;
  v_2 \cdot v_1
  \,,
\end{equation}

hence the (graded, super) {\it symmetric algebra} on the above super-vector space:
$$
  \mathbb{R}\big[ (v_i)_{i \in I} \big]
  \;\;
  :=
  \;\;
  \mathrm{Sym}
  \big(
    \mathbb{R}\big\langle 
      (v_i)_{i \in I} 
    \big\rangle
  \big).
$$
\item[\bf (iii)]
$\FDGCA\big[ (v_i)_{i \in I} \big]$ for the (free) differential graded-commutative algebra (dgca) generated by these elements and their \textit{differentials} 
$$(\dd v_i)_{i\in I}
$$ 
treated as primitive elements 
with $\mathrm{deg}(\dd e_i) = \deg(e_i) + (1,\mathrm{evn})$ and modulo the corresponding relation \eqref{TheSignRule}, with differential defined by 
\vspace{-1mm} 
\begin{align*} e_i \longmapsto \dd e_i  \hspace{1cm} ,  \hspace{1cm} 
\dd e_i \longmapsto 0
\end{align*}
and extended as a (graded) `derivation, hence the dgca 
\begin{equation}\label{FreeDGCA}
 \FDGCA \big[ (v_i)_{i \in I} \big]     \;
  :=
  \;
  \Big( \mathrm{Sym}
  \big(
    \mathbb{R}\big\langle 
      (v_i)_{i \in I}, \,  (\dd v_i)_{i \in I} 
    \big\rangle
  \big), \, \dd \Big).
\end{equation}

\end{itemize}

\vspace{-3mm} 
%%%%%%%%%%%%%%%%%%%%%%%%%%%%%%%
\subsection*{Spinors in 11d}
\label{SpinorsIn11d}
%%%%%%%%%%%%%%%%%%%%%%%%%%%%%%%

\vspace{-2mm} 
We briefly record the following standard facts (proofs and references may be found in \cite[\S 2.5]{MiemiecSchnakenburg06}\cite[\S 2.2.1]{GSS24-SuGra}):
There exists an $\mathbb{R}$-linear representation $\mathbf{32}$ of $\mathrm{Pin}^+(1,10)$ with generators
\begin{equation}
  \label{The11dMajoranaRepresentation}
  \Gamma_a 
  \;:\;
  \mathbf{32}
  \xrightarrow{\;\;}
  \mathbf{32}
\end{equation}
equipped with a 
$\mathrm{Spin}(1,10)$-equivariant
skew-symmetric and non-degenerate
bilinear form
\begin{equation}
  \label{TheSpinorPairing}
  \big(\hspace{.8pt}
    \overline{(-)}
    (-)\,
  \big)
  \;:\;
  \mathbf{32}
  \otimes
  \mathbf{32}
  \xrightarrow{\quad}
  \mathbb{R}
\end{equation}
which serves as the {\it spinor metric} whose components we denote $\big(\eta_{\alpha\beta}\big)_{\alpha,\beta = 1}^{32}$:

\newpage 
\begin{equation}
  \label{TheSpinorMetric}
  \psi^\alpha
  \,
  \eta_{\alpha \beta}
  \,
  \phi^\beta
  \;\;
    :=
  \;\;
  \big(\hspace{1pt}
    \overline{\psi}
    \,
    \phi
  \big)
  \,,
\end{equation}
that are skew-symmetric in their indices
\begin{equation}
  \label{SkewSymmetryOfSpinorMetric}
  \eta_{\alpha\beta}
  \;=\;
  -
  \eta_{\beta\alpha}
\end{equation}
which together with the inverse matrix $(\eta^{\alpha \beta})$ is
and used to lower and raise spinor indices by contraction ``from the right'' (the position of the terms is irrelevant, since the components $\eta_{\alpha\beta}$ are commuting numbers, but the order of the indices matters due to the skew-symmetry):
\begin{equation}
  \label{LoweringOfSpinorIndices}
  \psi_\alpha
  \;:=\;
  \psi^{\alpha'}
  \eta_{\alpha' \alpha}
  \,,
  \;\;\;\;\;
  \psi^\alpha
  \;=\;
  \psi_{\alpha'}
  \eta^{\alpha' \alpha}
  \,,
  \;\;\;\;\;
  \psi_\alpha \phi^\alpha
  \;=\;
  -
  \psi^\beta 
  \eta_{\beta \alpha}
  \eta^{\alpha \gamma}
  \phi_\gamma
  \;=\;
  -
  \psi^\alpha \phi_\alpha
  \,.
\end{equation}

This representation satisfies the following properties, where as 
usual we denote skew-symmetrized product of $k$ Clifford generators by
  \begin{equation}
    \label{CliffordBasisElements}
    \Gamma_{a_1 \cdots a_k}
    \;:=\;
    \tfrac{1}{k!}
    \underset{
      \sigma \in
      \mathrm{Sym}(k)
    }{\sum}
    \mathrm{sgn}(\sigma)
    \,
    \Gamma_{a_{\sigma(1)}}
    \cdot
    \Gamma_{a_{\sigma(2)}}
    \cdots
    \Gamma_{a_{\sigma(n)}}
    :
  \end{equation}

\begin{itemize}[leftmargin=.4cm]
  \item
  The Clifford generators square to the mostly plus Minkowski metric \eqref{MinkowskiMetric}
  \begin{equation}
    \label{CliffordDefiningRelation}
    \Gamma_a
    \Gamma_b
    +
    \Gamma_b
    \Gamma_a
    \;\;=\;\;
    +2 \, \eta_{a b}
    \,
    \mathrm{id}_{\mathbf{32}}
    \,.
  \end{equation}

  \item The Clifford product is given on the basis elements \eqref{CliffordBasisElements}
  as
\begin{equation}
  \label{GeneralCliffordProduct}
  \Gamma^{a_j \cdots a_1}
  \,
  \Gamma_{b_1 \cdots b_k}
  \;=\;
  \sum_{l = 0}^{
    \mathrm{min}(j,k)
  }
  \pm
  l!
\binom{j}{l}
 \binom{k}{l}
  \,
  \delta
   ^{[a_1 \cdots a_l}
   _{[b_1 \cdots b_l}
  \Gamma^{a_j \cdots a_{l+1}]}
  {}_{b_{l+1} \cdots b_k]}
  \,.
\end{equation}
  
  \item 
  The Clifford volume form equals the Levi-Civita symbol 
  \eqref{transversalizationOfLeviCivitaSymbol}:
  \begin{equation}
    \label{CliffordVolumeFormIn11d}
    \Gamma_{a_1 \cdots a_{11}}
    \;=\;
    \epsilon_{a_1 \cdots a_{11}}
    \mathrm{id}_{\mathbf{32}}
    \,.
  \end{equation}
\item The trace of all positive index Clifford basis elements vanishes:
  \begin{equation}
    \label{VanishingTraceOfCliffordElements}
    \mathrm{Tr}(
      \Gamma_{a_1 \cdots a_p}
    )
    \;=\;
    \left\{\!\!\!
    \def\arraystretch{1.1}
    \begin{array}{ccc}
      32 & \vert & p = 0
      \\
      0 & \vert & p > 0 
      \mathrlap{\,.}
    \end{array}
    \right.
  \end{equation}
\item The Hodge duality relation on Clifford elements is:
\begin{equation}
  \label{HodgeDualityOnCliffordAlgebra}
  \Gamma^{a_1 \cdots a_p}
  \;=\;
  \tfrac{
    (-1)^{
      (p+1)(p-2)/2
    }
  }{
    (11-p)!
  }
  \,
  \epsilon^{ 
    a_1 \cdots a_p
    \,
    b_1 \cdots a_{11-p}
  }
  \,
  \Gamma_{b_1 \cdots b_{11-p}}
  \,.
\end{equation}
For instance:
\begin{equation}
  \label{ExamplesOfHodgeDualCliffordElements}
  \def\arraystretch{1.6}
  \def\arraycolsep{10pt}
  \begin{array}{l}
    \Gamma^{
      a_1 \cdots a_{11}
    }
    =
    \epsilon^{a_1 \cdots a_{11}}
    \mathrm{Id}_{\mathbf{32}}
    \,,
    \;\;\;\;\;\;\;
    \Gamma^{a_1 \cdots a_6}
    \;=\;
    +
    \tfrac{
      1
    }{
      5!
    }
    \,
    \epsilon^{
      a_1 \cdots a_6
      \,
      \color{darkblue}
      b_1 \cdots b_5
    }
    \,
    \Gamma_{
      \color{darkblue}
      b_1 \cdots b_5
    }
    \,,
    \\
    \Gamma^{a_1 \cdots a_{\ten}}
    =
    \epsilon^{
      a_1 \cdots a_{\ten} 
      \color{darkblue}
      b
    }
    \,
    \Gamma_{
      \color{darkblue}
      b
    }
    \,,
    \;\;\;\;\;\;\;\;
    \Gamma^{a_1 \cdots a_5}
    \;=\;
    -
    \tfrac{
      1
    }{
      6!
    }
    \,
    \epsilon^{
      a_1 \cdots a_5
      \,
      \color{darkblue}
      b_1 \cdots b_6
    }
    \,
    \Gamma_{
      \color{darkblue}
      b_1 \cdots b_6
    }
    \,.
  \end{array}
\end{equation}

  \item The Clifford generators are skew self-adjoint with respect to the pairing \eqref{TheSpinorPairing}
  \vspace{1mm} 
  \begin{equation}
    \label{SkewSelfAdjointnessOfCliffordGenerators}
    \overline{\Gamma_a}
    \;=\;
    - \Gamma_a
    \;\;\;\;\;\;
    \mbox{in that}
    \;\;\;\;\;\;
    \underset{
      \phi,\psi \in \mathbf{32}
    }{\forall}
    \;\;
    \big(\,
      \overline{(\Gamma_a \phi)}
      \,
      \psi
    \big)
    \;=\;
    -
    \big(\,
      \overline{\phi}
      \,
      (\Gamma_a \psi)
    \big)
    \,,
  \end{equation}
  so that generally
  \vspace{1mm} 
  \begin{equation}
    \label{AdjointnessOfCliffordBasisElements}
    \overline{\Gamma_{a_1 \cdots a_p}}
    \;=\;
    (-1)^{
      p + p(p-1)/2
    }
    \,
    \Gamma_{a_1 \cdots a_p}
    \,.
  \end{equation}

  \item
  The $\mathbb{R}$-vector space of $\mathbb{R}$-linear  endomorphisms of $\mathbf{32}$ has a linear basis given by the $\leq 5$-index Clifford elements 
  \vspace{1mm} 
  \begin{equation}
    \label{CliffordElementsSpanningLinearMaps}
    \mathrm{End}_{\mathbb{R}}\big(
      \mathbf{32}
    \big)
    \;\;
    =
    \;\;
    \big\langle
      1
      ,\,
      \Gamma_{a_1}
      ,\,
      \Gamma_{a_1 a_2}
      ,\,
      \Gamma_{a_1 a_2 a_3}
      ,\,
      \Gamma_{a_1 \cdots a_4}
      ,\,
      \Gamma_{a_1 \cdots a_5}
    \big\rangle_{
      a_i = 0, 1, \cdots
    }
    \,,
    \end{equation}

  \item
  The $\mathbb{R}$-vector space space of {\it symmetric} bilinear forms on $\mathbf{32}$
  has a linear basis given by the expectation values with respect to \eqref{TheSpinorPairing} of the 1-, 2-, and 5-index Clifford basis elements:
  \begin{equation}
    \label{SymmetricSpinorPairings}
    \mathrm{Hom}_{\mathbb{R}}
    \big(
    (\mathbf{32}\otimes \mathbf{32})_{\mathrm{sym}}
    ,\,
    \mathbb{R}
    \big)
    \;\;
    \simeq
    \;\;
    \Big\langle
    \big(
      (\overline{-})
      \Gamma_a
      (-)
    \big)
    \,,\;\;
    \big(
      (\overline{-})
      \Gamma_{a_1 a_2}
      (-)
    \big)
    \,,\;\;
    \big(
      (\overline{-})
      \Gamma_{a_1 \cdots a_5}
      (-)
    \big)
    \Big\rangle_{
      a_i = 0, 1, \cdots
      \,,
    }
  \end{equation}
    which means in components that these Clifford generators are symmetric in their lowered indices \eqref{LoweringOfSpinorIndices}:
    \begin{equation}
      \label{SymmetryOfCliffordBasisElements}
      \Gamma^a_{\alpha \beta}
      \;=\;
      \Gamma^a_{\beta\alpha}
      \,,\;\;\;\;
      \Gamma^{a_1 a_2}_{\alpha \beta}
      \;=\;
      \Gamma^{a_1 a_2}_{\beta\alpha}
      \,,\;\;\;\;
      \Gamma^{a_1 \cdots a_5}_{\alpha \beta}
      \;=\;
      \Gamma^{a_1 \cdots a_5}_{\beta\alpha}
      \,,
    \end{equation}
  while a basis for the skew-symmetric bilinear forms is given by
  \begin{equation}
    \label{SkewSpinorPairings}
    \mathrm{Hom}_{\mathbb{R}}
    \big(
    (\mathbf{32}\otimes \mathbf{32})_{\mathrm{skew}}
    ,\,
    \mathbb{R}
    \big)
    \;\;
    \simeq
    \;\;
    \Big\langle
    \big(
      (\overline{-})
      (-)
    \big)
    \,,\;\;
    \big(
      (\overline{-})
      \Gamma_{a_1 a_2 a_3}
      (-)
    \big)
    \,,\;\;
    \big(
      (\overline{-})
      \Gamma_{a_1 \cdots a_4}
      (-)
    \big)
    \Big\rangle_{
      a_i = 0, 1, \cdots
      \,,
    }
  \end{equation}
    which means in components that these Clifford generators are skew-symmetric in their lowered indices \eqref{LoweringOfSpinorIndices}:
  \begin{equation}
    \eta_{\alpha\beta}
    \;=\;
    -
    \eta_{\beta\alpha}
    \,,\;\;\;\;
    \Gamma^{a_1 a_2 a_3}_{\alpha\beta}
    \;=\;
    -
    \Gamma^{a_1 a_2 a_3}_{\beta\alpha}
    \,,\;\;\;\;
    \Gamma^{a_1 \cdots a_5}_{\alpha\beta}
    \;=\;
    -
    \Gamma^{a_1 \cdots a_5}_{\beta\alpha}
  \end{equation}
\item
  Any linear endomorphism $\phi \in \mathrm{End}_{\mathbb{R}}(\mathbf{32})$ is uniquely a linear combination of Clifford elements as:
  \begin{equation}
    \label{CliffordExpansionOfEndomorphismOf32}
    \phi
      \;=\;
    \tfrac{1}{32}
    \sum_{p = 0}^5
    \;
    \tfrac{
      (-1)^{p(p-1)/2}
    }{ p! }
    \mathrm{Tr}\big(
      \phi \circ 
      \Gamma_{a_1 \cdots a_p}
    \big)
    \Gamma^{a_1 \cdots a_p}
    \,;
  \end{equation}

\item which implies in particular the Fierz expansion
\begin{equation}
  \label{FierzDecomposition}
  \hspace{-3mm} 
  \big(\,
  \overline{\phi}_1
  \,
  \psi
  \big)
  \big(\,
  \overline{\psi}
  \,
  \phi_2
  \big)
  \;
  =
  \;
  \tfrac{1}{32}\Big(
    \big(\,
      \overline{\psi}
      \,\Gamma^a\,
      \psi
    \big)
    \big(\,
      \overline{\phi}_1
      \,\Gamma_a\,
      \phi_2
    \big)
    -
    \tfrac{1}{2}
    \big(\,
      \overline{\psi}
      \,\Gamma^{a_1 a_2}\,
      \psi
    \big)
    \big(\,
      \overline{\phi}_1
      \,\Gamma_{a_1 a_2}\,
      \phi_2
    \big)
    +
    \tfrac{1}{5!}
    \big(\,
      \overline{\psi}
      \,\Gamma^{a_1 \cdots a_5}\,
      \psi
    \big)
    \big(\,
      \overline{\phi}_1
      \,\Gamma_{a_1 \cdots a_5}\,
      \phi_2
    \big)
 \Big).
\end{equation}

\end{itemize}

 \smallskip 
\begin{proposition}[{\bf The general Fierz identities} {\cite[(3.1-3) \& Table 2]{DF82}\cite[(II.8.69) \& Table II.8.XI]{CDF91}}]
 
  \noindent {\bf (i)}  The $\mathrm{Spin}(1,10)$-irrep decomposition of the first few symmetric tensor powers of $\mathbf{32}$ is:
  \begin{equation}
    \label{IrrepsInSymmetricPowersOf32}
    \def\arraystretch{1.3}
    \begin{array}{rcl}
       \big(
        \mathbf{32} \otimes \mathbf{32}
      \big)_{\mathrm{sym}}
      &\cong&
      \mathbf{11}
      \,\oplus\,
      \mathbf{55}
      \,\oplus\,
      \mathbf{462}
      \\
       \big(
        \mathbf{32}
          \otimes
        \mathbf{32} 
          \otimes 
        \mathbf{32}
      \big)_{\mathrm{sym}}
      &\cong&
      \mathbf{32}
      \,\oplus\,
      \mathbf{320}
      \,\oplus\,
      \mathbf{1408}
      \,\oplus\,
      \mathbf{4424}
      \\
       \big(
        \mathbf{32}
          \otimes
        \mathbf{32} 
          \otimes 
        \mathbf{32}
          \otimes 
        \mathbf{32}
      \big)_{\mathrm{sym}}
      &\cong&
      \mathbf{1}
      \,\oplus\,
      \mathbf{165}
      \,\oplus\,
      \mathbf{330}
      \,\oplus\,
      \mathbf{462}
      \,\oplus\,
      \mathbf{65}
      \,\oplus\,
      \mathbf{429}
      \,\oplus\,
      \mathbf{1144}
      \,\oplus\,
      \mathbf{17160}
      \,\oplus\,
      \mathbf{32604}\,.
    \end{array}
  \end{equation}
  \noindent 
  {\bf (ii)} In more detail, the irreps appearing on the right are tensor-spinors spanned by basis elements
  \begin{equation}
    \label{TheHigherTensorSpinors}
    \def\arraystretch{1.6}
    \begin{array}{l}
    \big\langle
      \Xi^\alpha_{a_1 \cdots a_p}
      \;=\;
      \Xi^\alpha_{[a_1 \cdots a_p]}
    \big\rangle_{
      a_i \in \{0,\cdots, 10\}, 
      \alpha \in \{1, \cdots 32\}
    }
    \;\;\;
    \in
    \;\;
    \mathrm{Rep}_{\mathbb{R}}\big(
      \mathrm{Spin}(1,10)
    \big)
    \\
    \mbox{\rm with}
    \;\;\;
    \Gamma^{a_1} \Xi_{a_1 a_2 \cdots a_p}
    \;=\;
    0
    \end{array}
  \end{equation}
  {\rm (jointly to be denoted $\Xi^{(N)}$ for the case of the irrep $\mathbf{N}$)}
  such that:
  \begin{equation}
    \label{GeneralCubicFierzIdentities}
    \def\arraycolsep{3pt}
    \def\arraystretch{1.6}
    \begin{array}{rcrrrr}
      \psi
      \big(\,
        \overline{\psi}
        \,\Gamma_a\,
        \psi
      \big)
      &=&
      \tfrac{1}{11}
      \,\Gamma_a\,
      \Xi^{(32)}
      &+\;
      \Xi^{(320)}_a
      \mathrlap{\,,}
      \\
      \psi 
      \big(\,
        \overline{\psi}
        \,\Gamma_{a_1 a_2}\,
        \psi
      \big)
      &=&
      \tfrac{1}{11}
      \,\Gamma_{a_1 a_2}\,
      \Xi^{(32)}
      &-\;
      \tfrac{2}{9}
      \,\Gamma_{[a_1}\,
      \Xi^{(320)}_{a_2]}
      &+\;
      \Xi^{(1408)}_{a_1 a_2}
      \mathrlap{\,,}
      \\
      \psi
      \big(\,
        \overline{\psi}
        \,\Gamma_{a_1 \cdots a_5}\,
        \psi
      \big)
      &=&
      -\tfrac{1}{77}
      \Gamma_{a_1 \cdots a_5}
      \Xi^{(32)}
      &+\;
      \tfrac{5}{9}
      \Gamma_{[a_1 \cdots a_4}
      \Xi^{(320)}_{a_5]}
      &+\;
      2
      \,\Gamma_{[a_1 a_2 a_3}\,
      \Xi^{(1408)}_{a_4 a_5]}
      &+\;
      \Xi^{(4224)}_{a_1 \cdots a_5}
      \mathrlap{\,.}
    \end{array}
  \end{equation}
\end{proposition}

%%%%%%%%%%%%%%%%%%%%%%%%%%%%%%%%%%
\subsection*{Super-Lie algebras}
\label{SuperLieAlgebras}
%%%%%%%%%%%%%%%%%%%%%%%%%%%%%%%%%%

\vspace{-2mm} 
Our ground field is the real numbers $\mathbb{R}$ and all super-vector spaces are assumed to be finite-dimensional.

Given a finite dimensional super-Lie algebra $\mathfrak{g} \,\simeq\, \mathfrak{g}_{\mathrm{evn}} \oplus \mathfrak{g}_{\mathrm{odd}}$,
the linear dual of the super-Lie bracket map
$$
  [\mbox{-},\mbox{-}]
  \;:\;
  \begin{tikzcd}
    \mathfrak{g}
    \vee 
    \mathfrak{g}
    \ar[r]
    &
    \mathfrak{g}
  \end{tikzcd}
$$
may be understood to map the first two the second exterior power of the underlying dual super-vector space, and as such it extends uniquely to a $\mathbb{Z}\!\times\!\ZTwo$-graded derivation $\mathrm{d}$ of degree=$(1,\mathrm{evn})$ on the exterior super-algebra (where the minus sign is just a convention)
$$
  \begin{tikzcd}
    \wedge^1
    \mathfrak{g}^{\ast}
    \ar[
      rr,
      "{ 
        -[\mbox{-},\mbox{-}]^\ast
      }"
    ]
    \ar[
      d,
      hook
    ]
    &&
    \wedge^2 \mathfrak{g}^{\ast}
    \ar[
      d,
      hook
    ]
    \\
    \wedge^\bullet 
    \mathfrak{g}^\ast
    \ar[
      rr,
      "{ \mathrm{d} }"
    ]
    &&
    \wedge^\bullet 
    \mathfrak{g}^\ast    
  \end{tikzcd}
$$
With this, the condition $d \circ d = 0$ is equivalently the super-Jacobi identity on $[\mbox{-},\mbox{-}]$, and the
the resulting differential graded super-commutative algebra is know as the {\it Chevalley-Eilenberg algebra} of $\mathfrak{g}$:
$$
  \mathrm{CE}\big(
    \mathfrak{g}
    ,\,
    [\mbox{-},\mbox{-}]
  \big)
  \;\;
  :=
  \;\;
  \big(
    \wedge^\bullet \mathfrak{g}^\ast
    ,\,
    \mathrm{d}
  \big)
$$
and this construction is fully faithful
$$
  \begin{tikzcd}
    \mathrm{sLieAlg}_{\mathbb{R}}
    \ar[
      rr,
      hook,
      "{ \mathrm{CE} }"
    ]
    &&
    \mathrm{sDGCAlg}_{\mathbb{R}}
      ^{\mathrm{op}}
  \end{tikzcd}
$$
in that (1) for every super-vector space $V$ a choice of such differential $\mathrm{d}$ on $\wedge^\bullet V^\ast$ uniquely comes from a super-Lie bracket $[\mbox{-},\mbox{-}]$ on $V$ this way, and (2)
super-Lie homomorphisms $\phi : \mathfrak{g} \xrightarrow{\;} \mathfrak{g}'$ 
are in bijection with sDGC-algebra homomorphisms $\phi^\ast \,:\, \mathrm{CE}(\mathfrak{g}') \xrightarrow{} \mathrm{CE}(\mathfrak{g})$.

More concretely, given $(T_i)_{i =1}^n$ a linear basis for $\mathfrak{g}$ with corresponding structure constants $\big(f^k_{i j} \in \mathbb{R}\big)_{i,j,k = 1}^n$, then the Chevalley-Eilenberg algebra is equivalently the graded-commutative polynomial algebra
$$
  \mathrm{CE}\big(
    \mathfrak{g}, 
    [\mbox{-},\mbox{-}]
  \big)
  \;\simeq\;
  \big(
    \mathbb{R}\big[
      t^1, \cdots, t^1
    \big]
    ,\,
    \mathrm{d}
  \big)
$$
on generators of degree $(1,\sigma_i)$ with corresponding structure constants for its differential:
\smallskip 
\begin{equation}
\label{RelationBetweenStructureConstants}
\mbox{
\def\arraystretch{1.7}
\begin{tabular}{|c||c|c|}
  \hline
  &
  \bf 
  \def\arraystretch{.95}
  \begin{tabular}{c}
    Super
    \\
    Lie algebra
  \end{tabular}
  &
  \bf 
  \def\arraystretch{.95}
  \begin{tabular}{c}
    Super
    \\
    dgc-algebra
  \end{tabular}
  \\
  \hline
  \hline
  Generators
  &
  $
    \big(
      \underbrace{
        T_i
      }_{ 
        \mathclap{
          \mathrm{deg} \,=\, (0,\sigma_i) 
        }
      }
    \big)_{i = 1}^n
  $
  &
  $
    \big(
      \underbrace{
        t^i
      }_{ 
        \mathclap{
          \mathrm{deg} \,=\, (1,\sigma_i) 
        }
      }
    \big)_{i = 1}^n
  $
  \\[-16pt]
  &&
  \\
  \hline
  \rowcolor{lightgray}
  Relations
  &
  $
    [T_i, T_j] 
    \,=\, 
    f^k_{i j}
    \,
    T_k
  $
  &
  $
    \mathrm{d}
    \,
    t^k
    \;=\;
    -
    \tfrac{1}{2}
    f^k_{i j}
    \, t^i t^j
  $
  \\
  \hline
\end{tabular}
}
\end{equation}


\medskip



\begin{thebibliography}{10}

\bibitem[Anc]{AncillaryFiles}
{Ancillary {\tt Mathematica} notebook}: 
[\href{https://ncatlab.org/schreiber/show/The+Hidden+M-Group#Anc}{\tt ncatlab.org/schreiber/show/The+Hidden+M-Group\#Anc}].


\bibitem[AD24]{AndrianopoliDAuria24}
L. Andrianopoli and R. D'Auria, {\it \color{darkblue}Supergravity in the Geometric Approach and its Hidden Graded Lie Algebra}, 
[\href{https://arxiv.org/abs/2404.13987}{\tt arXiv:2404.13987}].


\bibitem[ADR16]{ADR16}
L. Andrianopoli, R. D'Auria, and L. Ravera, 
{\it \color{darkblue}Hidden Gauge Structure of Supersymmetric Free Differential Algebras}, 
J. High Energy Phys. {\bf 1608} (2016) 095, 
[\href{https://doi.org/10.1007/JHEP08(2016)095}{\tt doi:10.1007/JHEP08(2016)095}],
[\href{https://arxiv.org/abs/1606.07328}{\tt arXiv:1606.07328}]. 

\bibitem[ADR17]{ADR17}
L. Andrianopoli, R. D'Auria, and L. Ravera, 
{\it \color{darkblue}More on the Hidden Symmetries of 11D Supergravity}, 
Phys. Lett. B {\bf 772} (2017), 578-585, 
[\href{https://doi.org/10.1016/j.physletb.2017.07.016}{\tt doi:10.1016/j.physletb.2017.07.016}], [\href{https://arxiv.org/abs/1705.06251}{\tt arXiv:1705.06251}].


% \bibitem[AdR20]{AdR20}
% M. F. Araujo de Resende,
% {\it \color{darkblue}A pedagogical overview on 2D and 3D Toric Codes and the origin of their topological orders}, 
% Rev. Math. Phys. {\bf 32} 02 (2020) 2030002,
% [\href{https://doi.org/10.1142/S0129055X20300022}{\tt doi:10.1142/S0129055X20300022}],
% [\href{https://arxiv.org/abs/1712.01258}{\tt arXiv:1712.01258}].

\bibitem[BW00]{BaerwaldWest00}
O. Baerwald and P. West, 
{\it \color{darkblue}Brane Rotating Symmetries and the Fivebrane Equations of Motion},
Phys. Lett. B {\bf 476} (2000), 157-164,
[\href{https://doi.org/10.1016/S0370-2693(00)00107-6}{\tt doi:10.1016/S0370-2693(00)00107-6}],
[\href{https://arxiv.org/abs/hep-th/9912226}{\tt arXiv:hep-th/9912226}].


% \bibitem[BaSh10]{BakulevShirkov10}
% A. P. Bakulev and D. Shirkov, 
% {\it \color{darkblue} Inevitability and Importance of Non-Perturbative Elements in Quantum Field Theory}, Proc. 6th Mathematical Physics Meeting, 
% Sep. 14–23 (2010), Belgrade, Serbia, 27–54, \newline 
% [\href{https://arxiv.org/abs/1102.2380}{\tt arXiv:1102.2380}].

\bibitem[Ba17]{Bandos17}
I. Bandos, 
{\it \color{darkblue}Exceptional field theories, superparticles in an enlarged 11D superspace and higher spin theories},
Nucl. Phys. B {\bf 925} (2017), 28-62,
[\href{https://doi.org/10.1016/j.nuclphysb.2017.10.001}{\tt doi:10.1016/j.nuclphysb.2017.10.001}],
[\href{https://arxiv.org/abs/1612.01321}{\tt arXiv:1612.01321}].



\bibitem[BDIPV04]{BDIPV04}
I. A. Bandos, J. A. de Azcarraga, J. M. Izquierdo, M. Picon, and O. Varela, 
{\it \color{darkblue}On the underlying gauge group structure of $D=11$ supergravity}, Phys. Lett. B {\bf 596} (2004), 145-155,
 [\href{https://arxiv.org/abs/hep-th/0406020}{\tt arXiv:hep-th/0406020}], \newline 
[\href{https://doi.org/10.1016/j.physletb.2004.06.079}{\tt 10.1016/j.physletb.2004.06.079}].


\bibitem[BDPV05]{BDPV05}
I. Bandos, J. de Azcárraga, M. Picon, and O. Varela, {\it\color{darkblue} On the formulation of $D=11$ supergravity and the composite nature of its three-from field}, Ann. Phys. {\bf 317} (2005), 238-279, [\href{https://doi.org/10.1016/j.aop.2004.11.016}{\tt doi:10.1016/j.aop.2004.11.016}],
[\href{https://arxiv.org/abs/hep-th/0409100}{\tt arXiv:hep-th/0409100}].

% \bibitem[BaSo23]{BandosSorokin23}
% I. A. Bandos and D. P. Sorokin, 
% {\it \color{darkblue} Superembedding approach to superstrings and super-$p$-branes}, in: {\it Handbook of Quantum Gravity}, Springer (2023), 
% [\href{https://doi.org/10.1007/978-981-19-3079-9_111-1}{\tt doi:10.1007/978-981-19-3079-9\_111-1}],
% [\href{https://arxiv.org/abs/2301.10668}{\tt arXiv:2301.10668}].

\bibitem[BW85]{BarrWells85}
M. Barr and C. Wells, 
{\it \color{darkblue}Toposes, Triples, and Theories}, Springer (1985); Reprints Theor. Appl. Categ. {\bf 12} (2005), 
1-287, [\href{http://www.tac.mta.ca/tac/reprints/articles/12/tr12abs.html}{\tt tac:tr12}].


\bibitem[Ba79]{Batchelor79}
M. Batchelor, {\it \color{darkblue}The structure of supermanifolds}, Trans. Amer. Math. Soc. {\bf 253} (1979), 329-338, 
\newline 
[\href{https://www.ams.org/journals/tran/1979-253-00/S0002-9947-1979-0536951-0}{\tt doi:1979-253-00/S0002-9947-1979-0536951-0}].

\bibitem[Ba84]{Batchelor84}
M. Batchelor, {\it \color{darkblue}Graded Manifolds and Supermanifolds}, in: {\it Mathematical Aspects of Superspace}, 
NATO ASI Series {\bf 132}, Springer (1984), 91-134, [\href{https://doi.org/10.1007/978-94-009-6446-4_4}{\tt doi:10.1007/978-94-009-6446-4\_4}].

\bibitem[Be87]{Berezin87}
F. A. Berezin (edited by A. A. Kirillov), 
{\it \color{darkblue} Lie Supergroups}, in: {\it 
Introduction to Superanalysis}, Mathematical Physics and Applied Mathematics {\bf 9}, Springer (1987), 
[\href{https://doi.org/10.1007/978-94-017-1963-6_8}{\tt doi:10.1007/978-94-017-1963-6\_8}].

\bibitem[BL75]{BerezinLeites75}
F. A. Berezin and D. A. Leites, 
{\it \color{darkblue}Supermanifolds}, (Russian) Dokl. Akad. Nauk SSSR {\bf 224} 3 (1975), 505-508; English translation: Soviet Math. Dokl. {\bf 16} 5 (1975), 1218-1222,
[\href{https://www.mathnet.ru/eng/dan39282}{\tt mathnet:eng/dan39282}].


% \bibitem[BB20]{BermanBlair20}
% D. Berman and C. D. A. Blair, 
% {\it \color{darkblue} The Geometry, Branes and Applications of Exceptional Field Theory}, 
% Int. J. Mod. Phys. A {\bf 35} 30 (2020) 2030014,
% [\href{https://doi.org/10.1142/S0217751X20300148}{\tt doi:10.1142/S0217751X20300148}],
% [\href{https://arxiv.org/abs/2006.09777}{\tt arXiv:2006.09777}].


% \bibitem[Bo${}^+$24]{BorstenEtAl24}
% L. Borsten,  M. Jalali Farahani, B. Jurco, H. Kim, J. Narozny, D. Rist, C. Saemann, and M. Wolf,
% {\it \color{darkblue} Higher Gauge Theory}, in
% {\it Encyclopedia of Mathematical Physics} 2nd ed, Elsevier (2024), 159-185, 
% [\href{https://arxiv.org/abs/2401.05275}{\tt arXiv:2401.05275}], [\href{https://doi.org/10.1016/B978-0-323-95703-8.00217-2}{\tt doi:10.1016/B978-0-323-95703-8.00217-2}].


% \bibitem[BKS19]{BossardKleinschmidtSezgin19}
% G. Bossard, A. Kleinschmidt, and E. Sezgin,
% {\it \color{darkblue} On supersymmetric $E_{11}$ exceptional field theory}, 
% J. High Energ. Phys. {\bf 2019} (2019) 165,
% [\href{https://doi.org/10.1007/JHEP10(2019)165}{\tt doi:10.1007/JHEP10(2019)165}],
% [\href{https://arxiv.org/abs/1907.02080}{\tt arXiv:1907.02080}].



% \bibitem[Br${}^+$14]{Brambilla14}
% Brambilla et al., 
% {\it \color{darkblue}  QCD and strongly coupled gauge theories -- challenges and perspectives}, 
% Eur. Phys. J. C Part. Fields {\bf 74} 10 (2014) 2981, 
% [\href{https://doi.org/10.1140/epjc/s10052-014-2981-5}{\tt doi:10.1140/epjc/s10052-014-2981-5}],
% [\href{https://arxiv.org/abs/1404.3723}{\tt arXiv:1404.3723}].


% \bibitem[BH80]{BH80}
% L. Brink and P. Howe, 
% {\it \color{darkblue}Eleven-Dimensional Supergravity on the Mass-Shell in Superspace}, Phys. Lett. B {\bf 91} (1980), 384-386, [\href{https://doi.org/10.1016/0370-2693(80)91002-3}{\tt doi:10.1016/0370-2693(80)91002-3}].

\bibitem[CR12]{CarchediRoytenberg12}
D. Carchedi and D. Roytenberg, 
{\it \color{darkblue}Homological Algebra for Superalgebras of Differentiable Functions}, \newline 
[\href{https://arxiv.org/abs/1212.3745}{\tt arXiv:1212.3745}].


\bibitem[Car06]{Carlevaro06}
L. Carlevaro, 
{\it \color{darkblue}Three approaches to M-theory}, PhD thesis,  Neuchatel University (2006), 
[\href{https://inspirehep.net/literature/1253257}{\tt spire:1253257}],
[\href{https://libra.unine.ch/handle/123456789/16186}{\tt hdl:123456789/16186}].


\bibitem[Cas11]{Castellani11}
L. Castellani,  
{\it \color{darkblue}Lie derivatives along antisymmetric tensors, and the M-theory superalgebra}, 
J. Phys. Math. {\bf 3} (2011), 1-7.
[\href{http://projecteuclid.org/euclid.jpm/1359468398}{\tt eculid:jpm/1359468398}],
[\href{https://arxiv.org/abs/hep-th/0508213}{\tt arXiv:hep-th/0508213}].

\bibitem[CDF91]{CDF91}
L. Castellani, R. D'Auria, and P. Fr{\'e}, {\it \color{darkblue} Supergravity and Superstrings -- A Geometric Perspective}, 
World Scientific (1991), 
[\href{https://doi.org/10.1142/0224}{\tt doi:10.1142/0224}].

% \bibitem[Ch18]{Chester18}
% S. Chester, 
% {\it \color{darkblue} Bootstrapping M-theory}, PhD thesis, Princeton (2018), 
% [\href{http://arks.princeton.edu/ark:/88435/dsp01kw52jb833}{\tt ark:/88435/dsp01kw52jb833}].


% \bibitem[CP18]{CP18}
% S. Chester and E. Perlmutter, 
% {\it \color{darkblue}M-Theory Reconstruction from $(2,0)$ CFT and the Chiral Algebra Conjecture}, 
% J. High Energ. Phys. {\bf 2018} (2018)  116, 
% [\href{https://doi.org/10.1007/JHEP08(2018)116}{\tt doi:10.1007/JHEP08(2018)116}],
% [\href{https://arxiv.org/abs/1805.00892}{\tt arXiv:1805.00892}].


\bibitem[Chr${}^+$00]{CdAIPB00}
C. Chrysso‌malakos, J. de Azc{\'a}rraga, J. M. Izquierdo, and C. Pérez Bueno, 
{\it \color{darkblue} The geometry of branes and extended superspaces}, 
Nucl. Phys. B {\bf 567} (2000), 293-330,  
[\href{https://doi.org/10.1016/S0550-3213(99)00512-X}{\tt doi:10.1016/S0550-3213(99)00512-X}], \newline 
[\href{https://arxiv.org/abs/hep-th/9904137}{\tt arXiv:hep-th/9904137}].

% \bibitem[Co07]{Cook07}
% P. Cook, 
% {\it \color{darkblue}  Connections between Kac-Moody algebras and M-theory}, 
% PhD thesis, King’s College London (2007),
% [\href{https://arxiv.org/abs/0711.3498}{\tt arXiv:0711.3498}].


\bibitem[CG04]{CorwinGreenleaf04}
L. Corwin and F. P. Greenleaf, 
{\it \color{darkblue}Representations of Nilpotent Lie Groups and their Applications -- Volume 1 Part 1: Basic Theory and Examples}, 
Cambridge University Press (2004),
[\href{https://www.cambridge.org/us/universitypress/subjects/mathematics/abstract-analysis/representations-nilpotent-lie-groups-and-their-applications-volume-1-part-1}{\tt ISBN:9780521604956}].


% \bibitem[Cr81]{Cremmer81}
% E. Cremmer, 
% {\it \color{darkblue}Supergravities in 5 dimensions}, in: {\it Superspace and Supergravity}, Cambridge University Press (1981), 
% [\href{https://inspirehep.net/literature/155020}{\tt spire:155020}].


% \bibitem[CF80]{CF80}
% E. Cremmer and S. Ferrara, 
% {\it \color{darkblue} Formulation of Eleven-Dimensional Supergravity in Superspace}, Phys. Lett. B {\bf 91} (1980), 61-66, 
% [\href{https://doi.org/10.1016/0370-2693(80)90662-0}{\tt doi:10.1016/0370-2693(80)90662-0}].


% \bibitem[CJ79]{CremmerJulia79}
% E. Cremmer and B. Julia, 
% {\it \color{darkblue}The $\mathrm{SO}(8)$ Supergravity}, 
% Nucl. Phys. B {\bf 159} (1979), 141-212, \newline  
% [\href{https://doi.org/10.1016/0550-3213(79)90331-6}{\tt doi:10.1016/0550-3213(79)90331-6}].


% \bibitem[DKN06]{DKN06}
% T. Damour, A. Kleinschmidt, and H.  Nicolai,
% {\it \color{darkblue}$K(E_{\ten})$, Supergravity and Fermions}, 
% J. High Energy Phys. {\bf 0608}  (2006) 056, 
% [\href{https://doi.org/10.1016/j.physletb.2006.04.007}{\tt doi;10.1016/j.physletb.2006.04.007}],
% [\href{https://arxiv.org/abs/hep-th/0606105}{\tt arXiv:hep-th/0606105}].


\bibitem[DF82]{DF82}
R. D'Auria and P. Fr\'e, 
{\it \color{darkblue} Geometric Supergravity in $D=11$ and its hidden supergroup}, 
Nucl. Phys. B {\bf 201} (1982), 101-140, 
[\href{https://doi.org/10.1016/0550-3213(82)90376-5}{\tt doi:10.1016/0550-3213(82)90376-5}].

% \bibitem[DFT79]{DFT79}
% R. D'Auria, P. Fr{\'e}, and T. Regge,
% {\it \color{darkblue}Geometrical formulation of supergravity as a theory on a supergroup manifold}, Supergravity Workshop, Stony Brook (1979), 85-92, [\href{https://inspirehep.net/literature/148583}{\tt spire:148583}].

\bibitem[dAz05]{Azcarraga05}
J. de Azc{\'a}rraga,
{\it \color{darkblue}Superbranes, $D=11$ CJS supergravity and enlarged superspace coordinates/fields correspondence}, 
AIP Conf. Proc. {\bf 767} (2005), 243-267, 
[\href{https://doi.org/10.1063/1.1923338}{\tt doi:10.1063/1.1923338}],
[\href{https://arxiv.org/abs/hep-th/0501198}{\tt arXiv:hep-th/0501198}]. 

% \bibitem[dBHP05]{dBuylHenneauxPaulot05}
% S. de Buyl, M. Henneaux, and L. Paulot,
% {\it \color{darkblue}Hidden Symmetries and Dirac Fermions}, Class. Quant. Grav. {\bf 22} (2005), 3595-3622,
% [\href{https://doi.org/10.1088/0264-9381/22/17/018}{\tt doi:10.1088/0264-9381/22/17/018}],
% [\href{https://arxiv.org/abs/hep-th/0506009}{\tt arXiv:hep-th/0506009}].




\bibitem[DF99]{DeligneFreed99}
P. Deligne and D. Freed, 
{\it \color{darkblue}Supersolutions}, in: {\it Quantum Fields and Strings, A Course for Mathematicians}, vol 1, 
Amer. Math. Soc. (1999),
[\href{https://bookstore.ams.org/qft-1-2-s}{\tt ISBN:978-0-8218-2014-8}],
[\href{https://arxiv.org/abs/hep-th/9901094}{\tt arXiv:hep-th/9901094}].

\bibitem[DM99]{DeligneMorgan99}
P. Deligne and J. Morgan, {\it \color{darkblue}Notes on Supersymmetry (following Joseph Bernstein)}, in: {\it Quantum Fields and Strings, A course for mathematicians}, Amer. Math. Soc. Providence (1999), 41-97, \newline 
[\href{https://bookstore.ams.org/qft-1-2-s}{\tt ISBN:978-0-8218-2014-8}].


\bibitem[DW84]{DeWitt92}
B. DeWitt, 
{\it \color{darkblue}Supermanifolds}, Monographs on Mathematical Physics, 
Cambridge University Press (1992), 
[\href{https://doi.org/10.1017/CBO9780511564000}{\tt doi:10.1017/CBO9780511564000}].


% \bibitem[dWN86]{deWitNicolai86}
% B. de Wit and H. Nicolai,
% {\it \color{darkblue}$D=11$ Supergravity With Local $\mathrm{SU}(8)$ Invariance}, 
% Nucl. Phys. B {\bf 274} (1986), 363-400, 
% [\href{https://doi.org/10.1016/0550-3213(86)90290-7}{\tt doi:10.1016/0550-3213(86)90290-7}].

% \bibitem[dWN01]{deWitNicolai01}
% B. de Wit and H. Nicolai,
% {\it \color{darkblue}Hidden Symmetries, Central Charges and All That}, 
% Class. Quant. Grav. {\bf 18} (2001), 3095-3112, 
% [\href{https://doi.org/10.1088/0264-9381/18/16/302}{\tt doi:10.1088/0264-9381/18/16/302}],
% [\href{https://arxiv.org/abs/hep-th/0011239}{\tt arXiv:hep-th/0011239}]. 


\bibitem[DFM07]{DFM07}
D.-E. Diaconescu, D. Freed, and G. Moore, 
{\it \color{darkblue} The M-theory 3-form and $E_8$-gauge theory}, chapter in {\it Elliptic Cohomology Geometry, Applications, and Higher Chromatic Analogues}, Cambridge University Press (2007), 
[\href{https://doi.org/10.1017/CBO9780511721489}{\tt doi:10.1017/CBO9780511721489}],
[\href{https://arxiv.org/abs/hep-th/0312069}{\tt arXiv:hep-th/0312069}].



% \bibitem[DW23]{DonagiWijnholt23}
% R. Donagi and M. Wijnholt, 
% {\it \color{darkblue} The M-Theory Three-Form and Singular Geometries},
% [\href{https://arxiv.org/abs/2310.05838}{\tt arXiv:2310.05838}].


% \bibitem[DGP13]{DGP13}
% A. Donos, J. P. Gauntlett, and C. Pantelidou, 
% {\it \color{darkblue} Semi-local quantum criticality in string/M-theory}, 
% J. High Energ. Phys. {\bf 2013}  (2013) 103,
% [\href{https://doi.org/10.1007/JHEP03(2013)103}{\tt doi:10.1007/JHEP03(2013)103}],
% [\href{https://arxiv.org/abs/1212.1462}{\tt arXiv:1212.1462}].



\bibitem[DK00]{DuistermaatKolk00}
H. Duistermaat and J. A. C. Kolk, 
{\it \color{darkblue}Lie groups}, 
Springer (2000), 
[\href{https://doi.org/10.1007/978-3-642-56936-4}{\tt doi:10.1007/978-3-642-56936-4}].

\bibitem[Du96]{Duff96}
M. Duff, 
{\it \color{darkblue} M-Theory (the Theory Formerly Known as Strings)}, Int. J. Mod. Phys. A {\bf 11} (1996), 5623-5642,  
[\href{https://doi.org/10.1142/S0217751X96002583}{\tt doi:10.1142/S0217751X96002583}],
[\href{https://arxiv.org/abs/hep-th/9608117}{\tt arXiv:hep-th/9608117}].

\bibitem[Du99a]{Duff99-MTheory}
M. Duff, 
{\it \color{darkblue} The World in Eleven Dimensions: Supergravity, Supermembranes and M-theory}, IoP Publishing (1999), 
[\href{https://www.crcpress.com/The-World-in-Eleven-Dimensions-Supergravity-supermembranes-and-M-theory/Duff/9780750306720}{\tt ISBN:9780750306720}].

\bibitem[Du20]{Duff20}
M. Duff,
{\it \color{darkblue} Perspectives on M-Theory},
interview \& opening remarks at {\it M-Theory and Mathematics 2020}, NYU Abu Dhabi (2020),
[\href{https://ncatlab.org/nlab/show/Perspectives+on+M-Theory}{\tt nlab/show/Perspectives+on+M-Theory}].

\bibitem[DNP86]{DNP86}
M. Duff, B. Nilsson, and C. Pope, {\it \color{darkblue}Kaluza-Klein supergravity}, Phys. Rep.
{\bf 130} (1986), 1-142, \newline 
[\href{https://doi.org/10.1016/0370-1573(86)90163-8}{\tt doi:10.1016/0370-1573(86)90163-8}].

\bibitem[Fa99]{Falkowski99}
A. Falkowski, 
{\it \color{darkblue}Five dimensional locally supersymmetric theories with branes}, MSc thesis, Warsaw (1999),
[\href{https://ncatlab.org/nlab/files/FalkowskiLecture.pdf}{\tt ncatlab.org/nlab/files/FalkowskiLecture.pdf}].



\bibitem[FIdO15]{FIdO15}
J. J. Fernandez, J. M. Izquierdo, and M. A. del Olmo, 
{\it \color{darkblue}Contractions from $osp(1\vert32)\oplus osp(1\vert32)$ to the M-theory superalgebra extended by additional fermionic generators}, 
Nucl. Phys. B {\bf 897} (2015), 87-97, \newline 
[\href{https://doi.org/10.1016/j.nuclphysb.2015.05.018}{\tt doi:10.1016/j.nuclphysb.2015.05.018}],
[\href{https://arxiv.org/abs/1504.05946}{\tt arXiv:1504.05946}]. 

% \bibitem[FGSS20]{FGSS20}
% A. Ferraz, K. S. Gupta, G. W. Semenoff, and P. Sodano (eds.), 
% {\it \color{darkblue} Strongly Coupled Field Theories for Condensed Matter and Quantum Information Theory}, 
% Proceedings in Physics {\bf 239}, Springer (2020), \newline 
% [\href{https://doi.org/10.1007/978-3-030-35473-2}{\tt doi:10.1007/978-3-030-35473-2}].


% \bibitem[FSS15]{FSS15-HigherWZW}
% D. Fiorenza, H. Sati, and U. Schreiber, 
% {\it \color{darkblue} Super Lie $n$-algebra extensions, higher WZW models and super p-branes with tensor multiplet fields}, 
% Int. J. Geom. Methods Mod. Phys. {\bf 12} (2015) 02,
% [\href{https://arxiv.org/abs/1308.5264}{\tt arXiv:1308.5264}], \newline 
% [\href{http://www.worldscientific.com/doi/abs/10.1142/S0219887815500188}{\tt doi:10.1142/S0219887815500188}].

% \bibitem[FSS17]{FSS17}
% D. Fiorenza, H. Sati, and U. Schreiber, {\it \color{darkblue} Rational sphere valued supercocycles in M-theory and type IIA string theory}, 
% J. Geom. Phys.
% {\bf 114} (2017), 91-108, 
% [\href{http://dx.doi.org/10.1016/j.geomphys.2016.11.024}{\tt doi:10.1016/j.geomphys.2016.11.024}],
% [\href{https://arxiv.org/abs/1606.03206}{\tt arXiv:1606.03206}].

% \bibitem[FSS18]{FSS18-TDuality}
% D. Fiorenza, H. Sati, and U. Schreiber, 
% {\it \color{darkblue} T-Duality from super Lie $n$-algebra cocycles for super $p$-branes}, 
% Adv. Theor. Math. Phys. {\bf 22}  (2018), 1209-1270,  
% [\href{https://dx.doi.org/10.4310/ATMP.2018.v22.n5.a3}{\tt doi:10.4310/ATMP.2018.v22.n5.a3}],
% [\href{https://arxiv.org/abs/1611.06536}{\tt arXiv:1611.06536}].


% \bibitem[FSS19]{FSS19-QHigher}
% D. Fiorenza, H. Sati, and U. Schreiber, {\it \color{darkblue} The rational higher structure of M-theory}, in: 
% {\it Proceedings of the LMS-EPSRC Durham Symposium: Higher Structures in M-Theory 2018}, Fortsch. Phys. {\bf 67}  (2019) 1910017,
%  [\href{https://doi.org/10.1002/prop.201910017}{\tt doi:10.1002/prop.201910017}],
% [\href{https://arxiv.org/abs/1903.02834}{\tt arXiv:1903.02834}].

\bibitem[FSS20a]{FSS20-HigherT}
D. Fiorenza, H. Sati, and U. Schreiber,
{\it \color{darkblue}Higher T-duality of super M-branes},
Adv. Theor.  Math. Phys.
{\bf 24} 3 (2020), 621-708,
[\href{https://dx.doi.org/10.4310/ATMP.2020.v24.n3.a3}{\tt doi:10.4310/ATMP.2020.v24.n3.a3}],
[\href{https://arxiv.org/abs/1803.05634}{\tt arXiv:1803.05634}].

\bibitem[FSS20b]{FSS20-H}
D. Fiorenza, H. Sati, and  U. Schreiber, {\it \color{darkblue} Twisted Cohomotopy implies M-theory anomaly cancellation on 8-manifolds}, 
Commun. Math. Phys. {\bf 377} (2020), 1961-2025, 
[\href{https://doi.org/10.1007/s00220-020-03707-2}{\tt doi:10.1007/s00220-020-03707-2}], \newline 
[\href{https://arxiv.org/abs/1904.10207}{\tt arXiv:1904.10207}].


\bibitem[FSS20c]{FSS20Exc}
D. Fiorenza, H. Sati, and U. Schreiber, 
{\it \color{darkblue} Super-exceptional geometry: Super-exceptional embedding construction of M5},
J. High Energy Phys. {\bf 2020}  (2020) 107, 
[\href{https://doi.org/10.1007/JHEP02(2020)107}{\tt doi:10.1007/JHEP02(2020)107}], \newline 
[\href{https://arxiv.org/abs/1908.00042}{\tt arXiv:1908.00042}].

\bibitem[FSS21a]{FSS21Exc}
D. Fiorenza, H. Sati, and U. Schreiber,
{\it \color{darkblue} Super-exceptional M5-brane model --
Emergence of $\mathrm{SU}(2)$-flavor sector},
J. Geom. Phys. {\bf 170} (2021) 104349,
[\href{https://doi.org/10.1016/j.geomphys.2021.104349}{\tt doi:10.1016/j.geomphys.2021.104349}],
[\href{https://arxiv.org/abs/2006.00012}{\tt arXiv:2006.00012}].

\bibitem[FSS21b]{FSS21Hopf}
D. Fiorenza, H. Sati, and U. Schreiber,
{\it \color{darkblue} Twisted Cohomotopy implies M5 WZ term level quantization},
Commun. Math. Phys. {\bf 384} (2021), 403-432,
[\href{https://doi.org/10.1007/s00220-021-03951-0}{\tt doi:10.1007/s00220-021-03951-0}],
[\href{https://arxiv.org/abs/1906.07417}{\tt arXiv:1906.07417}].

% \bibitem[FSS22]{FSS22-GS}
% D. Fiorenza, H. Sati, and U. Schreiber, 
% {\it \color{darkblue} Twistorial Cohomotopy implies Green-Schwarz anomaly cancellation},
% Rev. Math. Phys. {\bf 34} 05 (2022) 2250013,
% [\href{https://doi.org/10.1142/S0129055X22500131}{\tt doi:10.1142/S0129055X22500131}],
% [\href{https://arxiv.org/abs/2008.08544}{\tt arXiv:2008.08544}].

% \bibitem[FSS23]{FSS23-Char}
% D. Fiorenza, H. Sati, and U. Schreiber,
% {\it \color{darkblue} The Character map in Nonabelian Cohomology --- Twisted, Differential and Generalized},
% World Scientific, Singapore (2023),
% [\href{https://doi.org/10.1142/13422}{\tt doi:10.1142/13422}],
% [\href{https://arxiv.org/abs/2009.11909}{\tt arXiv:2009.11909}].

\bibitem[FrSS00]{FrappatEtAl00}
L. Frappat, A. Sciarrino, and P. Sorba, 
{\it \color{darkblue}Dictionary on Lie Superalgebras}, 
Academic Press (2000), 
\newline 
[{\tt ISBN:978-0122653407}], 
[\href{https://arxiv.org/abs/hep-th/9607161}{\tt arXiv:hep-th/9607161}].

\bibitem[Fr99]{Freed99}
D. Freed, 
{\it \color{darkblue}Five lectures on supersymmetry},
American Mathematical Society (1999), [\href{https://inspirehep.net/literature/517862}{\tt spire:517862}], \newline 
[\href{https://bookstore.ams.org/FLS}{\tt ISBN:978-0-8218-1953-1}].

\bibitem[Fr02]{Freed02}
D. Freed,
{\it \color{darkblue}Dirac charge quantization and generalized differential cohomology}, Surveys in Differential Geometry {\bf 7} (2002), 129-194,
[\href{https://dx.doi.org/10.4310/SDG.2002.v7.n1.a6}{\tt doi:10.4310/SDG.2002.v7.n1.a6}],
[\href{https://arxiv.org/abs/hep-th/0011220}{\tt arXiv:hep-th/0011220}].

\bibitem[FH21]{FreedHopkins21}
D. S. Freed and M. J. Hopkins, 
{\it \color{darkblue} Reflection positivity and invertible topological phases}, Geom. Topol. {\bf 25} (2021), 1165-1330, 
[\href{https://doi.org/10.2140/gt.2021.25.1165}{\tt doi:10.2140/gt.2021.25.1165}],
[\href{https://arxiv.org/abs/1604.06527}{\tt arXiv:1604.06527}].

\bibitem[FvP12]{FreedmanVOnProeyen12}
D. Freedman and A. Van Proeyen 
{\it \color{darkblue} Supergravity}, 
Cambridge University Press (2012), \newline 
[\href{https://doi.org/10.1017/CBO9781139026833}{\tt doi:10.1017/CBO9781139026833}].

% \bibitem[FS10]{FrishmanSonnenschein10}
% Y. Frishman and J. Sonnenschein, 
% {\it \color{darkblue}  Non-Perturbative Field Theory -- From Two Dimensional Conformal Field Theory to QCD in Four Dimensions}, 
% Cambridge University Press (2010), 
% [\href{https://doi.org/10.1017/9781009401654}{\tt doi:10.1017/9781009401654}],
% [\href{https://arxiv.org/abs/1004.4859}{\tt arXiv:1004.4859}].


% \bibitem[Fu12]{Fulde12}
% P. Fulde, 
% {\it \color{darkblue} Correlated Electrons in Quantum Matter}, World Scientific (2012), 
% [\href{https://doi.org/10.1142/8419}{\tt doi:10.1142/8419}].


% \bibitem[GJF19]{GJF19}
% D. Gaiotto and T. Johnson-Freyd, {\it \color{darkblue} Condensations in higher categories},
% [\href{https://arxiv.org/abs/1905.09566}{\tt arXiv:1905.09566}].


% \bibitem[GSW10a]{GSW10a}
% J. P. Gauntlett, J. Sonner, and T. Wiseman, 
% {\it \color{darkblue}Holographic superconductivity in M-Theory}, 
% Phys. Rev. Lett. {\bf 103} (2009) 151601,
% [\href{https://doi.org/10.1103/PhysRevLett.103.151601}{\tt 
% doi:10.1103/PhysRevLett.103.151601}], 
% [\href{https://arxiv.org/abs/0907.3796}{\tt arXiv:0907.3796}].

% \bibitem[GSW10b]{GSW10b}
% J. P. Gauntlett, J. Sonner, and T. Wiseman, 
% {\it \color{darkblue} Quantum Criticality and Holographic Superconductors in M-theory}, 
% J. High Energ. Phys. {\bf 2010} (2010) 60, 
% [\href{https://doi.org/10.1007/JHEP02(2010)060}{\tt doi:10.1007/JHEP02(2010)060}],
% [\href{https://arxiv.org/abs/0912.0512}{\tt arXiv:0912.0512}]. 





\bibitem[GSS24a]{GSS24-SuGra}
G. Giotopoulos, H. Sati, and U. Schreiber,
{\it \color{darkblue} Flux Quantization on 11d Superspace},
J. High Energy Phys. {\bf 2024}  (2024) 82,
[\href{https://doi.org/10.1007/JHEP07(2024)082}{\tt doi:10.1007/JHEP07(2024)082}],
[\href{https://arxiv.org/abs/2403.16456}{\tt arXiv:2403.16456}].

\bibitem[GSS24b]{GSS24-FluxOnM5}
G. Giotopoulos, H. Sati, and U. Schreiber, 
{\it \color{darkblue}Flux-Quantization on M5-Branes}, J. High Energy Phys. {\bf 2024} 140 (2024),
[\href{https://doi.org/10.1007/JHEP10(2024)140}{\tt doi:10.1007/JHEP10(2024)140}],
[\href{https://arxiv.org/abs/2406.11304}{\tt arXiv:2406.11304}].


\bibitem[GSS24c]{GSS24-M5Embedding}
G. Giotopoulos, H. Sati, and U. Schreiber, 
{\it \color{darkblue} Holographic M-Brane Super-Embeddings},
J. High Energy Physics (2024, in print),
[\href{https://arxiv.org/abs/2408.09921}{\tt arXiv:2408.09921}].

\bibitem[GSS24d]{GSS24-E11}
G. Giotopoulos, H. Sati, and U. Schreiber, 
{\it \color{darkblue}The M-Algebra completes the hierarchy of Super-Exceptional Tangent Spaces},
[\href{https://arxiv.org/abs/2411.03661}{\tt arXiv:2411.03661}].

\bibitem[GSS24e]{GSS25-TDuality}
G. Giotopoulos, H. Sati, and U. Schreiber, 
{\it \color{darkblue} Super-$\mathrm{Lie}_{\infty}$ T-Duality and M-Theory},
[\href{https://arxiv.org/abs/2411.10260}{\tt 	arXiv:2411.10260}]. 


\bibitem[GSS25a]{GSS24-SuperGeometry}
G. Giotopoulos, H. Sati, and U. Schreiber, 
{\it \color{darkblue} Field Theory via Higher Geometry II: Super-sets of fermionic fields} (in preparation).


\bibitem[GSS25b]{GSS25-AnyonsOnExceptional}
G. Giotopoulos, H. Sati, and U. Schreiber, 
{\it \color{darkblue} Abelian Anyons on the Super-Exceptional M5-Brane}
(in preparation).

% \bibitem[GKP19]{GomisKleinschmidtPalmkvist19}
% J. Gomis, A. Kleinschmidt, and J. Palmkvist,
% {\it \color{darkblue}Symmetries of M-theory and free Lie superalgebras}, 
% J. High Energ. Phys. {\bf 2019}  (2019) 160,
% [\href{https://doi.org/10.1007/JHEP03(2019)160}{\tt doi:10.1007/JHEP03(2019)160}],
% [\href{https://arxiv.org/abs/1809.09171}{\tt arXiv:1809.09171}].



\bibitem[GS22]{GS22}
 D. Grady and H. Sati, 
 {\it \color{darkblue} Ramond–Ramond fields and twisted differential K-theory}, 
 Adv. Theor. Math. Phys. {\bf 26} (2022), 1097–1155, 
 [\href{https://dx.doi.org/10.4310/ATMP.2022.v26.n5.a2}{\tt doi:10.4310/ATMP.2022.v26.n5.a2}],
[\href{https://arxiv.org/abs/1903.08843}{\tt 
 arXiv:1903.08843}].

\bibitem[Gr60]{Grothendieck60}
A. Grothendieck, 
{\it \color{darkblue}Technique de descente et th{\'e}or{\`e}mes d’existence en g{\'e}om{\'e}trie alg{\'e}briques. II: Le th{\'e}or{\`e}me d’existence en th{\'e}orie formelle des modules}, S{\'e}minaire Bourbaki {\bf 195} (1960), 
[\href{http://www.numdam.org/item/SB_1958-1960__5__369_0}{\tt doi:SB\_1958-1960\_\_5\_\_369\_0}].

\bibitem[Gr61]{Grothendieck61}
A. Grothendieck, 
{\it \color{darkblue}Techniques de construction et th{\'e}or{\`e}mes d’existence en g{\'e}om{\'e}trie alg{\'e}brique III: pr{\'e}sch{\'e}mas quotients}, S{\'e}minaire Bourbaki {\bf 212} (1961),
[\href{http://www.numdam.org/item/?id=SB_1960-1961__6__99_0}{\tt numdam:SB\_1960-1961\_\_6\_\_99\_0}].

% \bibitem[GPR10]{GPR10}
% S. S. Gubser, S. S. Pufu, and F. D. Rocha,
% {\it \color{darkblue}Quantum critical superconductors in string theory and M-theory}, Phys. Lett. B {\bf 683} (2010), 201-204, 
% [\href{https://doi.org/10.1016/j.physletb.2009.12.017}{\tt doi:10.1016/j.physletb.2009.12.017}],
% [\href{https://arxiv.org/abs/0908.0011}{\tt arXiv:0908.0011}]. 


% \bibitem[HLS18]{HLS18}
% S. Hartnoll, A. Lucas, and S. Sachdev, 
% {\it \color{darkblue} Holographic quantum matter}, MIT Press (2018),
% \newline 
% [\href{https://mitpress.ublish.com/book/holographic-quantum-matter}{\tt ISBN:9780262348010}],
% [\href{https://arxiv.org/abs/1612.07324}{\tt arXiv:1612.07324}].

% \bibitem[HKW91]{HKW91}
% Y. Hatsugai, M. Kohmoto, and Y.-S. Wu,
% {\it \color{darkblue} Anyons on a torus: Braid group, Aharonov-Bohm period, and numerical study}, 
% Phys. Rev. B {\bf 43} (1991) 10761, 
% [\href{https://doi.org/10.1103/PhysRevB.43.10761}{\tt doi:10.1103/PhysRevB.43.10761}].

\bibitem[He01]{Helgason01}
S. Helgason, 
{\it \color{darkblue} Differential geometry, Lie groups and symmetric spaces},
Amer. Math. Soc. (2001), \newline 
[\href{https://bookstore.ams.org/gsm-34}{\tt ams:gsm-34}].

% \bibitem[HKSS07]{HKSS07}
% C. P. Herzog, P. Kovtun, S. Sachdev, and  D. T. Son, {\it  \color{darkblue}
% Quantum critical transport, duality, and M-theory}, Phys. Rev. D {\bf 75} (2007) 085020, 
% [\href{https://doi.org/10.1103/PhysRevD.75.085020}{\tt doi:10.1103/PhysRevD.75.085020}],
% [\href{https://arxiv.org/abs/hep-th/0701036}{\tt arXiv:hep-th/0701036}].

\bibitem[HLZ13]{HLZ13}
O. Hohm, D. L{\"u}st, and B. Zwiebach, 
{\it \color{darkblue} The Spacetime of Double Field Theory: Review, Remarks, and Outlook}, Fortschr. Phys.
{\bf 61} 10 (2013), 926-966,
[\href{https://doi.org/10.1002/prop.201300024}{\tt doi:10.1002/prop.201300024}],
[\href{https://arxiv.org/abs/1309.2977}{\tt arXiv:1309.2977}].

% \bibitem[HS13a]{HohmSamtleben13a}
% O. Hohm and H. Samtleben,
% {\it \color{darkblue}U-duality covariant gravity}, 
% J. High Energ. Phys. {\bf 2013}  (2013) 80, \newline 
% [\href{https://doi.org/10.1007/JHEP09(2013)080}{\tt doi:10.1007/JHEP09(2013)080}],
% [\href{https://arxiv.org/abs/1307.0509}{\tt arXiv:1307.0509}].

% \bibitem[HS13]{HohmSamtleben13b}
% O. Hohm and H. Samtleben,
% {\it  \color{darkblue} Exceptional Form of $D=11$ Supergravity}, 
% Phys. Rev. Lett. {\bf 111} (2013)  231601,
% [\href{https://doi.org/10.1103/PhysRevLett.111.231601}{\tt doi:10.1103/PhysRevLett.111.231601}],
% [\href{https://arxiv.org/abs/1308.1673}{\tt arXiv:1308.1673}].


% \bibitem[HS14a]{HohmSamtleben14a}
% O. Hohm and H. Samtleben,
% {\it  \color{darkblue} Exceptional Field Theory I: $E_{6(6)}$ covariant Form of M-Theory and Type IIB}, 
% Phys. Rev. D {\bf 89} (2014) 066016,
% [\href{https://doi.org/10.1103/PhysRevD.89.066016}{\tt doi:10.1103/PhysRevD.89.066016}],
% [\href{https://arxiv.org/abs/1312.0614}{\tt arXiv:1312.0614}].


% \bibitem[HS14b]{HohmSamtleben14b}
% O. Hohm and H. Samtleben, {\it  \color{darkblue} Exceptional Field Theory II: $E_{7(7)}$}
% Phys. Rev. D {\bf 89} (2014) 066017, \newline 
% [\href{https://doi.org/10.1103/PhysRevD.89.066017}{\tt doi:10.1103/PhysRevD.89.066017}], 
% [\href{https://arxiv.org/abs/1312.4542}{\tt arXiv:1312.4542}].


% \bibitem[HS14c]{HohmSamtleben14c}
% O. Hohm and H. Samtleben,
% {\it \color{darkblue} Exceptional Field Theory III: $E_{8}$}, Phys. Rev. D {\bf 90} (2014) 066002, \newline 
% [\href{https://doi.org/10.1103/PhysRevD.90.066002}{\tt doi:10.1103/PhysRevD.90.066002}],
% [\href{https://arxiv.org/abs/1406.3348}{\tt arXiv:1406.3348}].



\bibitem[HKST11]{HKST11}
H. Hohnhold, M. Kreck, S. Stolz, and P. Teichner,
{\it \color{darkblue} Differential forms and 0-dimensional supersymmetric field theories}, Quantum Topology {\bf 2} 1 (2011),
1–41, 
[\href{http://dx.doi.org/10.4171/QT/12}{\tt doi:10.4171/QT/12}].


% \bibitem[HW96a]{HoravaWitten96a}
% P. Ho{\v r}ava and E. Witten, 
% {\it \color{darkblue} Heterotic and Type I string dynamics from eleven dimensions}, 
% Nucl. Phys. B {\bf 460} (1996) 506-524, 
% [\href{https://doi.org/10.1016/0550-3213(95)00621-4}{\tt doi:10.1016/0550-3213(95)00621-4}],
% [\href{https://arxiv.org/abs/hep-th/9510209}{\tt arXiv:hep-th/9510209}].

% \bibitem[HW96b]{HoravaWitten96b}
% P. Ho{\v r}ava and E. Witten, 
% {\it \color{darkblue} Eleven dimensional supergravity on a manifold with boundary}, 
% Nucl. Phys. B {\bf 475} (1996), 94-114,
% [\href{https://doi.org/10.1016/0550-3213(96)00308-2}{\tt doi:10.1016/0550-3213(96)00308-2}],
% [\href{https://arxiv.org/abs/hep-th/9603142}{\tt arXiv:hep-th/9603142}].

% \bibitem[HH92a]{HH92a}
% Y. Hosotani and C.-L. Ho, 
% {\it \color{darkblue} Anyons on a Torus}, 
% AIP Conf. Proc. {\bf 272} (1992), 1466–1469, \newline 
% [\href{https://doi.org/10.1063/1.43444}{\tt doi:10.1063/1.43444}],
% [\href{https://arxiv.org/abs/hep-th/9210112}{\tt arXiv:hep-th/9210112}].

% \bibitem[HH92b]{HH92b}
% Y. Hosotani and C.-L. Ho,
% {\it \color{darkblue} Anyon equation on a torus}, Int. J. Mod. Phys. A {\bf 07} 23 (1992),  5797-5831, \newline 
% 5797-5831, [\href{https://doi.org/10.1142/S0217751X92002647}{\tt doi:10.1142/S0217751X92002647}].

\bibitem[Ho97]{Howe97}
P. Howe, 
{\it \color{darkblue} Weyl Superspace}, Phys. Lett. B {\bf 415} 2 (1997), 149-155,
[\href{https://doi.org/10.1016/S0370-2693(97)01261-6}{\tt doi:10.1016/S0370-2693(97)01261-6}],
[\href{https://arxiv.org/abs/hep-th/9707184}{\tt arXiv:hep-th/9707184}].


\bibitem[Ho82]{Howe82}
P. Howe, 
{\it \color{darkblue}Supergravity in superspace}, 
Nucl. Phys. B {\bf 199} 2 (1982), 309-364, \newline
[\href{https://doi.org/10.1016/0550-3213(82)90349-2}{\tt doi:10.1016/0550-3213(82)90349-2}].


\bibitem[HoS97]{HoweSezgin97b}
P. Howe and E. Sezgin, 
{\it \color{darkblue} $D=11$, $p=5$}, Phys. Lett. B {\bf 394} (1997), 62-66, [\href{https://arxiv.org/abs/hep-th/9611008}{\tt arXiv:hep-th/9611008}],
\newline 
[\href{https://doi.org/10.1016/S0370-2693(96)01672-3}{\tt doi:10.1016/S0370-2693(96)01672-3}].


\bibitem[HuS18]{HS18}
J. Huerta and U. Schreiber,
{\it \color{darkblue} M-theory from the Superpoint}, Lett. Math. Phys. {\bf 108} (2018), 2695–2727,\newline 
[\href{https://doi.org/10.1007/s11005-018-1110-z}{\tt doi:10.1007/s11005-018-1110-z}],
[\href{https://arxiv.org/abs/1702.01774}{\tt arXiv:1702.01774}].

\bibitem[Hu98]{Hull98}
C. M. Hull, 
{\it \color{darkblue}Gravitational Duality, Branes and Charges}, 
Nucl. Phys. B {\bf 509} (1998), 216-251, 
\newline 
[\href{https://doi.org/10.1016/S0550-3213(97)00501-4}{\tt doi:10.1016/S0550-3213(97)00501-4}],
[\href{https://arxiv.org/abs/hep-th/9705162}{\tt arXiv:hep-th/9705162}].

% \bibitem[Hull07]{Hull07}
% C. M. Hull, 
% {\it \color{darkblue} Generalised Geometry for M-Theory}, 
% J. High Energy Phys. {\bf 0707} 
% (2007) 079,
% \newline
% [\href{https://doi.org/10.1088/1126-6708/2007/07/079}{\tt doi:10.1088/1126-6708/2007/07/079}],
% [\href{https://arxiv.org/abs/hep-th/0701203}{\tt arXiv:hep-th/0701203}].

% \bibitem[HT95]{HullTownsend95}
% C. Hull and P. Townsend,
% {\it \color{darkblue} Unity of Superstring Dualities}, 
% Nucl. Phys. B {\bf 438} (1995), 109-137,  \newline 
% [\href{https://doi.org/10.1016/0550-3213(94)00559-W}{\tt doi:10.1016/0550-3213(94)00559-W}],
% [\href{https://arxiv.org/abs/hep-th/9410167}{\tt arXiv:hep-th/9410167}]. 


\bibitem[MR91]{MoerdijkReyes91}
I. Moerdijk and G. Reyes, 
{\it \color{darkblue}Models for Smooth Infinitesimal Analysis}, Springer, Berlin (1991), \newline 
[\href{https://doi.org/10.1007/978-1-4757-4143-8}{\tt doi:10.1007/978-1-4757-4143-8}].


% \bibitem[IO94]{IO94}
% I. Ichinose and T. Ohbayashi, 
% {\it \color{darkblue}Exactly soluble model of multispecies anyons and the braid group on a torus}, Nucl. Phys. B {\bf 419} (1994),
% 529-552, [\href{https://doi.org/10.1016/0550-3213(94)90343-3}{\tt doi:10.1016/0550-3213(94)90343-3}].

% \bibitem[IL90]{IengoLechner90}
% R. Iengo and K. Lechner 
% {\it \color{darkblue}Quantum mechanics of anyons on a torus}, Nucl. Phys. B {\bf 346} (1990), 551-575,  \newline 
% [\href{https://doi.org/10.1016/0550-3213(90)90292-L}{\tt doi:10.1016/0550-3213(90)90292-L}].

% \bibitem[IL91]{IengoLechner91}
% R. Iengo and K. Lechner,
% {\it \color{darkblue}Exact results for anyons on a torus}, Nucl. Phys. B {\bf 364} (1991), 551-583, \newline  
% [\href{https://doi.org/10.1016/0550-3213(91)90277-5}{\tt doi:10.1016/0550-3213(91)90277-5}].


% \bibitem[Ka14a]{Kapustin14a}
% A. Kapustin,
% {\it \color{darkblue} Symmetry Protected Topological Phases, Anomalies, and Cobordisms: Beyond Group Cohomology}, 
% [\href{https://arxiv.org/abs/1403.1467}{\tt arXiv:1403.1467}].

% \bibitem[Ka14b]{Kapustin14b}
% A. Kapustin,
% {\it \color{darkblue} Bosonic Topological Insulators and Paramagnets: a view from cobordisms}, \newline 
% [\href{https://arxiv.org/abs/1404.6659}{\tt arXiv:1404.6659}].

% \bibitem[KM17]{KeimerMoore17}
% B. Keimer and J. E. Moore, 
% {\it \color{darkblue} The physics of quantum materials}, Nature Phys. {\bf 13} (2017), 1045–1055, \newline  
% [\href{https://doi.org/10.1038/nphys4302}{\tt doi:10.1038/nphys4302}].


% \bibitem[KVW93]{KVW93}
% E. Keski-Vakkuri and X.-G. Wen,
% {\it \color{darkblue} Ground state structure of hierarchical QH states on torus and modular transformation}, 
% Int. J. Mod. Phys. B {\bf 7} (1993), 4227-4259,
% [\href{https://doi.org/10.1142/S0217979293003644}{\tt doi:10.1142/S0217979293003644}], \newline 
% [\href{https://arxiv.org/abs/hep-th/9303155}{\tt arXiv:hep-th/9303155}].

\bibitem[KhS17]{KS17}
I. Khavkine and U. Schreiber, 
{\it \color{darkblue}Synthetic geometry of differential equations: I. Jets and comonad structure}, 
[\href{https://arxiv.org/abs/1701.06238}{\tt arXiv:1701.06238}].

% \bibitem[Ki03]{Kitaev03}
% A. Kitaev,  
% {\it \color{darkblue} Fault-tolerant quantum computation by anyons}, Ann. Phys. {\bf 303} 1 (2003), 2-30, \newline 
% [\href{https://doi.org/10.1016/S0003-4916(02)00018-0}{\tt doi:10.1016/S0003-4916(02)00018-0}],
% [\href{https://arxiv.org/abs/quant-ph/9707021}{\tt arXiv:quant-ph/9707021}].



% \bibitem[KKLN22]{KKLN22}
% A. Kleinschmidt, R. K{\"o}hl, R. Lautenbacher, and H. Nicolai,
% {\it \color{darkblue}Representations of involutory subalgebras of affine Kac-Moody algebras}, 
% Commun. Math. Phys. {\bf 392} (2022), 89–123, 
% [\href{https://doi.org/10.1007/s00220-022-04342-9}{\tt doi:10.1007/s00220-022-04342-9}],
% [\href{https://arxiv.org/abs/2102.00870}{\tt arXiv:2102.00870}]. 

% \bibitem[KN06]{KleinschmidtNicolai06}
% A. Kleinschmidt and H. Nicolai,
% {\it \color{darkblue}IIA and IIB spinors from $K(E_{\ten})$}, Phys. Lett. B {\bf 637} (2006), 107-112, \newline 
% [\href{https://doi.org/10.1016/j.physletb.2006.04.007}{\tt doi:10.1016/j.physletb.2006.04.007}],
% [\href{https://arxiv.org/abs/hep-th/0603205}{\tt arXiv:hep-th/0603205}].


% \bibitem[KN21]{KleinschmidtNicolai21}
% A. Kleinschmidt and H. Nicolai,
% {\it \color{darkblue}Generalised holonomies and $K(E_9)$}, 
% J. High Energy Phys. {\bf 2021}  (2021) 54,
% [\href{https://doi.org/10.1007/JHEP09(2021)054}{\tt doi:10.1007/JHEP09(2021)054}],
% [\href{https://arxiv.org/abs/2107.02445}{\tt arXiv:2107.02445}]. 

% \bibitem[KNP07]{KleinschmidtNicolaiPalmkvist07}
% A. Kleinschmidt, H. Nicolai, and J. Palmkvist,
% {\it \color{darkblue}$K(E_9)$ from $K(E_{\ten})$}, 
% J. High Energy Phys. {\bf 2007}  (2007) 06, 
% [\href{https://doi.org/10.1088/1126-6708/2007/06/051}{\tt doi:10.1088/1126-6708/2007/06/051}],
% [\href{https://arxiv.org/abs/hep-th/0611314}{\tt arXiv:hep-th/0611314}]. 


% \bibitem[KNV20]{KNV20}
% A. Kleinschmidt, H. Nicolai, and A. Vigan{\'o}, 
% {\it \color{darkblue}On spinorial representations of involutory subalgebras of Kac-Moody algebras}, in: 
% {\it Partition Functions and Automorphic Forms}, Moscow Lectures {\bf 5}, Springer (2020),
% [\href{https://doi.org/10.1007/978-3-030-42400-8_4}{\tt doi:10.1007/978-3-030-42400-8\_4}],
% [\href{https://arxiv.org/abs/1811.11659}{\tt arXiv:1811.11659}].



\bibitem[KoS05]{KS05}
D. Kochan and P. {\v S}evera, 
{\it \color{darkblue}Differential gorms, differential worms}, Mathematical Physics, World Scientific (2005), 128-130, 
[\href{https://doi.org/10.1142/9789812701862_0034}{\tt doi:10.1142/9789812701862\_0034}],
[\href{https://arxiv.org/abs/math/0307303}{\tt arXiv:math/0307303}].

% \bibitem[K{\"o}24]{Koenig24}
% B. K{\"o}nig,
% {\it \color{darkblue}$\mathfrak{k}$-structure of basic representation of affine algebras},
% [\href{https://arxiv.org/abs/2407.12748}{\tt arXiv:2407.12748}].

\bibitem[KMS93]{KMS93}
I. Kol{\'a}{\v r}, P. Michor, and J. Slov{\'a}k,
{\it \color{darkblue}Natural operations in differential geometry}, 
Springer, Berlin (1993), 
[\href{https://link.springer.com/book/10.1007/978-3-662-02950-3}{\tt doi:10.1007/978-3-662-02950-3}].


\bibitem[LS22]{LazaroiuShahbazi22}
C. Lazaroiu and C. S. Shahbazi, 
{\it \color{darkblue} The duality covariant geometry and DSZ quantization of abelian gauge theory}, 
Adv. Theor. Math. Phys. {\bf 26} (2022), 2213–2312,
[\href{https://dx.doi.org/10.4310/ATMP.2022.v26.n7.a5}{\tt doi:10.4310/ATMP.2022.v26.n7.a5}], \newline
[\href{https://arxiv.org/abs/2101.07236}{\tt arXiv:2101.07236}].


% \bibitem[LM22]{LiMong22}
% Z. Li and R. S. K. Mong, 
% {\it \color{darkblue} Detecting topological order from modular transformations of ground states on the torus}, 
% Phys. Rev. B {\bf 106} (2022) 235115,
% [\href{https://doi.org/10.1103/PhysRevB.106.235115}{\tt doi:10.1103/PhysRevB.106.235115}],
% [\href{https://arxiv.org/abs/2203.04329}{\tt arXiv:2203.04329}].


% \bibitem[Liu24]{Liu24}
% S. Liu,
% {\it \color{darkblue} Anyon quantum dimensions from an arbitrary ground state wave function}, 
% Nature Commun. {\bf 15} (2024) 5134, 
% [\href{https://doi.org/10.1038/s41467-024-47856-7}{\tt doi:10.1038/s41467-024-47856-7}],
% [\href{https://arxiv.org/abs/2106.15705}{\tt arXiv:2106.15705}].


% \bibitem[LHW16]{LHW16}
% Z.-X. Luo, Y.-T. Hu, and Y.-S. Wu, 
% {\it \color{darkblue} On Quantum Entanglement in Topological Phases on a Torus}, Phys. Rev. B {\bf 94} (2016) 075126,
% [\href{https://doi.org/10.1103/PhysRevB.94.075126}{\tt doi:10.1103/PhysRevB.94.075126}],
% [\href{https://arxiv.org/abs/1603.01777}{\tt arXiv:1603.01777}].


\bibitem[MT12]{MallerTesterman12}
G. Malle and D. Testerman, 
{\it \color{darkblue}Linear Algebraic Groups and Finite Groups of Lie Type}, 
Cambridge University Press (2012), 
[\href{https://doi.org/10.1017/CBO9780511994777}{\tt doi:10.1017/CBO9780511994777}].

\bibitem[MaS04]{MathaiSati}
V. Mathai and H. Sati,
{\it \color{darkblue} Some Relations between Twisted K-theory and $E_8$ Gauge Theory},
J. High Energy Phys. {\bf 0403} (2004), 016,
[\href{https://doi.org/10.1088/1126-6708/2004/03/016}{\tt doi:10.1088/1126-6708/2004/03/016}],
[\href{https://arxiv.org/abs/hep-th/0312033}{\tt  arXiv:hep-th/0312033}].


\bibitem[Me13]{Meinrenken13}
E. Meinrenken, 
{\it \color{darkblue}Clifford algebras and Lie groups}, 
Ergebn. Mathem. \& Grenzgeb., 
Springer (2013), \newline 
[\href{https://doi.org/10.1007/978-3-642-36216-3}{\tt doi:10.1007/978-3-642-36216-3}].

\bibitem[MiS06]{MiemiecSchnakenburg06}
A. Miemiec and I. Schnakenburg, {\it \color{darkblue}Basics of M-Theory}, Fortsch. Phys. {\bf 54} (2006),
5-72, \newline 
[\href{https://doi.org/10.1002/prop.200510256}{\tt doi:10.1002/prop.200510256}],
[\href{https://arxiv.org/abs/hep-th/0509137}{\tt arXiv:hep-th/0509137}].



\bibitem[Mi17]{Milne17}
J. Milne, 
{\it \color{darkblue}Algebraic Groups -- The theory of group schemes of finite type over a field}, 
Cambridge University Press (2017), 
[\href{https://doi.org/10.1017/9781316711736}{\tt doi:10.1017/9781316711736}].

% \bibitem[Mo14]{Moore14}
% G. Moore, 
% {\it \color{darkblue}Physical Mathematics and the Future},
% talk at {\it Strings 2014},
% \newline
% [\url{https://www.physics.rutgers.edu/~gmoore/PhysicalMathematicsAndFuture.pdf}].

% \bibitem[MySS24]{MySS24}
% D. J. Myers, H. Sati, and U. Schreiber, {\it \color{darkblue} Topological Quantum Gates in Homotopy Type Theory}, Commun. Math. Phys. {\bf 405} (2024) 172, 
% [\href{https://doi.org/10.1007/s00220-024-05020-8}{\tt doi:10.1007/s00220-024-05020-8}],
% [\href{https://arxiv.org/abs/2303.02382}{\tt arXiv:2303.02382}].


% \bibitem[Na17]{Nastase17}
% H. Nastase, 
% {\it \color{darkblue} String Theory Methods for Condensed Matter Physics}, Cambridge University Press (2017), \newline  
% [\href{https://doi.org/10.1017/9781316847978}{\tt doi:10.1017/9781316847978}].

% \bibitem[NR78]{NeemanRegge78}
% Y. Ne'eman and T. Regge, 
% {\it \color{darkblue} Gravity and supergravity as gauge theories on a group manifold}, 
% Phys. Lett. B {\bf 74} 1–2 (1978), 54-56, 
% [\href{https://doi.org/10.1016/0370-2693(78)90058-8}{\tt doi:10.1016/0370-2693(78)90058-8}].

% \bibitem[Ni87]{Nicolai87}
% H. Nicolai, 
% {\it \color{darkblue}$d=11$ Supergravity with local $\mathrm{SO}(16)$ invariance}, 
% Phys. Lett. B {\bf 187} (1987), 316-320, \newline 
% [\href{https://doi.org/10.1016/0370-2693(87)91102-6}{\tt doi:10.1016/0370-2693(87)91102-6}].


% \bibitem[Ni99]{Nicolai99}
% H. Nicolai, 
% {\it \color{darkblue}On M-Theory}, 
% J. Astrophys. Astron. {\bf 20} (1999), 149–164, 
% [\href{https://doi.org/10.1007/BF02702349}{\tt doi:10.1007/BF02702349}], \newline 
% [\href{https://arxiv.org/abs/hep-th/9801090}{\tt arXiv:hep-th/9801090}].






% \bibitem[NH98]{NicolaiHelling98}
% H. Nicolai and R. Helling,
% {\it \color{darkblue}Supermembranes and M(atrix) Theory}, in: {\it ICTP Spring School on Nonperturbative Aspects of String Theory and Supersymmetric Gauge Theories} (1998), 29-74, 
% [\href{https://arxiv.org/abs/hep-th/9809103}{\tt arXiv:hep-th/9809103}],
% [\href{https://inspirehep.net/literature/476366}{\tt spire:476366}].



% \bibitem[No23]{Noja23}
% S. Noja, 
% {\it \color{darkblue}On the Geometry of Forms on Supermanifolds}, Differential Geom. Appl. {\bf 88} (2023) 101999, \newline 
% [\href{https://doi.org/10.1016/j.difgeo.2023.101999}{\tt doi:10.1016/j.difgeo.2023.101999}],
% [\href{https://arxiv.org/abs/2111.12841}{\tt arXiv:2111.12841}].

% \bibitem[OP99]{ObersPiloine99}
% N. Obers and B. Pioline,
% {\it \color{darkblue}U-duality and M-Theory}, 
% Phys. Rept. {\bf 318} (1999), 113-225, \newline 
% [\href{https://doi.org/10.1016/S0370-1573(99)00004-6}{\tt doi:10.1016/S0370-1573(99)00004-6}],
% [\href{https://arxiv.org/abs/hep-th/9809039}{\tt arXiv:hep-th/9809039}]. 


\bibitem[Ov04]{Ovrut04}
B. Ovrut, 
{\it \color{darkblue}Lectures on Heterotic M-Theory}, in 
{\it Strings, Branes and Extra Dimensions}, TASI 2001, World Scientific (2004), 359-407, 
[\href{https://doi.org/10.1142/9789812702821_0007}{\tt doi:10.1142/9789812702821\_0007}],
[\href{https://arxiv.org/abs/hep-th/0201032}{\tt arXiv:hep-th/0201032}].



% \bibitem[Pi14]{Pires14}
% A. Pires, 
% {\it \color{darkblue} AdS/CFT correspondence in condensed matter}, Morgan \& Claypool (2014),  
% [\href{https://arxiv.org/abs/1006.5838}{\tt arXiv:1006.5838}],
% [\href{https://doi.org/10.1088/978-1-627-05309-9}{\tt doi:10.1088/978-1-627-05309-9}].


% \bibitem[PJ21]{PuJain21}
% S. Pu and J. K. Jain, 
% {\it \color{darkblue} Composite anyons on a torus}, 
% Phys. Rev. B {\bf 104} (2021) 115135, [\href{https://arxiv.org/abs/2106.15705}{\tt arXiv:2106.15705}], \newline 
% [\href{https://doi.org/10.1103/PhysRevB.104.115135}{\tt doi:10.1103/PhysRevB.104.115135}].

\bibitem[Ra84]{Rajaraman84}
R. Rajaraman,
{\it \color{darkblue}Solitons and Instantons}, 
North Holland (1984), 
[{\tt ISBN:9780444862297}].

\bibitem[Ra21]{Ravera21}
L. Ravera, 
{\it \color{darkblue}On the hidden symmetries of $D=11$ supergravity}, 
in: V. Dobrev (eds), {\it Lie Theory and Its Applications in Physics},
%Proceedings in Mathematics \& Statistics {\bf 396}, 
Springer (2021),
[\href{https://doi.org/10.1007/978-981-19-4751-3_15}{\tt doi:10.1007/978-981-19-4751-3\_15}],
[\href{https://arxiv.org/abs/2112.00445}{\tt arXiv:2112.00445}].



% \bibitem[RS20]{RobertsSchmidt20}
% C. Roberts and S. M. Schmidt, {\it \color{darkblue} Reflections upon the Emergence of Hadronic Mass}, 
% Eur. Phys. J. Special Topics {\bf 229} (2020), 3319–3340, 
% [\href{https://doi.org/10.1140/epjst/e2020-000064-6}{\tt doi:10.1140/epjst/e2020-000064-6}], [\href{https://arxiv.org/abs/2006.08782}{\tt arXiv:2006.08782}].

\bibitem[Ro07]{Rogers07}
A. Rogers, 
{\it  \color{darkblue} Supermanifolds: Theory and Applications}, 
World Scientific (2007), 
[\href{https://doi.org/10.1142/1878}{\tt doi:10.1142/1878}].


\bibitem[Sac08]{Sachse08}
C. Sachse, 
{\it \color{darkblue} A Categorical Formulation of Superalgebra and Supergeometry}, 
[\href{https://arxiv.org/abs/0802.4067}{\tt arXiv:0802.4067}].

\bibitem[Sam23]{Samtleben23}
H. Samtleben, 
{\it \color{darkblue}11D Supergravity and Hidden Symmetries}, in 
{\it Handbook of Quantum Gravity}, Springer (2023), 
[\href{https://doi.org/10.1007/978-981-19-3079-9}{\tt doi:10.1007/978-981-19-3079-9}],
[\href{https://arxiv.org/abs/2303.12682}{\tt arXiv:2303.12682}].


\bibitem[Sa10]{tcu}
H. Sati,
{\it \color{darkblue} Geometric and topological structures related to M-branes},
in: R. Doran, G. Friedman and J. Rosenberg (eds.),
{\it Superstrings, Geometry, Topology, and $C^\ast$-algebras},
Proc. Symp. Pure Math. {\bf 81}, AMS, Providence, 2010, pp. 181-236,
[\href{https://doi.org/10.1090/pspum/081}{\tt doi:10.1090/pspum/081}],
[\href{https://arxiv.org/abs/1001.5020}{\tt arXiv:1001.5020}].

%


\bibitem[SS17]{SS17-BPS}
H. Sati and U. Schreiber, 
{\it \color{darkblue} Lie $n$-algebras of BPS charges}, 
J. High Energ. Phys. {\bf 2017}  (2017) 87, \newline 
[\href{http://link.springer.com/article/10.1007/JHEP03(2017)087}{\tt doi:10.1007/JHEP03(2017)087}],
[\href{https://arxiv.org/abs/1507.08692}{\tt arXiv:1507.08692}].

\bibitem[SS20b]{SS20-Orb}
H. Sati and U. Schreiber,
{\it \color{darkblue} Proper Orbifold Cohomology},
[\href{https://arxiv.org/abs/2008.01101}{\tt arXiv:2008.01101}].


% \bibitem[SS20b]{SS20-TwChar}
% H. Sati and U. Schreiber,
% {\it \color{darkblue} The character map in equivariant twistorial Cohomotopy implies 
% the Green-Schwarz mechanism with heterotic M5-branes},
% [\href{https://arxiv.org/abs/2011.06533}{\tt arXiv:2011.06533}].


% \bibitem[SS23c]{SS23-ToplOrder}
% H. Sati and U. Schreiber, 
% {\it \color{darkblue} Anyonic topological order in TED K-theory}, 
% Rev. Math. Phys. (2023)
% {\bf 35} 03 (2023) 2350001,
% [\href{https://doi.org/10.1142/S0129055X23500010}{\tt doi:10.1142/S0129055X23500010}],
% [\href{https://arxiv.org/abs/2206.13563}{\tt arXiv:2206.13563}].

% \bibitem[SS24a]{SS24-PhaseSpace}
% H. Sati  and U. Schreiber, 
% {\it \color{darkblue} Flux Quantization on Phase Space},
% Ann. Henri Poincar{\' e} (2024), \newline 
% [\href{https://doi.org/10.1007/s00023-024-01438-x}{\tt doi:10.1007/s00023-024-01438-x}],
% [\href{https://arxiv.org/abs/2312.12517}{\tt arXiv:2312.12517}].


% \bibitem[SS24b]{SS24-AbAnyons}
% H. Sati and U. Schreiber,
% {\it \color{darkblue} Abelian Anyons on Flux-Quantized M5-Branes},
% [\href{https://arxiv.org/abs/2408.11896}{\tt arXiv:2408.11896}].

% \bibitem[SS24c]{SS24-ExAnyons}
% H. Sati and U. Schreiber,
% {\it \color{darkblue} The Possibility of Exotic Anyons},
% (in preparation).

\bibitem[SS25]{SS24-Flux}
H. Sati and U. Schreiber,
{\it \color{darkblue} Flux quantization},
Encyclopedia of Mathematical Physics, 2nd ed. 
{\bf 4} (2025), 281-324,
[\href{doi:10.1016/B978-0-323-95703-8.00078-1}{\tt doi:10.1016/B978-0-323-95703-8.00078-1}],
[\href{https://arxiv.org/abs/2402.18473}{\tt arXiv:2402.18473}].


% \bibitem[Sau17]{Sau17}
% J. Sau, 
% {\it \color{darkblue} A Roadmap for a Scalable Topological Quantum Computer}, Physics {\bf 10}  (2017) 68, 
% \newline
% [\href{https://physics.aps.org/articles/v10/68}{\tt physics.aps.org/articles/v10/68}].


% \bibitem[Schur1891]{Schur1891}
% F. Schur, 
% {\it \color{darkblue} Zur Theorie der endlichen Transformationsgruppen}, Math. Ann. {\bf 38} (1891), 263–286,
% \newline 
% [\href{https://doi.org/10.1007/BF01199254}{\tt doi:10.1007/BF01199254}].

\bibitem[Ser64]{Serre64}
J.-P. Serre, 
{\it \color{darkblue} Lie Algebras and Lie Groups -- 1964 Lectures given at Harvard University}, 
Lecture Notes in Mathematics {\bf 1500}, Springer (1992), 
[\href{https://doi.org/10.1007/978-3-540-70634-2}{\tt doi:10.1007/978-3-540-70634-2}].

\bibitem[Se97]{Sezgin97}
E. Sezgin, 
{\it \color{darkblue}The M-Algebra}, Phys. Lett. B {\bf 392} (1997), 323-331, 
[\href{https://doi.org/10.1016/S0370-2693(96)01576-6}{\tt doi:10.1016/S0370-2693(96)01576-6}], 
\newline 
[\href{https://arxiv.org/abs/hep-th/9609086}{\tt arXiv:hep-th/9609086}].

% \bibitem[Sh97]{Sharpe97}
% R. W. Sharpe, 
% {\it \color{darkblue}Differential geometry -- Cartan’s generalization of Klein’s Erlagen program}, 
% Graduate Texts in Mathematics {\bf 166}, Springer (1997), 
% [\href{https://link.springer.com/book/9780387947327}{\tt ISBN:9780387947327}].

\bibitem[So00]{Sorokin00}
D. Sorokin, 
{\it \color{darkblue} Superbranes and Superembeddings}, Phys. Rept. {\bf 329} (2000), 1-101, 
[\href{https://arxiv.org/abs/hep-th/9906142}{\tt arXiv:hep-th/9906142}], \newline
[\href{https://doi.org/10.1016/S0370-1573(99)00104-0}{\tt doi:10.1016/S0370-1573(99)00104-0}].


\bibitem[Str13]{Strocchi13}
F. Strocchi, 
{\it \color{darkblue} An Introduction to Non-Perturbative Foundations of Quantum Field Theory}, 
Oxford University Press (2013),
[\href{https://doi.org/10.1093/acprof:oso/9780199671571.001.0001}{\tt doi:10.1093/acprof:oso/9780199671571.001.0001}].


\bibitem[To95]{Townsend95}
P. Townsend, 
{\it \color{darkblue}$p$-Brane Democracy}, in Duff (ed.), {\it The World in Eleven Dimensions}, IoP (1999), 375-389,\newline 
[\href{https://www.crcpress.com/The-World-in-Eleven-Dimensions-Supergravity-supermembranes-and-M-theory/Duff/9780750306720}{\tt ISBN:9780750306720}],
[\href{https://arxiv.org/abs/hep-th/9507048}{\tt arXiv:hep-th/9507048}].

\bibitem[To98]{Townsend98}
P. Townsend, 
{\it \color{darkblue}M(embrane) theory on $T^9$}, 
Nucl. Phys. Proc. Suppl. {\bf 68} (1998), 11-16, \newline 
[\href{https://doi.org/10.1016/S0920-5632(98)00136-4}{\tt doi:10.1016/S0920-5632(98)00136-4}],
[\href{https://arxiv.org/abs/hep-th/9708034}{\tt arXiv:hep-th/9708034}]. 

\bibitem[To99]{Townsend99}
P. Townsend, 
{\it \color{darkblue} M-theory from its superalgebra}, in {\it Strings, Branes and Dualities}, NATO ASI Series {\bf 520}, Springer (1999), [\href{https://doi.org/10.1007/978-94-011-4730-9_5}{\tt doi:10.1007/978-94-011-4730-9\_5}], 
[\href{https://arxiv.org/abs/hep-th/9712004}{\tt arXiv:hep-th/9712004}].



% \bibitem[Ts04]{Tsimpis04}
% D. Tsimpis, 
% {\it \color{darkblue} Curved 11D supergeometry}, J. High Energy Phys. {\bf 11} (2004) 087,
% [\href{https://arxiv.org/abs/hep-th/0407244}{\tt arXiv:hep-th/0407244}],
% \newline 
% [\href{https://iopscience.iop.org/article/10.1088/1126-6708/2004/11/087}{\tt doi:1088/1126-6708/2004/11/087}].


\bibitem[Va04]{Varadarajan04}
V. Varadarajan, {\it \color{darkblue} Supersymmetry for mathematicians: An introduction}, Courant Lecture Notes in Mathematics {\bf 11}, American Mathematical Society (2004), 
[\href{http://dx.doi.org/10.1090/cln/011}{\tt doi:10.1090/cln/011}].

\bibitem[Var06]{Varela06}
O. Varela,
{\it \color{darkblue}Symmetry and holonomy in M Theory}, PhD thesis, Valencia (2006), 
[\href{https://arxiv.org/abs/hep-th/0607088}{\tt arXiv:hep-th/0607088}],
[\href{https://hdl.handle.net/10550/15484}{\tt hdl:10550/15484}].

\bibitem[Va07]{Vaula07}
S. Vaula, 
{\it \color{darkblue}On the underlying $E_{11}$-symmetry of the $D=11$ Free Differential Algebra}, 
J. High Energy Phys. {\bf 0703} (2007) 010,
[\href{https://doi.org/10.1088/1126-6708/2007/03/010}{\tt doi:10.1088/1126-6708/2007/03/010}],
[\href{https://arxiv.org/abs/hep-th/0612130}{\tt arXiv:hep-th/0612130}].


\bibitem[Wa24]{Waldorf24}
K. Waldorf, 
{\it \color{darkblue}Geometric T-duality: Buscher rules in general topology}, 
Ann. Henri Poincar{\'e} {\bf 25} (2024), 1285–1358,
[\href{https://doi.org/10.1007/s00023-023-01295-0}{\tt doi:10.1007/s00023-023-01295-0}],
[\href{https://arxiv.org/abs/2207.11799}{\tt arXiv:2207.11799}]. 


% \bibitem[Wen89]{Wen89}
% X.-G. Wen,
% {\it \color{darkblue}Vacuum Degeneracy of Chiral Spin State in Compactified Spaces}, 
% Phys. Rev. B {\bf 40}  (1989), \newline 7387--7390,
% [\href{https://doi.org/10.1103/PhysRevB.40.7387}{\tt doi:10.1103/PhysRevB.40.7387}].

% \bibitem[Wen90]{Wen90}
% X.-G. Wen, 
% {\it \color{darkblue}Topological Orders in Rigid States}, 
% Int. J. Mod. Phys. B {\bf 4} (1990), 239-271, \newline 
% [\href{https://doi.org/10.1142/S0217979290000139}{\tt doi:10.1142/S0217979290000139}].

% \bibitem[WDF90]{WDF90}
% X.-G. Wen, E. Dagotto, and E. Fradkin, {\it \color{darkblue}Anyons on a torus}, Phys. Rev. B {\bf 42} (1990), 6110--6123, \newline 
% [\href{https://doi.org/10.1103/PhysRevB.42.6110}{\tt doi:10.1103/PhysRevB.42.6110}].

% \bibitem[WN90]{WenNiu90}
% X.-G. Wen and Q. Niu, 
% {\it \color{darkblue} Ground state degeneracy of the FQH states in presence of random potential and on high genus Riemann surfaces}, Phys. Rev. B {\bf 41} (1990), 9377--9396,
% [\href{https://doi.org/10.1103/PhysRevB.41.9377}{\tt doi:10.1103/PhysRevB.41.9377}].

\bibitem[WZ77]{WessZumino77}
J. Wess and B. Zumino, 
{\it \color{darkblue} Superspace formulation of supergravity}, 
Phys. Lett. B {\bf 66} (1977), 361-364, \newline 
[\href{https://doi.org/10.1016/0370-2693(77)90015-6}{\tt doi:10.1016/0370-2693(77)90015-6}].

% \bibitem[We01]{West01}
% P. West, 
% {\it \color{darkblue}$E_{11}$ and M Theory}, 
% Class. Quant. Grav. {\bf 18} (2001), 4443-4460, [\href{https://arxiv.org/abs/hep-th/0104081}{\tt arXiv:hep-th/0104081}],
% \newline 
% [\href{https://doi.org/10.1088/0264-9381/18/21/305}{\tt doi:10.1088/0264-9381/18/21/305}].

\bibitem[We03]{West03}
P. West,
{\it \color{darkblue} $E_{11}$, $\mathrm{SL}(32)$ and Central Charges}, Phys. Lett. B {\bf 575} (2003), 333-342, 
[\href{https://arxiv.org/abs/hep-th/0307098}{\tt arXiv:hep-th/0307098}],
[\href{https://doi.org/10.1016/j.physletb.2003.09.059}{\tt doi:10.1016/j.physletb.2003.09.059}].




% \bibitem[Wi98]{Witten98}
% E. Witten,
% {\it \color{darkblue}Magic, Mystery, and Matrix}, Notices Amer. Math. Soc. {\bf 45} 9 (1998),
% [\href{https://www.ams.org/notices/199810}{\tt ams.org/notices/199810}].




\bibitem[Ya93]{Yagi93}
K. Yagi, 
{\it \color{darkblue} Super Lie Groups}, Adv. Stud. Pure Math. {\bf 22}, Progress in Differential Geometry (1993), 407-412, [\href{https://projecteuclid.org/euclid.aspm/1534359537}{\tt euclid:1534359537}].

% \bibitem[Yo19]{Yonekura19}
% K. Yonekura, 
% {\it \color{darkblue}On the cobordism classification of symmetry protected topological phases}, Commun. Math. Phys. {\bf 368} 
% (2019), 1121–1173, 
% [\href{https://doi.org/10.1007/s00220-019-03439-y}{\tt doi:10.1007/s00220-019-03439-y}],
% [\href{https://arxiv.org/abs/1803.10796}{\tt arXiv:1803.10796}].

% \bibitem[ZLSS15]{ZLSS15}
% J. Zaanen, Y. Liu, Y.-W. Sun, and K. Schalm, {\it \color{darkblue} Holographic Duality in Condensed Matter Physics}, Cambridge University Press (2015),
% [\href{https://doi.org/10.1017/CBO9781139942492}{\tt doi:10.1017/CBO9781139942492}].

\bibitem[Zee10]{Zee10}
A. Zee, 
{\it \color{darkblue} Quantum Field Theory in a Nutshell}, 2nd ed., 
Princeton University Press (2010), \newline [\href{https://press.princeton.edu/books/hardcover/9780691140346/quantum-field-theory-in-a-nutshell?srsltid=AfmBOoovOHTlGiQ8S-JqFUpn26ihzhfNQocoId3eN9SUZ8WYL83oyKG0}{\tt ISBN:9780691140346}].


% \bibitem[ZCZW19]{ZCZW19}
% B. Zeng, X. Chen, D.-L. Zhou, and X.-G. Wen, 
% {\it \color{darkblue} Quantum Information Meets Quantum Matter -- From Quantum Entanglement to Topological Phases of Many-Body Systems}, Quantum Science and Technology (QST), Springer (2019),
% [\href{https://doi.org/10.1007/978-1-4939-9084-9}{\tt doi:10.1007/978-1-4939-9084-9}],
% [\href{https://arxiv.org/abs/1508.02595}{\tt arXiv:1508.02595}].

\end{thebibliography}
\end{document}